\UseRawInputEncoding
\documentclass[11pt]{article}

\usepackage{amssymb,amsmath,amsthm,amsfonts,dsfont}
\usepackage[numbers,comma,sort&compress]{natbib}
\usepackage[colorlinks,hypertexnames=false]{hyperref}
\hypersetup{
	pdfstartview={FitH},
	pdfnewwindow=true,
	colorlinks=true,
	linkcolor=blue,
	citecolor=blue,
	filecolor=blue,
	urlcolor=blue}
\usepackage[capitalise]{cleveref}
\usepackage[all]{hypcap}
\usepackage{url}
\usepackage{graphicx}
\usepackage[font=small]{caption}
\usepackage[subrefformat=parens,labelformat=parens]{subcaption}
\usepackage{float}
\usepackage{footnote}
\usepackage{enumitem}
\usepackage[margin=1in]{geometry}

\newcommand{\abs}[1]{\left\lvert#1\right\rvert}
\newcommand{\norm}[1]{\left\lVert#1\right\rVert}

\DeclareMathOperator{\polylog}{polylog}

\newtheorem{theorem}{Theorem}
\newtheorem{lemma}{Lemma}
\newtheorem{proposition}[lemma]{Proposition}
\newtheorem{corollary}[lemma]{Corollary}
\theoremstyle{remark}
\newtheorem*{remark}{Remark}
\theoremstyle{plain}

\newcommand{\eq}[1]{\cref{eq:#1}}
\newcommand{\thm}[1]{\hyperref[thm:#1]{Theorem~\ref*{thm:#1}}}
\newcommand{\defn}[1]{\hyperref[defn:#1]{Definition~\ref*{defn:#1}}}
\newcommand{\lem}[1]{\hyperref[lem:#1]{Lemma~\ref*{lem:#1}}}
\newcommand{\prop}[1]{\hyperref[prop:#1]{Proposition~\ref*{prop:#1}}}
\newcommand{\fig}[1]{\hyperref[fig:#1]{Figure~\ref*{fig:#1}}}
\newcommand{\tab}[1]{\hyperref[tab:#1]{Table~\ref*{tab:#1}}}
\renewcommand{\sec}[1]{\hyperref[sec:#1]{Section~\ref*{sec:#1}}}
\newcommand{\append}[1]{\hyperref[append:#1]{Appendix~\ref*{append:#1}}}
\newcommand{\cor}[1]{\hyperref[cor:#1]{Corollary~\ref*{cor:#1}}}
\newcommand{\obs}[1]{\hyperref[obs:#1]{Observation~\ref*{obs:#1}}}

\newcommand{\ket}[1]{|#1\rangle}
\newcommand{\bra}[1]{\langle#1|}
\newcommand{\ketbra}[2]{\ket{#1}\!\bra{#2}}

\newcommand\blfootnote[1]{%
	\begingroup
	\renewcommand\thefootnote{}\footnote{#1}%
	\addtocounter{footnote}{-1}%
	\endgroup
}

\newtheorem{innercustomthm}{Theorem}
\newenvironment{customthm}[1]
{\renewcommand\theinnercustomthm{#1}\innercustomthm}
{\endinnercustomthm}

\usepackage{mathtools}

\DeclareFontFamily{U}{matha}{\hyphenchar\font45}
\DeclareFontShape{U}{matha}{m}{n}{
	<5> <6> <7> <8> <9> <10> gen * matha
	<10.95> matha10 <12> <14.4> <17.28> <20.74> <24.88> matha12
}{}
\DeclareSymbolFont{matha}{U}{matha}{m}{n}
\DeclareFontFamily{U}{mathx}{\hyphenchar\font45}
\DeclareFontShape{U}{mathx}{m}{n}{
	<5> <6> <7> <8> <9> <10>
	<10.95> <12> <14.4> <17.28> <20.74> <24.88>
	mathx10
}{}
\DeclareSymbolFont{mathx}{U}{mathx}{m}{n}
\DeclareMathSymbol{\obot}         {2}{matha}{"6B}
\DeclareMathSymbol{\bigobot}       {1}{mathx}{"CB}

\newcommand{\alphaU}{\alpha_{\scriptscriptstyle\mathbf{U}}}
\newcommand{\alphaT}{\alpha_{\scriptscriptstyle \widetilde{\mathbf{T}},\psi}}
\newcommand{\alphaFP}{\alpha_{\scriptscriptstyle \mathbf{F}'}}
\newcommand{\alphaF}{\alpha_{\scriptscriptstyle \mathbf{F},\psi}}
\newcommand{\pfail}{p_{\scriptscriptstyle \text{fail}}}
\newcommand{\sv}{\text{sv}}
\newcommand{\alphaPPsi}{\alpha_{\scriptscriptstyle p,\psi}}
\newcommand{\alphaFPsi}{\alpha_{\scriptscriptstyle \exp,\psi}}
\newcommand{\alphaPsiFaber}{\alpha_{\scriptscriptstyle \mathbf{\Psi}}}

\usepackage{nicematrix}
\usepackage{tikz}

\makeatletter
\def\newmaketag{%
  \def\maketag@@@##1{\hbox{\m@th\normalfont\normalsize##1}}%
  }
\makeatother

\usepackage[hyperpageref]{backref}
\renewcommand*{\backrefalt}[4]{%
\ifcase #1 %
No citations.%
\or
(Cited on page #2).%
\else
(Cited on pages #2).%
\fi
}

\usepackage{etoolbox}
\makeatletter
\patchcmd\NAT@citexnum{\let\NAT@last@num\NAT@num}{\MakeLinkTarget[cite]{}\Hy@backout{\@citeb\@extra@b@citeb}\let\NAT@last@num\NAT@num}{}{\fail}
\makeatother

\usepackage{etoolbox}
\apptocmd{\sloppy}{\hbadness 10000\relax}{}{}

\title{\huge\vspace{-1.cm} Quantum eigenvalue processing}
\author{Guang Hao Low$^1$ and Yuan Su$^2$}
\date{\vspace{-15mm}}

\begin{document}
\maketitle

\blfootnote{This is an enhanced version of the paper entitled \emph{Quantum eigenvalue processing} presented at the 65th IEEE Symposium on Foundations of Computer Science~\cite{QEVP2} and published in the SIAM Journal on Computing~\cite{QEVP}.}
\blfootnote{$^1$Azure Quantum, Microsoft, Redmond, WA 98052, USA. Now at Google Quantum AI, Venice, CA 90291, USA.}
\blfootnote{$^2$Azure Quantum, Microsoft, Redmond, WA 98052, USA. Now at the AWS Center for Quantum Computing, Pasadena, CA 91106, USA.}

\begin{abstract}
Many problems in linear algebra require processing eigenvalues of the input matrices. As eigenvalues are different from singular values for non-normal operators, these problems are out of reach of the existing quantum singular value algorithm and its descendants.

We present a Quantum EigenValue Estimation (QEVE) algorithm and a Quantum EigenValue Transformation (QEVT) algorithm that estimate and transform eigenvalues of high-dimensional matrices accessed by a quantum computer through block encoding oracles. We focus on input matrices with real spectra and Jordan forms---a broad class of operators that can describe non-Hermitian physics and transcorrelated quantum chemistry. However, our technique also handles general non-normal matrices with complex eigenvalues, and our method remains efficient even when the Jordan basis is ill conditioned.

Our QEVE estimates an eigenvalue of a diagonalizable matrix using $\mathbf{O}\left(\alpha\kappa/\epsilon\log(1/p)\right)$ queries to its block encoding and a unitary preparing the corresponding eigenstate, in terms of error $\epsilon$, failure probability $p$, normalization factor $\alpha$ of the block encoding, and condition number $\kappa$ of its basis transformation. This solves the eigenvalue estimation problem for a broad class of non-normal matrices with the Heisenberg-limited scaling, which naturally reduces to the optimal estimation of singular values that has long been known. Our approach is conceptually simple, based on reductions to the optimal scaling quantum linear system algorithm, improving over prior approaches using differential equation solvers which are polylogarithmic away from optimum.

Our QEVT implements transformations on eigenvalues of the input matrix through the Chebyshev and Faber approximations. As these expansions provide a close-to-best uniform polynomial approximation of functions over the complex plane, the query complexity of QEVT is expected to be nearly optimal. In particular, our eigenvalue algorithm achieves a performance comparable to previous singular value transformation results for the special
case of Hermitian inputs, where eigenvalues coincide with singular values in magnitude.

As an application, we present a quantum differential equation algorithm based on QEVT, whose query complexity scales strictly linear in the evolution time $t$ for an average-case diagonalizable input with imaginary spectra, whereas the best previous approach has a complexity with an extra multiplicative $\polylog(t)$ factor. We also develop a quantum algorithm for preparing the ground state of matrices with real spectra, which reduces to the nearly optimal result for Hermitian Hamiltonians from previous work.

Underlying both QEVE and QEVT is an efficient quantum algorithm for preparing the Chebyshev history state through its matrix generating function, encoding Chebyshev polynomials of the input matrix in quantum superposition, which may be of independent interest. Prior to our work, it was known how to efficiently create such a state only for Hermitian inputs. We then extend this result to prepare the Faber history state, achieving efficient eigenvalue transformation over the complex plane. Independently, we develop techniques to generate $n$ Fourier coefficients using $\mathbf{O}(\polylog(n))$ gates, improving over prior approaches with a cost of $\mathbf{\Theta}(n)$.

Our result thus provides a unifying framework for processing eigenvalues of matrices on a quantum computer.
\end{abstract}

\newpage
{
	\thispagestyle{empty}
	\clearpage\tableofcontents
	\thispagestyle{empty}
}
\newpage

\section{Introduction}
\label{sec:intro}
Many problems in linear algebra can be solved by processing eigenvalues of the input matrix. As eigenvalues of a non-normal matrix are different from its singular values, such problems are not approachable by the existing quantum singular value algorithm.

The goal of this work is to present an algorithmic framework to process eigenvalues of high-dimensional non-normal matrices on a quantum computer, going beyond previous quantum algorithms targeting at the singular values. We focus on operators with real spectra and Jordan forms---a broad class of non-normal matrices arising in the study of non-Hermitian physics and transcorrelated quantum chemistry.
However, our method also applies to more general operators with eigenvalues in the complex plane, with no dependence on the Jordan condition number.

Within the proposed framework, we develop: (i) a quantum algorithm for estimating eigenvalues of a non-normal matrix, with the complexity strictly achieving Heisenberg-limited scaling when the input is diagonalizable with real spectra; (ii) a quantum algorithm for applying arbitrary polynomial transformations to the input matrix, nearly reproducing previous singular value results when the input is Hermitian; (iii) a quantum algorithm for solving systems of linear differential equations, with a strictly linear scaling in the evolution time for an average diagonalizable input matrix with imaginary spectra; and (iv) a quantum algorithm for ground state preparation, recovering the nearly optimal scaling from previous work when the inputs are Hermitian Hamiltonians.
We summarize these algorithms in \tab{nonnormal}, explaining how we treat non-normality of the input in each case.

The core technique underpinning our results is the efficient creation of a history state encoding Chebyshev polynomials of the input matrix in quantum superposition, with further generalizations to Faber polynomials for eigenvalue processing over the complex plane. This is in turn achieved by implementing an operator version of the Chebyshev and Faber generating functions.
While the Chebyshev history state of a Hermitian matrix can be efficiently generated via quantum walk, no such a mechanism was available for non-normal operators. To our knowledge, our work elucidates the first connection between quantum computing and the vast field of Faber polynomials, which provide a nearly-optimal basis for uniform function approximation over the complex domain.

Our result thus suggests the use of matrix generating functions as a powerful methodology for solving linear algebraic problems on a quantum computer.

\subsection{Eigenvalue processing}
\label{sec:intro_process}

Quantum computers can operate quantum systems of exponentially large dimensions with only polynomially many resources. This feature underlies the exponential speedups found in various promising applications, including simulating quantum systems~\cite{Lloyd1996universal}, solving systems of linear equations~\cite{Harrow2009}, and factoring integers~\cite{sho_polynomialtime_1997}. These computational problems typically have inputs encoded by matrices such as multi-qubit unitaries and Hamiltonians, and their solutions can be obtained by processing the exponentially large matrices on a quantum computer using only a polynomial amount of resources. For quantum simulation, this involves applying the exponential function $H\mapsto e^{-itH}$ to the target Hamiltonian. For solving linear equations, this means implementing the inverse function $A\mapsto A^{-1}$ on a well-conditioned coefficient matrix. And for factoring integers, this entails estimating eigenvalues of the modular multiplication operator.

To directly harness this exponential power of quantum computers, an algorithmic technique known as the Quantum Singular Value Transformation (QSVT) was proposed~\cite{Gilyen2018singular}. Given the singular value decomposition $A=V\Sigma U^\dagger$ of the input matrix, QSVT applies polynomial functions $p$ to the singular values of $A$ in the following manner:
\begin{equation}
    A=V\Sigma U^\dagger\mapsto p_{\sv}(A)=Vp(\Sigma)U^\dagger.
\end{equation}
Notably, this can be realized on a quantum computer with a query complexity that depends on degree of the polynomial rather than dimensions of the underlying Hilbert space, avoiding an explicit calculation of the exponentially large basis transformations $U$ and $V$. As the notion of singular values plays a fundamental role in linear algebra, QSVT has found a host of applications in quantum linear algebra and has unified diverse quantum algorithms~\cite{Martyn21,childs2017lecture,LinLin22,deWolf19}, ranging from Hamiltonian simulation~\cite{Low2016HamSim} to fixed-point quantum search~\cite{Yoder2014Fixed} and eigenstate filtering~\cite{Lin2020optimalpolynomial}, while providing a systematic methodology to implement transformations on singular values with optimal query complexity~\cite{Montanaro2024quantum}.

Closely related to the singular value transformation is the problem of singular value estimation, whose solution underlies the quantum speedups for factoring integers~\cite{sho_polynomialtime_1997} and elucidating chemical reactions~\cite{vonBurg21,Lee21}. Here, an initial state close to a right singular vector $\ket{\psi_j}$ of $A$ is given, and the goal is to estimate the corresponding singular value $\sigma_j$. By generalizing the quantum phase estimation algorithm, a solution to the singular value estimation can be obtained with an optimal number of controlled queries to operators encoding the input matrix~\cite{kerenidis2016quantum,CGJ19,Gilyen2018singular}. In the special case where $A$ is Hermitian, essentially the same algorithm can be used for eigenvalue estimation, as eigenvalues and singular values coincide in absolute value.

However, many problems arising in practice require processing \emph{eigenvalues} of the input matrix, not its \emph{singular values}. For a diagonalizable input $A$, this means performing the polynomial transformation directly on $A$, which has the action
\begin{equation}
    A=S\Lambda S^{-1}\mapsto p(A)=Sp(\Lambda)S^{-1},
\end{equation}
with an invertible basis transformation $S$.
Such eigenvalue processing problems arise naturally in a variety of applications, including solving linear differential equations~\cite{Berry2017Differential}, simulating non-Hermitian physics~\cite{Bender07,Ashida20}, simulating transcorrelated quantum chemistry~\cite{McArdle20}, and estimating spectral properties of stochastic matrices~\cite{Zhang2024Nonnormal}. 
Of course, any singular value processing problem can always be reformulated as an eigenvalue problem through the Hermitian dilation $A\mapsto\ketbra{0}{1}\otimes A+\ketbra{1}{0}\otimes A^\dagger$.
But the converse does not hold for non-normal matrices. In fact, eigenvalue processing appears far more challenging than singular value processing in many respects:
(i) the basis transformation $S$ is invertible but not necessarily unitary;
(ii) eigenvalues are generally complex numbers, unlike singular values that are real and nonnegative; and
(iii) $A$ may only admit the Jordan form decomposition in general, where the factor $\Lambda$ is not diagonal.
See \sec{prelim_matrix} for formal definitions of the singular value, eigenvalue, and Jordan form transformations.
Hence, existing singular value algorithms are not applicable to such eigenvalue problems, and there is no unifying framework to solve them on a quantum computer with optimal query complexity.

\subsection{Chebyshev history state generation}
\label{sec:intro_history}

We present a Quantum EigenValue Estimation (QEVE) algorithm and a Quantum EigenValue Transformation (QEVT) algorithm that estimate and transform eigenvalues of high-dimensional non-normal matrices accessed by a quantum computer. Our main results and applications are illustrated diagrammatically in \fig{diagram} and tabulated in \tab{nonnormal}.

\begin{figure}[t]
	\centering
\includegraphics[width=0.95\textwidth]{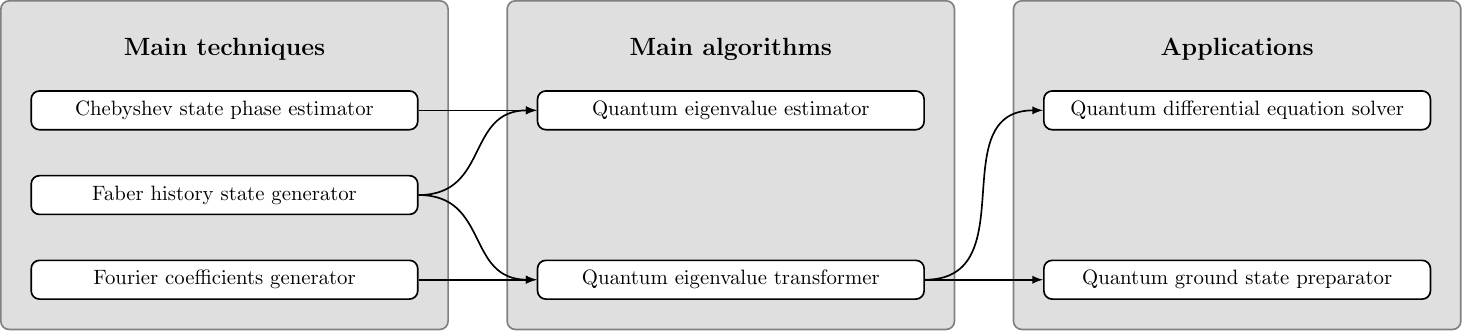}
\caption{A diagrammatic illustration of quantum eigenvalue processing and its applications. See \tab{nonnormal} for a summary of common treatments of non-normality of the input matrix and query complexity.}
\label{fig:diagram}
\end{figure}

The common tool underpinning both QEVE and QEVT is an efficient algorithm to prepare a history state encoding polynomials of the input matrix through the use of generating functions.
To be specific, consider the expansion of a polynomial $p$ with respect to a polynomial basis
\begin{equation}
\label{eq:poly_expansion}
    p(A)=\sum_{j=0}^{n-1}\beta_jp_j(A),
\end{equation}
where $p_j$ are degree-$j$ polynomials, and $A$ is for now assumed to satisfy $\norm{A}=1$ for simplicity ($\norm{\cdot}$ is the spectral norm). Given a block encoding of $A$ (see \sec{prelim_block} for the formal definition), if $A$ is Hermitian, prior art provides a highly efficient method based on quantum walk for generating $p_j(A)$ that are Chebyshev polynomials of $A$ using only $\mathbf{O}(j)$ queries to the block encoding~\cite{Childs2010,Low2019hamiltonian}. The QSVT algorithm in particular allows one to block encode an arbitrary $p(A)$ with query complexity
\begin{equation*}
    \mathbf{O}\left(n \norm{p}_{\max,[-1,1]}\right),
\end{equation*}
where $\norm{p}_{\max,[-1,1]}=\max_{x\in[-1,1]}\abs{p(x)}$.

However, when $A$ is non-normal, the QSVT technique does not apply and previous methods only allow one to efficiently generate monomials $A^j$ using $\mathbf{O}(j)$ queries. Hence if arbitrary polynomials $p(A)$ or $p_j(A)$ of $A$ are desired, previous methods are only able to block encode $p(A)$ by taking a linear combination of monomials, with a cost scaling like
\begin{equation}
    \mathbf{O}\left(n\sum_{j=0}^{n-1}\abs{\beta_j}\right)
    =\mathbf{O}\left(n\norm{\beta}_1\right),
\end{equation}
where $\beta_j$ are coefficients of the monomial expansion of $p$. Note that all the terms above add up constructively, so the result $\norm{\beta}_1=\sum_{j=0}^{n-1}\abs{\beta_j}$ can be significantly larger than the desired scaling with $\norm{p}_{\max,[-1,1]}$. In fact, this cost is exponentially large in $n$ for the differential equation and the ground state preparation problem to be discussed here. 

Our approach overcomes this exponential complexity of block encoding in a monomial basis by providing means to directly and efficiently generate a polynomial basis $p_j(A)$ of a non-normal matrix $A$, where $p_j$ may, for instance, include Chebyshev polynomials. Our approach is based on a matrix version of the generating function $\sum_{j=0}^\infty y^j p_j(x)=g(y,x)$ of the polynomial basis $p_j$. Specifically, we introduce the $n$-by-$n$ lower shift matrix $L$ and aim to implement 
\begin{equation}
    \sum_{j=0}^{n-1} L^j\otimes p_j\left(\frac{A}{\alpha_A}\right)=\sum_{j=0}^\infty L^j\otimes p_j\left(\frac{A}{\alpha_A}\right)=g\left(L\otimes I,I\otimes \frac{A}{\alpha_A}\right)
\end{equation}
with a complexity polynomial in $n$ (we will incorporate the normalization factor $\alpha_A$ of the block encoding of $A$ hereafter). 
The above equation follows from the definition of generating function, along with the substitution $y=L\otimes I$ and $x=I\otimes\frac{A}{\alpha_A}$.
This is then measured to estimate the target eigenvalue, or is further combined with a subroutine that generates the expansion coefficients $\beta_j$ to efficiently transform the eigenvalues.

When the polynomial basis is selected to be Chebyshev polynomials, we have the following \thm{generate_history} which will be established in \sec{history}.

\begin{customthm}{\ref{thm:generate_history}}[Chebyshev history state generation] 
    Let $A$ be a square matrix with only real eigenvalues, such that $A/\alpha_A$ is block encoded by $O_A$ with some normalization factor $\alpha_A\geq\norm{A}$. 
    Let $O_\psi\ket{0}=\ket{\psi}$ be the oracle preparing the initial state, 
    and $O_{\widetilde\beta}\ket{0}\propto\sum_{k=0}^{n-1}(\widetilde\beta_k-\widetilde\beta_{k+2})\ket{n-1-k}$ be the oracle preparing the shifting of coefficients $\widetilde\beta_k$ ($k=0,\ldots,n-1$). 
    Then, the quantum state proportional to
    \begin{equation}
    \label{eq:cheby_history}
        \ket{0}\sum_{l=0}^{n-1}\ket{l}
        \sum_{k=n-1-l}^{n-1}\widetilde{\beta}_k\widetilde{\mathbf{T}}_{k+l-n+1}\left(\frac{A}{\alpha_A}\right)\ket{\psi}
        +\sum_{s=1}^{\eta}\ket{s}\sum_{l=0}^{n-1}\ket{l}
        \sum_{k=0}^{n-1}\widetilde{\beta}_k\widetilde{\mathbf{T}}_{k}\left(\frac{A}{\alpha_A}\right)\ket{\psi}
    \end{equation}
    can be prepared with accuracy $\epsilon$ and probability $1-\pfail$ using
    \begin{equation}
        \mathbf{O}\left(\alphaU n(\eta+1)\log\left(\frac{1}{\epsilon}\right)\log\left(\frac{1}{\pfail}\right)\right)
    \end{equation}
    queries to controlled-$O_A$, controlled-$O_\psi$, controlled-$O_{\widetilde\beta}$, and their inverses, where 
    \begin{equation}
    \label{eq:alphaU}
        \alphaU\geq\max_{j=0,1,\ldots,n-1}\norm{\mathbf{U}_{j}\left(\frac{A}{\alpha_A}\right)},
    \end{equation}
    is an upper bound on Chebyshev polynomials of the second kind $\mathbf{U}_j(x)$,
    $\widetilde{\mathbf{T}}_k(x)$ are rescaled Chebyshev polynomials of the first kind to be defined in \sec{prelim_cheby_fourier}, and $\norm{\cdot}$ denotes the Euclidean norm for vectors and the spectral norm for operators.
\end{customthm}

The output state of our algorithm \eq{cheby_history} resembles those constructed by previous quantum differential equation algorithms. See~\cite{Berry2017Differential} for further explanations on how they relate to the history states commonly used in quantum complexity theory.
Specifically, this is a state of three quantum registers. The third register is the system register holding the input state, on which we perform the truncated Chebyshev expansion. The first register indicates whether the expansion has shifted indices, whereas the amount of shifting is further recorded in the second register. Here, the parameter $\eta$ controls the probability of preparing the Chebyshev partial sum. In our applications, we will adjust the parameter $\eta$ (choosing either $\eta=0$ or $\eta=1$) as well as the expansion coefficients $\widetilde{\beta}_k$, so that the resulting history state can be used in QEVE or QEVT respectively.

The complexity of our algorithm depends on largest size of the input operator under the polynomial basis mapping $\alphaU\geq\max_{j=0,1,\ldots,n-1}\norm{\mathbf{U}_{j}\left(\frac{A}{\alpha_A}\right)}$, which can be further bounded using properties of the input matrix $A$. For instance, if $A/\alpha_A=SJS^{-1}$ has a Jordan form decomposition, with $\kappa_S\geq\norm{S}\norm{S^{-1}}$ an upper bound on the \emph{Jordan condition number} and $d_{\max}$ size of the largest Jordan block, then $\alphaU=\mathbf{O}\left(n^{d_{\max}-1}\kappa_S\right)$ grows polynomially as a function in $n$. In particular, we have $\alphaU=\mathbf{O}\left(\kappa_S\right)$ if $d_{\max}=1$, resulting in the complexity
\begin{equation}
    \mathbf{O}\left(\kappa_Sn(\eta+1)\log\left(\frac{1}{\epsilon}\right)\right)
\end{equation}
for generating the Chebyshev history state of a diagonalizable input matrix $A$. As will be explained in \append{analysis_cheby}, this is essentially a bound based on the \emph{spectral abscissa}, corresponding to the analysis~\cite[Section 3.2]{Krovi2023improvedquantum} of differential equation solvers. Note however that there exist other bounds to determine $\alphaU$ each having its own strength and weakness, and there is no unifying one that completely dominates the others.
See \tab{nonnormal} for more details.
For the purpose of generality, we choose to express the complexity of our algorithm in terms of a general upper bound $\alphaU\geq\max_{j=0,1,\ldots,n-1}\norm{\mathbf{U}_{j}\left(\frac{A}{\alpha_A}\right)}$, which can be further refined when our algorithm is applied to a concrete problem.

To generate the Chebyshev history state, we use a matrix version of the Chebyshev generating function $\sum_{j=0}^\infty y^j \widetilde{\mathbf{T}}_j(x)=\frac{1}{2}\frac{1-y^2}{1-2yx+y^2}$:
\begin{equation}
    \sum_{j=0}^{n-1}L^j\otimes \widetilde{\mathbf{T}}_j\left(\frac{A}{\alpha_A}\right)
    =\sum_{j=0}^{\infty}L^j\otimes \widetilde{\mathbf{T}}_j\left(\frac{A}{\alpha_A}\right)
    =\frac{I\otimes I-L^2\otimes I}{2(I\otimes I+L^2\otimes I-2L\otimes \frac{A}{\alpha_A})},
\end{equation}
where $L$ is the $n$-by-$n$ lower shift matrix $L=\sum_{k=0}^{n-2}\ketbra{k+1}{k}$ such that $L^n=0$. When applied to an ancilla state encoding the Chebyshev coefficients $\widetilde\beta_k$, this generates the first term of our desired history state. The second term can then be generated by repeating the subterm flagged by $\ket{l}=\ket{n-1}$ a total number of $\sim\eta n$ times using the runaway padding trick~\cite{Berry2017Differential}. 

The application of the matrix Chebyshev generating function to the initial state can be formulated as solving a system of linear equations with coefficient matrix $C=(I\otimes I+L^2\otimes I-2L\otimes \frac{A}{\alpha_A})$ and target vector $b=\frac{I\otimes I-L^2\otimes I}{2}\left(\sum_{k=0}^{n-1}\widetilde\beta_k\ket{n-1-k}\ket{\psi}\right)\propto
\sum_{k=0}^{n-1}(\widetilde\beta_k-\widetilde\beta_{k+2})\ket{n-1-k}\ket{\psi}$, which can in turn be solved by a quantum linear system algorithm. To produce an $\epsilon$-approximate solution state $\ket{x}$ corresponding to the equation $Cx=b$, the fastest quantum linear system solver makes $\mathbf{O}\left(\alpha_C\alpha_{C^{-1}}\log\left(\frac{1}{\epsilon}\right)\right)$ queries to the block encoding of $C/\alpha_C$ and the unitary preparing the normalized version of $b$ as a quantum state~\cite{Costa22,OptInit,Dalzell2024shortcut}, where $\alpha_{C^{-1}}\geq\norm{C^{-1}}$ is an upper bound on size of the inverse operator. The claimed complexity is then established by showing explicitly how the padded version of $(I\otimes I+L^2\otimes I-2L\otimes \frac{A}{\alpha_A})$ can be block encoded using $1$ query to the block encoding of $A/\alpha_A$, and by upper bounding condition number of the resulting linear system.

\begin{figure}[t]
\centering
\begin{subfigure}[t]{.25\linewidth}
\hfill
\includegraphics[width = 0.95\textwidth]{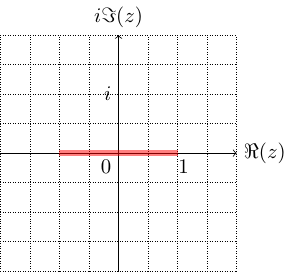}
\subcaption{}
\label{fig:region_interval}
\end{subfigure}
\begin{subfigure}[t]{.25\linewidth}
\hfill
\includegraphics[width = 0.95\textwidth]{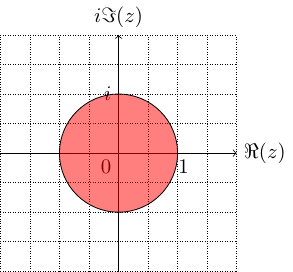}
\subcaption{}
\label{fig:region_disk}
\end{subfigure}
\begin{subfigure}[t]{.25\linewidth}
\hfill
\includegraphics[width = 0.95\textwidth]{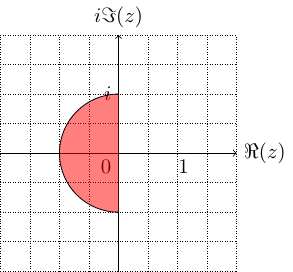}
\subcaption{}
\label{fig:region_general}
\end{subfigure}
\begin{subfigure}[t]{.25\linewidth}
\hfill
\includegraphics[width = 0.95\textwidth]{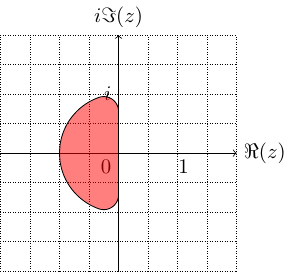}
\subcaption{}
\label{fig:region_general2}
\end{subfigure}
\caption{Illustration of regions in the complex plane that enclose eigenvalues of the input matrices. Subfigure (a) represents the real interval $[-1,1]$, where functions can be nearly best approximated by truncating the Chebyshev expansion. Subfigure (b) represents the unit disk, where functions can be nearly best approximated by truncating the Taylor expansion. Subfigures (c) and (d) represent more general regions in the complex plane, where functions can be nearly best approximated by truncating the Faber expansion. Subfigure (c) shows a unit semidisk on the left half-plane, with Faber expansion generated by the Elliott's conformal map~\cite{Coleman1987FaberCircularSectors}, whereas Subfigure (d) is a smooth deformation of (c).}
\label{fig:region}
\end{figure}

Chebyshev expansion provides a nearly best uniform polynomial approximation of functions over the real interval $[-1,1]$ (\fig{region_interval}). They are thus especially suitable for processing eigenvalues of matrices with real spectra. As aforementioned, matrices with real eigenvalues and Jordan forms already constitute a broad class of non-normal matrices that have applications in non-Hermitian physics~\cite{Bender07,Ashida20} and transcorrelated quantum chemistry~\cite{McArdle20} (complementary to the Hermitian case of~\cite{Motta20}). However, our technique is extendable to handle circular disks by implementing Taylor expansion (\fig{region_disk}), as well as more general complex regions by implementing Faber expansion~\cite{suetin1998series} (\fig{region_general} and \fig{region_general2}), providing a unified quantum algorithmic framework for eigenvalue processing.

\subsection{Quantum eigenvalue estimation}
\label{sec:intro_est}

The ability to generate the Chebyshev history states allows us to estimate and transform real eigenvalues of non-normal matrices.

Recall that although optimal quantum algorithms for the singular value estimation have long been known, the optimal eigenvalue estimation has remained elusive for non-normal matrices. Specifically, prior approaches to QEVE work by generating a state of the form
\begin{equation}
    \frac{\sum_{l=0}^{n-1}\ket{l}e^{2\pi il\frac{A}{\alpha_A}}\ket{\psi}}{\norm{\sum_{l=0}^{n-1}\ket{l}e^{2\pi il\frac{A}{\alpha_A}}\ket{\psi}}}
\end{equation}
using quantum differential equation algorithms~\cite{Shao2019eigenvalues,Shao2020GeneralizedEigenvalue}. When $\ket{\psi}\approx\ket{\psi_j}$ is close to an eigenstate of $A$ with eigenvalue $\lambda_j$, the resulting state (omitting the second register) is close to the Fourier state $\frac{1}{\sqrt{n}}\sum_{l=0}^{n-1}e^{2\pi il\frac{\lambda_j}{\alpha_A}}\ket{l}$.
Then the standard phase estimation algorithm suffices to estimate $\lambda_j/\alpha_A$ with accuracy $\mathbf{O}(1/n)$ and a constant success probability strictly larger than $1/2$, which can be boosted to at least $1-\pfail$ by repeating $\mathbf{O}(\log(1/\pfail))$ times and taking the median of measurement outcomes.
Unfortunately, existing quantum differential equation algorithms are not known to be optimal. For instance, consider a diagonalizable matrix $A/\alpha_A=S(\Lambda/\alpha_A) S^{-1}$. Then it takes
\addtocounter{equation}{1}
\begin{equation}
    \mathbf{O}\left(\kappa_Sn\polylog\left(\frac{\kappa_Sn}{\delta}\right)\right)
\end{equation}
queries to even produce the Fourier state with accuracy $\delta$. Choosing $\delta=\mathbf{\Theta}(1)$ sufficiently small and $n=\mathbf{\Theta}(\alpha_A/\epsilon)$, one gets a suboptimal quantum algorithm for eigenvalue estimation with cost
\begin{equation}
    \mathbf{O}\left(\frac{\alpha_A\kappa_S}{\epsilon}\polylog\left(\frac{\alpha_A\kappa_S}{\epsilon}\right)\log\left(\frac{1}{\pfail}\right)\right).
\end{equation}

In contrast, our approach starts by generating the following Chebyshev history state:
\begin{equation}
        \frac{\sum_{l=0}^{n-1}\ket{l}\widetilde{\mathbf{T}}_l\left(\frac{A}{\alpha_A}\right)\ket{\psi}}{\norm{\sum_{l=0}^{n-1}\ket{l}\widetilde{\mathbf{T}}_l\left(\frac{A}{\alpha_A}\right)\ket{\psi}}}
    \end{equation}
using only
\begin{equation}
        \mathbf{O}\left(\kappa_Sn\log\left(\frac{1}{\delta}\right)\right)
    \end{equation}
queries (for a diagonalizable input), where $\widetilde{\mathbf{T}}_l(x)$ are the rescaled Chebyshev polynomials as above.
To prepare this state, we can use our \thm{generate_history} with $\eta=0$ and set $\widetilde{\beta}_k=1$ when $k=n-1$ and $\widetilde{\beta}_k=0$ otherwise.
When $\ket{\psi}\approx\ket{\psi_j}$ is close to an eigenstate of $A$ with eigenvalue $\lambda_j$, the state we produce is close to the Chebyshev state $\frac{1}{\sqrt{\widetilde\alpha_{\lambda_j}}}\sum_{l=0}^{n-1}\widetilde{\mathbf{T}}_l\left(\frac{\lambda_j}{\alpha_A}\right)\ket{l}$,
where $\widetilde\alpha_{\lambda_j}=\sum_{l=0}^{n-1}\widetilde{\mathbf{T}}_l^2\left(\frac{\lambda_j}{\alpha_A}\right)$ is the normalization factor. This state is comparable to the Fourier state used by previous work, but its generation only takes $\mathbf{O}\left(\kappa_Sn\log\left(1/\delta\right)\right)$ queries and is thus significantly faster. It remains to explain how we estimate $\lambda_j$ given copies of such Chebyshev history states.

To this end, we use the observation that our generated state has an overlap of $\mathbf{\Omega}\left(1-1/\sqrt{n}\right)$ with the unrescaled Chebyshev state $\frac{1}{\sqrt{\alpha_{\lambda_j}}}\sum_{l=0}^{n-1}\mathbf{T}_l\left(\frac{\lambda_j}{\alpha_A}\right)\ket{l}
    =\frac{1}{\sqrt{\alpha_{\phi_j}}}\sum_{l=0}^{n-1}\cos\left(2\pi l\phi_j\right)\ket{l}$,
where $\phi_j=\frac{1}{2\pi}\arccos\left(\frac{\lambda_j}{\alpha_A}\right)$ and $\alpha_{\lambda_j}=\alpha_{\phi_j}=\sum_{l=0}^{n-1}\cos^2\left(2\pi l\phi_j\right)$.
Here, $\cos\left(2\pi l\phi_j\right)$ are periodic trigonometric functions whose discrete spectra are given by the Kronecker delta functions, so the phase angle $\phi_j$ should be extractable using the Fourier transform. To realize this intuition, we develop a variant of quantum phase estimation in \thm{chebyshev_qpe} for the Chebyshev history state, which may be of independent interest.

Given quantum state $\frac{1}{\sqrt{\alpha_\phi}}\sum_{l=0}^{n-1}\cos\left(2\pi l\phi\right)\ket{l}$, the Chebyshev state phase estimation algorithm outputs a value $l\in\{0,\ldots,n-1\}$ such that $\frac{l}{n}\approx\phi$ in modular distance.
Thus we can use $\alpha_A\cos\left(2\pi\frac{l}{n}\right)$ to estimate $\lambda_j$, and we achieve an accuracy $\epsilon$ by setting $n=\mathbf{O}\left(\alpha_A/\epsilon\right)$. This fails with a constant probability strictly smaller than $\frac{1}{2}$. By repeating $\mathbf{O}\left(\log(1/\pfail)\right)$ times and taking the median, the failure probability can be exponentially suppressed to below $\pfail$. We thus have the following \thm{qeve} which we preview here and prove in \sec{est}.

\begin{customthm}{\ref{thm:qeve}}[Quantum eigenvalue estimation]
    Let $A$ be a square matrix with only real eigenvalues, such that $A/\alpha_A$ is block encoded by $O_A$ with some normalization factor $\alpha_A\geq\norm{A}$.  
    Suppose that oracle $O_\psi\ket{0}=\ket{\psi}$ prepares an initial state within distance $\norm{\ket{\psi}-\ket{\psi_{\lambda_j}}}=\mathbf{O}(\sqrt{\epsilon/\alpha_A}/\alphaU)$ from an eigenstate $\ket{\psi_{\lambda_j}}$ such that $A\ket{\psi_{\lambda_j}}=\lambda_j\ket{\psi_{\lambda_j}}$, where $\alphaU$ satisfies \eq{alphaU} with
\addtocounter{equation}{3}
    \begin{equation}
        n=\mathbf{O}\left(\frac{\alpha_A}{\epsilon}\right).
    \end{equation}
    Then, the eigenvalue $\lambda_j$ can be estimated with accuracy $\epsilon$ and probability $1-\pfail$ using
    \begin{equation}
        \mathbf{O}\left(\frac{\alpha_A}{\epsilon}\alphaU\log\left(\frac{1}{\pfail}\right)\right)
    \end{equation}
    queries to controlled-$O_A$, controlled-$O_\psi$, and their inverses.
\end{customthm}

The setting in which the above algorithm works is pretty general as it only requires the input matrix $A$ to have real spectra. As is already explained, there are various methods one can use to further bound $\alphaU$. In the case where $A/\alpha_A=S(\Lambda/\alpha_A)S^{-1}$ is diagonalizable, $\alphaU=\mathbf{O}(\kappa_S)$ and the cost of QEVE algorithm becomes
\begin{equation}
    \mathbf{O}\left(\frac{\alpha_A\kappa_S}{\epsilon}\log\left(\frac{1}{\pfail}\right)\right).
\end{equation}
The scaling with the inverse precision $\sim\frac{\alpha_A}{\epsilon}$ cannot be improved even in the ideal case where the exact eigenstate is provided. It is closely related~\cite{Atia2017} to the \emph{Heisenberg scaling} in quantum metrology~\cite{Giovannetti06,Zwierz10}, and we will adopt this terminology throughout the remainder of the paper.

\subsection{Quantum eigenvalue transformation}
\label{sec:intro_transform}

Besides the eigenvalue estimation, the availability of the Chebyshev history state also allows us to apply polynomial functions to the eigenvalues of non-normal matrices. This is formally realized by the QEVT algorithm in \thm{qevt}, which we state here and establish in \sec{transform}.

\begin{customthm}{\ref{thm:qevt}}[Quantum eigenvalue transformation]
    Let $A$ be a square matrix with only real eigenvalues, such that $A/\alpha_A$ is block encoded by $O_A$ with some normalization factor $\alpha_A\geq\norm{A}$.
    Let $p(x)=\sum_{k=0}^{n-1}\widetilde\beta_k\widetilde{\mathbf{T}}_{k}(x)=\sum_{k=0}^{n-1}\beta_k{\mathbf{T}}_{k}(x)$ be the Chebyshev expansion of a degree-($n-1$) polynomial $p$.
    Let $O_\psi\ket{0}=\ket{\psi}$ be the oracle preparing the initial state, 
    and $O_{\widetilde\beta}\ket{0}\propto\sum_{k=0}^{n-1}(\widetilde\beta_k-\widetilde\beta_{k+2})\ket{n-1-k}$ be the oracle preparing the shifting of coefficients $\widetilde\beta_k$ ($k=0,\ldots,n-1$). 
    Then, the quantum state
    \begin{equation}
        \frac{p\left(\frac{A}{\alpha_A}\right)\ket{\psi}}{\norm{p\left(\frac{A}{\alpha_A}\right)\ket{\psi}}}
    \end{equation}
    can be prepared with accuracy $\epsilon$ and probability $1-\pfail$ using
    \begin{equation}
        \mathbf{O}\left(\frac{\alphaT}{\alphaPPsi}\alphaU n\log\left(\frac{\alphaT}{\alphaPPsi\epsilon}\right)\log\left(\frac{1}{\pfail}\right)\right)
    \end{equation}
    queries to controlled-$O_A$, controlled-$O_\psi$, controlled-$O_{\widetilde\beta}$, and their inverses,
    where $\alphaU$ satisfies \eq{alphaU} and
    \begin{equation}
    \label{eq:alphaT}
        \alphaT\geq
        \max_{l=0,1,\ldots,n-1}\norm{\sum_{k=l}^{n-1}\widetilde\beta_k\widetilde{\mathbf{T}}_{k-l}\left(\frac{A}{\alpha_A}\right)\ket{\psi}},\qquad
        \alphaPPsi\leq\norm{p\left(\frac{A}{\alpha_A}\right)\ket{\psi}}
    \end{equation}
    are upper bound on the maximum shifted partial sum of the Chebyshev expansion and lower bound on the transformed state.
\end{customthm}

Our QEVT algorithm proceeds by preparing a Chebyshev history state using \thm{generate_history} with $\eta=1$, followed by a fixed-point amplitude amplification. The number of amplitude amplification steps is determined by the ratio of the shifted partial sum $\alphaT\geq\max_l\norm{\sum_{k=l}^{n-1}\widetilde\beta_k\widetilde{\mathbf{T}}_{k-l}\left(\frac{A}{\alpha_A}\right)\ket{\psi}}$ and the desired $\alphaPPsi\leq\norm{p\left(\frac{A}{\alpha_A}\right)\ket{\psi}}$. This ratio arises in a similar way as (although is incomparable to) that of the quantum differential equation solvers~\cite{Krovi2023improvedquantum,Fang2023timemarchingbased,BerryCosta22}. Just like $\alphaU$, there are multiple ways one can further bound $\alphaT$, which we explain in \append{analysis_cheby}. For now, let us assume that the input matrix $A/\alpha_A=S(\Lambda/\alpha_A)S^{-1}$ is diagonalizable with some upper bound $\kappa_S\geq\norm{S}\norm{S^{-1}}$ on the condition number to simplify the discussion.

To analyze the matrix function $\sum_{k=l}^{n-1}\widetilde\beta_k\widetilde{\mathbf{T}}_{k-l}\left(\frac{A}{\alpha_A}\right)$, we can then diagonalize $A$ and consider instead the scalar function $\sum_{k=l}^{n-1}\widetilde\beta_k\widetilde{\mathbf{T}}_{k-l}(x)$ for the diagonal entries. For a given $l$, the shifted Chebyshev partial sum has a max-norm growing like
\begin{equation}
    \norm{\sum_{k=l}^{n-1}\widetilde\beta_k\widetilde{\mathbf{T}}_{k-l}}_{\max,[-1,1]}
    =\max_{x\in[-1,1]}\abs{\sum_{k=l}^{n-1}\widetilde\beta_k\widetilde{\mathbf{T}}_{k-l}(x)}
    =\mathbf{O}\left(\norm{p}_{\max,[-1,1]}\log(n)\right).
\end{equation}
Thus $\alphaT=\mathbf{O}\left(\norm{p}_{\max,[-1,1]}\kappa_S\log(n)\right)$ and our QEVT has the query complexity 
\begin{equation}
    \mathbf{O}\left(
    \frac{\norm{p}_{\max,[-1,1]}\kappa_S^2n}{\norm{p\left(\frac{A}{\alpha_A}\right)\ket{\psi}}}
    \log\left(
    \frac{\norm{p}_{\max,[-1,1]}\kappa_S\log(n)}{\norm{p\left(\frac{A}{\alpha_A}\right)\ket{\psi}}\epsilon}\right)\log\left(n\right)\log\left(\frac{1}{\pfail}\right)\right)
\end{equation}
in the worst case (for presentational purpose, we have used the actual value of $\norm{p\left(\frac{A}{\alpha_A}\right)\ket{\psi}}$ as opposed to its lower bound). However, we show in \append{analysis_cheby_carleson} that this shifted partial sum is actually smaller with respect to the $2$-norm:
\begin{equation}
    \norm{\sum_{k=l}^{n-1}\widetilde\beta_k\widetilde{\mathbf{T}}_{k-l}}_{2,[-1,1]}
    =\sqrt{\int_{-1}^{1}\mathrm{d}x\abs{\sum_{k=l}^{n-1}\widetilde\beta_k\widetilde{\mathbf{T}}_{k-l}(x)}^2}
    =\mathbf{O}\left(\norm{p}_{\max,[-1,1]}\right).
\end{equation}
This gives $\alphaT=\mathbf{O}\left(\norm{p}_{\max,[-1,1]}\kappa_S\right)$ and leads to the complexity
\begin{equation}
    \mathbf{O}\left(
    \frac{\norm{p}_{\max,[-1,1]}\kappa_S^2n}{\norm{p\left(\frac{A}{\alpha_A}\right)\ket{\psi}}}
    \log\left(
    \frac{\norm{p}_{\max,[-1,1]}\kappa_S}{\norm{p\left(\frac{A}{\alpha_A}\right)\ket{\psi}}\epsilon}\right)\log\left(\frac{1}{\pfail}\right)\right)
\end{equation}
when eigenvalues of the input matrix are randomly chosen, so our algorithm has a much better performance on average.
Many recent work have examined the use of randomness in improving quantum simulation algorithms.
However, most of those results have focused on the product-formula-based algorithms~\cite{Childs2019fasterquantum,Campbell18,ChenBrandao21,Zhao21}. Our result demonstrates that randomness can also be useful for speeding up more advanced quantum algorithms, which have many applications to the quantum simulation problem and beyond.

We emphasize that the QEVT algorithm as stated above not only solves the eigenvalue transformation problem for non-normal matrices, but actually provides a highly efficient solution, in the sense that its performance nearly recovers that of the QSVT algorithm for transforming singular values. Specifically, for polynomial functions $p(x)$ with $\norm{p}_{\max,[-1,1]}\leq 1$ and diagonalizable matrices, we can generate an $\epsilon$-approximate Chebyshev history state with query complexity $\mathbf{O}\left(\kappa_Sn\log\left(\frac{1}{\epsilon}\right)\right)$,
measuring which produces the quantum state $p\left(\frac{A}{\alpha_A}\right)\ket{\psi}$ with probability $\mathbf{\Omega}\left(\frac{\norm{p\left(\frac{A}{\alpha_A}\right)\ket{\psi}}^2}{\kappa_S^2}\right)$.
This is to be compared with the QSVT algorithm that uses $n$ queries and outputs the state $p_{\sv}\left(\frac{A}{\alpha_A}\right)\ket{\psi}$ with success probability $\norm{p_{\sv}\left(\frac{A}{\alpha_A}\right)\ket{\psi}}^2$.
By performing an additional fixed-point amplitude amplification, we obtain the normalized state
\addtocounter{equation}{2}
\begin{equation}
    \frac{p_{\sv}\left(\frac{A}{\alpha_A}\right)\ket{\psi}}{\norm{p_{\sv}\left(\frac{A}{\alpha_A}\right)\ket{\psi}}}
\end{equation}
using
\begin{equation}
    \mathbf{O}\left(\frac{n}{\norm{p_{\sv}\left(\frac{A}{\alpha_A}\right)\ket{\psi}}}\log\left(\frac{1}{\pfail}\right)\right)
\end{equation}
queries to the block encoding, whereas the complexity of QEVT is
\begin{equation}
    \mathbf{O}\left(
    \frac{\kappa_S^2 n}{\norm{p\left(\frac{A}{\alpha_A}\right)\ket{\psi}}}
    \log\left(
    \frac{\kappa_S}{\norm{p\left(\frac{A}{\alpha_A}\right)\ket{\psi}}\epsilon}\right)\log\left(\frac{1}{\pfail}\right)\right).
\end{equation}
On the common ground where the input matrix is Hermitian, our result has thus naturally recovered the complexity of QSVT for transforming singular values, up to a polylogarithmic factor (independent of $n$).
In particular, this implies a quantum algorithm for solving systems of linear differential equations with a strictly linear scaling in time for an average diagonalizable input, as well as a quantum algorithm for ground state preparation with a nearly optimal combined dependence on the inverse gap and inverse accuracy, which we discuss in the next subsection.

Our QEVT algorithm is formulated as a state preparation procedure, where the goal is to create a quantum state proportional to the transformed input matrix applied to the initial state. However, by using the quantum linear system solver~\cite[Corollary 69]{Gilyen2018singular} in place of~\cite{Costa22,OptInit,Dalzell2024shortcut}, it is fairly straightforward to derive a block encoding version of QEVT. We will not use this in our paper, as the block encoding introduces additional normalization factors that ruin our (nearly) optimal results for solving differential equations and preparing ground states. Nevertheless, we state this block encoding algorithm as \thm{qevt_block} in \sec{transform} for completeness, in hoping that it is useful in scenarios where QEVT serves as a subroutine.
 
In the actual circuit implementation of QEVT,
we need to implement a shifted version of the oracle $O_{\widetilde\beta}\ket{0}=\frac{1}{\norm{\widetilde\beta}}\sum_{k=0}^{n-1}\widetilde\beta_k\ket{k}$ preparing Chebyshev expansion coefficients of the target function in superposition. This state can be prepared using standard circuit techniques with a gate complexity of $\mathbf{\Theta}(n)$ (although no lower bound is known for this task). However, our truncate order $n$ in general scales polynomially with the input parameters (such as the evolution time and the inverse spectral gap), and can lead to a significant gate complexity overhead. 
We describe an alternative circuit implementation in \thm{fourier_coeff} that has gate complexity $\mathbf{O}(\polylog(n))$, by re-expressing the Chebyshev coefficients as Fourier coefficients and performing a cyclic convolution in the frequency domain. For presentational purpose, we defer a formal statement and proof of this result to \sec{fourier}.

\subsection{Applications}
\label{sec:intro_app}

Differential equations arise naturally in a broad range of scientific disciplines including engineering, physics, economics, and biology. However, classical differential equation solvers can struggle to handle problems of large dimensions, which motivates the development of quantum algorithms. To be concrete, consider the system of first-order linear differential equations $\frac{\mathrm{d}}{\mathrm{d}t}x(t)=Cx(t)$,
whose solution is given formally by $x(t)=e^{tC}x(0)$.
When $C$ has purely imaginary eigenvalues, we can prepare the solution state using QEVT by implementing the function $f(x)=e^{-i\alpha_Ctx}$ on the matrix $iC/\alpha_C$ (that has real spectra), which can be easily constructed from a block encoding of $C/\alpha_C$. We establish an equivalent version of this result for $A=iC$ as \thm{diff_eq}, whose proof will be given in \sec{app_diff_eq}.

Our algorithm proceeds by applying \thm{qevt} to the function $e^{-i\alpha_Atx}$ truncated at order
\addtocounter{equation}{2}
    \begin{equation}
        n=\mathbf{O}\left(\alpha_At+\log\left(\frac{\kappa_S}{\norm{e^{-itA}\ket{\psi}}\epsilon}\right)\right),
    \end{equation}
where $A/\alpha_A=SJS^{-1}$ has a Jordan condition number upper bounded by $\kappa_S\geq\norm{S}\norm{S}^{-1}$.
The complexity of our algorithm depends on the amplitude amplification ratio $\alphaT/\norm{e^{-itA}\ket{\psi}}$, as well as largest size $\alphaU$ of the block encoded operator under the mapping of polynomial basis. These two factors arise in a similar way as (although are not directly comparable to) those of previous differential equation solvers~\cite{Krovi2023improvedquantum}. However, these complexities become comparable when the input matrix is diagonalizable---a setting relevant for practical applications~\cite{Ashida20,McArdle20}. Then, we show that $\alphaU=\mathbf{O}\left(\kappa_S\right)$ (\append{analysis_cheby_bernstein}), whereas $\alphaT/\norm{e^{-itA}\ket{\psi}}=\mathbf{O}\left(\kappa_S\right)$ holds for an average input (\append{analysis_cheby_carleson}). This leads to the strictly linear scaling in the evolution time
\begin{equation}
    \mathbf{O}\left(\kappa_S^2\left(\alpha_At+\log\left(\frac{\kappa_S}{\epsilon}\right)\right)\log\left(\frac{\kappa_S}{\epsilon}\right)\log\left(\frac{1}{\pfail}\right)\right),
\end{equation}
shaving off a $\polylog(t)$ factor from the best previous result under the same setting.
This is reminiscent of the query complexity improvement of the Chebyshev-based method over the Taylor-based method for Hamiltonian simulation~\cite{Low2016HamSim}. However, to achieve this for solving differential equations, we would need both the new eigenvalue processing technique and the tighter analysis of Fourier truncation error.

It is worth noting that previous work proposed an alternative method to realize QEVT, based on the contour integration formula: $f(A)=1/(2\pi i)\int_{\mathcal{C}}\mathrm{d}z\ f(z)(zI-A)^{-1}$, where $\mathcal{C}$ is a contour enclosing all eigenvalues of $A$~\cite{Fang2023timemarchingbased,Takahira2020QuantumCauchy,Takahira21}. This method requires implementing a discrete version of the integral coherently on a quantum computer, and its performance depends largely on the choice of contours. With a circular contour, this method led to a quantum differential equation algorithm with a quadratic scaling in time~\cite{Fang2023timemarchingbased}. Rigorous analysis of a general contour becomes more complicated, and it is unclear how much improvement this method offers for other applications.

More recent work~\cite{AnChildsLin23} developed a quantum differential equation algorithm whose complexity is linear in the evolution time, along with additional dependence on an amplitude amplification cost that can be implicitly time dependent. That result is obtained under the assumption that the input matrix has a nonpositive numerical abscissa, similar to our Faber-based algorithm (\thm{diff_eq_faber}) to be introduced below. This is however incompatible with the setting of our \thm{diff_eq} where the input matrix has only imaginary eigenvalues. In fact, a matrix with nonpositive numerical abscissa has only imaginary eigenvalues, if and only if the matrix is anti-Hermitian (\append{analysis_faber_crouzeix}). Thus the result of~\cite{AnChildsLin23} is not immediately useful for applications such as transcorrelated quantum chemistry, where matrices have real spectra but are not necessarily Hermitian.

We now turn to our second application: the quantum ground state preparation. In the case where the input operator is a Hermitian Hamiltonian, this problem has been extensively studied by previous work such as~\cite{Poulin09,Ge19}, and can be solved near optimally on a quantum computer~\cite{Lin2020nearoptimalground}. Here, we extend the scope of previous results by considering non-normal matrices with real eigenvalues, which are relevant to applications in non-Hermitian physics and transcorrelated quantum chemistry. Specifically, let $A$ be a matrix with only real eigenvalues and an upper bound $\kappa_S$ on its Jordan condition number, block encoded with a normalization factor $\alpha_A$. Suppose that $\lambda_0$ is the smallest eigenvalue of $A$ with the corresponding eigenstate $\ket{\psi_0}$, and is \emph{nondefective} and \emph{nonderogatory}. That means, there is only one Jordan block in $A$ corresponding to the eigenvalue $\lambda_0$, and the size of that block is $1$.
Assume further that $\lambda_0$ is separated from the next eigenvalue $\lambda_1$: $\lambda_0\leq-\frac{\delta_A}{2}<0<\frac{\delta_A}{2}\leq\lambda_1$
for some spectral gap $\delta_A>0$.
Then our goal is to prepare a quantum state that $\epsilon$-approximates the ground state $\ket{\psi_0}$ up to a global phase, given an initial state $\ket{\psi}=\gamma_0\ket{\psi_0}+\sum_{l=1}^{d-1}\gamma_l\ket{\psi_l}$ 
expanded in the Jordan basis.
We achieve this using \thm{ground}, which is established in \sec{app_ground}.

Our algorithm proceeds by applying \thm{qevt} to the error function $1-\mathbf{Erf}(cx)=1-\frac{2}{\sqrt{\pi}}\int_{0}^{cx}\mathrm{d}y\ e^{-y^2}$ with a rescaling factor $c=\mathbf{O}\left(\frac{\alpha_A}{\delta_A}
        \sqrt{\log\left(\frac{\alpha_A}{\delta_A}\frac{\kappa_S}{|\gamma_0|\epsilon}\right)}\right)$ truncated at order 
\addtocounter{equation}{1}
    \begin{equation}
        n=\mathbf{O}\left(\frac{\alpha_A}{\delta_A}
        \log\left(\frac{\alpha_A}{\delta_A}\frac{\kappa_S}{|\gamma_0|\epsilon}\right)\right).
    \end{equation}
Similar as above, we will describe the algorithm in its full generality, keeping factors like $\alphaU$ and $\alphaT$ that can be further refined in concrete problems. For instance, if the input matrix is diagonalizable, then we show that $\alphaU=\mathbf{O}(\kappa_S)$ and $n=\mathbf{O}\left(\frac{\alpha_A}{\delta_A}\log\left(\frac{\kappa_S}{|\gamma_0|\epsilon}\right)\right)$, whereas the amplification ratio $\mathbf{O}\left(\frac{\kappa_S}{|\gamma_0|}\right)$ holds on average (with an additional $\log(n)$ for the worst-case input). This gives the query complexity
\begin{equation}
    \mathbf{O}\left(\frac{\kappa_S^2}{|\gamma_0|}\frac{\alpha_A}{\delta_A}\log^2\left(\frac{\kappa_S}{|\gamma_0|\epsilon}\right)
    \log\left(\frac{1}{\pfail}\right)
    \right).
\end{equation}
Thus when the input matrix $A$ is Hermitian, our result recovers the nearly optimal ground state preparation result~\cite{Lin2020nearoptimalground} up to a logarithmic factor. However, our algorithm is more general in that it applies to non-normal matrices with real eigenvalues whose ground states are still well defined.

\subsection{Eigenvalue processing over the complex plane}
\label{sec:intro_faber}

\begin{figure}[t]
	\centering
\includegraphics[width=0.75\textwidth]{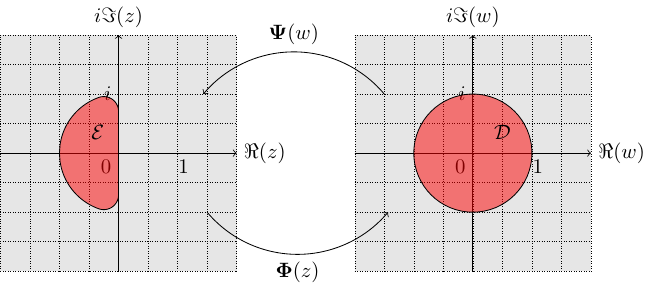}
\caption{Illustration of the unit disk $\mathcal{D}$, the target region $\mathcal{E}$, and the exterior Riemann mappings $\mathbf{\Psi}$, $\mathbf{\Phi}$ associated with the definition of Faber polynomials.}
\label{fig:faber_def}
\end{figure}

To simplify the analysis of algorithms, we have so far focused on the case where input matrices have only real eigenvalues from $[-1,1]$. In this case, we develop techniques to efficiently generate the Chebyshev history states, and our result is comparable to previous results for processing singular values. However, our techniques are applicable to more general matrices whose eigenvalues are enclosed by regions in the complex plane, thereby providing a quantum linear algebraic framework far more versatile than QSVT.

The core idea behind this generalization is to create the polynomial basis for a nearly-best uniform approximation over the target eigenvalue enclosing region. This is in turn achieved by passage from known polynomial basis for the real interval $[-1,1]$ or the unit disk $\mathcal{D}$, with the help of conformal maps. To be specific, consider a compact region $\mathcal{E}$ that includes eigenvalues of the input matrix. Under reasonable mathematical assumptions, there exists a unique conformal map
\begin{equation}
    \mathbf{\Phi}:\mathcal{E}^c\rightarrow\mathcal{D}^c,\qquad\mathbf{\Phi}(z)=w,
\end{equation}
known as the \emph{exterior Riemann map}, that sends the complement of $\mathcal{E}$ conformally onto the exterior of the unit disk $\mathcal{D}=\{|w|\leq1\}$ and satisfies $\mathbf{\Phi}(\infty)=\infty$, $\mathbf{\Phi}'(\infty)=\lim_{z\rightarrow\infty}\frac{\mathbf{\Phi}(z)}{z}=\zeta>0$,
with inverse
\begin{equation}
    \mathbf{\Psi}:\mathcal{D}^c\rightarrow\mathcal{E}^c,\qquad\mathbf{\Psi}(w)=z,
\end{equation}
where complement is taken with respect to the extended complex plane $\mathbb{C}\cup\{\infty\}$.
This implies that $\mathbf{\Phi}$ has a Laurent expansion in some neighborhood of $\infty$ as $\mathbf{\Phi}(z)=\zeta z+\zeta_0+\frac{\zeta_1}{z}+\frac{\zeta_2}{z^2}+\cdots$ and for the same reason $\mathbf{\Psi}(w)=\varsigma w+\varsigma_0+\frac{\varsigma_1}{w}+\frac{\varsigma_2}{w^2}+\cdots$. Then the $j$th Faber polynomial $\mathbf{F}_j(z)$ for the region $\mathcal{E}$ is defined as the polynomial part of the Laurent series of $\mathbf{\Phi}^j(z)$. See \fig{faber_def} for an illustration of regions and conformal maps relevant to the definition of Faber polynomials.

Faber polynomials provide a general methodology for constructing polynomial expansions that encompass Chebyshev and Taylor expansions as special cases. For instance, using the Joukowsky map $\mathbf{\Psi}(w)=\frac{w+w^{-1}}{2}$ and its inverse $\mathbf{\Phi}(z)=z+\sqrt{z^2-1}$, one can re-express Chebyshev polynomials $\mathbf{T}_j(x)=\cos(j\arccos(x))=\mathbf{F}_j(x)$ as Faber polynomials over the real interval $[-1,1]$. Similarly, using the affine map $\mathbf{\Phi}(z)=\frac{z-z_0}{\rho}$, one can identify power functions $\frac{(z-z_0)^j}{\rho^j}=\mathbf{F}_j(z)$ as Faber polynomials for the disk $\{|z-z_0|\leq\rho\}$. The significance of Faber polynomials however is that they provide a nearly-best uniform approximation of functions over the complex plane~\cite[Page 190]{suetin1998series}. Thus, by generating Faber polynomials in quantum superposition, one expects to obtain efficient quantum algorithms for matrices with complex eigenvalues, going beyond the Chebyshev-based algorithms discussed above.

To realize this idea, we use a matrix version of the Faber generating function
\begin{equation}
\sum_{j=0}^{n-1} L^j\otimes\mathbf{F}_j\left(\frac{A}{\alpha_A}\right)
=\sum_{j=0}^\infty L^j\otimes\mathbf{F}_j\left(\frac{A}{\alpha_A}\right)=\frac{\mathbf{\Psi}'(L^{-1})\otimes I}{L\mathbf{\Psi}(L^{-1})\otimes I-L\otimes \frac{A}{\alpha_A}},
\end{equation}
where $L$ is the $n$-by-$n$ lower shift matrix. Note that the Laurent series of $w\mathbf{\Psi}(w^{-1})$ is actually a power series and so $L\mathbf{\Psi}(L^{-1})$ is well defined, despite the fact that $L$ itself is not invertible. Similar to the Chebyshev case, we aim to bundle the numerator with a subroutine that prepares the Faber coefficients, and invert the denominator using a quantum linear system solver. The technical challenge, however, is that we need to implement operators such as $\mathbf{\Psi}'(L^{-1})$ and $L\mathbf{\Psi}(L^{-1})$ through efficient block encodings. We overcome this by using the Fourier expansions
\begin{equation}
    \mathbf{\Psi}'(e^{i\omega})=\varsigma -\varsigma_1e^{-2i\omega}-2\varsigma_2e^{-3i\omega}+\cdots,\qquad
	e^{-i\omega}\mathbf{\Psi}(e^{i\omega})=\varsigma +\varsigma_0e^{-i\omega}+\varsigma_1e^{-2i\omega}+\varsigma_2e^{-3i\omega}+\cdots
\end{equation}
This allows us to invoke \thm{fourier_coeff} to efficiently generate the Fourier coefficients, thereby producing the desired block encoding. We summarize our result for generating the Faber history states in \thm{faber_history}, which will be previewed below and further discussed in \sec{faber}.

\begin{customthm}{\ref{thm:faber_history}}[Faber history state generation]
Let $A$ be a square matrix such that $A/\alpha_A$ is block encoded by $O_A$ with some normalization factor $\alpha_A\geq\norm{A}$. Suppose that eigenvalues of $A/\alpha_A$ are enclosed by a Faber region $\mathcal{E}$ with associated conformal maps $\mathbf{\Phi}:\mathcal{E}^c\rightarrow\mathcal{D}^c$, $\mathbf{\Psi}:\mathcal{D}^c\rightarrow\mathcal{E}^c$ and Faber polynomials $\mathbf{F}_n(z)$.
Let $O_\psi\ket{0}=\ket{\psi}$ be the oracle preparing the initial state, and $O_{\beta}\ket{0}\propto\mathbf{\Psi}'(L_n^{-1})\sum_{k=0}^{n-1}\beta_k\ket{n-1-k}$ be the oracle preparing the shifting of coefficients $\beta$.
Then, the quantum state proportional to
	\begin{equation*}
  \ket{0}\sum_{l=0}^{n-1}\ket{l}
			\sum_{k=n-1-l}^{n-1}{\beta}_k{\mathbf{F}}_{k+l-n+1}\left(\frac{A}{\alpha_A}\right)\ket{\psi}
			+\sum_{s=1}^{\eta}\ket{s}\sum_{l=0}^{n-1}\ket{l}
			\sum_{k=0}^{n-1}{\beta}_k{\mathbf{F}}_{k}\left(\frac{A}{\alpha_A}\right)\ket{\psi}
	\end{equation*}
can be prepared with accuracy $\epsilon$ and probability $1-\pfail$ using
\begin{equation*}
    \mathbf{O}\left(\alphaFP n(\eta+1)\log\left(\frac{1}{\epsilon}\right)\log\left(\frac{1}{\pfail}\right)\right)
\end{equation*}
queries to controlled-$O_A$, controlled-$O_\psi$, controlled $O_{\widetilde\beta}$, and their inverses, where
\begin{equation*}
    \alphaFP\geq\max_{j=1,\ldots,n}\norm{\frac{\mathbf{F}_j'\left(\frac{A}{\alpha_A}\right)}{j}}
\end{equation*}
is an upper bound on the derivative of Faber polynomials.
\end{customthm}

With an additional fixed-point amplitude amplification, we obtain the Faber-based quantum eigenvalue transformation algorithm in \thm{qevt_faber}, previewed below.
As an application, we develop a quantum differential equation algorithm in \thm{diff_eq_faber} for general coefficient matrices by implementing Faber polynomials over a compact set (such as the one shown in \fig{region_general2}) on the left half of the complex plane. We also show in \thm{qeve_extreme} how to estimate \emph{leading eigenvalues} (eigenvalues of maximum absolute value) by directly implementing the Taylor expansion.
However, the success probability of our algorithm would decay drastically when applied to non-leading eigenvalues, which is partially addressed by the more recent method from~\cite{alase2024resolvent}.

\begin{customthm}{\ref{thm:qevt_faber}}[Quantum eigenvalue transformation, Faber version]
Let $A$ be a square matrix such that $A/\alpha_A$ is block encoded by $O_A$ with some normalization factor $\alpha_A\geq\norm{A}$. Suppose that eigenvalues of $A/\alpha_A$ are enclosed by a Faber region $\mathcal{E}$ with associated conformal maps $\mathbf{\Phi}:\mathcal{E}^c\rightarrow\mathcal{D}^c$, $\mathbf{\Psi}:\mathcal{D}^c\rightarrow\mathcal{E}^c$ and Faber polynomials $\mathbf{F}_n(z)$.
Let $p(z)=\sum_{k=0}^{n-1}\beta_k{\mathbf{F}}_{k}(z)$ be the Faber expansion of a degree-($n-1$) polynomial $p$.
Let $O_\psi\ket{0}=\ket{\psi}$ be the oracle preparing the initial state, and $O_{\beta}\ket{0}\propto\mathbf{\Psi}'(L_n^{-1})\sum_{k=0}^{n-1}\beta_k\ket{n-1-k}$ be the oracle preparing the shifting of coefficients $\beta$.
Then, the quantum state
    \begin{equation*}
        \frac{p\left(\frac{A}{\alpha_A}\right)\ket{\psi}}{\norm{p\left(\frac{A}{\alpha_A}\right)\ket{\psi}}}
    \end{equation*}
can be prepared with accuracy $\epsilon$ and probability $1-\pfail$ using
\begin{equation*}
    \mathbf{O}\left(\frac{\alphaF}{\alphaPPsi}\alphaFP n\log\left(\frac{\alphaF}{\alphaPPsi\epsilon}\right)\log\left(\frac{1}{\pfail}\right)\right)
\end{equation*}
queries to controlled-$O_A$, controlled-$O_\psi$, controlled $O_{\widetilde\beta}$, and their inverses, where 
\begin{equation*}
    \alphaFP\geq\max_{j=1,\ldots,n}\norm{\frac{\mathbf{F}_j'\left(\frac{A}{\alpha_A}\right)}{j}},\quad
    \alphaF\geq\max_{l=0,1,\ldots,n-1}\norm{\sum_{k=l}^{n-1}\beta_k\mathbf{F}_{k-l}\left(\frac{A}{\alpha_A}\right)\ket{\psi}},\quad
    \alphaPPsi\leq\norm{p\left(\frac{A}{\alpha_A}\right)\ket{\psi}}
\end{equation*}
are upper bound on the derivative of Faber polynomials, upper bound on the shifted Faber partial sum and lower bound on the transformed state respectively.
\end{customthm}

\begin{table}[t]
    \centering
    \resizebox{\textwidth}{!}{%
    \begin{tabular}{c|c|c}
        Algorithm & Query complexity & Measure of non-normality \\
        \hline
        \begin{tabular}{c}
             Quantum eigenvalue estimator\\
             (\thm{qeve})
        \end{tabular} & $\mathbf{O}\left(\frac{\alpha_A\kappa_S}{\epsilon}\log\left(\frac{1}{\pfail}\right)\right)$ & Jordan condition number\\
        \begin{tabular}{c}
             Quantum eigenvalue transformer\\
             (\thm{qevt})
        \end{tabular} & $\mathbf{O}\left(
    \frac{\norm{p}_{\max,[-1,1]}\kappa_S^2}{\norm{p\left(\frac{A}{\alpha_A}\right)\ket{\psi}}}
    n\log\left(
    \frac{\norm{p}_{\max,[-1,1]}\kappa_S}{\norm{p\left(\frac{A}{\alpha_A}\right)\ket{\psi}}\epsilon}\right)\log\left(\frac{1}{\pfail}\right)\right)$ & Jordan condition number\\
        \begin{tabular}{c}
             Quantum differential equation solver\\
             (\thm{diff_eq})
        \end{tabular} & $\mathbf{O}\left(\kappa_S^2\left(\alpha_At+\log\left(\frac{\kappa_S}{\epsilon}\right)\right)\log\left(\frac{\kappa_S}{\epsilon}\right)\log\left(\frac{1}{\pfail}\right)\right)$ & Jordan condition number\\
        \begin{tabular}{c}
             Quantum ground state preparator\\
             (\thm{ground})
        \end{tabular}
        & $\mathbf{O}\left(\frac{\kappa_S^2}{|\gamma_0|}\frac{\alpha_A}{\delta_A}\log^2\left(\frac{\kappa_S}{|\gamma_0|\epsilon}\right)
    \log\left(\frac{1}{\pfail}\right)
    \right)$ & Jordan condition number\\
    \begin{tabular}{c}
         Quantum eigenvalue transformer\\
         (\thm{qevt_faber})
    \end{tabular}
    & $\mathbf{O}\left(
    \frac{\norm{p}_{\max,\partial\mathcal{E}}}{\norm{p\left(\frac{A}{\alpha_A}\right)\ket{\psi}}}
    n\log\left(
    \frac{\norm{p}_{\max,\partial\mathcal{E}}}{\norm{p\left(\frac{A}{\alpha_A}\right)\ket{\psi}}\epsilon}\right)\log\left(\frac{1}{\pfail}\right)\right)$ & Numerical range/pseudospectrum\\
    \begin{tabular}{c}
         Quantum differential equation solver\\
         (\thm{diff_eq_faber})
    \end{tabular}
    & $\mathbf{O}\left(\frac{\alpha_At}{\norm{e^{tA}\ket{\psi}}}\polylog\left(\frac{\alpha_At}{\norm{e^{tA}\ket{\psi}}\epsilon}\right)\log\left(\frac{1}{\pfail}\right)\right)$ & Numerical range/pseudospectrum\\
    \begin{tabular}{c}
         Quantum eigenvalue estimator\\
         (\thm{qeve_extreme})
    \end{tabular}
    & $\mathbf{O}\left(\frac{\alpha_A}{\lambda_{\max}\epsilon}\kappa_S\log\left(\frac{1}{\pfail}\right)\right)$ & Jordan condition number\\
    \end{tabular}
    }
    \caption{Summary of common measures of non-normality and the corresponding complexity of quantum eigenvalue processing algorithms. The Jordan condition number, introduced in \sec{prelim_matrix}, is a commonly used measure of non-normality in numerical linear algebra~\cite[Page 444]{Trefethen05}. However, this measure is not suitable for problems with ill-conditioned Jordan basis. Alternatively, we can apply Faber approximations over the numerical range/pseudospectrum of the input matrix, to be formalized in \append{analysis_faber_crouzeix} and \append{analysis_faber_pseudospectrum} respectively, which leads to query complexity independent of the Jordan condition number. 
    See the relevant theorem statements for definitions of the remaining scaling parameters.}
    \label{tab:nonnormal}
\end{table}

For the purpose of generality, we will express the complexity of Faber-based algorithms using upper bounds like
\begin{equation}
\label{eq:alpha_f}
    \alphaFP\geq\max_{j=1,\ldots,n}\norm{\frac{\mathbf{F}_j'\left(\frac{A}{\alpha_A}\right)}{j}},\qquad
    \alphaF\geq\max_{l=0,1,\ldots,n-1}\norm{\sum_{k=l}^{n-1}\beta_k\mathbf{F}_{k-l}\left(\frac{A}{\alpha_A}\right)\ket{\psi}}.
\end{equation}
Similar to $\alphaU$ and $\alphaT$ in the Chebyshev case, these parameters can be further upper bounded when the Faber algorithms are applied to a concrete problem. Specifically, we consider the case where the target region $\mathcal{E}$ contains the \emph{numerical range} of the input matrix, which generalizes previous bounds for differential equation solvers based on the \emph{numerical abscissa}~\cite[Section 3.1]{Krovi2023improvedquantum}. A similar bound can be obtained when the \emph{pseudospectrum} of the input matrix is enclosed by $\mathcal{E}$. 
Intuitively, one can view ``numerical range'' and ``pseudospectrum'' as two common relaxations of the  notion of eigenvalues. These relaxations provide useful tools for bounding the size of matrix functions.
Indeed, as will be explained in \append{analysis_faber}, our resulting bounds are independent of the Jordan condition number $\kappa_S$ and therefore do not suffer from an ill-conditioned Jordan basis transformation that can arise in the Chebyshev case. See \tab{nonnormal} for a summary of common treatments of non-normality of input matrices and the corresponding complexity of quantum eigenvalue processing algorithms.

We summarize in \sec{prelim} preliminaries required to understand our results, and include in \sec{discuss} a brief summary of our work and a collection of questions for future work.

\section{Preliminaries}
\label{sec:prelim}
In this section, we present prerequisites that are necessary to understand our results on eigenvalue processing. We begin in \sec{prelim_notation} with an introduction of notation and terminology to be used throughout this paper. The next two subsections, \sec{prelim_cheby_fourier} and \sec{prelim_matrix}, summarize background material from mathematical analysis and linear algebra respectively. We also formally define in \sec{prelim_matrix} the \emph{eigenvalue transformation} and the \emph{singular value transformation} of a square matrix. A reader who has a sufficient background on these topics may proceed to \sec{prelim_block}, where we introduce the block encoding framework of quantum algorithms, on which our results are based.

\subsection{Notation and terminology}
\label{sec:prelim_notation}

We use lowercase Latin and Greek alphabets to represent vectors defined on a discrete set and functions defined on a continuum. For instance, we often write $\beta=\begin{bmatrix}
    \beta_0 & \beta_1 & \cdots
\end{bmatrix}$ to represent coefficients of a polynomial expansion and $f:[-1,1]\rightarrow\mathbb{C}$ to denote the target function for the eigenvalue transformation. To quantify the size of such vectors and functions, we use the \emph{$\ell_1$-norm} $\norm{\beta}_1=\sum_{j=0}^\infty\abs{\beta_j}$ and $\mathcal{L}_1$-norm $\norm{f}_{1,[-1,1]}=\int_{-1}^{1}\mathrm{d}x \abs{f(x)}$, the \emph{Euclidean norm} $\norm{\beta}=\sqrt{\sum_{j=0}^\infty\abs{\beta_j}^2}$ and $\mathcal{L}_2$-norm $\norm{f}_{2,[-1,1]}=\sqrt{\int_{-1}^{1}\mathrm{d}x \abs{f(x)}^2}$, and the \emph{max-norm} $\norm{\beta}_{\max}=\sup_{j=0,1,\ldots}\abs{\beta_j}$ and $\mathcal{L}_\infty$-norm $\norm{f}_{\max,[-1,1]}=\sup_{x\in[-1,1]}\abs{f(x)}$, dropping the underlying discrete set from the subscript of a vector norm if it is clear from the context. We will occasionally compute norms for anonymous functions. For instance, when writing $\norm{\sum_{j=0}^{\infty}\beta_je^{-ij(\cdot)}}_{2,[-\pi,\pi]}$, we are computing $\mathcal{L}_2$-norm $\sqrt{\int_{-\pi}^{\pi}\mathrm{d}\omega\abs{g(\omega)}^2}$ of the function $g(\omega)=\sum_{j=0}^{\infty}\beta_je^{-ij\omega}$.
We use the \emph{Dirac notation} $\ket{\psi}$ to denote a vector only when it has unit length $\norm{\ket{\psi}}=1$ with respect to the Euclidean norm. 

On the other hand, matrices and operators are represented by uppercase Latin and Greek letters in our paper, and their size is typically quantified by the \emph{operator norm} $\norm{\cdot}$ (also known as the \emph{spectral norm}). This notation of the operator norm is compatible with that of the Euclidean norm, as their values coincide for a vector when treated as a mapping from a one-dimensional space. Example notation of matrices includes: the input matrix of a quantum algorithm $A$ , the diagonal factor of an eigendecomposition $\Lambda$, the Jordan form of a matrix $J$ and the corresponding basis transformation $S$, the diagonal factor of a singular value decomposition $\Sigma$, the identity matrix $I$, and the lower shift matrix $L$ with $1$ on the first lower diagonal and $0$ elsewhere. When necessary, we will use subscripts to represent dimensions of the matrix. For instance, we have $L_n=\sum_{k=0}^{n-2}\ketbra{k+1}{k}$.

We can obtain a block matrix $B$ by stacking up submatrices $B_{jk}$:
\begin{equation}
    B=
    \begin{bmatrix}
        B_{11} & \cdots & B_{1n}\\
        \vdots & \vdots & \vdots\\
        B_{n1} & \cdots & B_{nn}
    \end{bmatrix}.
\end{equation}
We require that partition of the rows is the same as that of the columns, so that the block structure is respected under matrix multiplication. In the case where all $B_{jk}\in\mathbb{C}^{d\times d}$ are $d$-by-$d$ matrices, we can treat $B\in\mathbb{C}^{n\times n}\otimes\mathbb{C}^{d\times d}$ as acting on the tensor product space $\mathbb{C}^{n}\otimes\mathbb{C}^{d}$, writing
\begin{equation}
    B=
    \begin{bmatrix}
        B_{11} & \cdots & B_{1n}\\
        \vdots & \vdots & \vdots\\
        B_{n1} & \cdots & B_{nn}
    \end{bmatrix}
    =\sum_{j,k=1}^n\ketbra{j}{k}\otimes B_{jk}.
\end{equation}
We will slightly abuse the notation and sometimes use the above representation even when not all $B_{jk}$ are of the same size.
We give a bound in the following lemma on the spectral norm of a block matrix, using a generalized notion of max row and column sum norms defined for block matrices.

\begin{lemma}[Spectral norm bound for block matrices]
\label{lem:block_norm}
    For $d_j$-by-$d_k$ matrices $B_{jk}\in\mathbb{C}^{d_j\times d_k}$,
    \begin{equation}
        \norm{\begin{bmatrix}
        B_{11} & \cdots & B_{1n}\\
        \vdots & \vdots & \vdots\\
        B_{n1} & \cdots & B_{nn}
        \end{bmatrix}}
        \leq
        \sqrt{\max_{1\leq k\leq n}\sum_{j=1}^n\norm{B_{jk}}}
        \sqrt{\max_{1\leq j\leq n}\sum_{k=1}^n\norm{B_{jk}}}.
    \end{equation}
\end{lemma}
\begin{proof}
    The claimed bound follows from~\cite[5.6.P21]{horn2012matrix} and a direct verification that
    \begin{equation}
        \max_{1\leq k\leq n}\sum_{j=1}^n\norm{B_{jk}},\qquad
        \max_{1\leq j\leq n}\sum_{k=1}^n\norm{B_{jk}}
    \end{equation}
    are valid matrix norms for matrices with the underlying block structure~\cite[5.6.P55]{horn2012matrix}.
\end{proof}

We use boldface symbols to denote functions and operations having specific meanings. For instance, we write $\mathbf{T}_j(x)$ for \emph{Chebyshev polynomials of the first kind} ($\widetilde{\mathbf{T}}_j(x)$ for their rescaled version), $\mathbf{U}_j(x)$ for \emph{Chebyshev polynomials of the second kind}, $\mathbf{J}_j(t)$ for \emph{Bessel functions of the first kind}, $\mathbf{I}_j(k)$ for \emph{modified Bessel functions of the first kind}, $\mathbf{Erf}(x)=\frac{2}{\sqrt{\pi}}\int_{0}^{x}\mathrm{d}y\ e^{-y^2}$ for the \emph{error function}, and $\mathbf{F}_j(z)$ for \emph{Faber polynomials}, to be introduced in \sec{prelim_cheby_fourier} and \sec{faber_prelim} respectively.
In analyzing the Fourier expansion, we write $\mathbf{H}(\cdot)$ for the \emph{Hilbert transform}, 
$\mathbf{S}_{(-n,n)}(\cdot)$ for the $n$th partial sum of the Fourier expansion, and $\mathbf{S}_{*}(\cdot)$ for the \emph{Fourier maximal function}, to be defined in \sec{prelim_cheby_fourier}.
We denote the rank of an operator by $\mathbf{Rank}(\cdot)$, kernel of an operator by $\mathbf{Ker}(\cdot)$, and image of an operator by $\mathbf{Im}(\cdot)$, to be used in \sec{prelim_matrix} and \sec{prelim_block}.
We use $\mathbf{Floor}(\cdot)$ to denote the largest integer not exceeding a real number, and $\mathbf{CMod}_q(x)=x-q\ \mathbf{Floor}\left(\frac{x+\frac{q}{2}}{q}\right)\in[-\frac{q}{2},\frac{q}{2})$ to denote the centered modulus of $x$ modulo $q$ (see \sec{est_cmod}), reserving $\mathbf{Mod}_q(\cdot)\in[0,q)$ for the regular modulo-$q$ operation.
We write $\mathbf{O}(\cdot)$ or $\lesssim$ to mean asymptotically less than, $\mathbf{\Omega}(\cdot)$ or $\gtrsim$ to mean asymptotically more than, and $\mathbf{\Theta}(\cdot)$ or $\sim$ to represent quantities having the same asymptotic scaling, using $\mathbf{o}(1)$ to denote a positive number that approaches zero as some parameter grows.

Finally, we use calligraphic uppercase letters to represent sets. For instance, we denote the set of $q$th-power integrable functions over $[-1,1]$ by $\mathcal{L}_q([-1,1])$ (\sec{prelim_cheby_fourier}), a finite-dimensional Hilbert space by $\mathcal{G}$ or $\mathcal{H}$ (\sec{prelim_block}), the ball centered at $x$ with radius $\delta$ by $\mathcal{D}(x,\delta)$ (the unit disk being $\mathcal{D}$), the target region enclosing eigenvalues of the input operator by $\mathcal{E}$ (\sec{faber_prelim}), the numerical range of an operator by $\mathcal{W}(\cdot)$ (\append{analysis_faber_crouzeix}), and the $\delta$-pseudospectrum by $\mathcal{S}_\delta(\cdot)$ (\append{analysis_faber_pseudospectrum}). We denote the number of elements in a finite set $\mathcal{S}$ by $\#\mathcal{S}$, introduce the indicator function $\mathbf{Ind}_{\mathcal{S}}(\cdot)$ for a given set $\mathcal{S}$, and define the distance of subsets $\mathcal{E}_1$ and $\mathcal{E}_2$ in a Hilbert space by $\mathbf{Dist}(\mathcal{E}_1,\mathcal{E}_2)=\inf_{x_1\in\mathcal{E}_1,x_2\in\mathcal{E}_2}\norm{x_1-x_2}$. We adopt standard notations for number systems, writing $\mathbb Z$ for integers, $\mathbb{Q}$ for rational numbers, $\mathbb{R}$ for real numbers, and $\mathbb{C}$ for complex numbers, and we use superscripts like $\mathbb{C}^d$ and $\mathbb{C}^{d\times d}$ to denote sets of vectors and matrices of the given dimensions.

\subsection{Chebyshev and Fourier expansions}
\label{sec:prelim_cheby_fourier}

We begin this subsection by motivating the use of Chebyshev and Fourier expansions. Let $A$ be a diagonalizable matrix with only real eigenvalues, i.e., $A=S\Lambda S^{-1}$ where $\Lambda$ is a real diagonal matrix. Given a real analytic function $\widetilde{f}:\mathbb{R}\rightarrow\mathbb{C}$ and quantum state $\ket{\psi}$, the goal of QEVT is to produce a state that approximates $\frac{\widetilde{f}(A)\ket{\psi}}{\norm{\widetilde{f}(A)\ket{\psi}}}$.
For the purpose of developing quantum algorithms, however, the input matrix $A$ needs to be properly normalized as $A/\alpha_A$ with $\alpha_A\geq\norm{A}$, before it can be accessed by a quantum computer (we will make this point clearer in \sec{prelim_block}). Therefore, the problem of eigenvalue transformation should be reformulated as applying the rescaled function $f(\cdot)=\widetilde{f}(\alpha_A(\cdot))$ to the block encoded operator $A/\alpha_A$, approximately producing the state
\begin{equation}
    \frac{f\left(\frac{A}{\alpha_A}\right)\ket{\psi}}{\norm{f\left(\frac{A}{\alpha_A}\right)\ket{\psi}}}.
\end{equation}

For most problems of interest, we cannot directly implement the target function on a quantum computer. Instead, we aim to implement a polynomial $p(x)$ that approximates the rescaled function $f(\cdot)=\widetilde{f}(\alpha_A(\cdot))$. The above discussion suggests that this approximation should be made over the real interval $[-1,1]$. In this case, Chebyshev polynomials provide a nearly optimal solution to the uniform polynomial approximation problem, which we briefly review in the following.

We define the \emph{Chebyshev polynomials of the first kind} over the real interval $[-1,1]$ as
\begin{equation}
\label{eq:cheby}
    \mathbf{T}_j(x):=\cos(j\arccos(x)).
\end{equation}
Note that by setting $\theta=\arccos(x)$, $\cos(j\theta)=\sum_{\substack{0\leq k\leq j\\k\text{ is even}}}\binom{j}{k}i^k\cos^{j-k}(\theta)\left(1-\cos^2(\theta)\right)^{k/2}$, which implies $\mathbf{T}_j(x)=\sum_{\substack{0\leq k\leq j\\k\text{ is even}}}\binom{j}{k}i^kx^{j-k}\left(1-x^2\right)^{k/2}$.
Thus Chebyshev polynomials defined above are indeed polynomial functions, and the definition can be extended naturally to all $\mathbb{R}$.
Now consider the power series $\sum_{j=0}^\infty \mathbf{T}_j(x)y^j$ generated by the Chebyshev polynomials. Assuming $\abs{y}<1$, 
we use the substitution $\theta=\arccos(x)$ again to get the generating function
\begin{equation}
\begin{aligned}
    \sum_{j=0}^\infty \mathbf{T}_j(x)y^j
    =\sum_{j=0}^\infty \cos(j\theta)y^j
    =\sum_{j=0}^\infty \frac{e^{ij\theta}+e^{-ij\theta}}{2}y^j
    =\frac{\frac{1}{1-ye^{i\theta}}+\frac{1}{1-ye^{-i\theta}}}{2}
    =\frac{1-yx}{1+y^2-2yx}.
\end{aligned}
\end{equation}
It is sometimes convenient to rescale the first Chebyshev polynomial by a factor of $\frac{1}{2}$:
\begin{equation}
\label{eq:cheby_scale}
    \widetilde{\mathbf{T}}_j(x)=
    \begin{cases}
        \mathbf{T}_j(x),\quad &j\geq1,\\
        \frac{1}{2}\mathbf{T}_0(x),&j=0.
    \end{cases}
\end{equation}
In this case, we have the alternative generating function
\begin{equation}
\label{eq:cheby_gen}
    \sum_{j=0}^\infty \widetilde{\mathbf{T}}_j(x)y^j
    =\frac{1}{2}+\sum_{j=1}^\infty \mathbf{T}_j(x)y^j
    =\frac{1-y^2}{2(1+y^2-2yx)}.
\end{equation}
We also define \emph{Chebyshev polynomials of the second kind} as
\begin{equation}
    \mathbf{U}_j(x):=\frac{\sin((j+1)\arccos(x))}{\sin(\arccos(x))},
\end{equation}
which has the generating function for $\abs{y}<1$
\begin{equation}
\label{eq:cheby_gen2}
    \sum_{j=0}^\infty\mathbf{U}_j(x)y^j=\frac{1}{1+y^2-2yx}.
\end{equation}
The two kinds of Chebyshev polynomials are related as the solutions of the Pell equation $\mathbf{T}_j^2(x)-(x^2-1)\mathbf{U}_{j-1}^2(x)=1$.

Chebyshev polynomials are orthogonal in the sense that
\begin{equation}
    \int_{-1}^{1}\mathbf{T}_j(x)\mathbf{T}_k(x)\frac{\mathrm{d}x}{\sqrt{1-x^2}}=
    \begin{cases}
        0,\quad& j\neq k,\\
        \frac{\pi}{2},& j=k\neq 0,\\
        \pi,&j=k=0.
    \end{cases}
\end{equation}
Thus, if a function has a Chebyshev expansion, the expansion coefficients must be uniquely determined:
\begin{equation}
    f(x)=\sum_{j=0}^\infty\beta_j\mathbf{T}_j(x)\quad\Rightarrow\quad
    \beta_j=
    \begin{cases}
        \frac{2}{\pi}\int_{-1}^{1}\mathbf{T}_j(x)f(x)\frac{\mathrm{d}x}{\sqrt{1-x^2}},\quad&j\geq1,\\
        \frac{1}{\pi}\int_{-1}^{1}\mathbf{T}_0(x)f(x)\frac{\mathrm{d}x}{\sqrt{1-x^2}},\quad&j=0.\\
    \end{cases}
\end{equation}
For notational convenience, we sometimes rescale the first coefficient and define
\begin{equation}
    \widetilde{\beta}_j=
    \begin{cases}
        \beta_j,\quad &j\geq1,\\
        2\beta_0,&j=0,
    \end{cases}
\end{equation}
so that we can rewrite the expansion as
\begin{equation}
    f(x)=\sum_{j=0}^\infty\beta_j\mathbf{T}_j(x)=\sum_{j=0}^\infty\widetilde\beta_j\widetilde{\mathbf{T}}_j(x).
\end{equation}

As aforementioned, Chebyshev expansion provides a nearly best uniform polynomial approximation of functions over the real interval $[-1,1]$. Specifically, for a continuous function with the Chebyshev expansion $f(x)=\sum_{j=0}^\infty\beta_j\mathbf{T}_j(x)$, the maximum truncation error $\norm{f-\sum_{j=0}^{n-1}\beta_j\mathbf{T}_j}_{\max,[-1,1]}$ is larger by a factor at most $4+\frac{4}{\pi^2}\ln(n-1)$ than the error achieved by the (unique) best degree-($n-1$) polynomial~\cite{Trefethen2013approximation}.
Moreover, the error of approximating a Lipschitz function $f$ by the first $n$ terms of its Chebyshev expansion decreases rapidly with $n$. For differentiable functions where $f,f',\ldots,f^{(\nu-1)}$ are absolutely continuous with a uniformly bounded variation, the error decays polynomially like $\norm{f-\sum_{j=0}^{n-1}\beta_j\mathbf{T}_j}=\mathbf{O}(n^{-\nu})$. For functions $f$ analytic in $[-1, 1]$ that are analytically continuable to the  ellipse with foci $\pm1$ and major and minor semiaxis lengths summing to $\rho>1$, the approximation error decays geometrically like $\norm{f-\sum_{j=0}^{n-1}\beta_j\mathbf{T}_j}=\mathbf{O}(\rho^{-n})$. Entire analytic functions may converge even faster. Here, we bound the error of truncating Chebyshev expansion of the exponential function and the error function, which will be useful in analyzing the quantum differential equation algorithm and the ground state preparation algorithm respectively.

\begin{proposition}[Chebyshev expansion of exponential function {\cite{Low2016HamSim}}]
\label{prop:trunc_exp}
    Given $\tau>0$, the complex exponential function $e^{-i\tau x}$ has the Chebyshev expansion $e^{-i\tau x}=\sum_{j=0}^\infty\beta_j\mathbf{T}_j(x)$, where
    \begin{equation}
        \beta_j=
        \begin{cases}
            2i^j\mathbf{J}_j(\tau),\quad&j\geq 1,\\
            \mathbf{J}_0(\tau),&j=0,
        \end{cases}
    \end{equation}
    with $\mathbf{J}_j(t)$ Bessel functions of the first kind. Truncated at order $n$,
    \begin{equation}
        \norm{e^{-i\tau (\cdot)}-\sum_{j=0}^{n-1}\beta_j\mathbf{T}_j(\cdot)}_{\max,[-1,1]}
        =\mathbf{O}\left(\left(\frac{e\tau}{2n}\right)^n\right).
    \end{equation}
\end{proposition}

\begin{proposition}[Chebyshev expansion of error function {\cite{Low17,Wan22}}]
\label{prop:trunc_erf}
    Given $c>0$, the error function $\mathbf{Erf}(cx)
    =\frac{2}{\sqrt{\pi}}\int_{0}^{cx}\mathrm{d}u\ e^{-u^2}$ has the Chebyshev expansion $\mathbf{Erf}(cx)=\sum_{j=0}^\infty\beta_j\mathbf{T}_j(x)$, where
    \begin{equation}
        \beta_j=
        \begin{cases}
            \frac{2ce^{-\frac{c^2}{2}}}{j\sqrt{\pi}}(-1)^{\frac{j-1}{2}}\left(\mathbf{I}_{\frac{j-1}{2}}\left(\frac{c^2}{2}\right)+\mathbf{I}_{\frac{j+1}{2}}\left(\frac{c^2}{2}\right)\right),\quad&j\text{ is odd},\\
            0,&j\text{ is even},
        \end{cases}
    \end{equation}
    with $\mathbf{I}_j(c)$ modified Bessel functions of the first kind. Truncated at order $n$,
    \begin{equation}
        \norm{\mathbf{Erf}(c(\cdot))-\sum_{j=0}^{n-1}\beta_j\mathbf{T}_j(\cdot)}_{\max,[-1,1]}
        =\mathbf{O}\left(\frac{c}{n}\left(e^{-\frac{n^2}{2m}}+e^{-\frac{c^2}{2}}\left(\frac{ec^2}{2m}\right)^m\right)\right)
    \end{equation}
    for all $m=\mathbf{\Omega}(c^2)$ sufficiently large.
\end{proposition}

Now let us apply the substitution $\theta=\arccos(x)$ and re-express the Chebyshev expansion as
\begin{equation}
\label{eq:cheby_trig}
    f(\cos(\theta))=\sum_{j=0}^\infty\beta_j\cos(j\theta),
    \qquad
    \beta_j=
    \begin{cases}
        \frac{2}{\pi}\int_{0}^{\pi}\cos(j\theta)f(cos(\theta))\mathrm{d}\theta,\quad&j\geq1,\\
        \frac{1}{\pi}\int_{0}^{\pi}f(cos(\theta))\mathrm{d}\theta,\quad&j=0.\\
    \end{cases}
\end{equation}
We thus see that the Chebyshev expansion of $f(x)$ is closely related to the Fourier expansion of the even function $f(\cos(\theta))$, the latter of which has been extensively studied in previous literature. In what follows, we review a collection of results from Fourier analysis that are most relevant to our work, referring the reader to~\cite{allen2004signal,vetterli2014foundations} for a comprehensive treatment of this subject in the context of signal processing. We focus on the exponential form of the Fourier expansion for convenience, but it is straightforward to reformulate all the results in the trigonometric form.

Consider a function $g:[-\pi,\pi]\rightarrow\mathbb{C}$ and periodically extend its domain to the entire $\mathbb R$. Due to the orthogonality of $e^{-ij\omega}$ over the interval $[-\pi,\pi]$, if $g$ has a Fourier expansion, the expansion coefficients must be uniquely determined as:
\begin{equation}
    g(\omega)=\sum_{j=-\infty}^\infty\xi_je^{-ij\omega}
    \quad\Rightarrow\quad
    \xi_j=\frac{1}{2\pi}\int_{-\pi}^{\pi}\mathrm{d}\omega\ g(\omega)e^{ij\omega}.
\end{equation}
In engineering, the coefficients $\xi_j$ are often thought of as a signal in the time domain, and $g(\omega)$ is then the \emph{discrete time Fourier transform} of $\xi_j$ in the frequency domain. We do not explicitly use this terminology hereafter, but it is useful to have this time-frequency correspondence in mind when analyzing properties of the Fourier expansion. By the unitarity of the expansion, we have the Parseval-Plancherel identity:
\begin{equation}
    \sum_{j=-\infty}^\infty\abs{\xi_j}^2
    =\norm{\xi}^2
    =\frac{1}{2\pi}\norm{g}_{2,[-\pi,\pi]}^2
    =\frac{1}{2\pi}\int_{-\pi}^{\pi}\mathrm{d}\omega \abs{g(\omega)}^2.
\end{equation}

Supposing that $g(\omega)$ can be Fourier expanded with coefficients $\{\xi_j\}_{j=-\infty}^\infty$, we would like to find the function to which the shifted coefficients $\{\xi_{j-k}\}_{j=-\infty}^\infty$ correspond. A simple calculation yields that
\begin{equation}
    \sum_{j=-\infty}^\infty\xi_{j-k}e^{-ij\omega}
    =\sum_{j=-\infty}^\infty\xi_{j}e^{-i(j+k)\omega}
    =e^{-ik\omega}g(\omega).
\end{equation}
Thus a shift in the time domain corresponds to a phase shift in the frequency domain. Additionally, we may take the derivative in the frequency domain when appropriate, which yields the correspondence
\begin{equation}
    \sum_{j=-\infty}^\infty j\xi_{j}e^{-ij\omega}
    =i\frac{\mathrm{d}}{\mathrm{d}\omega}g(\omega).
\end{equation}
Now suppose that functions $g(\omega)$ and $h(\omega)$ can be Fourier expanded as $g(\omega)=\sum_{j=-\infty}^\infty \xi_je^{-ij\omega}$ and $h(\omega)=\sum_{j=-\infty}^\infty \zeta_je^{-ij\omega}$, respectively. Then the function corresponding to the pointwise product coefficients $\{\xi_j\zeta_j\}_{j=-\infty}^\infty$ can be obtained from the \emph{frequency domain convolution}:
\begin{equation}
    \sum_{j=-\infty}^\infty \xi_j\zeta_je^{-ij\omega}
    =\frac{1}{2\pi}\int_{-\pi}^{\pi}\mathrm{d}u\ g(u)h(\omega-u).
\end{equation}
The right-hand side of the above equation is sometimes known as the \emph{cyclic convolution} of periodic functions $g$ and $h$.

As an application of the frequency domain convolution, let us consider the one-sided Fourier expansion
\begin{equation}
    \widetilde{\mathbf{H}}(g)(\omega)=\sum_{j=0}^\infty \xi_je^{-ij\omega}
\end{equation}
and its relation to the two-sided Fourier expansion $g(\omega)=\sum_{j=-\infty}^\infty \xi_je^{-ij\omega}$. Note that the one-sided expansion can be obtained by multiplying the discrete Heaviside step function in the time domain, which corresponds to
\begin{equation}
    \sum_{j=0}^\infty e^{-ij\omega}
    =\frac{1}{1-e^{-i\omega}}
    +\pi\sum_{k=-\infty}^{\infty}\pmb{\delta}(\omega-2\pi k)
\end{equation}
in the frequency domain, where $\pmb{\delta}(\cdot)$ is the \emph{Dirac delta function}. Thus, by the frequency domain convolution theorem,
\begin{equation}
\begin{aligned}
    \widetilde{\mathbf{H}}(g)(\omega)
    &=\frac{1}{2\pi}\int_{-\pi}^{\pi}\mathrm{d}u\
    g(\omega-u)\left(
    \frac{1}{1-e^{-iu}}
    +\pi\sum_{k=-\infty}^{\infty}\pmb{\delta}(u-2\pi k)
    \right)\\
    &=\frac{1}{2\pi}\int_{-\pi}^{\pi}\mathrm{d}u\
    \frac{g(\omega-u)}{1-e^{-iu}}
    +\frac{1}{2}\int_{-\pi}^{\pi}\mathrm{d}u\
    g(\omega-u)\sum_{k=-\infty}^{\infty}\pmb{\delta}(u-2\pi k).
\end{aligned}
\end{equation}
For the second term, we have
\begin{equation}
\begin{aligned}
    \frac{1}{2}\int_{-\pi}^{\pi}\mathrm{d}u\
    g(\omega-u)\sum_{k=-\infty}^{\infty}\pmb{\delta}(u-2\pi k)
    &=\frac{1}{2}\sum_{k=-\infty}^{\infty}\int_{-\pi-2\pi k}^{\pi-2\pi k}\mathrm{d}u\
    g(\omega-2\pi k-u)\pmb{\delta}(u)\\
    &=\frac{1}{2}\int_{-\pi}^{\pi}\mathrm{d}u\
    g(\omega-u)\pmb{\delta}(u)
    =\frac{1}{2}g(\omega),
\end{aligned}
\end{equation}
whereas
\begin{equation}
    \frac{1}{1-e^{-iu}}
    =\frac{1}{2}-\frac{i}{2}\cot{\left(\frac{u}{2}\right)}
\end{equation}
for the first term.
So altogether, the one-sided expansion has the spectrum
\begin{equation}
    \widetilde{\mathbf{H}}(g)(\omega)
    =-\frac{i}{2}\mathbf{H}(g)(\omega)
    +\frac{1}{4\pi}\int_{-\pi}^{\pi}\mathrm{d}u\
    g(u)
    +\frac{1}{2}g(\omega)
\end{equation}
in the frequency domain, where
\begin{equation}
    \mathbf{H}(g)(\omega)=
    \frac{1}{2\pi}\int_{-\pi}^{\pi}\mathrm{d}u\
    g(u)\cot{\left(\frac{\omega-u}{2}\right)}
\end{equation}
is the \emph{Hilbert transform} of the $2\pi$-periodic function $g$.

The following Riesz inequality asserts that $\mathcal{L}_2$-norm does not increase under the Hilbert transform.
\begin{lemma}[Riesz inequality {\cite[(6.167)]{king_2009}}]
\label{lem:riesz}
    For $g\in \mathcal{L}_2([-\pi,\pi])$, its Hilbert transform
    \begin{equation}
        \mathbf{H}(g)(\omega)=
    \frac{1}{2\pi}\int_{-\pi}^{\pi}\mathrm{d}u\
    g(u)\cot{\left(\frac{\omega-u}{2}\right)}
    \end{equation}
    satisfies
    \begin{equation}
        \norm{\mathbf{H}(g)}_{2,[-\pi,\pi]}\leq\norm{g}_{2,[-\pi,\pi]},
    \end{equation}
    where $\norm{g}_{2,[-\pi,\pi]}=\sqrt{\int_{-\pi}^{\pi}\mathrm{d}\omega \abs{g(\omega)}^2}$.
\end{lemma}

Finally, we consider the growth of Fourier partial sum as a function of the truncate order. Specifically, given a Fourier expansion $g(\omega)=\sum_{j=-\infty}^\infty \xi_je^{-ij\omega}$ and $\mathcal{J}$ a collection of indices, we define $\mathbf{S}_{\mathcal{J}}(\omega)=\sum_{j\in\mathcal{J}}\xi_je^{-ij\omega}$. So for instance,
\begin{equation}
    \mathbf{S}_{(-n,n)}(g)(\omega)=\sum_{j=-n+1}^{n-1}\xi_je^{-ij\omega},\qquad
    \mathbf{S}_{[0,n)}(g)(\omega)=\sum_{j=0}^{n-1}\xi_je^{-ij\omega}.
\end{equation}
Note that these partial sums correspond to the pointwise multiplication of $\xi_j$ with rectangular window functions. Thus, by the frequency domain convolution theorem,
\begin{equation}
\begin{aligned}
    \mathbf{S}_{[0,n)}(g)(\omega)
    &=\frac{1}{2\pi}\int_{-\pi}^{\pi}\mathrm{d}u\
    g(\omega-u)
    \frac{1-e^{-in\omega}}{1-e^{-i\omega}}
    =\frac{1}{2\pi}\int_{-\pi}^{\pi}\mathrm{d}u\
    g(\omega-u)
    e^{-i\frac{(n-1)\omega}{2}}
    \frac{\sin\left(\frac{n\omega}{2}\right)}{\sin\left(\frac{\omega}{2}\right)},\\
    \mathbf{S}_{(-n,n)}(g)(\omega)
    &=\frac{1}{2\pi}\int_{-\pi}^{\pi}\mathrm{d}u\
    g(\omega-u)
    e^{i(n-1)\omega}\frac{1-e^{-i(2n-1)\omega}}{1-e^{-i\omega}}
    =\frac{1}{2\pi}\int_{-\pi}^{\pi}\mathrm{d}u\
    g(\omega-u)
    \frac{\sin\left(\frac{(2n-1)\omega}{2}\right)}{\sin\left(\frac{\omega}{2}\right)}.
\end{aligned}
\end{equation}
Note that while $\mathcal{L}_\infty$-norm of the \emph{Dirichlet kernel}
\begin{equation}
    \mathbf{D}_{n}(\omega)
    =\sum_{j=-n}^{n}e^{-ij\omega}
    =\frac{\sin\left((n+\frac{1}{2})\omega\right)}{\sin\left(\frac{\omega}{2}\right)}
\end{equation}
is given by
\begin{equation*}
    \norm{\mathbf{D}_{n}}_{\max,[-\pi,\pi]}
    =2n+1,
\end{equation*}
its $\mathcal{L}_1$-norm scales like
\begin{equation}
    \norm{\mathbf{D}_{n}}_{1,[-\pi,\pi]}
    =\frac{4}{\pi^2}\log\left(n\right)+\mathbf{O}(1).
\end{equation}
We thus see that in the worst case, the Fourier partial sums grow like $\sim\log n$ as well: in particular, if $g\in \mathcal{L}_\infty([-\pi,\pi])$, we have
\begin{equation}
\label{eq:worst_partial}
    \norm{\mathbf{S}_{[0,n)}(g)}_{\max,[-\pi,\pi]}
    ,\norm{\mathbf{S}_{(-n,n)}(g)}_{\max,[-\pi,\pi]}
    =\left(\frac{4}{\pi^2}\log n+\mathbf{O}(1)\right)\norm{g}_{\max,[-\pi,\pi]}.
\end{equation}
However, the scaling of the partial sum can be much smaller on average. In fact, we have the following highly nontrivial result due to Carleson and Hunt.

\begin{lemma}[{Carleson-Hunt theorem} {\cite[Theorem 12.8]{Reyna2002} and \cite{TaoBlog}}]
\label{lem:carleson_hunt}
Given $g\in \mathcal{L}_2([-\pi,\pi])$, its \emph{Fourier maximal operator}
\begin{equation}
    \mathbf{S}_*(g)(\omega)=\sup_{n=0,1,\ldots}\abs{\mathbf{S}_{(-n,n)}(g)(\omega)}.
\end{equation}
satisfies
\begin{equation}
    \norm{\mathbf{S}_*(g)}_{2,[-\pi,\pi]}
    \leq c\norm{g}_{2,[-\pi,\pi]}
\end{equation}
for some universal constant $c$, where $\norm{g}_{2,[-\pi,\pi]}=\sqrt{\int_{-\pi}^{\pi}\mathrm{d}\omega \abs{g(\omega)}^2}$.
\end{lemma}

\subsection{Matrix decompositions and transformations}
\label{sec:prelim_matrix}

We now introduce common classes of matrices and their decompositions relevant to our work, based on which we define transformations of these matrices. We refer the reader to~\cite{axler2023linear,roman2013advanced,horn2012matrix} for a more comprehensive coverage of matrix analysis and linear algebra.

We say a square matrix $A$ is \emph{diagonalizable} if $A=S\Lambda S^{-1}$ for some invertible matrix $S$ and diagonal matrix $\Lambda$.
In this case, $A=S\Lambda S^{-1}$ is called the \emph{eigendecomposition} of $A$.
Diagonalizable matrices arise in a variety of practical applications~\cite{Ashida20,McArdle20}, and they are especially easy to handle because, up to a change of basis, their actions are described by scalar multiplications by eigenvalues. But not all matrices are diagonalizable. For instance, if we are restricted to $\mathbb Q$ or $\mathbb{R}$, then matrices such as $\begin{bmatrix}
    0 & 1\\
    -1 & 0
\end{bmatrix}$ cannot be diagonalized due to the lack of rational or real eigenvalues. Even complex matrices can be non-diagonalizable if they fail to satisfy the following criterion.

\begin{proposition}[Eigendecomposition]
    The following statements hold for a matrix $A\in\mathbb C^{d\times d}$:
    \begin{enumerate}
        \item (Uniqueness of eigenvalues): The diagonal matrix 
        \begin{equation}
            \Lambda=
            \begin{bmatrix}
                \lambda_0 & & &\\
                & \lambda_1 & &\\
                & & \ddots &\\
                & & & \lambda_{d-1}
            \end{bmatrix}\in\mathbb{C}^{d\times d}
        \end{equation}
        satisfies $A=S\Lambda S^{-1}$ for some invertible matrix $S\in\mathbb{C}^{d\times d}$,
        if and only if
        \begin{equation}
        \label{eq:eigen_rank}
            \#\{\lambda_l\ |\ \lambda_l=\lambda\}
            =d-\mathbf{Rank}(A-\lambda I)
        \end{equation}
        for all $\lambda\in\mathbb C$.
        \item (Equivalence of eigenbasis): Invertible matrices $S,T\in\mathbb{C}^{d\times d}$ satisfy $A=S\Lambda S^{-1}=T\Lambda T^{-1}$ for some diagonal matrix
        \begin{equation}
            \Lambda=
            \begin{bmatrix}
                \lambda_0I_{d_0} & & &\\
                & \lambda_1I_{d_1} & &\\
                & & \ddots &\\
                & & & \lambda_{s-1}I_{d_{s-1}}
            \end{bmatrix}\in\mathbb{C}^{d\times d}
        \end{equation}
        with the same eigenvalue $\lambda_l$ grouped together as $\lambda_lI_{d_l}\in\mathbb{C}^{d_l\times d_l}$ ($\lambda_l$ are pairwise distinct),
        if and only if $S$ and $T$ are related by
        \begin{equation}
        \label{eq:eigen_unique}
            T=S
            \begin{bmatrix}
                R_0 & & &\\
                & R_1 & &\\
                & & \ddots &\\
                & & & R_{s-1}
            \end{bmatrix}
        \end{equation}
        for some invertible matrices $R_l\in\mathbb{C}^{d_l\times d_l}$~\cite[Theorem 1.3.27]{horn2012matrix}.
    \end{enumerate}
\end{proposition}

For a diagonalizable matrix $A$ with eigendecomposition $A=S\Lambda S^{-1}$, if a scalar function $f(x)$ is defined at all eigenvalues of $A$, we can simply let the \emph{eigenvalue transformation} be
\begin{equation}
\label{eq:eigenvalue_transformation}
    f(A)=Sf(\Lambda)S^{-1}=S
    \begin{bmatrix}
        f(\lambda_0) & & &\\
        & f(\lambda_1) & &\\
        & & \ddots &\\
        & & & f(\lambda_{d-1})
    \end{bmatrix}S^{-1}.
\end{equation}
This is indeed well defined. As eigenvalues and their multiplicities are determined by the rank condition \eq{eigen_rank}, we can simultaneously permute the diagonal entries of $\Lambda$ and $f(\Lambda)$ so that the same values are collected together. Then, diagonal blocks $\lambda_lI_{d_l}$ and $f(\lambda_l)I_{d_l}$ are all multiples of the identity matrix, which are invariant under the transformations $R_l(\cdot)R_l^{-1}$ from \eq{eigen_unique}.

In general, a complex square matrix admits the \emph{Jordan form decomposition} $A=SJS^{-1}$, where $S$ is an invertible matrix and $J$ is only block diagonal. Its existence and uniqueness are formally asserted by the following proposition.

\begin{proposition}[Jordan form decomposition]
    The following statements hold for a matrix $A\in\mathbb{C}^{d\times d}$:
    \begin{enumerate}
        \item (Existence): There exist an invertible matrix $S\in\mathbb{C}^{d\times d}$ and a block diagonal matrix 
        \begin{equation}
        \label{eq:jordan_block}
        \begin{aligned}
            J&=
            \begin{bmatrix}
                J(\lambda_0,d_0) & & &\\
                & J(\lambda_1,d_1) & &\\
                & &\ddots &\\
                & & &J(\lambda_{s-1},d_{s-1})
            \end{bmatrix}\in\mathbb{C}^{d\times d},\\
            J(\lambda_l,d_l)&=
            \begin{bmatrix}
                \lambda_l &  &  &  &  & \\
                1 & \lambda_l &  &  &  & \\
                0 & 1 & \lambda_l &  &  & \\
                \vdots & 0 & 1 & \ddots &  & \\
                \vdots & \ddots & \ddots & \ddots & \lambda_l & \\
                0 & \cdots & \cdots & 0 & 1 & \lambda_l\\
            \end{bmatrix}\in\mathbb{C}^{d_l\times d_l},
        \end{aligned}
        \end{equation}
        such that $A=SJS^{-1}$.
        \item (Uniqueness of Jordan form): The block diagonal matrix $J$ from \eq{jordan_block} satisfies $A=SJ S^{-1}$ for some invertible matrix $S\in\mathbb{C}^{d\times d}$, if and only if
        \begin{equation}
            \#\{J(\lambda_l,d_l)\ |\ \lambda_l=\lambda,d_l\geq k\}
            =\mathbf{Rank}\left((A-\lambda I)^{k-1}\right)
            -\mathbf{Rank}\left((A-\lambda I)^{k}\right)
        \end{equation}
        for all $\lambda\in\mathbb{C}$ and positive integer $k$~\cite[Lemma 3.1.18]{horn2012matrix}.
        \item (Equivalence of Jordan basis): Invertible matrices $S,T\in\mathbb{C}^{d\times d}$ satisfy $A=SJS^{-1}=TJT^{-1}$ for some Jordan form
        \begin{equation}
            J=
            \begin{bmatrix}
                \widetilde{J}_{d_{s_0}}(\lambda_0) & & &\\
                & \widetilde{J}_{d_{s_1}}(\lambda_1) & &\\
                & &\ddots &\\
                & & &\widetilde{J}_{d_{s_{g-1}}}(\lambda_{g-1})
            \end{bmatrix}\in\mathbb{C}^{d\times d}
        \end{equation}
        where blocks with the same eigenvalue $\lambda_l$ are grouped together as $\widetilde{J}_{d_{s_l}}(\lambda_l)\in\mathbb{C}^{d_{s_l}\times d_{s_l}}$ ($\lambda_l$ are pairwise distinct),
        if and only if $S$ and $T$ are related by
        \begin{equation}
        \label{eq:jordan_unique}
            T=S
            \begin{bmatrix}
                \widetilde{R}_0 & & &\\
                & \widetilde{R}_1 & &\\
                & & \ddots &\\
                & & & \widetilde{R}_{g-1}
            \end{bmatrix}
        \end{equation}
        for some invertible matrices $\widetilde{R}_l\in\mathbb{C}^{d_{s_l}\times d_{s_l}}$ with Toeplitz-type block structures~\cite[Theorem 3.2]{FERRER20133945}.
    \end{enumerate}
\end{proposition}

Given a matrix $A$ with Jordan form decomposition $A=SJS^{-1}$, if a scalar function $f(x)$ is analytic at all eigenvalues of $A$, we let the \emph{Jordan form transformation} be
\begin{equation}
\label{eq:jordan_form_transformation}
\begin{aligned}
    f(A)&=Sf(J)S^{-1}
    =S
    \begin{bmatrix}
        f(J(\lambda_0,d_0)) & & &\\
        & f(J(\lambda_1,d_1)) & &\\
        & &\ddots &\\
        & & &f(J(\lambda_{s-1},d_{s-1}))     
    \end{bmatrix}S^{-1},\\
    f(J(\lambda_l,d_l))&=
    \begin{bmatrix}
        f(\lambda_l) &  &  &  &  & \\
        f'(\lambda_l) & f(\lambda_l) &  &  &  & \\
        \frac{f^{(2)}(\lambda_l)}{2!} & f'(\lambda_l) & f(\lambda_l) &  &  & \\
        \vdots & \frac{f^{(2)}(\lambda_l)}{2!} & f'(\lambda_l) & \ddots &\\
        \vdots & \ddots & \ddots & \ddots & \ddots & \\
        \frac{f^{(d_l-1)}(\lambda_l)}{(d_l-1)!} & \cdots & \cdots & \frac{f^{(2)}(\lambda_l)}{2!} & f'(\lambda_l) & f(\lambda_l)\\
    \end{bmatrix}
    =\sum_{r=0}^{d_l-1}\frac{f^{(r)}(\lambda_l)}{r!}L_{\scriptscriptstyle d_l}^{r}.
\end{aligned}
\end{equation}
In other words, we need not only the function $f$ itself, but also its higher derivatives: at least, $f$ should be sufficiently smooth at $\lambda_l$ to match largest size of the $\lambda_l$-Jordan blocks. 
The validity of this matrix function definition can be seen as follows. We permute the Jordan blocks so that blocks with the same eigenvalue $\lambda_l$ are grouped together as $\widetilde{J}_{d_{s_l}}(\lambda_l)\in\mathbb{C}^{d_{s_l}\times d_{s_l}}$. By the above equivalence result, we can without loss of generality restrict ourselves within each $\widetilde{J}_{d_{s_l}}(\lambda_l)$, since this is the only place where one has freedom to choose the basis transformation. We thus have $\widetilde{J}_{d_{s_l}}(\lambda_l)=\widetilde{R}_l\widetilde{J}_{d_{s_l}}(\lambda_l)\widetilde{R}_l^{-1}$ from \eq{jordan_unique}, which implies $\widetilde{L}_{d_{s_l}}=\widetilde{R}_l\widetilde{L}_{d_{s_l}}\widetilde{R}_l^{-1}$ for the blocked lower shift matrix $\widetilde{L}_{d_{s_l}}=\widetilde{J}_{d_{s_l}}(\lambda_l)-\lambda_lI_{d_{s_l}}$. This means that $\sum_{r=0}^{d_{s_l}-1}\frac{f^{(r)}(\lambda_l)}{r!}\widetilde L_{\scriptscriptstyle d_{s_l}}^{r}$ is also invariant under the action of $\widetilde{R}_l(\cdot)\widetilde{R}_l^{-1}$, from which validity of the definition is justified.

If a matrix $A\in\mathbb{C}^{d\times d}$ has only real spectra, we order its eigenvalues increasingly and write
\begin{equation}
    \lambda_{\min}(A)=\lambda_0(A)\leq\lambda_{1}(A)\leq\cdots\leq\lambda_{d-1}(A)=\lambda_{\max}(A).
\end{equation}
Otherwise, we assign an arbitrary ordering to the eigenvalues $\lambda_l(A)$ if $A$ has complex spectra.
We drop the dependence on $A$ when the underlying matrix is clear from the context.
We call an eigenvalue $\lambda$ \emph{nonderogatory} if there is only one $\lambda$-Jordan block, and \emph{nondefective} if all $\lambda$-Jordan blocks have size one. When both conditions are satisfied, we have the Jordan form decomposition $A=S\begin{bmatrix}
    \lambda & 0\\
    0 & \widetilde J
\end{bmatrix}S^{-1}$, where the bottom right $\widetilde J$ has no eigenvalue $\lambda$. In solving the ground state preparation problem, we will assume that the ground state energy is both nondefective and nonderogatory to simplify our analysis of the algorithm.
We will use a similar notation for singular values (to be introduced below):
\begin{equation*}
    0\leq\sigma_{\min}(A)=\sigma_0(A)\leq\sigma_{1}(A)\leq\cdots\leq\sigma_{d-1}(A)=\sigma_{\max}(A)=\norm{A}.
\end{equation*}

Now, we consider matrix decompositions involving unitary transformations, i.e., we consider $A=U\Lambda U^\dagger$ with $\Lambda$ diagonal and $U$ unitary. It is well known in linear algebra that such a \emph{spectral decomposition} exists when and only when $A$ is normal.
\begin{proposition}[Spectral decomposition]
    The following statements hold for a normal matrix $A\in\mathbb{C}^{d\times d}$ ($AA^\dagger=A^\dagger A$):
    \begin{enumerate}
        \item (Existence): There exist a unitary matrix $U\in\mathbb{C}^{d\times d}$ and a diagonal matrix 
        \begin{equation}
        \label{eq:spectral_eigen}
            \Lambda=
            \begin{bmatrix}
                \lambda_0 & & &\\
                & \lambda_1 & &\\
                & & \ddots &\\
                & & & \lambda_{d-1}
            \end{bmatrix}\in\mathbb{C}^{d\times d}
        \end{equation}
        such that $A=U\Lambda U^{\dagger}$.
        \item (Uniqueness of eigenvalues): The diagonal matrix $\Lambda$ from \eq{spectral_eigen} satisfies $A=U\Lambda U^{\dagger}$ for some unitary matrix $U\in\mathbb{C}^{d\times d}$, if and only if
        \begin{equation}
            \#\{\lambda_l\ |\ \lambda_l=\lambda\}
            =d-\mathbf{Rank}(A-\lambda I)
        \end{equation}
        for all $\lambda\in\mathbb C$.
        \item (Equivalence of orthonormal eigenbasis): Unitary matrices $U,V\in\mathbb{C}^{d\times d}$ satisfy $A=U\Lambda U^{\dagger}=V\Lambda V^{\dagger}$ for some diagonal matrix
        \begin{equation}
            \Lambda=
            \begin{bmatrix}
                \lambda_0I_{d_0} & & &\\
                & \lambda_1I_{d_1} & &\\
                & & \ddots &\\
                & & & \lambda_{s-1}I_{d_{s-1}}
            \end{bmatrix}\in\mathbb{C}^{d\times d}
        \end{equation}
        with the same eigenvalue $\lambda_l$ grouped together as $\lambda_lI_{d_l}\in\mathbb{C}^{d_l\times d_l}$ ($\lambda_l$ are pairwise distinct),
        if and only if $U$ and $V$ are related by
        \begin{equation}
        \label{eq:spectral_unique}
            V=U
            \begin{bmatrix}
                W_0 & & &\\
                & W_1 & &\\
                & & \ddots &\\
                & & & W_{s-1}
            \end{bmatrix}
        \end{equation}
        for some unitary matrices $W_l\in\mathbb{C}^{d_l\times d_l}$~\cite[Theorem 2.5.4]{horn2012matrix}. 
    \end{enumerate}
\end{proposition}

Spectral decomposition offers a useful perspective that differentiates a normal matrix from a general one. This can be understood as follows. 
For a square matrix $A$, we introduce its \emph{Jordan condition number} as $\inf_{J,S}\norm{S}\norm{S^{-1}}$, where the minimization is over all possible Jordan form decompositions $A=SJS^{-1}$. Since the factor $J$ is unique up to a permutation of diagonal blocks, it suffices to only minimize over the basis transformations $\inf_{S}\norm{S}\norm{S^{-1}}$. We then have that Jordan condition number of a matrix is always $\geq1$, and the matrix is normal if and only if it is diagonalizable with Jordan conditon number $=1$~\cite[Condition 72]{Elsner1998}.
In fact, Jordan condition number serves as a commonly used measure of nonnormality in numerical linear algebra~\cite[Page 444]{Trefethen05}, and it may be nearly optimally bounded by separating the spectrum of input matrix into a disjoint union of eigenvalues~\cite[Section 5.3]{JordanCondition}.
This perspective can be useful in analyzing our eigenvalue algorithms. For instance, we bound the runtime of our Chebyshev-based algorithms in \append{analysis_cheby} using known upper bounds $\kappa_S$ on the Jordan condition number, with the understanding that this condition number approaches $1$ when the target matrix is close to normal. This treatment is similar to previous analysis of differential equations~\cite[Section 3.2]{Krovi2023improvedquantum} based on the \emph{spectral abscissa}.
However, we also show in \append{analysis_faber} that complexity of the Faber-based algorithms can be independent of the Jordan condition number, which avoids the issue with a potentially ill-conditioned Jordan basis similar to previous work~\cite[Section 3.1]{Krovi2023improvedquantum} based on the \emph{numerical abscissa}.
See \tab{nonnormal} for more details.

Below we list subclasses of normal matrices that we will commonly refer to in this work. They all admit the spectral decomposition, with eigenvalues belonging to different subsets of the complex plane.
\begin{enumerate}
    \item Normal matrices: $NN^\dagger=N^\dagger N$, if and only if $\norm{N\ket{\psi}}=\norm{N^\dagger\ket{\psi}}$ for all $\norm{\ket{\psi}}=1$, if and only $N$ has a spectral decomposition.
    \item Unitary matrices: $UU^\dagger=U^\dagger U=I$, if and only if $\norm{U\ket{\psi}}=\norm{U^\dagger\ket{\psi}}=1$ for all $\norm{\ket{\psi}}=1$, if and only if $U$ has a spectral decomposition with all eigenvalues $\abs{\lambda_l}=1$.
    \item Hermitian matrices: $HH^\dagger=H^2$, if and only if $H^\dagger=H$, if and only if $\bra{\psi}H\ket{\psi}\in\mathbb R$ for all $\norm{\ket{\psi}}=1$, if and only if $H$ has a spectral decomposition with all eigenvalues $\lambda_l\in\mathbb R$. The first characterization can be proved by showing the trace $\mathbf{Tr}\left(\left(H-H^\dagger\right)\left(H-H^\dagger\right)^\dagger\right)=0$.
    \item Anti-Hermitian matrices: $KK^\dagger=-K^2$, if and only if $K^\dagger=-K$, if and only if $\bra{\psi}K\ket{\psi}\in i\mathbb R$ for all $\norm{\ket{\psi}}=1$, if and only if $K$ has a spectral decomposition with all eigenvalues $\lambda_l\in i\mathbb R$ purely imaginary. The first characterization follows from $\mathbf{Tr}\left(\left(K+K^\dagger\right)\left(K+K^\dagger\right)^\dagger\right)=0$.
    \item Positiv semidefinite matrices: $P=CC^\dagger$ for some matrix $C$, if and only if $\bra{\psi}P\ket{\psi}\geq0$ for all $\norm{\ket{\psi}}=1$, if and only if $P$ has a spectral decomposition with all eigenvalues $\lambda_l\geq0$.
    \item Orthogonal projection matrices: $\Pi=GG^\dagger$ for some isometry $G^\dagger G=I$, if and only if $\Pi^2=\Pi$ and $\norm{\Pi\ket{\psi}}\leq1$ for all $\norm{\ket{\psi}}=1$, if and only if $\Pi$ has a spectral decomposition with all eigenvalues $\lambda_l=0,1$. See~\cite[Theorem 10.5]{roman2013advanced} or~\cite[Corollary 3.4.3.3]{horn2012matrix} for a proof of the second characterization.
\end{enumerate}

Finally, we introduce the \emph{singular value decomposition}, which simplifies the action of a matrix with two orthonormal bases. We will only introduce this decomposition for a square matrix, to facilitate a direct comparison with the eigendecomposition. Nevertheless, it is fairly straightforward to extend the following discussion to an arbitrary non-square matrix.

\begin{proposition}[Singular value decomposition]
    The following statements hold for a matrix $A\in\mathbb{C}^{d\times d}$:
    \begin{enumerate}
        \item (Existence): There exist unitary matrices $U,V\in\mathbb{C}^{d\times d}$ and a diagonal matrix 
        \begin{equation}
        \label{eq:singular_val}
            \Sigma=
            \begin{bmatrix}
                \sigma_0 & & &\\
                & \sigma_1 & &\\
                & & \ddots &\\
                & & & \sigma_{d-1}
            \end{bmatrix}\in\mathbb{C}^{d\times d}
        \end{equation}
        with nonnegative entries $\sigma_l\geq0$, such that $A=V\Sigma U^{\dagger}$.
        \item (Uniqueness of singular values): The diagonal matrix $\Sigma$ from \eq{singular_val} with nonnegative entries satisfies $A=V\Sigma U^{\dagger}$ for some unitary matrices $U,V\in\mathbb{C}^{d\times d}$, if and only if
        \begin{equation}
        \label{eq:singular_rank}
            \#\{\sigma_l\ |\ \sigma_l=\sigma\}
            =d-\mathbf{Rank}\left(A^\dagger A-\sigma^2 I\right)
        \end{equation}
        for all $\sigma\geq0$.
        \item (Equivalence of singular vector basis): Unitary matrices $U,V,\widetilde U,\widetilde V\in\mathbb{C}^{d\times d}$ satisfy $A=V\Sigma U^{\dagger}=\widetilde V\Sigma \widetilde U^{\dagger}$ for some diagonal matrix
        \begin{equation}
            \Sigma=
            \begin{bmatrix}
                0_{d_0} & & &\\
                & \sigma_1I_{d_1} & &\\
                & & \ddots &\\
                & & & \sigma_{s-1}I_{d_{s-1}}
            \end{bmatrix}\in\mathbb{C}^{d\times d}
        \end{equation}
        with the same singular value $\sigma_l$ grouped together as $\sigma_lI_{d_l}\in\mathbb{C}^{d_l\times d_l}$ ($\sigma_l$ are pairwise distinct),
        if and only if $V$ and $\widetilde V$, $U$ and $\widetilde U$ are related by
        \begin{equation}
        \label{eq:singular_unique}
            \widetilde V=V
            \begin{bmatrix}
                W_{0,\text{left}} & & &\\
                & W_1 & &\\
                & & \ddots &\\
                & & & W_{s-1}
            \end{bmatrix},\qquad
            \widetilde U=U
            \begin{bmatrix}
                W_{0,\text{right}} & & &\\
                & W_1 & &\\
                & & \ddots &\\
                & & & W_{s-1}
            \end{bmatrix}
        \end{equation}
        for some unitary matrices $W_{0,\text{left}},W_{0,\text{right}}\in\mathbb{C}^{d_0\times d_0},W_l\in\mathbb{C}^{d_l\times d_l}$~\cite[Theorem 2.6.5]{horn2012matrix}.
    \end{enumerate}
\end{proposition}

Let us now define the \emph{singular value transformation} of matrices, which is central to the QSVT algorithm to be introduced in the next subsection. 
Specifically, suppose we are given a function $f(x)$ satisfying a parity constraint, i.e., we have either $f(x)=-f(-x)$ for all $x$ (odd parity) or $f(x)=f(-x)$ for all $x$ (even parity). Our goal is to apply $f$ to the singular values of the target matrix $A$. Assuming $A$ is a square matrix and has the singular value decomposition $A=V\Sigma U^\dagger$, our singular value transformation is then defined as
\begin{equation}
\label{eq:singular_value_transformation}
    f_{\sv}(A)=
    \begin{cases}
        Vf(\Sigma)U^\dagger,\qquad&f\text{ is odd},\\[0.5em]
        Uf(\Sigma)U^\dagger,\qquad&f\text{ is even},\\
    \end{cases}
    \qquad
    f(\Sigma)=
            \begin{bmatrix}
                f(\sigma_0) & & &\\
                & f(\sigma_1) & &\\
                & & \ddots &\\
                & & & f(\sigma_{d-1})
            \end{bmatrix}.
\end{equation}
We claim that the above transformation is mathematically well defined. Indeed, as singular values and their multiplicities are determined by the rank condition \eq{singular_rank}, we can simultaneously permute the diagonal entries of $\Sigma$ and $f(\Sigma)$ so that the same values are collected together. Then, diagonal blocks $\sigma_lI_{d_l}$ and $f(\sigma_l)I_{d_l}$ are all multiples of the identity matrix, which are invariant under the unitary conjugation $W_l(\cdot)W_l^{\dagger}$ from \eq{singular_unique}. In the case where $f$ is an odd function, we have an additional transformation of the form $W_{0,\text{left}}(\cdot)W_{0,\text{right}}^\dagger$. But this corresponds to the block with the zero singular value, and so the result is always the zero matrix when mapped by an odd function: $W_{0,\text{left}}f(0)W_{0,\text{right}}^\dagger=W_{0,\text{left}}0W_{0,\text{right}}^\dagger=0$. We thus conclude that the singular value transformation $f_{\sv}(\cdot)$ is indeed well defined in all the cases.

\subsection{Block encoding}
\label{sec:prelim_block}

In this subsection, we review basic facts about block encoding and quantum algorithms developed within this framework, on which our eigenvalue processing results are based.

We will define the block encoding in its full generality using unitaries and isometries. 
We say an operator $G:\mathcal{G}\rightarrow\mathcal{H}$ is an \emph{isometry} if $G^\dagger G=I$. By definition, $G$ is injective and $G^\dagger$ is surjective, whereas $GG^\dagger$ is an orthogonal projection on $\mathcal{H}$ with kernel $\mathbf{Ker}(GG^\dagger)=\mathbf{Ker}(G^\dagger)$ and image $\mathbf{Im}(GG^\dagger)=\mathbf{Im}(G)$. We thus have the Hilbert space embedding
\begin{equation}
    \mathcal{G}
    \xrightleftharpoons[G^\dagger]{G}
    \mathbf{Im}(GG^\dagger)\subseteq\mathcal{H}.
\end{equation}
Choosing any two orthonormal bases that respect this embedding, we have the matrix representation
\begin{equation}
    G=
    \begin{bmatrix}
        I\\
        0
    \end{bmatrix},\qquad
    G^\dagger=
    \begin{bmatrix}
        I & 0
    \end{bmatrix}.
\end{equation}
Examples of isometries include: (i) unitaries $U$; (ii) quantum states $\ket{\psi}$; (iii) tensor product $G_1\otimes G_2$, if $G_1$ and $G_2$ are both isometries; and (iv) composition $G_2G_1$, if $G_1$ and $G_2$ are both isometries and the composition makes sense.

Now given Hilbert spaces $\mathcal{G}_0$, $\mathcal{G}_1$ and $\mathcal{H}$, we say an operator $A:\mathcal{G}_0\rightarrow\mathcal{G}_1$ is \emph{block encoded} by isometries $G_0:\mathcal{G}_0\rightarrow\mathcal{H}$, $G_1:\mathcal{G}_1\rightarrow\mathcal{H}$, and a unitary $U:\mathcal{H}\rightarrow\mathcal{H}$, if
\begin{equation}
    A=G_1^\dagger UG_0.
\end{equation}
Choosing bases with respect to the orthogonal decompositions $\mathcal{H}=\mathbf{Im}(G_0G_0^\dagger)\obot\mathbf{Ker}(G_0G_0^\dagger)=\mathbf{Im}(G_1G_1^\dagger)\obot\mathbf{Ker}(G_1G_1^\dagger)$, we have the matrix representation
\begin{equation}
    U=
    \begin{bmatrix}
        A & *\\
        * & *
    \end{bmatrix},
\end{equation}
where $A$ is exhibited as the top-left block of $U$; hence the name \emph{block encoding}.
Note that this is essentially a unitary dilation problem, and it is mathematically feasible if and only if $\norm{A}\leq1$~\cite[2.7.P2]{horn2012matrix}: if the norm condition is satisfied, we can simply let~\cite{halmos1950normal}
\begin{equation}
    U
    =\begin{bmatrix}
        A & -\sqrt{I-AA^\dagger}\\
        \sqrt{I-A^\dagger A} & A^\dagger
    \end{bmatrix}.
\end{equation}
However, when a block encoding is realized by quantum circuits, additional normalization factors will likely be introduced.

Specifically, to block encode a square matrix on a system register using quantum circuits, we can prepare states $\ket{G_0}$, $\ket{G_1}$ on an ancilla register and perform a unitary $U$ acting jointly on the ancilla and the system register. 
Then, the operator block encoded by this circuit is given by $\left(\bra{G_1}\otimes I\right)U\left(\ket{G_0}\otimes I\right)$.
For a given operator $A$ on the system register, we need to introduce a normalization factor $\alpha_A$ so that $A/\alpha_A$ can be properly encoded. Our above discussion suggests that $\alpha_A\geq\norm{A}$ is a prerequisite for the existence of a block encoding. But the corresponding construction typically requires an explicit computation of the singular value decomposition of $A$, which is prohibitive for high-dimensional input matrices. So the strict inequality $\alpha_A>\norm{A}$ often holds in practice. 

In developing our eigenvalue algorithms, we sometimes assume that the input matrix can be block encoded with a normalization factor $\alpha\geq2\norm{A}$.
We note that the block-encoding-based model covers a host of input matrices with sparsity constraints or linear-combination-of-unitary expansions~\cite{Low2019hamiltonian,Gilyen2018singular}. If a normalization factor satisfies $2\norm{A}>\alpha_A\geq\norm{A}$, 
we can block encode a rescaling constant
\begin{equation}
\label{eq:block_rescaling}
    \frac{\sqrt{\alpha+\alpha_A}\bra{0}+\sqrt{\alpha-\alpha_A}\bra{1}}{\sqrt{2\alpha}}\left(\ketbra{0}{0}-\ketbra{1}{1}\right)\frac{\sqrt{\alpha+\alpha_A}\ket{0}+\sqrt{\alpha-\alpha_A}\ket{1}}{\sqrt{2\alpha}}=\frac{\alpha_A}{\alpha}
\end{equation}
with $\alpha>\alpha_A$ and artificially increase the normalization factor to meet the desired assumption~\cite[Appendix A.2]{FractionalQuery14}.
This rescaling of normalization factor from $\alpha_A$ to $\alpha>\alpha_A$ has no query overhead.

The block-encoding framework allows quantum computers to efficiently perform arithmetic operations on the input matrices, including linear combinations, multiplications, and tensor products. Moreover, there exists the QSVT technique to apply polynomial functions to the singular values of a block encoded matrix~\cite{Gilyen2018singular}. We briefly explain the idea of taking linear combination of block encodings as well as the algorithm of QSVT, which we compare with our main result QEVT on transforming eigenvalues.

Suppose we have $A_j/\alpha_{A_j}$ block encoded by $\ket{0}\otimes I$ and $U_j$ for $j=0,1,\ldots,n-1$, and we want to implement the linear combination $\sum_{j=0}^{n-1}\beta_jA_j$ for some coefficients $\beta_j\geq0$. Here, the choice of the reference state $\ket{0}$ is without loss of generality, as any state preparation can be absorbed into the definition of $U_j$. Then, we define
\begin{equation}
    \ket{G}=\frac{1}{\sqrt{\sum_{k=0}^{n-1}\beta_k\alpha_{A_k}}}\sum_{j=0}^{n-1}\sqrt{\beta_j\alpha_{A_j}}\ket{j}\ket{0},\qquad
    U=\sum_{j=0}^{n-1}\ketbra{j}{j}\otimes U_j,
\end{equation}
so that 
\begin{equation}
    \left(\bra{G}\otimes I\right)U\left(\ket{G}\otimes I\right)=\frac{\sum_{j=0}^{n-1}\beta_jA_j}{\sum_{k=0}^{n-1}\beta_k\alpha_{A_k}}.
\end{equation}
This block encoding consists of a \emph{state preparation} subroutine that can be implemented using $\mathbf{\Theta}(n)$ gates~\cite{ShendeBullockMarkov06}, as well as an \emph{operator selection} subroutine that can be realized by generating all binary strings of length $\sim\log(n)$ with gate complexity $\mathbf{\Theta}(n)$~\cite[Appendix G.4]{CMNRS18}. The normalization factor
\begin{equation}
    \sum_{j=0}^{n-1}\beta_j\alpha_{A_j}\geq\sum_{j=0}^{n-1}\beta_j\norm{A_j}\geq\norm{\sum_{j=0}^{n-1}\beta_jA_j}
\end{equation}
is larger than the spectral norm $\norm{\sum_{j=0}^{n-1}\beta_jA_j}$ as expected. But in fact, this normalization factor can be exponentially large if the above construction is directly applied to implement the polynomial expansion \eq{poly_expansion} for eigenvalue processing. We overcome this catastrophe by developing efficient methods for generating polynomial basis, including the Chebyshev basis for matrices with real eigenvalues, and the Faber basis for matrices with complex eigenvalues.

As for the singular value transformation, consider a polynomial $p$ with max-norm $\norm{p}_{\max,[-1,1]}=\max_{x\in[-1,1]}\abs{p(x)}\leq1$. Then the QSVT algorithm implements a block encoding of the transformed matrix $p_{\sv}\left(\frac{A}{\alpha_A}\right)$ with a query complexity proportional to the degree $n$ of $p$. When performed on an input state $\ket{\psi}$, QSVT outputs the state $p_{\sv}\left(\frac{A}{\alpha_A}\right)\ket{\psi}$ with success probability
\begin{equation}
    \norm{p_{\sv}\left(\frac{A}{\alpha_A}\right)\ket{\psi}}^2.
\end{equation}
Obtaining the normalized state 
$\frac{p_{\sv}\left(\frac{A}{\alpha_A}\right)\ket{\psi}}{\norm{p_{\sv}\left(\frac{A}{\alpha_A}\right)\ket{\psi}}}$
with a success probability $1-\pfail$ would cost an additional factor of
$\mathbf{O}\left(\frac{1}{\norm{p_{\sv}\left(\frac{A}{\alpha_A}\right)\ket{\psi}}}\log\left(\frac{1}{\pfail}\right)\right)$
queries to the block encoding.

In contrast, our goal is to apply polynomial functions to the eigenvalues like $p\left(\frac{A}{\alpha_A}\right)$. If the input matrix is not normal, this problem is out of reach of the QSVT algorithm and its descendants. But otherwise, we achieve a complexity comparable to that of QSVT for Hermitian matrices. Suppose that the input matrix $A$ is diagonalizable with a condition number $\kappa_S$ of the basis transformation. Then in the worst case, our QEVT algorithm $\epsilon$-approximates a Chebyshev history state using $\mathbf{O}(\kappa_Sn\log(1/\epsilon))$ queries to the block encoding of $A/\alpha_A$, which if measured will produce $p\left(\frac{A}{\alpha_A}\right)\ket{\psi}$ with probability
\begin{equation}
    \mathbf{\Theta}\left(\frac{\norm{p\left(\frac{A}{\alpha_A}\right)\ket{\psi}}^2}{\kappa_S^2\log^2 (n)}\right).
\end{equation}
So obtaining the normalized state
$\frac{p\left(\frac{A}{\alpha_A}\right)\ket{\psi}}{\norm{p\left(\frac{A}{\alpha_A}\right)\ket{\psi}}}$
would require a fixed-point amplitude amplification with an additional query complexity of
$\mathbf{O}\left(\frac{\kappa_S\log(n)}{\norm{p\left(\frac{A}{\alpha_A}\right)\ket{\psi}}}\log\left(\frac{1}{\pfail}\right)\right)$.
Furthermore, we show that QEVT can achieve a better performance on average: we can shave off the additional $\log(n)$ factor from the above expressions when eigenvalues of the input matrix are randomly chosen. Therefore, under the common assumption where the input matrix is Hermitian (which holds if and only if it is diagonalizable with real eigenvalues and $\kappa_S=1$~\cite[Condition 72]{Elsner1998}), our QEVT naturally recovers QSVT. In particular, this yields a quantum algorithm for solving linear differential equations on the imaginary axis with a strictly linear scaling in the evolution time, improving over the best previous result under the same setting.

Finally, we introduce several quantum algorithms within the block encoding framework which will be used as subroutines in eigenvalue processing. This includes an optimal scaling quantum linear system solver (to be used in \sec{history} and \sec{faber} to generate the Chebyshev and Faber history states), a block encoded quantum linear system solver (to be used in \sec{transform_block_summary} to realize the block encoded version of QEVT), and a block encoding amplifier (to be used in \sec{faber} to block encode the matrix Faber generating function).

\begin{lemma}[Optimal scaling quantum linear system algorithm {\cite{Costa22,OptInit,Dalzell2024shortcut}}]
\label{lem:opt_lin}
    Let $C$ be a matrix such that $C/\alpha_C$ is block encoded by $O_C$ with some normalization factor $\alpha_C\geq\norm{C}$.
    Let $O_b$ be the oracle preparing the initial state $\ket{b}$. Then the quantum state
    \begin{equation}
        \frac{C^{-1}\ket{b}}{\norm{C^{-1}\ket{b}}}
    \end{equation}
    can be prepared with accuracy $\epsilon$ and probability $1-\pfail$ using
    \begin{equation}
    \label{eq:opt_lin_cost}
        \mathbf{O}\left(\alpha_C\alpha_{C^{-1}}\log\left(\frac{1}{\epsilon}\right)\log\left(\frac{1}{\pfail}\right)\right)
    \end{equation}
    queries to controlled-$O_C$, controlled-$O_b$, and their inverses, where $\alpha_{C^{-1}}\geq\norm{C^{-1}}$ is an upper bound on norm of the inverse matrix.
\end{lemma}
\begin{remark}
    Unless otherwise stated, we will choose a sufficiently large (but constant) success probability when invoking the quantum linear system solver as a subroutine, and boost the probability at the very end of the entire quantum algorithm.
\end{remark}

\begin{lemma}[Block encoding inversion {\cite[Corollary 69]{Gilyen2018singular}}]
\label{lem:inv_block}
     Let $C$ be a matrix such that $C/\alpha_C$ is block encoded by $O_C$ with some normalization factor $\alpha_C\geq\norm{C}$.
     Then the operator
     \begin{equation}
         \frac{C^{-1}}{2\alpha_{C^{-1}}}
     \end{equation}
     can be block encoded with accuracy $\epsilon$ using
     \begin{equation}
        \mathbf{O}\left(\alpha_C\alpha_{C^{-1}}\log\left(\frac{1}{\epsilon}\right)\right)
    \end{equation}
    queries to the controlled-$O_C$ and its inverse, where $\alpha_{C^{-1}}\geq\norm{C^{-1}}$ is an upper bound on norm of the inverse matrix.
\end{lemma}

\begin{lemma}[Block encoding amplification {\cite[Theorem 2]{Low17}}]
\label{lem:amp_block}
    Let $C$ be a matrix such that $C/\alpha$ is block encoded by $O_C$ with some normalization factor $\alpha$.
    Then given any $\alpha_C\geq\norm{C}$, the operator
    \begin{equation}
        \frac{C}{2\alpha_C}
    \end{equation}
    can be block encoded with accuracy $\epsilon$ using
    \begin{equation}
        \mathbf{O}\left(\frac{\alpha}{\alpha_C}\log\left(\frac{1}{\epsilon}\right)\right)
    \end{equation}
    queries to the controlled-$O_C$ and its inverse.
\end{lemma}
\begin{remark}
In the interesting regime where $\alpha\gg2\alpha_C$, this rescaling decreases the normalization from $\alpha$ to $2\alpha_C$ and necessarily incurs a query overhead as above.
    Note also that this is lower than the complexity $\mathbf{O}\left(\frac{\alpha}{\alpha_C}\log\left(\frac{\alpha}{\alpha_C\epsilon}\right)\right)$ stated in~\cite[Theorem 30]{Gilyen2018singular} by a logarithmic factor, which results from a looser error analysis. In the proof of~\cite[Theorem 30]{Gilyen2018singular}, the constructed odd polynomial $p_{\Re}$ has degree $\mathbf{O}\left(\frac{\alpha}{\alpha_C}\log\left(\frac{\alpha}{\alpha_C\epsilon}\right)\right)$ and actually approximates the linear function $f(x)=\frac{\alpha}{\alpha_C} x$ on the domain $[-\frac{\alpha_C}{2\alpha},\frac{\alpha_C}{2\alpha}]$ with additive error
    \begin{equation*}
        \norm{p_{\Re}-f}_{\max,[-\frac{\alpha_C}{2\alpha},\frac{\alpha_C}{2\alpha}]}=\mathbf{O}\left(\frac{\alpha_C\epsilon}{\alpha}\right).
    \end{equation*}
    This is tighter than what was claimed in~\cite{Gilyen2018singular}:
    \begin{equation*}
        \norm{p_{\Re}-f}_{\max,[-\frac{\alpha_C}{2\alpha},\frac{\alpha_C}{2\alpha}]}=\mathbf{O}(\epsilon).
    \end{equation*}
\end{remark}

\section{Chebyshev history state generation}
\label{sec:history}
In this section, we present our main technique---a quantum algorithm that efficiently generates Chebyshev history states for non-normal input matrices. This is achieved using a matrix version of the Chebyshev generating function, which is overviewed in \sec{history_generate}. We then present a quantum circuit implementing the matrix Chebyshev generating function in \sec{history_block}, together with an explicit normalization factor for the block encoding. With the help of a quantum linear system solver, we obtain the desired quantum algorithm for creating the Chebyshev history state, which we summarize as \thm{generate_history} and analyze in \sec{history_summary}.

\subsection{Matrix Chebyshev generating function}
\label{sec:history_generate}

As is discussed in \sec{intro_history}, the main insight of our approach is to efficiently create a history state, encoding a polynomial basis of the input matrix in quantum superposition, even when the input matrix is not Hermitian or normal. For Chebyshev polynomials, this is accomplished by using the following matrix version of the generating function:
\begin{equation}
\label{eq:matrix_cheby_gen}
\begin{aligned}
    \sum_{j=0}^{n-1}L^j\otimes \widetilde{\mathbf{T}}_j\left(\frac{A}{\alpha_A}\right)
    &=\begin{bmatrix}
        \frac{I}{2} & 0 & 0 & \cdots & 0\\
        {\mathbf{T}}_1\left(\frac{A}{\alpha_A}\right) & \frac{I}{2} & 0 & \ddots & \vdots\\
        {\mathbf{T}}_2\left(\frac{A}{\alpha_A}\right) & {\mathbf{T}}_1\left(\frac{A}{\alpha_A}\right) & \frac{I}{2} & \ddots & \vdots\\
        \vdots & \ddots & \ddots & \ddots & \vdots\\
        {\mathbf{T}}_{n-1}\left(\frac{A}{\alpha_A}\right) & \cdots & {\mathbf{T}}_2\left(\frac{A}{\alpha_A}\right) & {\mathbf{T}}_1\left(\frac{A}{\alpha_A}\right) & \frac{I}{2}
    \end{bmatrix}\\
    &=\sum_{j=0}^{\infty}L^j\otimes \widetilde{\mathbf{T}}_j\left(\frac{A}{\alpha_A}\right)
    =\frac{I\otimes I-L^2\otimes I}{2(I\otimes I+L^2\otimes I-2L\otimes \frac{A}{\alpha_A})},
\end{aligned}
\end{equation}
where the input matrix is block encoded as $A/\alpha_A$ and $L=\sum_{k=0}^{n-2}\ketbra{k+1}{k}$ is the $n$-by-$n$ lower shift matrix such that $L^n=0$.
\eq{matrix_cheby_gen} follows from \eq{cheby_gen} by substituting $x=I\otimes\frac{A}{\alpha_A}$ and $y=L\otimes I$. This substitution is valid because $L$ has zero eigenvalues only, whereas both sides of \eq{cheby_gen} have the same derivatives at $y=0$ of any order. See~\cite[Chapter 6]{roger1994topics} or~\cite[Chapter 2]{henrici1974applied} for a complete mathematical justification.

Now consider the problem of eigenvalue processing. Given a target Chebyshev expansion $\sum_{k=0}^{n-1}{\beta}_k\mathbf{T}_{k}$, we apply the matrix generating function to the initial state
\begin{equation}
    \frac{\sum_{k=0}^{n-1}\widetilde{\beta}_k\ket{n-1-k}}{\norm{\widetilde \beta}}\ket{\psi},\qquad
    \widetilde\beta_k=
    \begin{cases}
        \beta_k,\quad&k\neq0,\\
        2\beta_0,&k=0.
    \end{cases}
\end{equation}
Up to a normalization factor, we obtain 
\begin{equation}
\begin{aligned}
    \left(\sum_{j=0}^{n-1}L^j\otimes \widetilde{\mathbf{T}}_j\left(\frac{A}{\alpha_A}\right)\right)
    \left(\sum_{k=0}^{n-1}\widetilde{\beta}_k\ket{n-1-k}\ket{\psi}\right)
    &=\sum_{k=0}^{n-1}\widetilde{\beta}_k\sum_{j=0}^{k}\ket{n-1-k+j}\widetilde{\mathbf{T}}_j\left(\frac{A}{\alpha_A}\right)\ket{\psi}\\
    &=\sum_{k=0}^{n-1}\widetilde{\beta}_k\sum_{l=n-1-k}^{n-1}\ket{l}\widetilde{\mathbf{T}}_{l+k-n+1}\left(\frac{A}{\alpha_A}\right)\ket{\psi}\\
    &=\sum_{l=0}^{n-1}\ket{l}
    \sum_{k=n-1-l}^{n-1}\widetilde{\beta}_k\widetilde{\mathbf{T}}_{k+l-n+1}\left(\frac{A}{\alpha_A}\right)\ket{\psi}.
\end{aligned}
\end{equation}
If we now measure the first register and get the outcome $l=n-1$, the second register will have the desired state proportional to
\begin{equation}
    \sum_{k=0}^{n-1}\widetilde{\beta}_k\widetilde{\mathbf{T}}_{k}\left(\frac{A}{\alpha_A}\right)\ket{\psi}
    =\sum_{k=0}^{n-1}{\beta}_k\mathbf{T}_{k}\left(\frac{A}{\alpha_A}\right)\ket{\psi}.
\end{equation}
However, we will also get unwanted components for $l=0,\ldots,n-2$, leading to a failure of the algorithm.

To boost the success probability, we use the runaway padding trick~\cite{Berry2017Differential} to repeat the desired state $\eta n$ times. This can be understood via the following formula for inverting lower block matrices.

\begin{lemma}[Lower block matrix inversion]
\label{lem:lower_block_inv}
For invertible square matrices $A_{11}$ and $A_{22}$,
\begin{equation}
    \begin{bmatrix}
        A_{11} & 0\\
        A_{21} & A_{22}
    \end{bmatrix}^{-1}=
    \begin{bmatrix}
        A_{11}^{-1} & 0\\
        -A_{22}^{-1}A_{21}A_{11}^{-1} & A_{22}^{-1}
    \end{bmatrix}.
\end{equation}
\end{lemma}

In our case, we let
\begin{equation}
    A_{11}=I_n\otimes I+L_n^2\otimes I-2L_n\otimes \frac{A}{\alpha_A}
    =\begin{bmatrix}
        I & 0 & \cdots & \cdots & \cdots &0\\
        -\frac{2A}{\alpha_A} & I & \ddots & \ddots & \ddots & \vdots\\
        I & -\frac{2A}{\alpha_A} & I & \ddots &\ddots & \vdots\\
        0 & I & \frac{-2A}{\alpha_A} & I & \ddots & \vdots\\
        \vdots & \ddots & \ddots & \ddots & \ddots & 0\\
        0 & \cdots & 0 & I & -\frac{2A}{\alpha_A} & I \\
    \end{bmatrix}
\end{equation}
corresponding to the denominator of the matrix Chebyshev generating function, so that from \eq{cheby_gen2}, 
\begin{equation}
    A_{11}^{-1}=\frac{1}{I_n\otimes I+L_n^2\otimes I-2L_n\otimes \frac{A}{\alpha_A}}=\begin{bmatrix}
        \mathbf{U}_{0}\left(\frac{A}{\alpha_A}\right) & 0 & \cdots & 0 \\
        \mathbf{U}_{1}\left(\frac{A}{\alpha_A}\right) & \mathbf{U}_{0}\left(\frac{A}{\alpha_A}\right) &\ddots & \vdots \\[.2cm]
        \vdots & \ddots & \ddots & 0\\[.2cm]
        \mathbf{U}_{n-1}\left(\frac{A}{\alpha_A}\right) & \cdots &\mathbf{U}_{1}\left(\frac{A}{\alpha_A}\right) & \mathbf{U}_{0}\left(\frac{A}{\alpha_A}\right) \\
    \end{bmatrix}.
\end{equation} 
Here, we have used subscripts to explicitly represent dimensions of the identity and the lower shift matrix on the ancilla register.
Now we take the $\eta n$-by-$n$ block matrix
\begin{equation}
    A_{21}=\ketbra{0}{n-1}\otimes(-I)=
    \begin{bmatrix}
        0 & 0 & \cdots & -I\\
        0 & 0 & \cdots & 0\\
        \vdots & \vdots & \vdots & \vdots\\
        0 & 0 & \cdots & 0\\
    \end{bmatrix}.
\end{equation}
Thus the action of $-A_{21}A_{11}^{-1}$ is to copy the last row of $A_{11}^{-1}$ to the first row while padding the remaining rows with zeros. Now we simply need to copy the first row of $-A_{21}A_{11}^{-1}$ to the remaining rows, which can be achieved by setting
\begin{equation}
    A_{22}=\left(I_{\eta n}-L_{\eta n}\right)\otimes I=
    \begin{bmatrix}
        I & 0 & 0 & \cdots & 0\\
        -I & I & 0 & \cdots & 0\\
        0 & -I & I & \ddots & 0\\
        \vdots & \ddots & \ddots & \ddots & \vdots\\
        0 & \cdots & \ddots & -I & I\\
    \end{bmatrix}\
    \Rightarrow\
    A_{22}^{-1}=\begin{bmatrix}
        I & 0 & 0 & \cdots & 0\\
        I & I & 0 & \cdots & 0\\
        I & I & I & \ddots & \vdots\\
        \vdots & \ddots & \ddots & \ddots & \vdots\\
        I & \cdots & \ddots & \ddots & I
    \end{bmatrix}.
\end{equation}
We will bundle the numerator of the Chebyshev generating function
\begin{equation}
    B_{11}=\frac{I_n\otimes I-L_n^2\otimes I}{2}
\end{equation}
with the state preparation subroutine and discuss it later in \sec{fourier}.

To summarize, after the padding, the denominator of the matrix Chebyshev generating function becomes
\begin{equation}
\NiceMatrixOptions
{
    custom-line = 
    {
        letter = I , 
        command = hdashedline , 
        tikz = {dashed,dash phase=3pt} ,
        width = \pgflinewidth
    }
}
\begin{aligned}
\label{eq:pad_a}
    \mathbf{Pad}(A)&=\ketbra{0}{0}\otimes A_{11}+\ketbra{1}{0}\otimes A_{21}+\ketbra{1}{1}\otimes A_{22}\\
    &=\ketbra{0}{0}\otimes\left(I_n\otimes I+L_n^2\otimes I-2L_n\otimes \frac{A}{\alpha_A}\right)\\
    &\quad+\ketbra{1}{0}\otimes\ketbra{0}{n-1}\otimes (-I)
    +\ketbra{1}{1}\otimes(I_{\eta n}-L_{\eta n})\otimes I\\
    &=\begin{bNiceArray}{ccccccIccc}
        I & 0 & \cdots & \cdots & \cdots &0   &0 &\cdots &\cdots\\
        -\frac{2A}{\alpha_A} & I & \ddots & \ddots & \ddots & \vdots   &\vdots &\vdots &\vdots \\
        I & -\frac{2A}{\alpha_A} & I & \ddots &\ddots & \vdots   &\vdots &\vdots &\vdots \\
        0 & I & \frac{-2A}{\alpha_A} & I & \ddots & \vdots   &\vdots &\vdots &\vdots \\
        \vdots & \ddots & \ddots & \ddots & \ddots & 0   &\vdots &\vdots &\vdots \\
        0 & \cdots & 0 & I & -\frac{2A}{\alpha_A} & I   &0 &\cdots &\cdots \\
        \hdashedline
        0 & \cdots &\cdots &0 & 0 & -I   &I &0 &\cdots \\
        0 & \cdots & \cdots &\cdots & 0 & 0   &-I &I &\ddots \\
        \vdots & \vdots & \vdots & \vdots & \vdots & \vdots   &\ddots &\ddots &\ddots \\
    \end{bNiceArray},\\
\end{aligned}
\end{equation}
which has inverse
\begin{equation}
\label{eq:pad_a_inv}
\NiceMatrixOptions
{
    custom-line = 
    {
        letter = I , 
        command = hdashedline , 
        tikz = {dashed,dash phase=3pt} ,
        width = \pgflinewidth
    }
}
    \mathbf{Pad}(A)^{-1}=
    \begin{bNiceArray}{ccccIccc}
        \mathbf{U}_{0}\left(\frac{A}{\alpha_A}\right) & 0 & \cdots & 0 & 0 & \cdots & \cdots \\
        \mathbf{U}_{1}\left(\frac{A}{\alpha_A}\right) & \mathbf{U}_{0}\left(\frac{A}{\alpha_A}\right) & \ddots & \vdots & \vdots & \vdots & \vdots \\
        \vdots & \ddots & \ddots & 0 & \vdots & \vdots & \vdots \\
        \mathbf{U}_{n-1}\left(\frac{A}{\alpha_A}\right) & \cdots & \mathbf{U}_{1}\left(\frac{A}{\alpha_A}\right) & \mathbf{U}_{0}\left(\frac{A}{\alpha_A}\right) & 0 & \cdots & \cdots \\
        \hdashedline
        \mathbf{U}_{n-1}\left(\frac{A}{\alpha_A}\right) & \cdots & \mathbf{U}_{1}\left(\frac{A}{\alpha_A}\right) & \mathbf{U}_{0}\left(\frac{A}{\alpha_A}\right) & I & 0 & \cdots \\
        \mathbf{U}_{n-1}\left(\frac{A}{\alpha_A}\right) & \cdots & \mathbf{U}_{1}\left(\frac{A}{\alpha_A}\right) & \mathbf{U}_{0}\left(\frac{A}{\alpha_A}\right) & I & I & \ddots \\
        \vdots & \vdots & \vdots & \vdots & \vdots & \vdots & \ddots 
    \end{bNiceArray}.
\end{equation}
The numerator
\begin{equation}
\begin{aligned}
\label{eq:pad_b}
    \mathbf{Pad}(B)&=\ketbra{0}{0}\otimes B_{11}+\ketbra{1}{1}\otimes I_{\eta n}\otimes I\\
    &=\ketbra{0}{0}\otimes\frac{I_n-L_n^2}{2}\otimes I
    +\ketbra{1}{1}\otimes I_{\eta n}\otimes I
\end{aligned}
\end{equation}
will be bundled with the initial state, which is now augmented with an additional ancilla $\ket{0}$:
\begin{equation}
\label{eq:pad_psi}
    \ket{0}\frac{\sum_{k=0}^{n-1}\widetilde{\beta}_k\ket{n-1-k}}{\norm{\widetilde \beta}}\ket{\psi}.
\end{equation}
Here, we have slightly abused the notation and used $\ketbra{0}{0}$, $\ketbra{0}{1}$, and $\ketbra{1}{1}$ to index matrix blocks with different sizes.

Ignoring the normalization factor, we obtain
\begin{equation}
\label{eq:pad_output}
\begin{aligned}
    &\ \mathbf{Pad}(A)^{-1}\mathbf{Pad}(B)
    \left(\ket{0}\sum_{k=0}^{n-1}\widetilde{\beta}_k\ket{n-1-k}\ket{\psi}\right)\\
    =&\ \left(\ketbra{0}{0}\otimes A_{11}^{-1}+\ketbra{1}{0}\otimes\left(-A_{22}^{-1}A_{21}A_{11}^{-1}\right)+\ketbra{1}{1}\otimes A_{22}^{-1}\right)
    \left(\ket{0}B_{11}\sum_{k=0}^{n-1}\widetilde{\beta}_k\ket{n-1-k}\ket{\psi}\right)\\
    =&\ \ket{0}\sum_{l=0}^{n-1}\ket{l}
        \sum_{k=n-1-l}^{n-1}\widetilde{\beta}_k\widetilde{\mathbf{T}}_{k+l-n+1}\left(\frac{A}{\alpha_A}\right)\ket{\psi}
        +\sum_{s=1}^{\eta}\ket{s}\sum_{l=0}^{n-1}\ket{l}
        \sum_{k=0}^{n-1}\widetilde{\beta}_k\widetilde{\mathbf{T}}_{k}\left(\frac{A}{\alpha_A}\right)\ket{\psi},
\end{aligned}
\end{equation}
which pads the original state
\begin{equation}
\begin{aligned}
    A_{11}^{-1}B_{11}\left(\sum_{k=0}^{n-1}\widetilde{\beta}_k\ket{n-1-k}\ket{\psi}\right)
    &=\frac{I_n\otimes I-L_n^2\otimes I}{2(I_n\otimes I+L_n^2\otimes I-2L_n\otimes \frac{A}{\alpha_A})}\left(\sum_{k=0}^{n-1}\widetilde{\beta}_k\ket{n-1-k}\ket{\psi}\right)\\
    &=\left(\sum_{j=0}^{n-1}L_n^j\otimes \widetilde{\mathbf{T}}_j\left(\frac{A}{\alpha_A}\right)\right)
    \left(\sum_{k=0}^{n-1}\widetilde{\beta}_k\ket{n-1-k}\ket{\psi}\right)\\
    &=\sum_{l=0}^{n-1}\ket{l}
        \sum_{k=n-1-l}^{n-1}\widetilde{\beta}_k\widetilde{\mathbf{T}}_{k+l-n+1}\left(\frac{A}{\alpha_A}\right)\ket{\psi}
\end{aligned}
\end{equation}
as desired. This can alternatively be understood using the fact that 
\begin{equation}
\NiceMatrixOptions
{
    custom-line = 
    {
        letter = I , 
        command = hdashedline , 
        tikz = {dashed,dash phase=3pt} ,
        width = \pgflinewidth
    }
}
    \mathbf{Pad}(A)^{-1}\mathbf{Pad}(B)=
    \begin{bNiceArray}{ccccIccc}
        \frac{I}{2} & 0 & \cdots & 0 & 0 & \cdots & \cdots \\
        \mathbf{T}_{1}\left(\frac{A}{\alpha_A}\right) & \frac{I}{2} & \ddots & \vdots & \vdots & \vdots & \vdots \\
        \vdots & \ddots & \ddots & 0 & \vdots & \vdots & \vdots \\
        \mathbf{T}_{n-1}\left(\frac{A}{\alpha_A}\right) & \cdots & \mathbf{T}_{1}\left(\frac{A}{\alpha_A}\right) & \frac{I}{2} & 0 & \cdots & \cdots \\
        \hdashedline
        \mathbf{T}_{n-1}\left(\frac{A}{\alpha_A}\right) & \cdots & \mathbf{T}_{1}\left(\frac{A}{\alpha_A}\right) & \frac{I}{2} & I & 0 & \cdots \\
        \mathbf{T}_{n-1}\left(\frac{A}{\alpha_A}\right) & \cdots & \mathbf{T}_{1}\left(\frac{A}{\alpha_A}\right) & \frac{I}{2} & I & I & \ddots \\
        \vdots & \vdots & \vdots & \vdots & \vdots & \vdots & \ddots 
    \end{bNiceArray}
\end{equation}
pads the original matrix Chebyshev generating function \eq{matrix_cheby_gen}.

Note that there is no need to worry about the normalization factor of the initial or the output state, simply because such a factor is automatically canceled out by the quantum linear system solver~\lem{opt_lin} which outputs a normalized quantum state. However, we must consider the normalization factor for block encoding $\mathbf{Pad}(A)$ as that will affect the condition number of the linear system. We discuss this issue in the next subsection.

\subsection{Block encoding implementation}
\label{sec:history_block}

Our goal is to invert the denominator of the padded matrix Chebyshev generating function $\mathbf{Pad}(A)$ in \eq{pad_a} using the optimal quantum linear system algorithm~\lem{opt_lin}. The complexity of this inversion further depends on the specific way in which we perform the block encoding. So in this subsection, we describe a block encoding of the padded matrix $\mathbf{Pad}(A)$ with an efficient circuit implementation.

Let us first simplify $\mathbf{Pad}(A)$ as
\begin{equation}
\label{eq:pad_a_simplify}
\begin{aligned}
    \mathbf{Pad}(A)
    &=\ketbra{0}{0}\otimes\left(I_n\otimes I+L_n^2\otimes I-2L_n\otimes \frac{A}{\alpha_A}\right)\\
    &\quad-\ketbra{1}{0}\otimes\ketbra{0}{n-1}\otimes I
    +\ketbra{1}{1}\otimes(I_{\eta n}-L_{\eta n})\otimes I\\
    &=\left(\ketbra{0}{0}\otimes I_n\otimes I
    +\ketbra{1}{1}\otimes I_{\eta n}\otimes I\right)
    +\ketbra{0}{0}\otimes L_n^2\otimes I\\
    &\quad-\left(\ketbra{0}{0}\otimes L_n\otimes \frac{2A}{\alpha_A}
    +\ketbra{1}{0}\otimes\ketbra{0}{n-1}\otimes I
    +\ketbra{1}{1}\otimes L_{\eta n}\otimes I\right)\\
    &=I_{\eta+1}\otimes I_n\otimes I
    +\ketbra{0}{0}\otimes L_n^2\otimes I\\
    &\quad-\left(\ketbra{0}{0}\otimes I_n\otimes\frac{2A}{\alpha_A}+\ketbra{1}{1}\otimes I_{\eta n}\otimes I\right)\\
    &\qquad\cdot\left(\ketbra{0}{0}\otimes L_n\otimes I+\ketbra{1}{0}\otimes\ketbra{0}{n-1}\otimes I+\ketbra{1}{1}\otimes L_{\eta n}\otimes I\right)\\
    &=I_{\eta+1}\otimes I_n\otimes I
    +\ketbra{0}{0}\otimes L_n^2\otimes I\\
    &\quad-\left(\ketbra{0}{0}\otimes I_n\otimes \frac{2A}{\alpha_A}+\sum_{s=1}^\eta\ketbra{s}{s}\otimes I_n\otimes I\right)
    \left(L_{(\eta+1)n}\otimes I\right),\\
\end{aligned}
\end{equation}
where we have slightly abused the Dirac notation to represent block matrices of different sizes and applied the matrix equality $\ketbra{1}{1}\otimes I_{\eta n}
    =\sum_{s=1}^\eta\ketbra{s}{s}\otimes I_n$.
Alternatively, we can understand the above simplification through a direct matrix computation. For instance, if $n=3$ and $\eta=1$,
\begin{equation}
\NiceMatrixOptions
{
    custom-line = 
    {
        letter = I , 
        command = hdashedline , 
        tikz = {dashed,dash phase=3pt} ,
        width = \pgflinewidth
    }
}
\begin{aligned}
    \begin{bNiceArray}{cccIccc}
        I & & & & &\\
        -\frac{2A}{\alpha_A} & I & & & &\\
        I & -\frac{2A}{\alpha_A} & I & & &\\
        \hdashedline
        & & -I & I & &\\
        & & & -I & I &\\
        & & & & -I & I\\
    \end{bNiceArray}
    &=
    \begin{bNiceArray}{cccIccc}
        I & & & & &\\
        & I & & & &\\
        & & I & & &\\
        \hdashedline
        & & & I & &\\
        & & & & I &\\
        & & & & & I\\
    \end{bNiceArray}
    +
    \begin{bNiceArray}{cccIccc}
        0 & & & & &\\
        0 & 0 & & & &\\
        I & 0 & 0 & & &\\
        \hdashedline
        & & & 0 & &\\
        & & & & 0 &\\
        & & & & & 0\\
    \end{bNiceArray}\\
    &\quad-
    \begin{bNiceArray}{cccIccc}
        \frac{2A}{\alpha_A} & & & & &\\
        & \frac{2A}{\alpha_A} & & & &\\
        & & \frac{2A}{\alpha_A} & & &\\
        \hdashedline
        & & & I & &\\
        & & & & I &\\
        & & & & & I\\
    \end{bNiceArray}
    \begin{bNiceArray}{cccIccc}
        0 & & & & &\\
        I & 0 & & & &\\
        & I & 0 & & &\\
        \hdashedline
        & & I & 0 & &\\
        & & & I & 0 &\\
        & & & & I & 0\\
    \end{bNiceArray}.
\end{aligned}
\end{equation}
Then we need to block encode the lower shift matrix $L$, and further combine it with $A/\alpha_A$ by applying matrix arithmetics within the block encoding framework.

Without loss of generality, let us consider an $n$-by-$n$ lower shift matrix $L_n$. To block encode $L_n$ and its $j$th power (for $j=0,1,\ldots,n-1$)
\begin{equation}
    L_n=\sum_{k=0}^{n-2}\ketbra{k+1}{k},\qquad
    L_n^j=\sum_{k=0}^{n-1-j}\ketbra{k+j}{k},
\end{equation}
we enlarge the Hilbert space and consider the cyclic shift operator
\begin{equation}
\begin{aligned}
    X_{2n}&=\sum_{k=0}^{2n-2}\ketbra{k+1}{k}+\ketbra{0}{2n-1},\\
    X_{2n}^j&=\sum_{k=0}^{2n-1-j}\ketbra{k+j}{k}+\sum_{k=2n-j}^{2n-1}\ketbra{k+j-2n}{k}
    =\sum_{k=0}^{2n-1}\ketbra{\mathbf{Mod}_{2n}(k+j)}{k}.
\end{aligned}
\end{equation}
The operator $X_{2n}$ is unitary and can be used to block encode $L_{n}$, as long as we ``zero out'' some of its redundant entries.
This is achieved by the following comparing operation
\begin{equation}
    \mathrm{CMP}_{2n,2}=\sum_{k=0}^{n-1}\ketbra{k}{k}\otimes I_2+\sum_{k=n}^{2n-1}\ketbra{k}{k}\otimes X_2,
\end{equation}
whose action is to flag an ancilla qubit (using the regular Pauli gate $X$) when the given register overflows. Note that when $n$ is a power of $2$, this is simply a CNOT gate with the most significant bit of $k$ as the control and the ancilla qubit as the target. Our observation is that the cyclic shift operator $X_{2n}$ and the lower shift operator $L_{n}$ are equivalent, up to the action of the comparator.

\begin{lemma}[Lower shift matrix block encoding]
\label{lem:shift}
For integers $j=0,1,\ldots,n-1$, it holds
\begin{equation}
    L_{n}^j=\left(I_{2n}\otimes\bra{0}\right)\mathrm{CMP}_{2n,2}\left(I_{2n}\otimes\ket{0}\right)X_{2n}^j\left(I_{2n}\otimes\bra{0}\right)\mathrm{CMP}_{2n,2}\left(I_{2n}\otimes\ket{0}\right).
\end{equation}
\end{lemma}
\begin{proof}
    Note that
    \begin{equation}
        \left(I_{2n}\otimes\bra{0}\right)\mathrm{CMP}_{2n,2}\left(I_{2n}\otimes\ket{0}\right)
        =\sum_{k=0}^{n-1}\ketbra{k}{k}.
    \end{equation}
    Therefore,
    \begin{equation}
    \begin{aligned}
        &\ \left(I_{2n}\otimes\bra{0}\right)\mathrm{CMP}_{2n,2}\left(I_{2n}\otimes\ket{0}\right)X_{2n}^j\left(I_{2n}\otimes\bra{0}\right)\mathrm{CMP}_{2n,2}\left(I_{2n}\otimes\ket{0}\right)\\
        =&\ \left(\sum_{k_1=0}^{n-1}\ketbra{k_1}{k_1}\right)
        \left(\sum_{k_2=0}^{2n-1}\ketbra{\mathbf{Mod}_{2n}(k_2+j)}{k_2}\right)
        \left(\sum_{k_3=0}^{n-1}\ketbra{k_3}{k_3}\right)\\
        =&\ \sum_{\substack{k_2=0,\ldots,n-1\\\mathbf{Mod}_{2n}(k_2+j)=0,\ldots,n-1}}\ketbra{\mathbf{Mod}_{2n}(k_2+j)}{k_2}.\\
    \end{aligned}
    \end{equation}
    For $j,k_2=0,1,\ldots,n-1$, it always holds $k_2+j\leq 2n-2$, which then implies
    \begin{equation}
        \sum_{\substack{k_2=0,\ldots,n-1\\\mathbf{Mod}_{2n}(k_2+j)=0,\ldots,n-1}}\ketbra{\mathbf{Mod}_{2n}(k_2+j)}{k_2}
        =\sum_{\substack{k_2=0,\ldots,n-1\\k_2+j=0,\ldots,n-1}}\ketbra{k_2+j}{k_2}
        =\sum_{k_2=0}^{n-1-j}\ketbra{k_2+j}{k_2}
        =L_n^j,
    \end{equation}
    as claimed.
\end{proof}

The above construction can be understood through a direct matrix calculation. For instance, if $n=3$ and $j=1$, we have
\begin{equation}
\NiceMatrixOptions
{
    custom-line = 
    {
        letter = I , 
        command = hdashedline , 
        tikz = {dashed,dash phase=3pt} ,
        width = \pgflinewidth
    }
}
    \begin{bmatrix}
        0 & 0 & 0\\
        1 & 0 & 0\\
        0 & 1 & 0\\
    \end{bmatrix}
    =
    \begin{bNiceArray}{cccIccc}
        1 & & & 0 & &\\
        & 1 & & & 0 &\\
        & & 1 & & & 0\\
    \end{bNiceArray}
    \begin{bNiceArray}{cccIccc}
        0 & & & & & 1\\
        1 & 0 & & & &\\
        & 1 & 0 & & &\\
        \hdashedline
        & & 1 & 0 & &\\
        & & & 1 & 0 &\\
        & & & & 1 & 0\\
    \end{bNiceArray}
    \begin{bNiceArray}{ccc}
        1 & & \\
        & 1 & \\
        & & 1 \\
        \hdashedline
        0 & &\\
        & 0 &\\
        & & 0\\
    \end{bNiceArray}.
\end{equation}
Note that the result not only works for $L_n$ but also its powers $I_n,L_n,\ldots,L_n^{n-1}$. This general version will be used in \sec{fourier} to construct an efficient circuit that creates $n$ Fourier coefficients with only $\sim\polylog(n)$ gates. Alternatively, we can also realize the block encoding as the linear combination
\begin{equation}
\begin{aligned}
    L_n^j&=\sum_{k=0}^{n-1-j}\ketbra{k+j}{k}
    =\frac{1+1}{2}\sum_{k=0}^{n-1-j}\ketbra{k+j}{k}+\frac{1-1}{2}\sum_{k=n-j}^{n-1}\ketbra{k+j-n}{k}\\
    &=\frac{1}{2}\left(\sum_{k=0}^{n-1}\ketbra{\mathbf{Mod}_{n}(k+j)}{k}
    +\sum_{k=0}^{n-1}(-1)^{\mathbf{Ind}_{n-j\leq k\leq n-1}(k)}\ketbra{\mathbf{Mod}_{n}(k+j)}{k}\right),
\end{aligned}
\end{equation}
with the indicator function $\mathbf{Ind}_{n-j\leq k\leq n-1}(k)$ implemented by inequality testings~\cite{IneqTest19}, exemplified by the matrix computation
\begin{equation}
    \begin{bmatrix}
        0 & 0 & 0\\
        1 & 0 & 0\\
        0 & 1 & 0\\
    \end{bmatrix}
    =\frac{1}{2}\left(
    \begin{bmatrix}
        0 & 0 & 1\\
        1 & 0 & 0\\
        0 & 1 & 0\\
    \end{bmatrix}
    +
    \begin{bmatrix}
        0 & 0 & -1\\
        1 & 0 & 0\\
        0 & 1 & 0\\
    \end{bmatrix}
    \right).
\end{equation}
In any case, let us focus on $j=1$ for now and conclude that the $n$-by-$n$ lower shift matrix $L_n$ can be block encoded with normalization factor $1$.

To complete the block encoding of $\mathbf{Pad}(A)$ in \eq{pad_a_simplify}, we consider the rewriting
\begin{equation}
\begin{aligned}
    \ketbra{0}{0}\otimes I_n\otimes \frac{2A}{\alpha_A}+\sum_{s=1}^\eta\ketbra{s}{s}\otimes I_n\otimes I
    &=\frac{3}{2}\left(\ketbra{0}{0}\otimes I_n\otimes \frac{A}{\alpha_A}+\sum_{s=1}^\eta\ketbra{s}{s}\otimes I_n\otimes I\right)\\
    &\quad+\frac{1}{2}\left(\ketbra{0}{0}\otimes I_n\otimes \frac{A}{\alpha_A}-\sum_{s=1}^\eta\ketbra{s}{s}\otimes I_n\otimes I\right).
\end{aligned}
\end{equation}
This means we can implement the block encoding of $\left(\ketbra{0}{0}\otimes I_n\otimes \frac{2A}{\alpha_A}+\sum_{s=1}^\eta\ketbra{s}{s}\otimes I_n\otimes I\right)/2$ with normalization factor $2$. Putting it altogether, we conclude that $\mathbf{Pad}(A)/4$ can be block encoded with a normalization factor of $4$, using one query to the controlled block encoding of $A/\alpha_A$.

\subsection{Summary of Chebyshev history state generation}
\label{sec:history_summary}

We now summarize the quantum algorithm for generating the Chebyshev history state.

\begin{enumerate}
    \item Construct a block encoding of the lower shift matrix $L$ using \lem{shift}.
    \item Combine $L$ and the input matrix $A$ to construct $\mathbf{Pad}(A)/4$ of \eq{pad_a_simplify} within the block encoding framework.
    \item Invoke the quantum linear system algorithm~\lem{opt_lin} with $\mathbf{Pad}(A)/4$ as the coefficient matrix, and $\ket{0}\frac{1}{\alpha_{\widetilde\beta}}\sum_{k=0}^{n-1}(\widetilde\beta_k-\widetilde\beta_{k+2})\ket{n-1-k}\ket{\psi}$ as the initial state $\left(\alpha_{\widetilde\beta}=\sqrt{\sum_{k=0}^{n-1}|\widetilde{\beta}_k-\widetilde{\beta}_{k+2}|^2}\right)$.
\end{enumerate}

\begin{theorem}[Chebyshev history state generation]
\label{thm:generate_history}
    Let $A$ be a square matrix with only real eigenvalues, such that $A/\alpha_A$ is block encoded by $O_A$ with some normalization factor $\alpha_A\geq\norm{A}$. 
    Let $O_\psi\ket{0}=\ket{\psi}$ be the oracle preparing the initial state, 
    and $O_{\widetilde\beta}\ket{0}=\frac{1}{\alpha_{\widetilde\beta}}\sum_{k=0}^{n-1}(\widetilde\beta_k-\widetilde\beta_{k+2})\ket{n-1-k}$ be the oracle preparing the shifting of coefficients $\widetilde\beta_k$ ($k=0,\ldots,n-1$) with  $\alpha_{\widetilde\beta}=\sqrt{\sum_{k=0}^{n-1}|\widetilde{\beta}_k-\widetilde{\beta}_{k+2}|^2}$. 
    Then, the quantum state
    \begin{equation}
        \frac{\ket{0}\sum_{l=0}^{n-1}\ket{l}
        \sum_{k=n-1-l}^{n-1}\widetilde{\beta}_k\widetilde{\mathbf{T}}_{k+l-n+1}\left(\frac{A}{\alpha_A}\right)\ket{\psi}
        +\sum_{s=1}^{\eta}\ket{s}\sum_{l=0}^{n-1}\ket{l}
        \sum_{k=0}^{n-1}\widetilde{\beta}_k\widetilde{\mathbf{T}}_{k}\left(\frac{A}{\alpha_A}\right)\ket{\psi}}
        {\sqrt{\sum_{l=0}^{n-1}\norm{\sum_{k=n-1-l}^{n-1}\widetilde{\beta}_k\widetilde{\mathbf{T}}_{k+l-n+1}\left(\frac{A}{\alpha_A}\right)\ket{\psi}}^2
        +\eta n\norm{\sum_{k=0}^{n-1}\widetilde{\beta}_k\widetilde{\mathbf{T}}_{k}\left(\frac{A}{\alpha_A}\right)\ket{\psi}}^2}},
    \end{equation}
    can be prepared with accuracy $\epsilon$ and probability $1-\pfail$ using
    \begin{equation}
        \mathbf{O}\left(\alphaU n(\eta+1)\log\left(\frac{1}{\epsilon}\right)\log\left(\frac{1}{\pfail}\right)\right),
    \end{equation}
    queries to controlled-$O_A$, controlled-$O_\psi$, controlled-$O_{\widetilde\beta}$, and their inverses, where 
    \begin{equation}
    \label{eq:alphaU2}
        \alphaU\geq\max_{j=0,1,\ldots,n-1}\norm{\mathbf{U}_{j}\left(\frac{A}{\alpha_A}\right)},
    \end{equation}
    is an upper bound on Chebyshev polynomials of the second kind $\mathbf{U}_j(x)$.
\end{theorem}
\begin{proof}
    The analysis of \sec{history_block} shows that $\mathbf{Pad}(A)/4$ can be block encoded with $1$ query to the controlled block encoding of $A/\alpha_A$ (i.e., the controlled-$O_A$).
    The quantum linear system algorithm of \lem{opt_lin} outputs a state $\epsilon$-close to
    \begin{equation}
        \frac{\frac{4}{\mathbf{Pad}(A)}\left(\ket{0}\frac{\sum_{k=0}^{n-1}(\widetilde\beta_k-\widetilde\beta_{k+2})\ket{n-1-k}}{\alpha_{\widetilde\beta}}\ket{\psi}\right)}
        {\norm{\frac{4}{\mathbf{Pad}(A)}\left(\ket{0}\frac{\sum_{k=0}^{n-1}(\widetilde\beta_k-\widetilde\beta_{k+2})\ket{n-1-k}}{\alpha_{\widetilde\beta}}\ket{\psi}\right)}}
        =\frac{\mathbf{Pad}(A)^{-1}\mathbf{Pad}(B)\left(\ket{0}\sum_{k=0}^{n-1}\widetilde{\beta}_k\ket{n-1-k}\ket{\psi}\right)}
        {\norm{\mathbf{Pad}(A)^{-1}\mathbf{Pad}(B)\left(\ket{0}\sum_{k=0}^{n-1}\widetilde{\beta}_k\ket{n-1-k}\ket{\psi}\right)}},
    \end{equation}
    which follows from the fact that
    \begin{equation}
    \label{eq:shift_coeff}
        (I_n-L_n^2)\sum_{k=0}^{n-1}\widetilde{\beta}_k\ket{n-1-k}
        =\sum_{k=0}^{n-1}\widetilde\beta_k\ket{n-1-k}
    -\sum_{k=2}^{n-1}\widetilde\beta_k\ket{n+1-k}
        =\sum_{k=0}^{n-1}(\widetilde\beta_k-\widetilde\beta_{k+2})\ket{n-1-k},
    \end{equation}
    under the convention $\widetilde\beta_{n}=\widetilde{\beta}_{n+1}=\cdots=0$.
    This is exactly the padded Chebyshev history state because of \eq{pad_output}.
    
    We have the norm bound on the inverse padded matrix
    \begin{equation}
        \norm{\mathbf{Pad}(A)^{-1}}=\mathbf{O}\left((\eta+1)n\alphaU\right),
    \end{equation}
    for any upper bound $\alphaU\geq\max_{j=0,1,\ldots,n-1}\norm{\mathbf{U}_{j}\left(\frac{A}{\alpha_A}\right)}$ on Chebyshev polynomials of the second kind, which follows from \lem{block_norm} and the matrix representation of $\mathbf{Pad}(A)^{-1}$ in \eq{pad_a_inv}.
    The claimed complexity now follows from \eq{opt_lin_cost}.
\end{proof}
\begin{remark}
    For the purpose of generality, we have expressed the complexity of our algorithm in terms of $\alphaU$, which shares a similar spirit with recent results on solving linear differential equations~\cite{Krovi2023improvedquantum}. This analysis can be further refined when the algorithm is applied in a concrete setting. For instance, if the input matrix has the Jordan form decomposition $A/\alpha_A=SJS^{-1}$ with upper bound $\kappa_S\geq\norm{S}\norm{S}^{-1}$ on the Jordan condition number and size $d_{\max}$ of the largest Jordan block, then it holds $\alphaU=\mathbf{O}\left(n^{d_{\max-1}}\kappa_S\right)$ provided that $\alpha_A\geq2\norm{A}$, which is always achievable per the rescaling technique of \eq{block_rescaling}. In particular, we have $\alphaU=\mathbf{O}\left(\kappa_S\right)$ for diagonalizable matrices and the complexity becomes
    \begin{equation}
        \mathbf{O}\left(\kappa_Sn(\eta+1)\log\left(\frac{1}{\epsilon}\right)\right).
    \end{equation}
    See \append{analysis_cheby_bernstein} for more details.

    The generation of Chebyshev history state relies on preparation of the shifted coefficients $\frac{1}{\alpha_{\widetilde\beta}}\sum_{k=0}^{n-1}(\widetilde\beta_k-\widetilde\beta_{k+2})\ket{n-1-k}$ with $\alpha_{\widetilde\beta}=\sqrt{\sum_{k=0}^{n-1}|\widetilde{\beta}_k-\widetilde{\beta}_{k+2}|^2}$. This is an $n$-dimensional quantum state, and can thus be prepared using the conventional approach~\cite{ShendeBullockMarkov06} with gate complexity $\mathbf{\Theta}(n)$. However, our $n$ typically scales polynomially with the input parameters, leading to a considerable overhead. Fortunately, this overhead is avoidable for both QEVE and QEVT. For the eigenvalue estimation, we set $\widetilde\beta_k=0$ if $k\neq n-1$, so this is actually a $2$-dimensional state which can be prepared with $\mathbf{O}(1)$ gates. For the eigenvalue transformation, we can treat $\widetilde\beta_k$ as the coefficients from a truncated Fourier expansion and coherently implement a convolution in the frequency domain, which has a cost of $\mathbf{O}\left(\polylog(n)\right)$. See \sec{est} and \sec{fourier} for further details.

    Finally, note that the query complexity of initial state preparation can be improved using the block preconditioning technique of~\cite{OptInit}.
\end{remark}

\section{Quantum eigenvalue estimation}
\label{sec:est}
With the Chebyshev history state at our disposal, we now present a solution to the quantum eigenvalue estimation problem. We start by introducing the \emph{centered modulus} together with its properties in \sec{est_cmod} which is useful for our algorithmic analysis. We develop a variant of phase estimation in \sec{est_cheby_phase} to extract the phase information from a Chebyshev state, establishing \thm{chebyshev_qpe}. This essentially solves the eigenvalue estimation problem when the input state is an exact eigenstate. We then consider the general case of imperfect eigenstate in \sec{est_imperfect}. Finally, we state and analyze the quantum eigenvalue estimation algorithm in \sec{est_summary}, obtaining \thm{qeve}.

\subsection{Centered modulus and its properties}
\label{sec:est_cmod}

In our analysis of the eigenvalue estimation algorithm, we will make extensive use of the \emph{centered modulus} which we now introduce. Fixing a positive number $q>0$, every real number $x\in\mathbb{R}$ can be uniquely written as
\begin{equation}
    x=kq+r,\qquad
    k\in\mathbb Z,\
    -\frac{q}{2}\leq r<\frac{q}{2}.
\end{equation}
Indeed, the fact that such an expansion exists follows by taking
\begin{equation}
    k=\mathbf{Floor}\left(\frac{x+\frac{q}{2}}{q}\right),\qquad
    r=x-q\mathbf{Floor}\left(\frac{x+\frac{q}{2}}{q}\right)
\end{equation}
with $\mathbf{Floor}(\cdot)$ the largest integer not exceeding the input.
As for the uniqueness, assume that $k_1q+r_1=k_2q+r_2$, which implies $(k_1-k_2)q=r_2-r_1$. Now the requirement $-q<r_2-r_1=(k_1-k_2)q<q$ forces that $k_1=k_2$ and $r_1=r_2$. We can thus well define
\begin{equation}
    \mathbf{CMod}_{q}(x)=r=x-q\mathbf{Floor}\left(\frac{x+\frac{q}{2}}{q}\right)\in\left[-\frac{q}{2},\frac{q}{2}\right)
\end{equation}
as the \emph{centered modulus} of $x$ modulo $q$, which is basically $x$ modulo $q$ with offset $-\frac{q}{2}$.

In the following, we include a list of properties of the centered modulus which are useful for our analysis.

\begin{lemma}[Properties of centered modulus]
The following properties hold for the centered modulus:
\begin{enumerate}
    \item Periodicity: $\mathbf{CMod}_q(x+lq)=\mathbf{CMod}_q(x)$ for $l\in\mathbb Z$.
    \item Boundedness: $\abs{\mathbf{CMod}_q(x)}\leq\abs{x}$.
    \item Negation: $\abs{\mathbf{CMod}_q(-x)}=\abs{\mathbf{CMod}_q(x)}$.
    \item Positive scaling: $\mathbf{CMod}_q(cx)=c\mathbf{CMod}_{q/c}(x)$ for $c>0$.
    \item Triangle inequality: $\abs{\mathbf{CMod}_q(x+y)}\leq\abs{\mathbf{CMod}_q(x)}+\abs{\mathbf{CMod}_q(y)}$.
\end{enumerate}
\end{lemma}
\begin{proof}
Periodicity directly follows from definition of the centered modulus.
To see the second property, note that $\abs{\mathbf{CMod}_q(x)}\leq\frac{q}{2}$ always holds true. So if $\abs{x}\geq\frac{q}{2}$, there is nothing to prove. But the remaining case $\abs{x}<\frac{q}{2}$ means that $\mathbf{CMod}_q(x)=x$, so the claimed bound holds in both cases. 
For the third property, note that
\begin{equation}
    x=kq+r,\qquad
    k\in\mathbb Z,\
    -\frac{q}{2}\leq r<\frac{q}{2}
    \quad\Rightarrow\quad
    -x=(-k)q+(-r),\qquad
    -k\in\mathbb Z,\
    -\frac{q}{2}< -r\leq\frac{q}{2}.
\end{equation}
Thus if $x\neq(k-\frac{1}{2})q$, we have $\mathbf{CMod}_q(-x)=-\mathbf{CMod}_q(x)$. But on the other hand if $x=(k-\frac{1}{2})q$, we have $\mathbf{CMod}_q(-x)=-\frac{q}{2}=\mathbf{CMod}_q(x)$, so the claimed equality holds in both cases.
The positive scaling property follows from the observation that
\begin{equation}
    x=k\frac{q}{c}+r,\qquad
    k\in\mathbb Z,\
    -\frac{q}{2c}\leq r<\frac{q}{2c}
    \quad\Rightarrow\quad
    cx=kq+cr,\qquad
    k\in\mathbb Z,\
    -\frac{q}{2}\leq cr<\frac{q}{2}.
\end{equation}
To prove the triangle inequality, assume that
\begin{equation}
    x=jq+r,\
    y=kq+s,\qquad
    j,k\in\mathbb Z,\
    -\frac{q}{2}\leq r,s<\frac{q}{2}.
\end{equation}
Then, we have
\begin{equation}
    \abs{\mathbf{CMod}_q(x+y)}
    =\abs{\mathbf{CMod}_q(r+s)}
    \leq\abs{r+s}
    \leq\abs{r}+\abs{s}
    =\abs{\mathbf{CMod}_q(x)}+\abs{\mathbf{CMod}_q(y)},
\end{equation}
which completes the proof.
\end{proof}

Although absolute value of the centered modulus is not a valid norm due to the potential change of modulus in the rescaling, it does induce a natural distance metric $\abs{\mathbf{CMod}_q(x-y)}$ for the interval $\left[-\frac{q}{2},\frac{q}{2}\right)$ under the periodic boundary condition. In particular, it satisfies the following defining properties of a distance metric:
\begin{enumerate}
    \item Nonnegativity: $\abs{\mathbf{CMod}_q(x-y)}\geq0$.
    \item Positivity: $\abs{\mathbf{CMod}_q(x-y)}=0$ if and only if $\mathbf{CMod}_q(x)=\mathbf{CMod}_q(y)$.
    \item Symmetry: $\abs{\mathbf{CMod}_q(x-y)}=\abs{\mathbf{CMod}_q(y-x)}$.
    \item Triangle inequality: $\abs{\mathbf{CMod}_q(x-y)}\leq\abs{\mathbf{CMod}_q(x-z)}+\abs{\mathbf{CMod}_q(z-y)}$.
\end{enumerate}

Centered modulus can be used to significantly simplify reasonings about trigonometric functions. In particular, we will use the following two inequalities in analyzing the eigenvalue estimation algorithm:
\begin{lemma}[Trigonometric bounds with centered modulus]
The following bounds on trigonometric functions hold:
\begin{enumerate}
    \item $\abs{\sin(x)}\geq\frac{2}{\pi}\abs{\mathbf{CMod}_{\pi}(x)}$.
    \item $\abs{\cos(x)-\cos(y)}\leq\abs{\mathbf{CMod}_{2\pi}(x-y)}$.
\end{enumerate}
\end{lemma}
\begin{proof}
To prove the first inequality, we assume that $x=k\pi+r$ for $k\in\mathbb Z$ and $-\frac{\pi}{2}\leq r<\frac{\pi}{2}$.
Then,
\begin{equation}
    \abs{\sin{x}}
    =\abs{\sin(k\pi+r)}
    =\abs{\sin(r)}
    \geq\frac{2}{\pi}\abs{r}
    =\frac{2}{\pi}\abs{\mathbf{CMod}_{\pi}(x)}.
\end{equation}
The second inequality follows from the fact that
\begin{equation}
\begin{aligned}
    \abs{\cos(x)-\cos(y)}
    &=\abs{\Re(e^{ix}-e^{iy})}
    \leq\abs{e^{ix}-e^{iy}}
    =\abs{e^{i(x-y)}-1}\\
    &=\abs{e^{i\mathbf{CMod}_{2\pi}(x-y)}-1}
    =\abs{\int_{0}^{\mathbf{CMod}_{2\pi}(x-y)}\mathrm{d}\theta\
    ie^{i\theta}}
    \leq\abs{\mathbf{CMod}_{2\pi}(x-y)},
\end{aligned}
\end{equation}
where $\Re$ is the real part of a complex number.
\end{proof}

\subsection{Chebyshev state phase estimation}
\label{sec:est_cheby_phase}

Recall that in the eigenvalue estimation problem, we are given a matrix $A$ with only real eigenvalues, which has an eigenstate $\ket{\psi_\lambda}$ corresponding to the eigenvalue $\lambda$, such that $A\ket{\psi_\lambda}=\lambda\ket{\psi_\lambda}$. Our goal is to estimate $\lambda$ assuming that the input state $\ket{\psi}$ is sufficiently close to $\ket{\psi_\lambda}$.

We start by preparing a Chebyshev state encoding the target eigenvalue in the phase of the coefficients. Specifically, we invoke \thm{generate_history} with $\eta=0$ and
\begin{equation}
\label{eq:coeff_est}
    \widetilde{\beta}_k=
    \begin{cases}
        1,\quad&k=n-1,\\
        0,\quad&k\neq n-1.
    \end{cases}
\end{equation}
This allows us to prepare the following version of the Chebyshev history state
    \begin{equation}
        \frac{\sum_{l=0}^{n-1}\ket{l}\widetilde{\mathbf{T}}_l\left(\frac{A}{\alpha_A}\right)\ket{\psi}}{\norm{\sum_{l=0}^{n-1}\ket{l}\widetilde{\mathbf{T}}_l\left(\frac{A}{\alpha_A}\right)\ket{\psi}}}
    \end{equation}
with accuracy $\delta$ using 
\begin{equation}
    \mathbf{O}\left(\alphaU n\log\left(\frac{1}{\delta}\right)\right)
\end{equation}
queries to the oracle $O_A$ that block encodes $A/\alpha_A$ and the oracle $O_\psi$ that prepares the initial state $\ket{\psi}$. The complexity reduces to $\mathbf{O}\left(\kappa_Sn\log\left(\frac{1}{\delta}\right)\right)$ for diagonalizable matrices with upper bound $\kappa_S$ on the Jordan condition number.

The state we prepare is close to
\begin{equation}
    \frac{\sum_{j=0}^{n-1}\mathbf{T}_j\left(\frac{\lambda}{\alpha_A}\right)\ket{j}}{\norm{\sum_{j=0}^{n-1}\mathbf{T}_j\left(\frac{\lambda}{\alpha_A}\right)\ket{j}}}
    =\frac{1}{\sqrt{\alpha_\phi}}\sum_{j=0}^{n-1}\cos(2\pi j\phi)\ket{j}
    =\frac{1}{\sqrt{\alpha_\phi}}\sum_{j=0}^{n-1}\frac{e^{i2\pi j\phi}+e^{-i2\pi j\phi}}{2}\ket{j},
\end{equation}
where $\phi=\frac{1}{2\pi}\arccos{\frac{\lambda}{\alpha_A}}$ and $\alpha_\phi$ is the corresponding normalization factor. In light of the rescaling trick in \eq{block_rescaling}, we may assume $\alpha_A\geq2\norm{A}$ without loss of generality, which implies $\frac{\lambda}{\alpha_A}\in\left[-\frac{1}{2},\frac{1}{2}\right]$ and
\begin{equation}
    \phi\in\left[\frac{1}{6},\frac{1}{3}\right].
\end{equation}
In the following, we develop a variant of the quantum phase estimation algorithm for estimating the unknown $\phi$ given a Chebyshev state as above.
For presentational purpose, we will first assume that the input state is exactly the eigenstate and the Chebyshev history state can be prepared perfectly. The general case will be handled in the next subsection.

Let us first analyze the normalization factor $\alpha_\phi$. We have
\begin{equation}
    \alpha_\phi=\sum_{j=0}^{n-1}\cos^2(2\pi j\phi)
    =\sum_{j=0}^{n-1}\frac{\cos(4\pi j\phi)+1}{2}
    =\frac{n}{2}+\frac{1}{4}\sum_{j=0}^{n-1}\left(e^{i4\pi j\phi}+e^{-i4\pi j\phi}\right).
\end{equation}
Since $e^{i4\pi\phi}\neq1$ by our assumption,
\begin{equation}
\begin{aligned}
    \alpha_\phi&=\frac{n}{2}+\frac{1}{4}\left(\frac{1-e^{i4\pi n\phi}}{1-e^{i4\pi\phi}}+\frac{1-e^{-i4\pi n\phi}}{1-e^{-i4\pi\phi}}\right)\\
    &=\frac{n}{2}+\frac{1}{4}\left(\frac{e^{-i2\pi n\phi}-e^{i2\pi n\phi}}{e^{-i2\pi\phi}-e^{i2\pi\phi}}e^{i2\pi(n-1)\phi}+\frac{e^{i2\pi n\phi}-e^{-i2\pi n\phi}}{e^{i2\pi\phi}-e^{-i2\pi\phi}}e^{-i2\pi(n-1)\phi}\right)\\
    &=\frac{n}{2}+\frac{1}{2}\frac{\sin(2\pi n\phi)}{\sin(2\pi\phi)}\cos((n-1)2\pi\phi).\\
\end{aligned}
\end{equation}
Note that because $\phi\in\left[\frac{1}{6},\frac{1}{3}\right]$, we have $\frac{\sqrt{3}}{2}\leq\sin(2\pi\phi)\leq1$, which implies $\abs{\alpha_\phi-\frac{n}{2}}\leq\frac{1}{2\frac{\sqrt{3}}{2}}
    =\frac{\sqrt{3}}{3}$ and
\begin{equation}
    \alpha_\phi\in\left[\frac{n}{2}-\frac{\sqrt{3}}{3},\frac{n}{2}+\frac{\sqrt{3}}{3}\right].
\end{equation}
We formulate this observation in terms of properties of Chebyshev polynomials as follows.
\begin{lemma}[$\ell_2$-norm bounds for Chebyshev polynomials]
\label{lem:cheby_l2}
    For any $x\in\left[-\frac{1}{2},\frac{1}{2}\right]$, it holds that
    \begin{equation}
    \begin{aligned}
        \frac{n}{2}-\frac{\sqrt{3}}{3}
        &\leq\sum_{j=0}^{n-1}\mathbf{T}_j^2(x)
        \leq\frac{n}{2}+\frac{\sqrt{3}}{3},\\
        \frac{n}{2}-\frac{\sqrt{3}}{3}-\frac{3}{4}
        &\leq\sum_{j=0}^{n-1}\widetilde{\mathbf{T}}_j^2(x)
        \leq\frac{n}{2}+\frac{\sqrt{3}}{3}-\frac{3}{4}.
    \end{aligned}
    \end{equation}
\end{lemma}

Applying the quantum Fourier transform $F\ket{j}=\frac{1}{\sqrt{n}}\sum_{l=0}^{n-1}e^{-i\frac{2\pi jl}{n}}\ket{l}$, we obtain
\begin{equation}
\begin{aligned}
    F\frac{1}{\sqrt{\alpha_\phi}}\sum_{j=0}^{n-1}\cos(2\pi j\phi)\ket{j}
    &=\frac{1}{\sqrt{\alpha_\phi}}\sum_{j=0}^{n-1}\frac{e^{i2\pi j\phi}+e^{-i2\pi j\phi}}{2}
    \frac{1}{\sqrt{n}}\sum_{l=0}^{n-1}e^{-i\frac{2\pi jl}{n}}\ket{l}\\
    &=\frac{1}{2\sqrt{\alpha_\phi n}}\sum_{l=0}^{n-1}\sum_{j=0}^{n-1}\left(e^{i2\pi j(\phi-\frac{l}{n})}+e^{-i2\pi j(\phi+\frac{l}{n})}\right)\ket{l}.
\end{aligned}
\end{equation}
Because $\phi\in\left[\frac{1}{6},\frac{1}{3}\right]$ and $\frac{l}{n}\in[0,1)$, the first summation is degenerate when $\phi=\frac{l}{n}$, whereas the second summation is degenerate when $\phi=\frac{n-l}{n}$.
For the case where $\phi=\frac{l}{n}$,
\begin{equation}
    F\frac{1}{\sqrt{\alpha_\phi}}\sum_{j=0}^{n-1}\cos(2\pi j\phi)\ket{j}
    =\frac{1}{2\sqrt{\alpha_\phi n}}\left(n\ket{l}+n\ket{n-l}\right).
\end{equation}
Here, the normalization factor is
\begin{equation}
    \alpha_\phi=\frac{n}{2}+\frac{1}{2}\frac{\sin\left(2\pi n\frac{l}{n}\right)}{\sin\left(2\pi\frac{l}{n}\right)}\cos\left((n-1)2\frac{l}{n}\phi\right)
    =\frac{n}{2}.
\end{equation}
So we actually obtain the state
\begin{equation}
    \frac{1}{\sqrt{2}}\left(\ket{l}+\ket{n-l}\right),
\end{equation}
from which the phase/eigenvalue can be recovered deterministically. A similar analysis applies to $\phi=\frac{n-l}{n}$.

Assuming this is not the case hereafter, we have
\begin{equation}
\begin{aligned}
    F\frac{1}{\sqrt{\alpha_\phi}}\sum_{j=0}^{n-1}\cos(2\pi j\phi)\ket{j}
    &=\frac{1}{2\sqrt{\alpha_\phi n}}\sum_{l=0}^{n-1}\sum_{j=0}^{n-1}\left(e^{i2\pi j(\phi-\frac{l}{n})}+e^{-i2\pi j(\phi+\frac{l}{n})}\right)\ket{l}\\
    &=\frac{1}{2\sqrt{\alpha_\phi n}}\sum_{l=0}^{n-1}\left(\frac{1-e^{i2\pi(n\phi-l)}}{1-e^{i2\pi(\phi-\frac{l}{n})}}+\frac{1-e^{-i2\pi(n\phi+l)}}{1-e^{-i2\pi(\phi+\frac{l}{n})}}\right)\ket{l}\\
    &=\frac{1}{2\sqrt{\alpha_\phi n}}\sum_{l=0}^{n-1}\bigg(\frac{\sin\pi(n\phi-l)}{\sin\pi(\phi-\frac{l}{n})}e^{i\pi(n-1)(\phi-\frac{l}{n})}\\
    &\qquad\qquad\qquad\quad+\frac{\sin\pi(n\phi+l)}{\sin\pi(\phi+\frac{l}{n})}e^{-i\pi(n-1)(\phi+\frac{l}{n})}\bigg)\ket{l}.\\
\end{aligned}
\end{equation}
In other words,
\begin{equation}
    F\frac{1}{\sqrt{\alpha_\phi}}\sum_{j=0}^{n-1}\cos(2\pi j\phi)\ket{j}
    =\frac{1}{2\sqrt{\alpha_\phi n}}\sum_{l=0}^{n-1}\alpha_l\ket{l},
\end{equation}
where
\begin{equation}
    \abs{\alpha_l}^2
    \leq2\left(\frac{\sin^2\pi(n\phi-l)}{\sin^2\pi(\phi-\frac{l}{n})}+\frac{\sin^2\pi(n\phi+l)}{\sin^2\pi(\phi+\frac{l}{n})}\right)
    \leq2\bigg(\underbrace{\frac{1}{\sin^2\pi(\phi-\frac{l}{n})}}_{\alpha_{1,l}^2}+\underbrace{\frac{1}{\sin^2\pi(\phi+\frac{l}{n})}}_{\alpha_{2,l}^2}\bigg).
\end{equation}

Let us see that we can approximately recover the eigenvalue if $l$ is close to $\pm\mathbf{Floor}(n\phi)$ in centered modulus. Indeed,
\begin{equation}
\begin{aligned}
    &\quad \abs{\mathbf{CMod}_{n}(l\pm\mathbf{Floor}(n\phi))}\leq n_1-1
    \quad \Leftrightarrow\quad \frac{1}{n}\abs{\mathbf{CMod}_{n}(l\pm\mathbf{Floor}(n\phi))}\leq\frac{n_1-1}{n}\\
    \Leftrightarrow&\quad \abs{\mathbf{CMod}_{1}\left(\frac{l\pm\mathbf{Floor}(n\phi)}{n}\right)}\leq\frac{n_1-1}{n}
    \quad \Rightarrow\quad \abs{\mathbf{CMod}_{1}\left(\frac{l}{n}\pm\phi\right)}\leq\frac{n_1}{n}\\
    \Rightarrow&\quad \abs{\cos\left(2\pi\frac{l}{n}\right)-\cos(2\pi\phi)}
    \leq\abs{\mathbf{CMod}_{2\pi}\left(2\pi\frac{l}{n}\pm2\pi\phi\right)}
    \leq2\pi\abs{\mathbf{CMod}_{1}\left(\frac{l}{n}\pm\phi\right)}
    \leq\frac{2\pi n_1}{n}\\
    \Rightarrow&\quad \abs{\alpha_A\cos\left(2\pi\frac{l}{n}\right)-\lambda}\leq\frac{2\pi\alpha_A n_1}{n},
\end{aligned}
\end{equation}
where in the second line we have used the fact that
\begin{equation}
    \abs{\mathbf{CMod}_1\left(\frac{\pm n\phi\mp\mathbf{Floor}(n\phi)}{n}\right)}
    \leq\abs{\frac{n\phi-\mathbf{Floor}(n\phi)}{n}}\leq\frac{1}{n}.
\end{equation}
We now set 
\begin{equation}
    n=n_0n_1.
\end{equation}
To ensure that the eigenvalue is estimated to accuracy $\epsilon$, it suffices to take
\begin{equation}
    \frac{2\pi\alpha_A}{n_0}\leq\epsilon\quad\Rightarrow\quad
    n_0\geq\frac{2\pi\alpha_A}{\epsilon}.
\end{equation}

In the following, we will derive tail bounds for $\alpha_{1,l}^2$ and $\alpha_{2,l}^2$, which we then use to analyze the success probability of the algorithm. To this end, we consider those $l$ whose centered modulus to $\pm\mathbf{Floor}(n\phi)$ is larger than some threshold value $n_1-1$. We have
\begin{equation}
\begin{aligned}
    &\ \frac{1}{2\alpha_\phi n}\sum_{\abs{\mathbf{CMod}_{n}(l-\mathbf{Floor}(n\phi))}>n_1-1}\alpha_{1,l}^2\\
    =&\ \frac{1}{2\alpha_\phi n}\sum_{\abs{\mathbf{CMod}_{n}(l-\mathbf{Floor}(n\phi))}>n_1-1}\frac{1}{\sin^2\pi(\phi-\frac{l}{n})}\\
    \leq&\ \frac{1}{2\alpha_\phi n}\sum_{\abs{\mathbf{CMod}_{n}(l-\mathbf{Floor}(n\phi))}>n_1-1}\frac{1}{\frac{4}{\pi^2}\abs{\mathbf{CMod}_\pi(\pi(\phi-\frac{l}{n}))}^2}\\
    =&\ \frac{1}{8\alpha_\phi n}\sum_{\abs{\mathbf{CMod}_{n}(l-\mathbf{Floor}(n\phi))}>n_1-1}\frac{1}{\abs{\mathbf{CMod}_1(\phi-\frac{l}{n})}^2}\\
    \leq&\ \frac{1}{8\alpha_\phi n}\sum_{\abs{\mathbf{CMod}_{n}(l-\mathbf{Floor}(n\phi))}>n_1-1}\frac{1}{\left(
    \abs{\mathbf{CMod}_1\left(\frac{\mathbf{Floor}(n\phi)-l}{n}\right)}
    -\abs{\mathbf{CMod}_1\left(\frac{n\phi-\mathbf{Floor}(n\phi)}{n}\right)}
    \right)^2}.
\end{aligned}
\end{equation}
Since
\begin{equation}
\begin{aligned}
    \abs{\mathbf{CMod}_1\left(\frac{n\phi-\mathbf{Floor}(n\phi)}{n}\right)}
    &\leq\abs{\frac{n\phi-\mathbf{Floor}(n\phi)}{n}}
    \leq\frac{1}{n},\\
    \abs{\mathbf{CMod}_1\left(\frac{\mathbf{Floor}(n\phi)-l}{n}\right)}
    &=\frac{1}{n}\abs{\mathbf{CMod}_n\left(\mathbf{Floor}(n\phi)-l\right)},
\end{aligned}
\end{equation}
this implies that
\begin{equation}
\begin{aligned}
    &\ \frac{1}{2\alpha_\phi n}\sum_{\abs{\mathbf{CMod}_{n}(l-\mathbf{Floor}(n\phi))}>n_1-1}\alpha_{1,l}^2\\
    \leq&\ \frac{n}{8\alpha_\phi}\sum_{\abs{\mathbf{CMod}_{n}(l-\mathbf{Floor}(n\phi))}>n_1-1}\frac{1}{\left(\abs{\mathbf{CMod}_n\left(\mathbf{Floor}(n\phi)-l\right)}-1\right)^2}\\
    \leq&\ \frac{n}{4\alpha_\phi}\sum_{j=n_1-1}^{\infty}\frac{1}{j^2}
    \leq\frac{n}{4\alpha_\phi}\int_{n_1-2}^{\infty}\mathrm{d}x\ \frac{1}{x^2}
    =\frac{n}{4\alpha_\phi(n_1-2)}.
\end{aligned}
\end{equation}
Similarly,
\begin{equation}
    \frac{1}{2\alpha_\phi n}\sum_{\abs{\mathbf{CMod}_{n}(l+\mathbf{Floor}(n\phi))}>n_1-1}\alpha_{2,l}^2
    \leq\frac{n}{4\alpha_\phi(n_1-2)}.
\end{equation}

Now, our success probability can be lower bounded as
\begin{equation}
\begin{aligned}
    &\ \mathbf{P}\left(\abs{\alpha_A\cos\left(2\pi\frac{l}{n}\right)-\lambda}\leq\epsilon\right)
    \geq\mathbf{P}\left(\abs{\alpha_A\cos\left(2\pi\frac{l}{n}\right)-\lambda}\leq\frac{2\pi\alpha_An_1}{n}\right)\\
    \geq&\ \mathbf{P}\Big(\abs{\mathbf{CMod}_{n}(l-\mathbf{Floor}(n\phi))}\leq n_1-1\text{ OR }\abs{\mathbf{CMod}_{n}(l+\mathbf{Floor}(n\phi))}\leq n_1-1\Big)\\
    =&\ 1-\mathbf{P}\Big(\abs{\mathbf{CMod}_{n}(l-\mathbf{Floor}(n\phi))}>n_1-1\text{ AND }\abs{\mathbf{CMod}_{n}(l+\mathbf{Floor}(n\phi))}>n_1-1\Big)\\
    =&\ 1-\frac{1}{4\alpha_\phi n}\sum_{\substack{\abs{\mathbf{CMod}_{n}(l-\mathbf{Floor}(n\phi))}>n_1-1\\
    \text{AND}\\
    \abs{\mathbf{CMod}_{n}(l+\mathbf{Floor}(n\phi))}>n_1-1}}\abs{\alpha_{l}}^2
    \geq1-\frac{1}{2\alpha_\phi n}\sum_{\substack{\abs{\mathbf{CMod}_{n}(l-\mathbf{Floor}(n\phi))}>n_1-1\\
    \text{AND}\\
    \abs{\mathbf{CMod}_{n}(l+\mathbf{Floor}(n\phi))}>n_1-1}}\left(\alpha_{1,l}^2+\alpha_{2,l}^2\right)\\
    \geq&\ 1-\frac{1}{2\alpha_\phi n}\sum_{\substack{\abs{\mathbf{CMod}_{n}(l-\mathbf{Floor}(n\phi))}>n_1-1}}\alpha_{1,l}^2
    -\frac{1}{2\alpha_\phi n}\sum_{\substack{\abs{\mathbf{CMod}_{n}(l+\mathbf{Floor}(n\phi))}>n_1-1}}\alpha_{2,l}^2\\
    \geq&\ 1-\frac{n}{2\alpha_\phi(n_1-2)},
\end{aligned}
\end{equation}
where the failure probability is further upper bounded as
\begin{equation}
\begin{aligned}
    \frac{n}{2\alpha_\phi(n_1-2)}
    \leq\frac{n}{\left(n-\frac{2\sqrt{3}}{3}\right)(n_1-2)}
    =\frac{n_1}{\left(n_1-\frac{2\sqrt{3}}{3n_0}\right)(n_1-2)}
    \leq\frac{n_1}{\left(n_1-\frac{2\sqrt{3}}{3}\right)(n_1-2)}.\\
\end{aligned}
\end{equation}
This is $\leq0.433$ for $n_1\geq5$. In other words, we can choose $n_1$ sufficiently large to succeed with a probability strictly greater than $\frac{1}{2}$. The success probability can then be boosted to at least $1-\pfail$ using the median amplification by repeating the algorithm $\mathbf{O}\left(\log(1/\pfail)\right)$ times. We summarize the core idea of this analysis as follows:

\begin{theorem}[Chebyshev state phase estimation]
\label{thm:chebyshev_qpe}
    Given Chebyshev state $\frac{1}{\sqrt{\alpha_\phi}}\sum_{l=0}^{n-1}\cos\left(2\pi l\phi\right)\ket{l}$ with $\phi\in\left[\frac{1}{6},\frac{1}{3}\right]$, there exists a quantum algorithm that uses one copy of the state and outputs a value $l\in\{0,\ldots,n-1\}$ satisfying
    \begin{equation}
        \abs{\mathbf{CMod}_{1}\left(\frac{l}{n}\pm\phi\right)}\leq\frac{n_1}{n},
    \end{equation}
    with probability at least
    \begin{equation}
        1-\frac{n_1}{\left(n_1-\frac{2\sqrt{3}}{3}\right)(n_1-2)}.
    \end{equation}
    The algorithm performs a quantum Fourier transform followed by a measurement in the computational basis.
    For $n_1\geq5$, the success probability is at least $1-\frac{5}{15-2\sqrt{3}}\approx0.566$ strictly larger than $\frac{1}{2}$.
\end{theorem}

\subsection{Analysis of imperfect eigenstate}
\label{sec:est_imperfect}

In describing the Chebyshev state phase estimation algorithm, we have assumed that a perfect eigenstate is given a prior. In this subsection, we discuss how this assumption can be relaxed to allow for imperfect eigenstates, which is more common in practical applications.

Specifically, the error comes from the following three sources:
\begin{enumerate}
    \item the quantum linear system solver we use only outputs an approximate solution state;
    \item the initial state $\ket{\psi}$ only approximates an eigenstate $\ket{\psi_\lambda}$; and
    \item the output state from the linear system solver corresponds to the rescaled Chebyshev polynomials $\widetilde{\mathbf{T}}_j$, which approximates that of the regular Chebyshev polynomials $\mathbf{T}_j$.
\end{enumerate}
The first error is easy to analyze. If $\ket{\psi}$ is the input and $\ket{\varphi}$ is the output state of the quantum linear system solver, then we have
\begin{equation}
    \norm{\ket{\varphi}
    -\frac{\sum_{j=0}^{n-1}\ket{j}\widetilde{\mathbf{T}}_j\left(\frac{A}{\alpha_A}\right)\ket{\psi}}{\norm{\sum_{j=0}^{n-1}\ket{j}\widetilde{\mathbf{T}}_j\left(\frac{A}{\alpha_A}\right)\ket{\psi}}}}
    \leq\epsilon_{\text{lin}}
\end{equation}
where $\epsilon_{\text{lin}}$ is the accuracy of the linear system algorithm.

The second error is essentially the error of solving linear equations with an imperfect initial state. This can be analyzed as follows.
\begin{lemma}
	Let $C$ and $\widetilde{C}$ be invertible matrices of the same size. It holds that
	\begin{equation}
		\widetilde{C}^{-1}-C^{-1}=-\widetilde{C}^{-1}(\widetilde{C}-C)C^{-1}
		=-C^{-1}(\widetilde{C}-C)\widetilde{C}^{-1}.
	\end{equation}
\end{lemma}
\begin{corollary}[Quantum linear system with perturbation]
\label{cor:lin_sys_perturb}
	Let $C$ and $\widetilde C$ be invertible matrices of the same size, acting on (normalized) quantum states $\ket{\psi}$ and $\ket{\widetilde\psi}$. We have
	\begin{equation}
		\norm{\frac{\widetilde{C}^{-1}\ket{\widetilde\psi}}{\norm{\widetilde{C}^{-1}\ket{\widetilde\psi}}}-\frac{C^{-1}\ket{\psi}}{\norm{C^{-1}\ket{\psi}}}}
		\leq
		\frac{2\norm{C^{-1}}\norm{\ket{\widetilde\psi}-\ket{\psi}}}{\norm{C^{-1}\ket{\psi}}}
		+\frac{2\norm{\widetilde C^{-1}}\norm{C^{-1}}\norm{\widetilde C-C}}{\norm{C^{-1}\ket{\psi}}}.
	\end{equation}
\end{corollary}
\begin{proof}
    We use the triangle inequality to upper bound the left-hand side as
    \begin{equation}
    \begin{aligned}
        \norm{\frac{C^{-1}\ket{\psi}}{\norm{C^{-1}\ket{\psi}}}-\frac{\widetilde{C}^{-1}\ket{\widetilde\psi}}{\norm{\widetilde{C}^{-1}\ket{\widetilde\psi}}}}
		\leq\norm{\frac{C^{-1}\ket{\psi}}{\norm{C^{-1}\ket{\psi}}}-\frac{\widetilde{C}^{-1}\ket{\widetilde\psi}}{\norm{C^{-1}\ket{\psi}}}}
		+\norm{\frac{\widetilde{C}^{-1}\ket{\widetilde\psi}}{\norm{C^{-1}\ket{\psi}}}-\frac{\widetilde{C}^{-1}\ket{\widetilde\psi}}{\norm{\widetilde{C}^{-1}\ket{\widetilde\psi}}}}.
    \end{aligned}
    \end{equation}
    For the first term, we have
    \begin{equation}
        \norm{\frac{C^{-1}\ket{\psi}}{\norm{C^{-1}\ket{\psi}}}-\frac{\widetilde{C}^{-1}\ket{\widetilde\psi}}{\norm{C^{-1}\ket{\psi}}}}
        =\frac{\norm{C^{-1}\ket{\psi}-\widetilde{C}^{-1}\ket{\widetilde\psi}}}{\norm{C^{-1}\ket{\psi}}},
    \end{equation}
    whereas the second term can be further bounded similarly as
    \begin{equation}
        \norm{\frac{\widetilde{C}^{-1}\ket{\widetilde\psi}}{\norm{C^{-1}\ket{\psi}}}-\frac{\widetilde{C}^{-1}\ket{\widetilde\psi}}{\norm{\widetilde{C}^{-1}\ket{\widetilde\psi}}}}
        =\norm{\widetilde{C}^{-1}\ket{\widetilde\psi}}\abs{\frac{1}{\norm{C^{-1}\ket{\psi}}}-\frac{1}{\norm{\widetilde{C}^{-1}\ket{\widetilde\psi}}}}
        \leq\frac{\norm{C^{-1}\ket{\psi}-\widetilde{C}^{-1}\ket{\widetilde\psi}}}{\norm{C^{-1}\ket{\psi}}}.
    \end{equation}
    Thus, it remains to analyze $\frac{\norm{C^{-1}\ket{\psi}-\widetilde{C}^{-1}\ket{\widetilde\psi}}}{\norm{C^{-1}\ket{\psi}}}$.

    For the denominator, one can further bound
    \begin{equation}
    \label{eq:denominator_bnd}
        1=\norm{CC^{-1}\ket{\psi}}\leq\norm{C}\norm{C^{-1}\ket{\psi}}
        \quad\Rightarrow\quad
        \frac{1}{\norm{C^{-1}\ket{\psi}}}\leq\norm{C}.
    \end{equation}
    But we will keep it for the time being, as $\norm{C^{-1}\ket{\psi}}$ represents size of the solution vector, which we may have direct knowledge about in applications.
    As for the numerator,
    \begin{equation}
    \begin{aligned}
        \norm{C^{-1}\ket{\psi}-\widetilde{C}^{-1}\ket{\widetilde\psi}}
        &\leq\norm{C^{-1}\ket{\psi}-C^{-1}\ket{\widetilde\psi}}
        +\norm{C^{-1}\ket{\widetilde\psi}-\widetilde{C}^{-1}\ket{\widetilde\psi}}\\
        &\leq\norm{C^{-1}}\norm{\ket{\psi}-\ket{\widetilde\psi}}+\norm{C^{-1}-\widetilde{C}^{-1}}.
    \end{aligned}
    \end{equation}
    The claimed bound now follows from the proceeding lemma.
\end{proof}
In applications where $\norm{\widetilde{C}^{-1}}$ is unknown, we may remove its dependence using a strategy similar to that for proving~\cite[Proposition 5.7.7]{humpherys2017foundations}.
Note that our above perturbation analysis is more general than is needed here, as it bounds the error of quantum linear system where the coefficient matrix and the input state can both be imperfect. This general bound will be used later in \sec{faber_history} to analyze complexity of generating the Faber history state.
For the time being, let us assume that the input state has distance $\norm{\ket{\psi}-\ket{\psi_\lambda}}\leq\epsilon_{\text{init}}$ to a true eigenstate.
By~\lem{cheby_l2}, we have $\norm{\sum_{j=0}^{n-1}\ket{j}\widetilde{\mathbf{T}}_j\left(\frac{A}{\alpha_A}\right)\ket{\psi_\lambda}}=\mathbf{\Theta}(\sqrt{n})$, which implies
\begin{equation}
    \norm{\frac{\sum_{j=0}^{n-1}\ket{j}\widetilde{\mathbf{T}}_j\left(\frac{A}{\alpha_A}\right)\ket{\psi}}{\norm{\sum_{j=0}^{n-1}\ket{j}\widetilde{\mathbf{T}}_j\left(\frac{A}{\alpha_A}\right)\ket{\psi}}}
    -\frac{\sum_{j=0}^{n-1}\ket{j}\widetilde{\mathbf{T}}_j\left(\frac{A}{\alpha_A}\right)\ket{\psi_\lambda}}{\norm{\sum_{j=0}^{n-1}\ket{j}\widetilde{\mathbf{T}}_j\left(\frac{A}{\alpha_A}\right)\ket{\psi_\lambda}}}}
    =\mathbf{O}\left(\sqrt{n}\alphaU\epsilon_{\text{init}}\right),
\end{equation}
where the second term further simplifies to
\begin{equation}
    \frac{\sum_{j=0}^{n-1}\ket{j}\widetilde{\mathbf{T}}_j\left(\frac{A}{\alpha_A}\right)\ket{\psi_\lambda}}{\norm{\sum_{j=0}^{n-1}\ket{j}\widetilde{\mathbf{T}}_j\left(\frac{A}{\alpha_A}\right)\ket{\psi_\lambda}}}
    =
    \frac{\sum_{j=0}^{n-1}\widetilde{\mathbf{T}}_j\left(\frac{\lambda}{\alpha_A}\right)\ket{j}\ket{\psi_\lambda}}{\norm{\sum_{j=0}^{n-1}\widetilde{\mathbf{T}}_j\left(\frac{\lambda}{\alpha_A}\right)\ket{j}\ket{\psi_\lambda}}}.
\end{equation}

Finally, we analyze the error of performing eigenvalue estimation on the rescaled Chebyshev state as opposed to the regular Chebyshev state. This is handled by the following bound.
\begin{corollary}[Distance between rescaled and regular Chebyshev states]
    For any $x\in\left[-\frac{1}{2},\frac{1}{2}\right]$, it holds that
    \begin{equation}
        \norm{\frac{\sum_{j=0}^{n-1}\widetilde{\mathbf{T}}_j(x)\ket{j}}{\sqrt{\sum_{k=0}^{n-1}\widetilde{\mathbf{T}}_k^2(x)}}
        -\frac{\sum_{j=0}^{n-1}\mathbf{T}_j(x)\ket{j}}{\sqrt{\sum_{k=0}^{n-1}\mathbf{T}_k^2(x)}}}
        \leq\frac{\frac{3}{8}}{\sqrt{\frac{n}{2}-\frac{\sqrt{3}}{3}}\sqrt{\frac{n}{2}-\frac{\sqrt{3}}{3}-\frac{3}{4}}}
        +\frac{\frac{1}{2}}{\sqrt{\frac{n}{2}-\frac{\sqrt{3}}{3}}}
        =\mathbf{O}\left(\frac{1}{\sqrt{n}}\right).
    \end{equation}
\end{corollary}
\begin{proof}
    We use the triangle inequality to get
    \begin{equation}
    \begin{aligned}
        &\ \norm{\frac{\sum_{j=0}^{n-1}\widetilde{\mathbf{T}}_j(x)\ket{j}}{\sqrt{\sum_{k=0}^{n-1}\widetilde{\mathbf{T}}_k^2(x)}}
        -\frac{\sum_{j=0}^{n-1}\mathbf{T}_j(x)\ket{j}}{\sqrt{\sum_{k=0}^{n-1}\mathbf{T}_k^2(x)}}}\\
        \leq&\ \norm{\frac{\sum_{j=0}^{n-1}\widetilde{\mathbf{T}}_j(x)\ket{j}}{\sqrt{\sum_{k=0}^{n-1}\widetilde{\mathbf{T}}_k^2(x)}}
        -\frac{\sum_{j=0}^{n-1}\widetilde{\mathbf{T}}_j(x)\ket{j}}{\sqrt{\sum_{k=0}^{n-1}\mathbf{T}_k^2(x)}}}
        +\norm{\frac{\sum_{j=0}^{n-1}\widetilde{\mathbf{T}}_j(x)\ket{j}}{\sqrt{\sum_{k=0}^{n-1}\mathbf{T}_k^2(x)}}
        -\frac{\sum_{j=0}^{n-1}\mathbf{T}_j(x)\ket{j}}{\sqrt{\sum_{k=0}^{n-1}\mathbf{T}_k^2(x)}}}\\
        =&\ \frac{\abs{\sqrt{\sum_{k=0}^{n-1}\mathbf{T}_k^2(x)}
        -\sqrt{\sum_{k=0}^{n-1}\widetilde{\mathbf{T}}_k^2(x)}}}{\sqrt{\sum_{k=0}^{n-1}\mathbf{T}_k^2(x)}}
        +\frac{\frac{1}{2}}{\sqrt{\sum_{k=0}^{n-1}\mathbf{T}_k^2(x)}}.
    \end{aligned}
    \end{equation}
    Here, the numerator of the first term can be further bounded by
    \begin{equation}
        \abs{\sqrt{\sum_{k=0}^{n-1}\mathbf{T}_k^2(x)}
        -\sqrt{\sum_{k=0}^{n-1}\widetilde{\mathbf{T}}_k^2(x)}}
        =\abs{\int_{\sum_{k=0}^{n-1}\widetilde{\mathbf{T}}_k^2(x)}^{\sum_{k=0}^{n-1}\mathbf{T}_k^2(x)}
        \mathrm{d}u\ \frac{1}{2\sqrt{u}}}
        \leq\frac{\frac{3}{8}}{\sqrt{\sum_{k=0}^{n-1}\widetilde{\mathbf{T}}_k^2(x)}}.
    \end{equation}
    The claimed bound now follows from \lem{cheby_l2}.
\end{proof}

Putting it altogether, we finally obtain that output state of the quantum linear system algorithm $\ket{\varphi}$ has error at most
\begin{equation}
\label{eq:est_imperfect}
    \norm{\ket{\varphi}
    -\frac{\sum_{j=0}^{n-1} \mathbf{T}_j\left(\frac{\lambda}{\alpha_A}\right)\ket{j}\ket{\psi_\lambda}}{\norm{\sum_{j=0}^{n-1} \mathbf{T}_j\left(\frac{\lambda}{\alpha_A}\right)\ket{j}\ket{\psi_\lambda}}}}
    =\epsilon_{\text{lin}}+\mathbf{O}\left(\sqrt{n}\alphaU\epsilon_{\text{init}}\right)+\mathbf{O}\left(\frac{1}{\sqrt{n}}\right).
\end{equation}

\subsection{Summary of quantum eigenvalue estimation}
\label{sec:est_summary}

We now summarize the quantum algorithm for estimating eigenvalues.
\begin{enumerate}
    \item If necessary, rescale the input block encoding using \eq{block_rescaling}, so that $\alpha_A\geq2\norm{A}$.
    \item Invoke the Chebyshev state generation algorithm \thm{generate_history} with $\eta=0$, coefficients $\widetilde{\beta}_k$ from \eq{coeff_est}, and a state $\ket{\psi}$ close to the target eigenstate.
    \item Perform the quantum Fourier transform and measure in the computational basis.
    \item Perform the median amplification to boost the success probability.
\end{enumerate}

\begin{theorem}[Quantum eigenvalue estimation]
\label{thm:qeve}
    Let $A$ be a square matrix with only real eigenvalues, such that $A/\alpha_A$ is block encoded by $O_A$ with some normalization factor $\alpha_A\geq\norm{A}$. 
    Suppose that oracle $O_\psi\ket{0}=\ket{\psi}$ prepares an initial state within distance $\norm{\ket{\psi}-\ket{\psi_{\lambda_j}}}=\mathbf{O}(\sqrt{\epsilon/\alpha_A}/\alphaU)$ from an eigenstate $\ket{\psi_{\lambda_j}}$ such that $A\ket{\psi_{\lambda_j}}=\lambda_j\ket{\psi_{\lambda_j}}$, where $\alphaU$ satisfies \eq{alphaU2} with
    \begin{equation}
        n=\mathbf{O}\left(\frac{\alpha_A}{\epsilon}\right).
    \end{equation}
    Then, the eigenvalue $\lambda_j$ can be estimated with accuracy $\epsilon$ and probability $1-\pfail$ using
    \begin{equation}
        \mathbf{O}\left(\frac{\alpha_A}{\epsilon}\alphaU\log\left(\frac{1}{\pfail}\right)\right)
    \end{equation}
    queries to controlled-$O_A$, controlled-$O_\psi$, and their inverses.
\end{theorem}
\begin{proof}
    Assume that the input matrix $A/\alpha_A$ is block encoded with normalization factor $\alpha_A\geq2\norm{A}$, that the input state $\ket{\psi}=\ket{\psi_{\lambda_j}}$ is the exact eigenstate, and that the quantum linear system solver makes no error, while ignoring the distinction between $\mathbf{T}_0(x)=1$ and $\widetilde{\mathbf{T}}_0(x)=\frac{1}{2}$. Then the Chebyshev state phase estimation of \sec{est_cheby_phase} (in particular, \thm{chebyshev_qpe}) shows that we can get a measurement outcome $l\in\{0,\ldots,n-1\}$ satisfying 
    \begin{equation}
        \abs{\alpha_A\cos\left(2\pi\frac{l}{n}\right)-{\lambda_j}}\leq\frac{2\pi\alpha_A n_1}{n}
    \end{equation}
    with probability at least
    \begin{equation}
        1-\frac{n_1}{\left(n_1-\frac{2\sqrt{3}}{3}\right)(n_1-2)}.
    \end{equation}
    For $n_1\geq5$, the success probability is at least $1-\frac{5}{15-2\sqrt{3}}\approx0.566$ strictly larger than $\frac{1}{2}$.
    We can then choose 
    \begin{equation}
        n=\mathbf{O}\left(\frac{\alpha_A}{\epsilon}\right)
    \end{equation}
    so that the target eigenvalue ${\lambda_j}$ is estimated with accuracy $\epsilon$.

    Now consider the general case. The analysis of \sec{est_imperfect} shows that output of the quantum linear system solver is close to the ideal Chebyshev state with Euclidean distance at most
    \begin{equation}
        \epsilon_{\text{lin}}+\mathbf{O}\left(\sqrt{n}\alphaU\epsilon_{\text{init}}\right)+\mathbf{O}\left(\frac{1}{\sqrt{n}}\right).
    \end{equation}
    But for two quantum states with distance $\norm{\ket{\varphi_1}-\ket{\varphi_2}}\leq\delta$, if we apply a unitary $U$ followed by an orthogonal projection $\Pi$, the amplitudes differ at most
\begin{equation}
    \abs{\norm{\Pi U\ket{\varphi_1}}-\norm{\Pi U\ket{\varphi_2}}}
    \leq\norm{\Pi U\ket{\varphi_1}-\Pi U\ket{\varphi_2}}\leq\delta,
\end{equation}
which implies difference of the probabilities 
\begin{equation}
    \abs{\norm{\Pi U\ket{\varphi_1}}^2-\norm{\Pi U\ket{\varphi_2}}^2}
    \leq2\abs{\norm{\Pi U\ket{\varphi_1}}-\norm{\Pi U\ket{\varphi_2}}}
    \leq2\delta.
\end{equation}
Thus if we have $\delta=\mathbf{O}(1)$ sufficiently small, we still guarantee a success probability strictly larger than $\frac{1}{2}$ in the Chebyshev state phase estimation, even after accounting for the success probability of the quantum linear system solver (\lem{opt_lin}), which can then be boosted using the median amplification. To achieve this, we let $\epsilon_{\text{lin}}=\mathbf{O}(1)$ sufficiently small, $\epsilon_{\text{init}}=\mathbf{O}(1/(\sqrt{n}\alphaU))$ sufficiently small, and $n=\mathbf{\Omega}(1)$ sufficiently large. However, we already have the stronger requirement $n=\mathbf{\Theta}\left(\alpha_A/\epsilon\right)$ to achieve the desired accuracy in the Chebyshev state phase estimation, so $\epsilon_{\text{init}}=\mathbf{O}(\sqrt{\epsilon/\alpha_A}/\alphaU)$.
Our proof is now complete with the claimed complexity from \thm{generate_history}.
\end{proof}
\begin{remark}
    The complexity of our algorithm depends on size $\alphaU$ of the input matrix under Chebyshev polynomials of the second kind. See the remark succeeding \thm{generate_history} for more discussions about this parameter. When the input matrix is diagonalizable with a known upper bound $\kappa_S$ on its Jordan condition number, we have $\alphaU=\mathbf{O}(\kappa_S)$ and the complexity becomes
    \begin{equation}
        \mathbf{O}\left(\frac{\alpha_A\kappa_S}{\epsilon}\log\left(\frac{1}{\pfail}\right)\right).
    \end{equation}
    This achieves the so-called \emph{Heisenberg scaling}~\cite{Giovannetti06,Zwierz10} in quantum metrology and is provably optimal for the eigenvalue estimation.

    In the actual circuit implementation, we need to prepare a quantum state $O_{\widetilde\beta}\ket{0}=\frac{1}{\alpha_{\widetilde\beta}}\sum_{k=0}^{n-1}(\widetilde\beta_k-\widetilde\beta_{k+2})\ket{n-1-k}$ encoding the shifted coefficients with normalization $\alpha_{\widetilde\beta}=\sqrt{\sum_{k=0}^{n-1}|\widetilde{\beta}_k-\widetilde{\beta}_{k+2}|^2}$. For eigenvalue estimation, we have all $\widetilde{\beta}_k=0$ except $\widetilde{\beta}_{n-1}=1$. The resulting state $\frac{\ket{0}-\ket{2}}{\sqrt{2}}$ is $2$-dimensional, and can be prepared with $\mathbf{O}(1)$ cost.

    Finally, note that the query complexity of initial state preparation can be improved using the block preconditioning technique of~\cite{OptInit}.
\end{remark}

\section{Quantum eigenvalue transformation}
\label{sec:transform}
In this section, we use the Chebyshev history state to construct a quantum algorithm that transforms eigenvalues of a high-dimensional input matrix. 
With all the technical preliminaries already in place, we describe this algorithm and analyze its complexity in \thm{qevt} of \sec{transform_summary}. We also describe a variant of the algorithm in \thm{qevt_block} of \sec{transform_block_summary} based on a block encoded quantum linear system solver, which can be useful when QEVT is used as a subroutine in desigining other quantum algorithms.

\subsection{Summary of quantum eigenvalue transformation}
\label{sec:transform_summary}

We now summarize the quantum algorithm for transforming eigenvalues.
\begin{enumerate}
    \item Invoke the Chebyshev state generation algorithm \thm{generate_history} with $\eta=1$, coefficients $\widetilde{\beta}_k$ from the expansion $p(x)=\sum_{k=0}^{n-1}\widetilde\beta_k\widetilde{\mathbf{T}}_{k}(x)$, and input state $\ket{\psi}$.
    \item Perform a fixed-point amplitude amplification on part of the state flagged by the ancilla $\ket{1}$.
\end{enumerate}

\begin{theorem}[Quantum eigenvalue transformation]
\label{thm:qevt}
    Let $A$ be a square matrix with only real eigenvalues, such that $A/\alpha_A$ is block encoded by $O_A$ with some normalization factor $\alpha_A\geq\norm{A}$.
    Let $p(x)=\sum_{k=0}^{n-1}\widetilde\beta_k\widetilde{\mathbf{T}}_{k}(x)=\sum_{k=0}^{n-1}\beta_k{\mathbf{T}}_{k}(x)$ be the Chebyshev expansion of a degree-($n-1$) polynomial $p$.
    Let $O_\psi\ket{0}=\ket{\psi}$ be the oracle preparing the initial state, 
    and $O_{\widetilde\beta}\ket{0}=\frac{1}{\alpha_{\widetilde\beta}}\sum_{k=0}^{n-1}(\widetilde\beta_k-\widetilde\beta_{k+2})\ket{n-1-k}$ be the oracle preparing the shifting of coefficients $\widetilde\beta_k$ ($k=0,\ldots,n-1$) with  $\alpha_{\widetilde\beta}=\sqrt{\sum_{k=0}^{n-1}|\widetilde{\beta}_k-\widetilde{\beta}_{k+2}|^2}$. 
    Then, the quantum state
    \begin{equation}
        \frac{p\left(\frac{A}{\alpha_A}\right)\ket{\psi}}{\norm{p\left(\frac{A}{\alpha_A}\right)\ket{\psi}}}
    \end{equation}
    can be prepared with accuracy $\epsilon$ and probability $1-\pfail$ using
    \begin{equation}
        \mathbf{O}\left(\frac{\alphaT}{\alphaPPsi}\alphaU n\log\left(\frac{\alphaT}{\alphaPPsi\epsilon}\right)\log\left(\frac{1}{\pfail}\right)\right)
    \end{equation}
    queries to controlled-$O_A$, controlled-$O_\psi$, controlled-$O_{\widetilde\beta}$, and their inverses,
    where $\alphaU$ is defined in \eq{alphaU2} and
    \begin{equation}
    \label{eq:alphaT_alphaPPsi}
        \alphaT\geq
        \max_{l=0,1,\ldots,n-1}\norm{\sum_{k=l}^{n-1}\widetilde\beta_k\widetilde{\mathbf{T}}_{k-l}\left(\frac{A}{\alpha_A}\right)\ket{\psi}},\qquad
        \alphaPPsi\leq\norm{p\left(\frac{A}{\alpha_A}\right)\ket{\psi}}
    \end{equation}
    are upper bound on the maximum shifted partial sum of the Chebyshev expansion and lower bound on the transformed state.
\end{theorem}
\begin{proof}
    With $\eta=1$, the Chebyshev history state from \thm{generate_history} reads
    \begin{equation}
        \frac{\ket{0}\sum_{l=0}^{n-1}\ket{l}
        \sum_{k=n-1-l}^{n-1}\widetilde{\beta}_k\widetilde{\mathbf{T}}_{k+l-n+1}\left(\frac{A}{\alpha_A}\right)\ket{\psi}
        +\ket{1}\sum_{l=0}^{n-1}\ket{l}
        \sum_{k=0}^{n-1}\widetilde{\beta}_k\widetilde{\mathbf{T}}_{k}\left(\frac{A}{\alpha_A}\right)\ket{\psi}}
        {\sqrt{\sum_{l=0}^{n-1}\norm{\sum_{k=n-1-l}^{n-1}\widetilde{\beta}_k\widetilde{\mathbf{T}}_{k+l-n+1}\left(\frac{A}{\alpha_A}\right)\ket{\psi}}^2
        +n\norm{\sum_{k=0}^{n-1}\widetilde{\beta}_k\widetilde{\mathbf{T}}_{k}\left(\frac{A}{\alpha_A}\right)\ket{\psi}}^2}}.
    \end{equation}
    Applying the fixed-point amplitude amplification on the ancilla state $\ket{1}$ then produces a quantum state close to
    \begin{equation}
        \frac{\sum_{k=0}^{n-1}\widetilde{\beta}_k\widetilde{\mathbf{T}}_{k}\left(\frac{A}{\alpha_A}\right)\ket{\psi}}{\norm{\sum_{k=0}^{n-1}\widetilde{\beta}_k\widetilde{\mathbf{T}}_{k}\left(\frac{A}{\alpha_A}\right)\ket{\psi}}}
        =\frac{p\left(\frac{A}{\alpha_A}\right)\ket{\psi}}{\norm{p\left(\frac{A}{\alpha_A}\right)\ket{\psi}}}
    \end{equation}
    as desired.

    To achieve a success probability at least $1-\pfail$, we take a number of amplification steps scaling like
    \begin{equation}
    \begin{aligned}
        &\mathbf{O}\left(\frac{\sqrt{\sum_{l=0}^{n-1}\norm{\sum_{k=n-1-l}^{n-1}\widetilde{\beta}_k\widetilde{\mathbf{T}}_{k+l-n+1}\left(\frac{A}{\alpha_A}\right)\ket{\psi}}^2
        +n\norm{\sum_{k=0}^{n-1}\widetilde{\beta}_k\widetilde{\mathbf{T}}_{k}\left(\frac{A}{\alpha_A}\right)\ket{\psi}}^2}}{\sqrt{n\norm{\sum_{k=0}^{n-1}\widetilde{\beta}_k\widetilde{\mathbf{T}}_{k}\left(\frac{A}{\alpha_A}\right)\ket{\psi}}^2}}\log\left(\frac{1}{\pfail}\right)\right)\\
        &=\mathbf{O}\left(\frac{\alphaT}{\alphaPPsi}\log\left(\frac{1}{\pfail}\right)\right),
    \end{aligned}
    \end{equation}
    where $\alphaT$ defined in \eq{alphaT_alphaPPsi} is an upper bound on the shifted partial sum of the Chebyshev expansion. To achieve a total accuracy $\epsilon$, we require each preparation of the Chebyshev history state to have error at most $\mathbf{O}\left(\alphaPPsi\epsilon/\alphaT\right)$. 
    The claimed complexity now follows from \thm{generate_history} and the above analysis.
\end{proof}
\begin{remark}
    The complexity of our algorithm depends on parameters such as $\alphaU$ and $\alphaT$, the former of which has already been discussed in the remark succeeding \thm{generate_history}. The parameter $\alphaT$ denotes maximum size of the shifted Chebyshev expansion, and it can be further upper bounded in terms of the Jordan condition number. For instance, if $A/\alpha_A=SJS^{-1}$ is the Jordan form decomposition of the input matrix, with an upper bound $\kappa_S\geq\norm{S}\norm{S}^{-1}$ on the Jordan condition number and $d_{\max}$ size of the largest Jordan block, then we show in \append{analysis_cheby_bernstein} that $\alphaT=\mathbf{O}\left(\kappa_Sn^{d_{\max}-1}\log(n)\norm{p}_{\max,[-1,1]}\right)$. In particular, we have $\alphaT=\mathbf{O}\left(\kappa_S\log(n)\norm{p}_{\max,[-1,1]}\right)$ for diagonalizable matrices, leading to the complexity
    \begin{equation}
        \mathbf{O}\left(\frac{\norm{p}_{\max,[-1,1]}\kappa_S^2n}{\norm{p\left(\frac{A}{\alpha_A}\right)\ket{\psi}}} \log\left(\frac{\norm{p}_{\max,[-1,1]}\kappa_S\log(n)}{\norm{p\left(\frac{A}{\alpha_A}\right)\ket{\psi}}\epsilon}\right)\log(n)\log\left(\frac{1}{\pfail}\right)\right)
    \end{equation}
    for a worst-case input. However, we further show that the $\log(n)$ factors can be shaved off
    \begin{equation}
        \mathbf{O}\left(\frac{\norm{p}_{\max,[-1,1]}\kappa_S^2n}{\norm{p\left(\frac{A}{\alpha_A}\right)\ket{\psi}}} \log\left(\frac{\norm{p}_{\max,[-1,1]}\kappa_S}{\norm{p\left(\frac{A}{\alpha_A}\right)\ket{\psi}}\epsilon}\right)\log\left(\frac{1}{\pfail}\right)\right)
    \end{equation}
    when running the algorithm on an average input matrix (\append{analysis_cheby_carleson}). 
    Specifically, if the distribution of eigenvalues of $A$ does not depend on the target polynomial degree $n$, then the vector norm $\norm{\sum_{k=l}^{n-1}\widetilde\beta_k\widetilde{\mathbf{T}}_{k-l}\left(\frac{A}{\alpha_A}\right)\ket{\psi}}$ admits an upper bound independent of $n$.
    This is reminiscent of the fact that Fourier expansions can converge much faster on average, and is made rigorous by the Carleson-Hunt inequality from \lem{carleson_hunt}.

    The circuit implementation of QEVT requires preparing the state $O_{\widetilde\beta}\ket{0}=\frac{1}{\alpha_{\widetilde\beta}}\sum_{k=0}^{n-1}(\widetilde\beta_k-\widetilde\beta_{k+2})\ket{n-1-k}$ with  $\alpha_{\widetilde\beta}=\sqrt{\sum_{k=0}^{n-1}|\widetilde{\beta}_k-\widetilde{\beta}_{k+2}|^2}$ encoding shifted coefficients from the Chebyshev expansion. We show how such a state can be prepared using $\mathbf{O}\left(\polylog(n)\right)$ gates in \sec{fourier}, improving over the standard state preparation method with complexity $\mathbf{\Theta}(n)$.

    Finally, note that the query complexity of initial state preparation can be improved using the block preconditioning technique of~\cite{OptInit}.
\end{remark}

\subsection{Summary of quantum eigenvalue transformation, block encoded version}
\label{sec:transform_block_summary}

We now summarize a variant of the quantum eigenvalue transformation algorithm with a block encoding output.
\begin{enumerate}
    \item Construct a block encoding of $\mathbf{Pad}(A)/4$ as in \sec{history_block} with $\eta=1$.
    \item Invoke the block encoding version of quantum linear system solver \lem{inv_block} with $\mathbf{Pad}(A)/4$ as the input.
    \item Prepare the state $\ket{0}\frac{1}{\alpha_{\widetilde\beta}}\sum_{k=0}^{n-1}(\widetilde\beta_k-\widetilde\beta_{k+2})\ket{n-1-k}$ with $\alpha_{\widetilde\beta}=\sqrt{\sum_{k=0}^{n-1}|\widetilde{\beta}_k-\widetilde{\beta}_{k+2}|^2}$. Unprepare the state $\ket{1}\frac{1}{\sqrt{n}}\sum_{k=0}^{n-1}\ket{k}$.
    \item Amplify the block encoding using \lem{amp_block}.
\end{enumerate}

\begin{theorem}[Quantum eigenvalue transformation, block encoded version]
\label{thm:qevt_block}
    Let $A$ be a square matrix with only real eigenvalues, such that $A/\alpha_A$ is block encoded by $O_A$ with some normalization factor $\alpha_A\geq\norm{A}$.
    Let $p(x)=\sum_{k=0}^{n-1}\widetilde\beta_k\widetilde{\mathbf{T}}_{k}(x)=\sum_{k=0}^{n-1}\beta_k{\mathbf{T}}_{k}(x)$ be the Chebyshev expansion of a degree-$n$ polynomial $p$.
    Then for any $\alpha_p\geq\norm{p\left(\frac{A}{\alpha_A}\right)}$, the operator
    \begin{equation}
        \frac{p\left(\frac{A}{\alpha_A}\right)}{2\alpha_p}
    \end{equation}
    can be block encoded with accuracy $\epsilon$ using
    \begin{equation}
        \begin{aligned}
        {\mathbf{O}\left(\frac{\norm{p(\cos)\sin}_{2,[-\pi,\pi]}\sqrt{n}\alphaU}{\alpha_p}n\alphaU\log\left(\frac{\norm{p(\cos)\sin}_{2,[-\pi,\pi]}\sqrt{n}\alphaU}{\alpha_p\epsilon}\right)\log\left(\frac{1}{\epsilon}\right)\right)}
        \end{aligned}
    \end{equation}
    queries to controlled-$O_A$ and its inverse, where $\alphaU$ satisfies \eq{alphaU2}
\end{theorem}
\begin{proof}
    Applying \lem{inv_block} to $\mathbf{Pad}(A)/4$, we get a block encoding of
    \begin{equation}
        \frac{\mathbf{Pad}(A)^{-1}}{2\alpha_{\scriptscriptstyle \mathbf{Pad(A)}^{-1}}}
    \end{equation}
    with $\alpha_{\scriptscriptstyle \mathbf{Pad(A)}^{-1}}=\mathbf{O}\left(\alphaU n\right)$. 
    Together with the preparation of $\ket{0}\frac{1}{\alpha_{\widetilde\beta}}\sum_{k=0}^{n-1}(\widetilde\beta_k-\widetilde\beta_{k+2})\ket{n-1-k}$ and unpreparation of $\ket{1}\frac{1}{\sqrt{n}}\sum_{k=0}^{n-1}\ket{k}$, we obtain the block encoding
    \begin{equation}
    \begin{aligned}
        &\left(\bra{1}\frac{1}{\sqrt{n}}\sum_{k=0}^{n-1}\bra{k}\otimes I\right)
        \frac{\mathbf{Pad}(A)^{-1}}{2\alpha_{\scriptscriptstyle \mathbf{Pad(A)}^{-1}}}
        \left(\ket{0}\frac{1}{\alpha_{\widetilde\beta}}\sum_{k=0}^{n-1}(\widetilde\beta_k-\widetilde\beta_{k+2})\ket{n-1-k}\otimes I\right)\\
        &=\frac{1}{\sqrt{n}\alpha_{\scriptscriptstyle \mathbf{Pad(A)}^{-1}}\alpha_{\widetilde\beta}}
        \left(\bra{1}\sum_{k=0}^{n-1}\bra{k}\otimes I\right)
        \mathbf{Pad}(A)^{-1}
        \mathbf{Pad}(B)
        \left(\ket{0}\sum_{k=0}^{n-1}\widetilde{\beta}_k\ket{n-1-k}\otimes I\right)\\
        &=\frac{1}{\sqrt{n}\alpha_{\scriptscriptstyle \mathbf{Pad(A)}^{-1}}\alpha_{\widetilde\beta}}
        \left(\sum_{k=0}^{n-1}\bra{k}\otimes I\right)
        \left(\sum_{l=0}^{n-1}\ket{l}
        \sum_{k=0}^{n-1}\widetilde{\beta}_k\widetilde{\mathbf{T}}_{k}\left(\frac{A}{\alpha_A}\right)\otimes I\right)
        =\frac{p\left(\frac{A}{\alpha_A}\right)}{\alpha_{p,\text{pre}}}
    \end{aligned}
    \end{equation}
    with
    \begin{equation}
        \alpha_{p,\text{pre}}=\frac{\alpha_{\scriptscriptstyle \mathbf{Pad(A)}^{-1}}\alpha_{\widetilde\beta}}{\sqrt{n}}
        =\mathbf{O}\left(\alphaU \sqrt{n}\alpha_{\widetilde\beta}\right),
    \end{equation}
    where the first equality follows from \eq{shift_coeff} and \eq{pad_b}, and the second equality follows from \eq{pad_output}.

    We now claim that Euclidean norm of the shifted coefficients scales like
    \begin{equation}
        \alpha_{\widetilde\beta}=\sqrt{\sum_{k=0}^{n-1}|\widetilde{\beta}_k-\widetilde{\beta}_{k+2}|^2}
        =\mathbf{O}\left(\norm{p(\cos)\sin}_{2,[-\pi,\pi]}\right).
    \end{equation}
    Here the Chebyshev coefficients are shifted in the time domain, so the target function will have a phase shift in the frequency domain. Indeed, a direct calculation shows that
\begin{equation}
\begin{aligned}
    \widetilde\beta_{j}-\widetilde\beta_{j+2}
    &=\frac{2}{\pi}\int_{0}^{\pi}\mathrm{d}\theta\
    p(\cos(\theta))\left(\cos(j\theta)-\cos((j+2)\theta)\right)\\
    &=\frac{4}{\pi}\int_{0}^{\pi}\mathrm{d}\theta\
    p(\cos(\theta))\sin(\theta)\sin((j+1)\theta)\\
    &=\frac{2}{\pi}\int_{-\pi}^{\pi}\mathrm{d}\theta\
    p(\cos(\theta))\sin(\theta)\sin((j+1)\theta).
\end{aligned}
\end{equation}
Thus, $\widetilde\beta_{j}-\widetilde\beta_{j+2}$ can be seen as the Fourier coefficients of the odd function $2p(\cos{\theta})\sin{\theta}$. Invoking Parseval's theorem with the convention that $\widetilde{\beta}_{n}=\widetilde{\beta}_{n+1}=\cdots=0$, we have
\begin{equation}
    \alpha_{\widetilde\beta}
    =\sqrt{\sum_{k=0}^{\infty}|\widetilde{\beta}_k-\widetilde{\beta}_{k+2}|^2}
    =\sqrt{\frac{1}{\pi}\int_{-\pi}^{\pi}\mathrm{d}\theta\
    4\abs{p(\cos{\theta})\sin{\theta}}^2}
    \sim\norm{p(\cos)\sin}_{2,[-\pi,\pi]}.
\end{equation}
The theorem now follows from \lem{amp_block}.
\end{proof}
\begin{remark}
    For a discussion about the scaling of $\alphaU$, see the remark succeeding \thm{generate_history}. This block encoding version of the eigenvalue transformation algorithm (\thm{qevt_block}) underperforms the state version (\thm{qevt}) for the differential equation problem and the ground state preparation problem to be considered in \sec{app}. This is because both versions have a similar gate complexity to generate the Chebyshev history state. But the block encoding algorithm introduces an additional normalization factor $\alpha_{p,\text{pre}}$, which needs to be further amplified. It is for this reason that \thm{qevt_block} will not be used in the remainder of the paper. However, the block-encoded version can become useful when QEVT is invoked as a subroutine in designing other quantum algorithms.
    The scaling $\mathbf{O}(n^{1.5})$ with degree of the target polynomial is improved to $\mathbf{O}(n)$ by recent work~\cite{OptInit} through block preconditioning.

    To implement this algorithm with quantum circuits, we need to prepare the state $\frac{1}{\alpha_{\widetilde\beta}}\sum_{k=0}^{n-1}(\widetilde\beta_k-\widetilde\beta_{k+2})\ket{n-1-k}$ encoding the shifted Chebyshev coefficients. As discussed in the remark succeeding \thm{qevt}, this state can be prepared using the technique from \sec{fourier}. However, unlike \thm{qevt}, we only need to prepare this state once to get the preamplified block encoding. Hence, the cost of the state preparation is much less than that for inverting $\mathbf{Pad}(A)/4$, so it may also be acceptable to prepare this state using the conventional approach~\cite{ShendeBullockMarkov06}.

    Finally, the normalization factor $\norm{p(\cos)\sin}_{2,[-\pi,\pi]}$ corresponds to the $\mathcal{L}_2$-norm of the phase shifted function $p(\cos(\theta))\sin(\theta)$ evaluated in the frequency domain. We assume that this norm can be efficiently computed to arbitrary precision on a classical computer as is often the case; otherwise, we can replace it with a known upper bound $\alpha_{p(\cos)\sin}\geq\norm{p(\cos)\sin}_{2,[-\pi,\pi]}$. Asymptotically, this is better than the more familiar $\norm{p}_{\max,[-1,1]}$ since
    \begin{equation}
    \begin{aligned}
        \norm{p(\cos)\sin}_{2,[-\pi,\pi]}
        &=\sqrt{\int_{-\pi}^{\pi}\mathrm{d}\theta\
    \abs{p(\cos{\theta})\sin{\theta}}^2}
        =\sqrt{2\int_{0}^{\pi}\mathrm{d}\theta\
    \abs{p(\cos{\theta})}^2\sin^2{\theta}}\\
        &=\sqrt{-2\int_{1}^{-1}\mathrm{d}x\
    \abs{p(x)}^2\sqrt{1-x^2}}
        =\mathbf{O}\left(\norm{p}_{2,[-1,1]}\right)
        =\mathbf{O}\left(\norm{p}_{\max,[-1,1]}\right).
    \end{aligned}
    \end{equation}
\end{remark}

\section{Fourier coefficients generation}
\label{sec:fourier}
In this section, we describe an efficient quantum circuit for generating $n$ Fourier coefficients, encoded by the amplitudes of a quantum state. 
This state can be prepared using standard circuit techniques with a gate complexity of $\mathbf{\Theta}(n)$. 
However, our truncate order $n$ in general scales polynomially with the input parameters (such as the evolution time and the inverse spectral gap), and can lead to a significant overhead. 
Our new result has gate complexity $\mathbf{O}(\polylog(n))$.

We begin by introducing the problem in \sec{fourier_conv}, explaining how the generation of Fourier coefficients can be achieved with a frequency domain convolution. Such a convolution is given by a Riemann integral, which we implement using the circuit described in \sec{fourier_block}. However, due to presence of the Dirichlet kernel, integrand of the convolution changes dramatically throughout the entire domain, which can be costly to implement directly. We apply a rescaling principle for Riemann integrals to significantly reduce the implementation cost in \sec{fourier_rescale}. Finally, we summarize the quantum circuit for generating the Fourier coefficients as \thm{fourier_coeff} and analyze it in \sec{fourier_summary}. The generation of Chebyshev coefficients follows immediately as Chebyshev expansions can be reformulated as Fourier expansions through a change of variables.

\subsection{Fourier coefficients generation with frequency domain convolution}
\label{sec:fourier_conv}

Let $g$ be a $2\pi$-periodic function with the Fourier expansion
\begin{equation}
\label{eq:full_fourier}
	g(\omega)=\sum_{j=-\infty}^\infty\xi_je^{-ij\omega}
	=\cdots+\xi_{-2}e^{2i\omega}+\xi_{-1}e^{i\omega}+\xi_{0}
	+\xi_{1}e^{-i\omega}+\xi_{2}e^{-2i\omega}+\cdots
\end{equation}
Then the problem of generating Fourier coefficients is to construct a block encoding for the operator
\begin{equation}
\label{eq:gen_fourier}
    \sum_{j=0}^{n-1}\xi_jL^j=\begin{bmatrix}
    \xi_0 & & &\\
    \xi_1 & \xi_0 & &\\
    \vdots & \ddots & \ddots &\\
    \xi_{n-1} & \cdots & \xi_1 & \xi_0\\
    \end{bmatrix},
\end{equation}
where $\{\xi_j\}_{j=0}^{n-1}$ are the first $n$ Fourier coefficients with nonnegative indices, and $L$ is the $n$-by-$n$ lower shift matrix.

There are a number of places in the paper where we need an efficient quantum circuit for generating Fourier coefficients. For instance, in the Chebyshev eigenvalue transformation algorithm, we need to prepare an initial state of the following form encoding the shifted Chebyshev coefficients
\begin{equation}
    O_{\widetilde\beta}\ket{0}=\frac{1}{\alpha_{\widetilde\beta}}\sum_{k=0}^{n-1}(\widetilde\beta_k-\widetilde\beta_{k+2})\ket{n-1-k},\qquad \alpha_{\widetilde\beta}=\sqrt{\sum_{k=0}^{n-1}|\widetilde{\beta}_k-\widetilde{\beta}_{k+2}|^2}.
\end{equation}
This can be achieved by treating the Chebyshev expansion as a Fourier expansion, and applying the above block encoding to $\ket{0}$. This generates the state $\frac{\sum_{j=0}^{n-1}\widetilde{\beta}_j\ket{j}}{\norm{\widetilde\beta}}$, which can be shifted to the desired form by applying $(I-L)/2$ followed by a further amplitude amplification.
On the other hand, for the Faber-based algorithm to be discussed in \sec{faber_history}, we need to block encode operators such as
\begin{equation}
	\mathbf{\Psi}'(L^{-1}),\qquad
	L\mathbf{\Psi}(L^{-1}),
\end{equation}
where both functions have Laurent expansions with only nonnegative powers (they are purely power series expansions):
\begin{equation}
	\mathbf{\Psi}'(w^{-1})=\varsigma -\varsigma_1w^2-2\varsigma_2w^3+\cdots\qquad
	w\mathbf{\Psi}(w^{-1})=\varsigma +\varsigma_0w+\varsigma_1w^2+\varsigma_2w^3+\cdots
\end{equation}
To handle this, we take $w\rightarrow e^{-i\omega}$ and turn the above Laurent expansions into Fourier expansions
\begin{equation}
	\mathbf{\Psi}'(e^{i\omega})=\varsigma -\varsigma_1e^{-2i\omega}-2\varsigma_2e^{-3i\omega}+\cdots\qquad
	e^{-i\omega}\mathbf{\Psi}(e^{i\omega})=\varsigma +\varsigma_0e^{-i\omega}+\varsigma_1e^{-2i\omega}+\varsigma_2e^{-3i\omega}+\cdots
\end{equation}
Thus we can invoke our circuit with respect to the Fourier expansions of $\mathbf{\Psi}'(e^{i\omega})$ and $e^{-i\omega}\mathbf{\Psi}(e^{i\omega})$ and obtain the desired block encoding.

The key idea behind our approach is to generate the Fourier coefficients using a frequency domain convolution. To be more specific, recall from \lem{shift} that for $j=0,1,\ldots,n-1$,
\begin{equation}
    L_{n}^j=\left(I_{2n}\otimes\bra{0}\right)\mathrm{CMP}_{2n,2}\left(I_{2n}\otimes\ket{0}\right)X_{2n}^j\left(I_{2n}\otimes\bra{0}\right)\mathrm{CMP}_{2n,2}\left(I_{2n}\otimes\ket{0}\right).
\end{equation}
Here, the cyclic shift operator $X_{2n}$ can be diagonalized by the quantum Fourier transform $F_{2n}=\frac{1}{\sqrt{2n}}\sum_{l,m=0}^{2n-1}e^{-i\frac{2\pi}{2n}lm}\ketbra{l}{m}$ as $F_{2n} X_{2n}F_{2n}^\dagger=\sum_{m=0}^{2n-1}e^{-i\frac{2\pi}{2n}m}\ketbra{m}{m}=Z_{2n}$, resulting in
\begin{equation}
\label{eq:block_L_Z}
    L_{n}^j=\left(\left(I_{2n}\otimes\bra{0}\right)\mathrm{CMP}_{2n,2}\left(I_{2n}\otimes\ket{0}\right)F_{2n}^\dagger\right)
    Z_{2n}^j
    \left(F_{2n}\left(I_{2n}\otimes\bra{0}\right)\mathrm{CMP}_{2n,2}\left(I_{2n}\otimes\ket{0}\right)\right).
\end{equation}
Thus for each diagonal element of $Z_{2n}$ labeled by $m=0,1,\ldots,2n-1$, we need to implement
\begin{equation}
	\label{eq:trunc_fourier}
	\xi_0+\xi_1e^{-i\frac{2\pi}{2n}m}+\xi_2e^{-2i\frac{2\pi}{2n}m}+\cdots
	+\xi_{n-1}e^{-(n-1)i\frac{2\pi}{2n}m}.
\end{equation}
Comparing \eq{full_fourier} and \eq{trunc_fourier}, we arrive at the integral representation
\begin{equation}
\label{eq:fourier_coeff_int}
	\xi_0+\xi_1e^{-i\frac{2\pi}{2n}m}+\xi_2e^{-2i\frac{2\pi}{2n}m}+\cdots
	+\xi_{n-1}e^{-(n-1)i\frac{2\pi}{2n}m}
	=\frac{1}{2\pi}\int_{-\pi}^{\pi}\mathrm{d}u\ g\left(\frac{m\pi}{n}-u\right)\frac{1-e^{-niu}}{1-e^{-iu}}
\end{equation}
from the frequency domain convolution theorem. The circuit implementation of such a Riemann integral will be discussed in the next subsection.

\subsection{Block encoding Riemann integrals}
\label{sec:fourier_block}

Suppose we have a function $h:[a,b]\rightarrow\mathbb C$ with 
max-norm $\norm{h}_{\max,[a,b]}$
that is $\mu_h$-Lipschitz continuous such that $\abs{h(z)-h(w)}\leq\mu_h\abs{z-w}$ for all $z,w\in[a,b]$. Our goal here is to construct a block encoding of the Riemann integral
\begin{equation}
    \frac{\int_{a}^{b}\mathrm{d}x\ h(x)}{(b-a)\norm{h}_{\max,[a,b]}}.
\end{equation}

This is achieved by approximating the integral with a Riemann sum which can then be implemented using standard  techniques. Here we include a description of this block encoding for completeness.
To this end, let us first describe the oracular access model for the input function $h$. We assume that the input register takes $n_{\text{in}}$ values, and two output registers take ($n_{\text{Abs}}+1$) and $n_{\text{Arg}}$ values, holding the absolute value and argument of $h$ respectively. Then we introduce the oracles
    \begin{equation}
    \begin{aligned}
        O_{\text{Abs}}\ket{s,0}
        &=\bigg|s,\mathbf{Floor}\left(n_{\text{Abs}}\frac{\abs{h\left(a+\frac{b-a}{n_{\text{in}}}s\right)}}{\norm{h}_{\max,[a,b]}}\right)\bigg\rangle,\\
        O_{\text{Arg}}\ket{s,0}
        &=\bigg|s,\mathbf{Floor}\left(n_{\text{Arg}}\frac{\mathbf{Arg}\left(h\left(a+\frac{b-a}{n_{\text{in}}}s\right)\right)}{2\pi}\right)\bigg\rangle.
    \end{aligned}
    \end{equation}
Here, we segment $[a,b]$ into $n_{\text{in}}$ subintervals labeled by $s=0,1,\ldots,n_{\text{in}}-1$, so that $a+\frac{b-a}{n_{\text{in}}}s$ approximately represents a general number from $[a,b]$. The absolute value $\abs{h\left(a+\frac{b-a}{n_{\text{in}}}s\right)}$ is normalized by the max-norm $\norm{h}_{\max,[a,b]}$, so that $\mathbf{Floor}\left(n_{\text{Abs}}\frac{\abs{h\left(a+\frac{b-a}{n_{\text{in}}}s\right)}}{\norm{h}_{\max,[a,b]}}\right)=0,1,\ldots,n_{\text{Abs}}$ can be held by the first output register. Similarly, the argument $\mathbf{Arg}\left(h\left(a+\frac{b-a}{n_{\text{in}}}s\right)\right)\in[0,2\pi)$ is normalized by $2\pi$, so that $\mathbf{Floor}\left(n_{\text{Arg}}\frac{\mathbf{Arg}\left(h\left(a+\frac{b-a}{n_{\text{in}}}s\right)\right)}{2\pi}\right)=0,1,\ldots,n_{\text{Arg}}-1$ can be represented in the second output register.
Note that the value of max-norm $\norm{h}_{\max,[a,b]}$ can often be efficiently computed to arbitrary precision on a classical computer. But our method still works if the exact value of max-norm is replaced by its upper bound $\alpha_{h,\max}\geq\norm{h}_{\max,[a,b]}$.
We now implement the block encoding as follows.

\begin{lemma}[Block encoding Riemann integrals]
\label{lem:block_int}
    Let $h:[a,b]\rightarrow\mathbb C$ be a $\mu_h$-Lipschitz function with max-norm $\norm{h}_{\max,[a,b]}$. Suppose that its absolute value and argument are provided by the oracles
    \begin{equation}
    \begin{aligned}
        O_{\text{Abs}}\ket{s,0}
        &=\bigg|s,\mathbf{Floor}\left(n_{\text{Abs}}\frac{\abs{h\left(a+\frac{b-a}{n_{\text{in}}}s\right)}}{\norm{h}_{\max,[a,b]}}\right)\bigg\rangle,\\
        O_{\text{Arg}}\ket{s,0}
        &=\bigg|s,\mathbf{Floor}\left(n_{\text{Arg}}\frac{\mathbf{Arg}\left(h\left(a+\frac{b-a}{n_{\text{in}}}s\right)\right)}{2\pi}\right)\bigg\rangle,
    \end{aligned}
    \end{equation}
    with an $n_{\text{in}}$-value input register, and ($n_{\text{Abs}}+1$)- and $n_{\text{Arg}}$-value output registers respectively.
    Then, the normalized integral
    \begin{equation}
        \frac{\int_{a}^{b}\mathrm{d}x\ h(x)}{(b-a)\norm{h}_{\max,[a,b]}}
    \end{equation}
    can be block encoded with accuracy $\delta$ by setting $n_{\text{in}}=\mathbf{O}(\mu_h(b-a)/(\norm{h}_{\max,[a,b]}\delta)),n_{\text{Abs}},n_{\text{Arg}}=\mathbf{O}(1/\delta)$, using $2$ queries to $O_{\text{Abs}}$ and $1$ query to $O_{\text{Arg}}$, together with
    \begin{equation}
        \mathbf{O}\left(\log\left(\frac{\mu_h(b-a)}{\norm{h}_{\max,[a,b]}\delta}\right)\right)
    \end{equation}
    two-qubit gates.
\end{lemma}
\begin{proof}
    We perform the block encoding as follows:
    \begin{enumerate}
        \item We prepare a uniform superposition state 
        \begin{equation}
            \frac{1}{\sqrt{n_{\text{in}}}}\sum_{s=0}^{n_{\text{in}}-1}\ket{s}
        \end{equation}
        for $n_{\text{in}}$ sufficiently large. This has gate complexity $\mathbf{O}(\log(n_{\text{in}}))$.
        \item We apply the oracle $O_{\text{Abs}}$. This has query complexity $1$.
        \item We introduce the uniform superposition state
        \begin{equation}
            \frac{1}{\sqrt{n_{\text{Abs}}}}\sum_{x=0}^{n_{\text{Abs}}-1}\ket{x}
        \end{equation}
        and test the inequality
        \begin{equation}
            \mathbf{Floor}\left(n_{\text{Abs}}\frac{\abs{h\left(a+\frac{b-a}{n_{\text{in}}}s\right)}}{\norm{h}_{\max,[a,b]}}\right)
            \leq x.
        \end{equation}
        The state of the entire system becomes
        \begin{footnotesize}
        \newmaketag
        \begin{equation}
            \frac{1}{\sqrt{n_{\text{in}}}}\sum_{s=0}^{n_{\text{in}}-1}\ket{s}
            \left(\sqrt{\frac{\mathbf{Floor}\left(n_{\text{Abs}}\frac{\abs{h\left(a+\frac{b-a}{n_{\text{in}}}s\right)}}{\norm{h}_{\max,[a,b]}}\right)}{n_{\text{Abs}}}}\ket{\phi_0}\ket{0}+\sqrt{1-\frac{\mathbf{Floor}\left(n_{\text{Abs}}\frac{\abs{h\left(a+\frac{b-a}{n_{\text{in}}}s\right)}}{\norm{h}_{\max,[a,b]}}\right)}{n_{\text{Abs}}}}\ket{\phi_1}\ket{1}\right)
        \end{equation}
        \end{footnotesize}
        for some auxiliary states $\ket{\phi_0}$ and $\ket{\phi_1}$. 
        This has gate complexity $\mathbf{O}(\log(n_{\text{Abs}}))$.
        \item We apply the oracle $O_{\text{Arg}}$. This has query complexity $1$.
        \item We prepare a phase gradient state over $n_{\text{Arg}}$ values, and add the argument to it with gate complexity $\mathbf{O}(\log(n_{\text{Arg}}))$~\cite[Appendix A]{SandersCompilation20}. Omitting unnecessary registers, we obtain
        \begin{footnotesize}
        \newmaketag
        \begin{equation}
        \begin{aligned}
            &\frac{1}{\sqrt{n_{\text{in}}}}\sum_{s=0}^{n_{\text{in}}-1}\exp\left(i\frac{2\pi}{n_{\text{Arg}}}\mathbf{Floor}\left(n_{\text{Arg}}\frac{\mathbf{Arg}\left(h\left(a+\frac{b-a}{n_{\text{in}}}s\right)\right)}{2\pi}\right)\right)\ket{s}\\
            &\qquad\otimes\left(\sqrt{\frac{\mathbf{Floor}\left(n_{\text{Abs}}\frac{\abs{h\left(a+\frac{b-a}{n_{\text{in}}}s\right)}}{\norm{h}_{\max,[a,b]}}\right)}{n_{\text{Abs}}}}\ket{\phi_0}\ket{0}+\sqrt{1-\frac{\mathbf{Floor}\left(n_{\text{Abs}}\frac{\abs{h\left(a+\frac{b-a}{n_{\text{in}}}s\right)}}{\norm{h}_{\max,[a,b]}}\right)}{n_{\text{Abs}}}}\ket{\phi_1}\ket{1}\right).
        \end{aligned}
        \end{equation}
        \end{footnotesize}
        \item We introduce an ancilla register in state $\ket{0}$ and swap out outcome of the inequality test
        \begin{footnotesize}
        \newmaketag
        \begin{equation}
        \begin{aligned}
            &\frac{1}{\sqrt{n_{\text{in}}}}\sum_{s=0}^{n_{\text{in}}-1}\exp\left(i\frac{2\pi}{n_{\text{Arg}}}\mathbf{Floor}\left(n_{\text{Arg}}\frac{\mathbf{Arg}\left(h\left(a+\frac{b-a}{n_{\text{in}}}s\right)\right)}{2\pi}\right)\right)\ket{s}\\
            &\qquad\otimes\left(\sqrt{\frac{\mathbf{Floor}\left(n_{\text{Abs}}\frac{\abs{h\left(a+\frac{b-a}{n_{\text{in}}}s\right)}}{\norm{h}_{\max,[a,b]}}\right)}{n_{\text{Abs}}}}\ket{\phi_0}\ket{0,0}+\sqrt{1-\frac{\mathbf{Floor}\left(n_{\text{Abs}}\frac{\abs{h\left(a+\frac{b-a}{n_{\text{in}}}s\right)}}{\norm{h}_{\max,[a,b]}}\right)}{n_{\text{Abs}}}}\ket{\phi_1}\ket{0,1}\right).
        \end{aligned}
        \end{equation}
        \end{footnotesize}
        \item Finally, we reverse the first four steps. This is described by the bra vector (omitting unnecessary registers)
        \begin{footnotesize}
        \newmaketag
        \begin{equation}
            \frac{1}{\sqrt{n_{\text{in}}}}\sum_{s=0}^{n_{\text{in}}-1}\bra{s}
            \left(\sqrt{\frac{\mathbf{Floor}\left(n_{\text{Abs}}\frac{\abs{h\left(a+\frac{b-a}{n_{\text{in}}}s\right)}}{\norm{h}_{\max,[a,b]}}\right)}{n_{\text{Abs}}}}\bra{\phi_0}\bra{0,0}+\sqrt{1-\frac{\mathbf{Floor}\left(n_{\text{Abs}}\frac{\abs{h\left(a+\frac{b-a}{n_{\text{in}}}s\right)}}{\norm{h}_{\max,[a,b]}}\right)}{n_{\text{Abs}}}}\bra{\phi_1}\bra{1,0}\right).
        \end{equation}
        \end{footnotesize}
    \end{enumerate}
    
    The above procedure allows us to block encode
    \begin{equation}
    \begin{aligned}
        &\ \frac{1}{n_{\text{in}}}\sum_{s=0}^{n_{\text{in}}-1}\exp\left(i\frac{2\pi}{n_{\text{Arg}}}\mathbf{Floor}\left(n_{\text{Arg}}\frac{\mathbf{Arg}\left(h\left(a+\frac{b-a}{n_{\text{in}}}s\right)\right)}{2\pi}\right)\right)
        \frac{\mathbf{Floor}\left(n_{\text{Abs}}\frac{\abs{h\left(a+\frac{b-a}{n_{\text{in}}}s\right)}}{\norm{h}_{\max,[a,b]}}\right)}{n_{\text{Abs}}}\\
        \overset{(1)}{\approx}&\
        \frac{1}{n_{\text{in}}}\sum_{s=0}^{n_{\text{in}}-1}e^{i\mathbf{Arg}\left(h\left(a+\frac{b-a}{n_{\text{in}}}s\right)\right)}
        \frac{\mathbf{Floor}\left(n_{\text{Abs}}\frac{\abs{h\left(a+\frac{b-a}{n_{\text{in}}}s\right)}}{\norm{h}_{\max,[a,b]}}\right)}{n_{\text{Abs}}}\\
        \overset{(2)}{\approx}&\
        \frac{1}{n_{\text{in}}}\sum_{s=0}^{n_{\text{in}}-1}e^{i\mathbf{Arg}\left(h\left(a+\frac{b-a}{n_{\text{in}}}s\right)\right)}
        \frac{\abs{h\left(a+\frac{b-a}{n_{\text{in}}}s\right)}}{\norm{h}_{\max,[a,b]}}\\
        =&\ \frac{1}{n_{\text{in}}}\sum_{s=0}^{n_{\text{in}}-1}
        \frac{h\left(a+\frac{b-a}{n_{\text{in}}}s\right)}{\norm{h}_{\max,[a,b]}}
        \overset{(3)}{\approx}
        \frac{\int_{a}^{b}\mathrm{d}x\ h(x)}{(b-a)\norm{h}_{\max,[a,b]}}.
    \end{aligned}
    \end{equation}
    It is easy to see that the first error is at most $\frac{2\pi}{n_{\text{Arg}}}$ and the second error is at most $\frac{1}{n_{\text{Abs}}}$. So our remaining task is to bound the third error, which comes from discretizing the Riemann integral as a Riemann sum.

    We discretize the integral as
    \begin{equation}
        \int_{a}^{b}\mathrm{d}x\ h(x)
        =\sum_{s=0}^{n_{\text{in}}-1}\int_{a+\frac{b-a}{n_{\text{in}}}s}^{a+\frac{b-a}{n_{\text{in}}}(s+1)}\mathrm{d}x\ h(x)
        \approx\sum_{s=0}^{n_{\text{in}}-1}\frac{b-a}{n_{\text{in}}}
        h\left(a+\frac{b-a}{n_{\text{in}}}s\right).
    \end{equation}
    Since $h$ is $\mu_h$-Lipschitz, within each subinterval the error contribution is at most
    \begin{equation}
        \abs{\int_{a+\frac{b-a}{n_{\text{in}}}s}^{a+\frac{b-a}{n_{\text{in}}}(s+1)}\mathrm{d}x\left(h(x)-h\left(a+\frac{b-a}{n_{\text{in}}}s\right)\right)}
        \leq\int_{a+\frac{b-a}{n_{\text{in}}}s}^{a+\frac{b-a}{n_{\text{in}}}(s+1)}\mathrm{d}x\ \mu_h\abs{x-a-\frac{b-a}{n_{\text{in}}}s}
        =\mu_h\frac{(b-a)^2}{2n_{\text{in}}^2}.
    \end{equation}
    So altogether,
    \begin{equation}
        \abs{
        \frac{1}{n_{\text{in}}}\sum_{s=0}^{n_{\text{in}}-1}
        \frac{h\left(a+\frac{b-a}{n_{\text{in}}}s\right)}{\norm{h}_{\max,[a,b]}}
        -
        \frac{\int_{a}^{b}\mathrm{d}x\ h(x)}{(b-a)\norm{h}_{\max,[a,b]}}
        }
        \leq\frac{\mu_h(b-a)}{2n_{\text{in}}\norm{h}_{\max,[a,b]}}.
    \end{equation}
    The proof is now complete by choosing $n_{\text{in}}, n_{\text{Abs}}, n_{\text{Arg}}$ such that $\frac{2\pi}{n_{\text{Arg}}},\frac{1}{n_{\text{Abs}}},\frac{\mu_h(b-a)}{2n_{\text{in}}\norm{h}_{\max,[a,b]}}=\mathbf{O}(\delta)$.
\end{proof}

\subsection{Rescaling principle for Riemann integrals}
\label{sec:fourier_rescale}

Now that we have the quantum circuit for implementing Riemann integrals, it is tempting to apply it directly to the frequency domain convolution in \eq{fourier_coeff_int}. 
However, this would result in a gate complexity scaling linear in $n$, same as that of the standard approach for state preparation. The underlying reason is that our convolution in \eq{fourier_coeff_int} is essentially defined with respect to the Dirichlet kernel, whose function value changes dramatically over the domain of interest. In particular, $\mathcal{L}_\infty$-norm of the Dirichlet kernel grows like $\mathbf{\Theta}(n)$ which determines the block encoding normalization factor as per \lem{block_int}, although its $\mathcal{L}_1$-norm only scales like $\mathbf{\Theta}(\log(n))$.

Our idea to overcome this catastrophe is to use the following rescaling principle for Riemann integrals. 
\begin{lemma}[Rescaling principle for Riemann integrals]
\label{lem:rescale}
    Let $h:[a,b]\rightarrow\mathbb C$ and $\alpha:[a,b]\rightarrow\mathbb R_{>0}$ be continuous. Then the primitive function
    \begin{equation}
        \alpha_{1}(t)=\int_{a}^{t}\mathrm{d}\tau\ \alpha(\tau)
    \end{equation}
    is positive, continuously differentiable, and monotonically increasing on $[a,b]$, such that
    \begin{equation}
        \int_{a}^{b}\mathrm{d}t\ h(t)
        =\int_{0}^{\alpha_{1}(b)}\mathrm{d}s\ \frac{h\left(\alpha_{1}^{-1}(s)\right)}{\alpha\left(\alpha_{1}^{-1}(s)\right)}
        =\int_{0}^{\norm{\alpha}_{1,[a,b]}}\mathrm{d}s\ \frac{h\left(\alpha_{1}^{-1}(s)\right)}{\alpha\left(\alpha_{1}^{-1}(s)\right)}.
    \end{equation}
\end{lemma}
\begin{proof}
    This is just a change of variables with $s=\alpha_{1}(t)$.
\end{proof}

A similar rescaling principle was previously introduced for simulating time-ordered integrators~\cite{Berry2020timedependent}. Here, we use this technique to improve the circuit implementation of scalar Riemann integrals. To this end, we choose $\alpha(\tau)\sim\abs{h(\tau)}$ so that absolute value of the rescaled integrand
\begin{equation}
    \abs{\frac{h\left(\alpha_{1}^{-1}(s)\right)}{\alpha\left(\alpha_{1}^{-1}(s)\right)}}
    \sim\abs{\frac{h\left(\alpha_{1}^{-1}(s)\right)}{h\left(\alpha_{1}^{-1}(s)\right)}}=1
\end{equation}
is almost uniform, over the rescaled interval of length
\begin{equation}
    a_{1}(b)=
    \norm{\alpha}_{1,[a,b]}=\int_{a}^{b}\mathrm{d}\tau\ \alpha(\tau)
    \sim\int_{a}^{b}\mathrm{d}\tau\ \abs{h(\tau)}=\norm{h}_{1,[a,b]}.
\end{equation}
This technique thus helps ``flatten out'' a function whose instantaneous value may change dramatically over the domain.

Toward generating the Fourier coefficients via \eq{fourier_coeff_int}, we further simplify the frequency domain convolution as
\begin{equation}
	\begin{aligned}
		\frac{1}{2\pi}\int_{-\pi}^{\pi}\mathrm{d}u\ g\left(\frac{m\pi}{n}-u\right)\frac{1-e^{-niu}}{1-e^{-iu}}
		&=\frac{1}{2\pi}\int_{-\pi}^{\pi}\mathrm{d}u\ g\left(\frac{m\pi}{n}-u\right)e^{-\frac{n-1}{2}iu}\frac{\sin\left(\frac{nu}{2}\right)}{\sin\left(\frac{u}{2}\right)}\\
		&=\frac{1}{\pi}\int_{-\frac{\pi}{2}}^{\frac{\pi}{2}}\mathrm{d}v\ g\left(\frac{m\pi}{n}-2v\right)e^{-(n-1)iv}\frac{\sin\left(nv\right)}{\sin\left(v\right)}.\\
	\end{aligned}
\end{equation}
We will break the integral into the following $4$ parts:
\begin{enumerate}
    \item $[0,\frac{\pi}{2n}]$: In this case, we have the integral
	\begin{equation}
		\frac{1}{\pi}\int_{0}^{\frac{\pi}{2n}}\mathrm{d}v\ g\left(\frac{m\pi}{n}-2v\right)e^{-(n-1)iv}\frac{\sin\left(nv\right)}{\sin\left(v\right)}.
	\end{equation}
	The $\mathcal{L}_\infty$-norm of the integrand multiplied by length of the interval is bounded by
	\begin{equation}
		\frac{\pi}{2n}\frac{1}{\pi}\max_v\abs{g\left(\frac{m\pi}{n}-2v\right)e^{-(n-1)iv}\frac{\sin\left(nv\right)}{\sin\left(v\right)}}
		\leq\frac{1}{2n}\norm{g}_{\max,[-\pi,\pi]}\max_v\frac{nv}{\frac{2v}{\pi}}
		=\frac{\pi}{4}\norm{g}_{\max,[-\pi,\pi]}.
	\end{equation}
    \item $[\frac{\pi}{2n},\frac{\pi}{2}]$: In this case, we have the integral
	\begin{equation}
		\frac{1}{\pi}\int_{\frac{\pi}{2n}}^{\frac{\pi}{2}}\mathrm{d}v\ g\left(\frac{m\pi}{n}-2v\right)e^{-(n-1)iv}\frac{\sin\left(nv\right)}{\sin\left(v\right)}.
	\end{equation}
	In this case, the instantaneous absolute value of the integrand changes dramatically, and we need to rescale the integral by letting
	\begin{equation}
		v=e^s,
	\end{equation}
	so that
	\begin{equation}
		\frac{1}{\pi}\int_{\frac{\pi}{2n}}^{\frac{\pi}{2}}\mathrm{d}v\ g\left(\frac{m\pi}{n}-2v\right)e^{-(n-1)iv}\frac{\sin\left(nv\right)}{\sin\left(v\right)}
		=\frac{1}{\pi}\int_{\ln\frac{\pi}{2n}}^{\ln\frac{\pi}{2}}\mathrm{d}s\ e^sg\left(\frac{m\pi}{n}-2e^s\right)e^{-(n-1)ie^s}\frac{\sin\left(ne^s\right)}{\sin\left(e^s\right)}.
	\end{equation}
	After the rescaling, the $\mathcal{L}_\infty$-norm of the integrand multiplied by length of interval is now bounded by
	\begin{equation}
		\begin{aligned}
		&\left(\ln\frac{\pi}{2}-\ln\frac{\pi}{2n}\right)\frac{1}{\pi}\max_s\abs{e^sg\left(\frac{m\pi}{n}-2e^s\right)e^{-(n-1)ie^s}\frac{\sin\left(ne^s\right)}{\sin\left(e^s\right)}}\\
		&\quad\leq\frac{\ln(n)}{\pi}\norm{g}_{\max,[-\pi,\pi]}\max_s\abs{\frac{e^s}{\frac{2}{\pi}e^s}}
		=\frac{\ln(n)}{2}\norm{g}_{\max,[-\pi,\pi]}.
		\end{aligned}
	\end{equation}
    \item $[-\frac{\pi}{2},-\frac{\pi}{2n}]$: In this case, we have
	\begin{equation}
			\frac{1}{\pi}\int_{-\frac{\pi}{2}}^{-\frac{\pi}{2n}}\mathrm{d}v\ g\left(\frac{m\pi}{n}-2v\right)e^{-(n-1)iv}\frac{\sin\left(nv\right)}{\sin\left(v\right)}
			=\frac{1}{\pi}\int_{\frac{\pi}{2n}}^{\frac{\pi}{2}}\mathrm{d}v\ g\left(\frac{m\pi}{n}+2v\right)e^{(n-1)iv}\frac{\sin\left(nv\right)}{\sin\left(v\right)}.
	\end{equation}
	We then proceed as in the second case.
    \item $[-\frac{\pi}{2n},0]$: In this case, we have
	\begin{equation}
		\begin{aligned}
			\frac{1}{\pi}\int_{-\frac{\pi}{2n}}^{0}\mathrm{d}v\ g\left(\frac{m\pi}{n}-2v\right)e^{-(n-1)iv}\frac{\sin\left(nv\right)}{\sin\left(v\right)}
			=\frac{1}{\pi}\int_{0}^{\frac{\pi}{2n}}\mathrm{d}v\ g\left(\frac{m\pi}{n}+2v\right)e^{(n-1)iv}\frac{\sin\left(nv\right)}{\sin\left(v\right)}.
		\end{aligned}
	\end{equation}
	We then proceed as in the first case.
\end{enumerate}

By taking a linear combination of the above integrals, we can block encode
\begin{equation}
    \frac{\frac{1}{2\pi}\int_{-\pi}^{\pi}\mathrm{d}u\ g\left(\frac{m\pi}{n}-u\right)\frac{1-e^{-niu}}{1-e^{-iu}}}{2\left(\frac{\pi}{4}\norm{g}_{\max,[-\pi,\pi]}+\frac{\ln(n)}{2}\norm{g}_{\max,[-\pi,\pi]}\right)}
	=\frac{\frac{1}{2\pi}\int_{-\pi}^{\pi}\mathrm{d}u\ g\left(\frac{m\pi}{n}-u\right)\frac{1-e^{-niu}}{1-e^{-iu}}}{\left(\frac{\pi}{2}+\ln(n)\right)\norm{g}_{\max,[-\pi,\pi]}}.
\end{equation}
If we perform the block encoding for each quantum value $m=0,1,\ldots,2n-1$ with the same normalization factor, we obtain
\begin{equation}
\label{eq:block_fourier_Z}
    \frac{\sum_{j=0}^{n-1}\xi_jZ_{2n}^j}{\left(\frac{\pi}{2}+\ln(n)\right)\norm{g}_{\max,[-\pi,\pi]}},
\end{equation}
which block encodes
\begin{equation}
\label{eq:block_fourier_L}
    \frac{\sum_{j=0}^{n-1}\xi_jL_{n}^j}{\left(\frac{\pi}{2}+\ln(n)\right)\norm{g}_{\max,[-\pi,\pi]}}
\end{equation}
per the analysis of \sec{fourier_conv}. We analyze the complexity of this circuit in more detail in the next subsection.

\subsection{Summary of Fourier coefficients generation}
\label{sec:fourier_summary}

We now summarize the circuit for generating Fourier coefficients.
\begin{enumerate}
    \item Implement the four rescaled convolutions of \sec{fourier_rescale} using \lem{block_int}, controlled by the quantum value $m=0,1,\ldots,2n-1$ from an ancilla register.
    \item Take a linear combination to block encode \eq{block_fourier_Z}.
    \item Construct a block encoding of \eq{block_fourier_L} using \eq{block_L_Z}. 
\end{enumerate}

\begin{theorem}[Fourier coefficients generation]
\label{thm:fourier_coeff}
Let $g$ be a $2\pi$-periodic function having the Fourier expansion $g(\omega)=\sum_{j=-\infty}^\infty\xi_je^{-ij\omega}$, 
with max-norm $\norm{g}_{\max,[-\pi,\pi]}$ of the function value and maximum derivative $\norm{g'}_{\max,[-\pi,\pi]}$.
Suppose that the absolute value and argument of $g$ are provided by the oracles
    \begin{equation}
    \begin{aligned}
        O_{\text{Abs}}\ket{m,s,0}
        &=\bigg|m,s,\mathbf{Floor}\left(n_{\text{Abs}}\frac{\abs{g\left(\frac{m\pi}{n}+\pi-\frac{2\pi}{n_{\text{in}}}s\right)}}{\norm{g}_{\max,[-\pi,\pi]}}\right)\bigg\rangle,\\
        O_{\text{Arg}}\ket{m,s,0}
        &=\bigg|m,s,\mathbf{Floor}\left(n_{\text{Arg}}\frac{\mathbf{Arg}\left(g\left(\frac{m\pi}{n}+\pi-\frac{2\pi}{n_{\text{in}}}s\right)\right)}{2\pi}\right)\bigg\rangle,
    \end{aligned}
    \end{equation}
with $2n$-value register holding $\ket{m}$, $n_{\text{in}}$-value register holding $\ket{s}$, and ($n_{\text{Abs}}+1$)- and $n_{\text{Arg}}$-value output registers respectively.
Then the operator
\begin{equation}
    \frac{\sum_{j=0}^{n-1}\xi_jL_{n}^j}{\left(\frac{\pi}{2}+\ln(n)\right)\norm{g}_{\max,[-\pi,\pi]}}
\end{equation}
can be block encoded to accuracy $\epsilon$ by setting size of the registers $n_{\text{Abs}},n_{\text{Arg}}=\mathbf{O}(1/\epsilon)$ and $n_{\text{in}}=\mathbf{O}\left(\frac{\norm{g'}_{\max,[-\pi,\pi]}n\log(n)}{\norm{g}_{\max,[-\pi,\pi]}\epsilon}+\frac{n^2\log(n)}{\epsilon}\right)$,
using $\mathbf{O}(1)$ queries to $O_{\text{Abs}}$ and $O_{\text{Arg}}$, together with
    \begin{equation}
        \mathbf{O}\left(\polylog\left(\frac{\norm{g'}_{\max,[-\pi,\pi]}n}{\norm{g}_{\max,[-\pi,\pi]}\epsilon}+\frac{n^2}{\epsilon}\right)\right)
    \end{equation}
two qubit gates.
\end{theorem}
\begin{proof}
    The query complexity and normalization factor of the block encoding are already analyzed in \sec{fourier_rescale}.
    To establish the gate complexity, we begin with the four convolutions from \sec{fourier_rescale}. Consider first the integral
        \begin{equation}
		\frac{1}{\pi}\int_{0}^{\frac{\pi}{2n}}\mathrm{d}v\ g\left(\frac{m\pi}{n}-2v\right)e^{-(n-1)iv}\frac{\sin\left(nv\right)}{\sin\left(v\right)}.
	\end{equation}
    We have already shown that $\mathcal{L}_\infty$-norm of the integrand is bounded by
    \begin{equation}
        \norm{\frac{1}{\pi}g\left(\frac{m\pi}{n}-2(\cdot)\right)e^{-(n-1)i(\cdot)}\frac{\sin\left(n(\cdot)\right)}{\sin\left(\cdot\right)}}_{\max,[0,\frac{\pi}{2n}]}
        \leq\frac{n}{2}\norm{g}_{\max,[-\pi,\pi]},
    \end{equation}
    whereas length of the interval is $\frac{\pi}{2n}$.
    For the Lipschitz constant, we take the derivative of the integrand
\begin{equation}
\begin{aligned}
    \mu&=\max_v\abs{\frac{1}{\pi}\frac{\mathrm{d}}{\mathrm{d}v}\left(g\left(\frac{m\pi}{n}-2v\right)
    e^{-(n-1)iv}\frac{\sin(nv)}{\sin(v)}\right)}
    =\frac{1}{\pi}\max_v\abs{\frac{\mathrm{d}}{\mathrm{d}v}\left(g\left(\frac{m\pi}{n}-2v\right)
    \sum_{j=0}^{n-1}e^{-i2jv}\right)}\\
    &=\frac{1}{\pi}\max_v\abs{-2g'\left(\frac{m\pi}{n}-2v\right)\sum_{j=0}^{n-1}e^{-i2jv}
    -i2g\left(\frac{m\pi}{n}-2v\right)\sum_{j=0}^{n-1}je^{-i2jv}}\\
    &=\mathbf{O}\left(n\norm{g'}_{\max,[-\pi,\pi]}+n^2\norm{g}_{\max,[-\pi,\pi]}\right).
\end{aligned}
\end{equation}
    Moreover, the integrand has absolute value
    \begin{equation}
        \abs{\frac{1}{\pi}g\left(\frac{m\pi}{n}-2v\right)e^{-(n-1)iv}\frac{\sin\left(nv\right)}{\sin\left(v\right)}}
        =\frac{1}{\pi}\abs{g\left(\frac{m\pi}{n}-2v\right)}\frac{\sin\left(nv\right)}{\sin\left(v\right)}
    \end{equation}
    and argument
    \begin{equation}
        \mathbf{Arg}\left(\frac{1}{\pi}g\left(\frac{m\pi}{n}-2v\right)e^{-(n-1)iv}\frac{\sin\left(nv\right)}{\sin\left(v\right)}\right)
        =\mathbf{Mod}_{2\pi}\left(\mathbf{Arg}\left(g\left(\frac{m\pi}{n}-2v\right)\right)-(n-1)v\right).
    \end{equation}
    Invoking \lem{block_int} with additional arithmetics to compute the trigonometric and exponential functions, this integral can be implemented to accuracy $\mathbf{O}(\epsilon)$ with gate complexity
    \begin{equation}
        \mathbf{O}\left(\polylog\left(\frac{\norm{g'}_{\max,[-\pi,\pi]}}{ n\norm{g}_{\max,[-\pi,\pi]}\epsilon}+\frac{1}{\epsilon}\right)\right).
    \end{equation}

    Now, consider the rescaled convolution
	\begin{equation}
		\frac{1}{\pi}\int_{\frac{\pi}{2n}}^{\frac{\pi}{2}}\mathrm{d}v\ g\left(\frac{m\pi}{n}-2v\right)e^{-(n-1)iv}\frac{\sin\left(nv\right)}{\sin\left(v\right)}
		=\frac{1}{\pi}\int_{\ln\frac{\pi}{2n}}^{\ln\frac{\pi}{2}}\mathrm{d}s\ e^sg\left(\frac{m\pi}{n}-2e^s\right)e^{-(n-1)ie^s}\frac{\sin\left(ne^s\right)}{\sin\left(e^s\right)}.
	\end{equation}
    We have already shown that $\mathcal{L}_\infty$-norm of the integrand is bounded by
    \begin{equation}
        \norm{\frac{1}{\pi}e^{(\cdot)}g\left(\frac{m\pi}{n}-2e^{(\cdot)}\right)e^{-(n-1)ie^{(\cdot)}}\frac{\sin\left(ne^{(\cdot)}\right)}{\sin\left(e^{(\cdot)}\right)}}_{\max,[\ln\frac{\pi}{2n},\ln\frac{\pi}{2}]}
        \leq\frac{\norm{g}_{\max,[-\pi,\pi]}}{2},
    \end{equation}
    whereas length of the interval is $\ln\frac{\pi}{2}-\ln\frac{\pi}{2n}=\ln n$.
    Denoting the integrand by $h(s)$, we have
    \begin{equation}
        \frac{\mathrm{d}}{\mathrm{d}s}h(s)
        =\frac{\mathrm{d}}{\mathrm{d}v}h(v)\frac{\mathrm{d}}{\mathrm{d}s}v(s)
        =\frac{\mathrm{d}}{\mathrm{d}v}h(v)e^s
    \end{equation}
    from the chain rule. Thus the Lipschitz constant has the scaling
    \begin{equation}
    \begin{aligned}
        \mu&=\max_s\abs{\frac{1}{\pi}\frac{\mathrm{d}}{\mathrm{d}s}\left(e^sg\left(\frac{m\pi}{n}-2e^s\right)e^{-(n-1)ie^s}\frac{\sin\left(ne^s\right)}{\sin\left(e^s\right)}\right)}\\
        &=\mathbf{O}\left(n\log(n)\norm{g'}_{\max,[-\pi,\pi]}+n^2\log(n)\norm{g}_{\max,[-\pi,\pi]}\right).
    \end{aligned}
    \end{equation}
    Moreover, the integrand has absolute value
    \begin{equation}
        \abs{\frac{1}{\pi}e^sg\left(\frac{m\pi}{n}-2e^s\right)e^{-(n-1)ie^s}\frac{\sin\left(ne^s\right)}{\sin\left(e^s\right)}}
        =\frac{1}{\pi}e^s\abs{g\left(\frac{m\pi}{n}-2e^s\right)\frac{\sin\left(ne^s\right)}{\sin\left(e^s\right)}}
    \end{equation}
    and argument
    \begin{equation}
    \begin{aligned}
        &\ \mathbf{Arg}\left(\frac{1}{\pi}e^sg\left(\frac{m\pi}{n}-2e^s\right)e^{-(n-1)ie^s}\frac{\sin\left(ne^s\right)}{\sin\left(e^s\right)}\right)\\
        =&\ \mathbf{Mod}_{2\pi}\left(\mathbf{Arg}\left(g\left(\frac{m\pi}{n}-2e^s\right)\right)
        -(n-1)e^s
        +\mathbf{Arg}\left(\sin\left(ne^s\right)\right)-\mathbf{Arg}\left(\sin\left(e^s\right)\right)
        \right).
    \end{aligned}
    \end{equation}
    Invoking \lem{block_int} with additional arithmetics to compute the trigonometric and exponential functions (for $s\in[\ln\frac{\pi}{2n},\ln\frac{\pi}{2}]$), this integral can be implemented to accuracy $\mathbf{O}(\epsilon)$ with gate complexity
    \begin{equation}
        \mathbf{O}\left(\polylog\left(\frac{n\norm{g'}_{\max,[-\pi,\pi]}}{\norm{g}_{\max,[-\pi,\pi]}\epsilon}+\frac{n^2}{\epsilon}\right)\right).
    \end{equation}
    This completes the proof since analysis of the remaining two convolutions proceeds in a similar way as above.
\end{proof}
\begin{remark}
    Although we have only constructed the circuit to generate Fourier coefficients in the exponential-form expansion, the construction can be trivially adapted to other settings such as the trigonometric form. For instance, suppose we have 
    \begin{equation}
        g(\omega)=\sum_{j=-\infty}^\infty\xi_je^{-ij\omega}
        =\xi_0+\sum_{j=1}^{\infty}\left(\left(\xi_j+\xi_{-j}\right)\cos(j\omega)+\left(-i\xi_j+i\xi_{-j}\right)\sin(j\omega)\right),
    \end{equation}
    and we want to implement $\sum_{j=0}^{n-1}\left(\xi_j+\xi_{-j}\right)L^j$, with the first coefficient rescaled. Because
    \begin{equation}
        g(\omega)+g(-\omega)=\sum_{j=-\infty}^{\infty}\left(\xi_j+\xi_{-j}\right)e^{-ij\omega},
    \end{equation}
    it suffices for us to invoke the above circuit with the even function $g(\omega)+g(-\omega)$.
    
    We can thus apply this to generate the Chebyshev coefficients, since Chebyshev expansion can be recast as a Fourier expansion of an even function in the trigonometric form, as per \eq{cheby_trig}. Specifically, given a Chebyshev expansion $p(x)=\sum_{j=0}^{n-1}\widetilde\beta_j\widetilde{\mathbf{T}}_j(x)$ of a degree-($n-1$) polynomial, we can block encode
    $\sum_{j=0}^{n-1}\widetilde{\beta}_jL_{n}^j/\alpha$ with some normalization factor $\alpha=\mathbf{O}\left(\norm{p}_{\max,[-1,1]}\log(n)\right)$. When applied to $\ket{0}$, this gives the Chebyshev coefficient state $\sum_{j=0}^{n-1}\widetilde\beta_j\ket{j}/\norm{\widetilde\beta}$ with probability $\norm{\widetilde\beta}^2/\alpha^2$.
    We can then subtract $n-1$ and negate the quantum register to get $\sum_{j=0}^{n-1}\widetilde\beta_j\ket{n-1-j}/\norm{\widetilde\beta}$ with the same probability.
    By a further application of $(I_n-L_n)/2$ through block encoding, we obtain the shifted Chebyshev coefficients 
    \begin{equation}
        \frac{\sum_{k=0}^{n-1}(\widetilde\beta_k-\widetilde\beta_{k+2})\ket{n-1-k}}{\alpha_{\widetilde\beta}}
    \end{equation}
    for $\alpha_{\widetilde\beta}=\sqrt{\sum_{k=0}^{n-1}|\widetilde{\beta}_k-\widetilde{\beta}_{k+2}|^2}=\mathbf{O}\left(\norm{p(\cos)\sin}_{2,[-\pi,\pi]}\right)$, with success probability $\alpha_{\widetilde\beta}^2/\alpha^2$. Typically, this ratio can be computed exactly on a classical computer, so the success probability can be boosted exactly to unity using 
    \begin{equation}
        \mathbf{O}\left(\frac{\norm{p}_{\max,[-1,1]}\log(n)}{\norm{p(\cos)\sin}_{2,[-\pi,\pi]}}\right)
    \end{equation}
    steps of amplitude amplification.
    Altogether, this gives the gate complexity
    \begin{equation}
        \mathbf{O}\left(\frac{\norm{p}_{\max,[-1,1]}\log(n)}{\norm{p(\cos)\sin}_{2,[-\pi,\pi]}}
        \polylog\left(\frac{\norm{p'}_{\max,[-1,1]}n}{\norm{p(\cos)\sin}_{2,[-\pi,\pi]}\epsilon}
        +\frac{\norm{p}_{\max,[-1,1]}n^2}{\norm{p(\cos)\sin}_{2,[-\pi,\pi]}\epsilon}\right)\right)
    \end{equation}
    to prepare the shifting of Chebyshev coefficients, fulfilling the requirement of QEVT in \thm{qevt}.
    See the remark succeeding \thm{qevt_block} for further discussions on the scaling of $\norm{p(\cos)\sin}_{2,[-\pi,\pi]}$.
\end{remark}

\section{Applications}
\label{sec:app}
We now apply QEVT to solve linear differential equations in \sec{app_diff_eq} and prepare ground states in \sec{app_ground}. For both applications, the underlying idea is to implement truncated Chebyshev expansions that approximate the target functions, similar to previous QSVT-based results~\cite{Low2016HamSim,Low2019hamiltonian,Lin2020nearoptimalground}. However, the challenge here is that our input matrices are no longer Hermitian or even diagonalizable. So to establish the claimed complexities in \thm{diff_eq} and \thm{ground}, we will develop more general error bounds for truncating Chebyshev expansions of matrix functions, which may be of independent interest.

\subsection{Quantum algorithm for linear differential equations}
\label{sec:app_diff_eq}

We first consider the homogeneous linear differential equations
\begin{equation}
    \frac{\mathrm{d}}{\mathrm{d}t}x(t)=Cx(t),
\end{equation}
whose solution is given by
\begin{equation}
    x(t)=e^{tC}x(0).
\end{equation}
When $C$ has only imaginary eigenvalues, we can use QEVT to implement the function $e^{-i\alpha_Ctx}$ on a block encoding of $iC/\alpha_C$, which can be easily constructed from a block encoding of $C/\alpha_C$. For presentational purpose, we change the variable to $A=iC$ and describe our result as to implement $e^{-i\alpha_A tx}$ on a block encoding of $A/\alpha_A$.

Suppose that the matrix exponential function has the (rescaled) Chebyshev expansion $e^{-itA}=\sum_{j=0}^{\infty}\widetilde{\beta}_j\widetilde{\mathbf{T}}_j\left(\frac{A}{\alpha_A}\right)$. We wish to choose the truncate order $n$ sufficiently large so that the error $\norm{e^{-itA}-\sum_{j=0}^{n-1}\widetilde{\beta}_j\widetilde{\mathbf{T}}_j\left(\frac{A}{\alpha_A}\right)}$ is at most $\epsilon$. The problem can then be solved by implementing the truncated series using QEVT. When $A$ is Hermitian, we can use \prop{trunc_exp} to get an $n$ scaling like $n=\mathbf{O}\left(\alpha_At+\log\left(\frac{1}{\epsilon}\right)\right)$, which leads to the optimal Hamiltonian simulation results from previous work~\cite{Low2016HamSim,Low2019hamiltonian}. In the general case where $A$ is not diagonalizable, \prop{trunc_exp} is no longer applicable, and we instead develop the following bound for truncating matrix exponentials.

\begin{lemma}[Chebyshev truncation of matrix exponentials]
\label{lem:trunc_matrix_exp}
    Let $\widetilde A$ be a matrix with eigenvalues belonging to the real interval $[-\frac{1}{2},\frac{1}{2}]$.
    Suppose that $\widetilde A=SJS^{-1}$ has a Jordan form decomposition with upper bound $\kappa_S\geq\norm{S}\norm{S}^{-1}$ on the Jordan condition number and size $d_{\max}$ of the largest Jordan block.
    Given $\tau>0$, let $e^{-i\tau x}=\sum_{j=0}^\infty\widetilde\beta_j\widetilde{\mathbf{T}}_j(x)$ be the Chebyshev expansion of the complex exponential function $e^{-i\tau x}$.
    Then,
    \begin{equation}
        \norm{e^{-i\tau\widetilde A}-\sum_{j=0}^{n-1}\widetilde{\beta}_j\widetilde{\mathbf{T}}_j\left(\widetilde A\right)}
        =\mathbf{O}\left(\kappa_S\left(\frac{ed_{\max}\tau}{2n}\right)^n\right).
    \end{equation}
\end{lemma}
\begin{proof}
    We start with the triangle inequality estimate
    \begin{equation}
    \begin{aligned}
        \norm{e^{-i\tau\widetilde A}-\sum_{j=0}^{n-1}\widetilde{\beta}_j\widetilde{\mathbf{T}}_j\left(\widetilde A\right)}
        =\norm{\sum_{j=n}^{\infty}\widetilde{\beta}_j\widetilde{\mathbf{T}}_j\left(\widetilde A\right)}
        \leq\sum_{j=n}^{\infty}\abs{\widetilde{\beta}_j}\norm{\widetilde{\mathbf{T}}_j\left(\widetilde A\right)}.
    \end{aligned}
    \end{equation}
    Here, the Chebyshev expansion coefficients are given by Bessel functions of the first kind as $\widetilde{\beta}_j=2i^j\mathbf{J}_j(\tau)$, which are bounded by
    \begin{equation}
        \abs{\widetilde{\beta}_j}=2\abs{\mathbf{J}_j(\tau)}\leq\frac{2}{j!}\left(\frac{\tau}{2}\right)^j.
    \end{equation}
    It remains to analyze size of the matrix polynomial $\norm{\widetilde{\mathbf{T}}_j\left(\widetilde A\right)}$.

    In \cor{matrix_poly_bound} of \append{analysis_cheby}, we will derive a general bound for matrix polynomial functions that reads 
    \begin{equation}
        \norm{p_j(C)}=\mathbf{O}\left(\kappa_S\left(\frac{j}{\sqrt{\delta}}\right)^{d_{\max}-1}\norm{p_j}_{\max,[a,b]}\right),
    \end{equation}
    where $p_j$ is a degree-$j$ polynomial and 
    eigenvalues of $C$ are all contained in $[a+\delta,b-\delta]$.
    The prefactor of the bound only depends on $d_{\max}$ and $[a,b]$ which we treat as constant, and is independent of the polynomial degree $j$ and the margin $\delta$.
    
    We now apply this bound to the rescaled Chebyshev polynomials $\widetilde{\mathbf{T}}_j(\widetilde A)$. To this end, we set $[a,b]=[-1,1]$, which contains all eigenvalues of $\widetilde A$, and is constant-distance gapped from $[-\frac{1}{2},\frac{1}{2}]$ that also encloses the eigenvalues.
    We have
    \begin{equation}
        \norm{\widetilde{\mathbf{T}}_j\left(\widetilde A\right)}
        =\mathbf{O}\left(\kappa_Sj^{d_{\max}-1}\norm{\widetilde{\mathbf{T}}_j}_{\max,[-1,1]}\right)
        =\mathbf{O}\left(\kappa_Sj^{d_{\max}-1}\right)
        =\mathbf{O}\left(\kappa_Sd_{\max}^j\right).
    \end{equation}
    Thus, there exists a constant $c$ for which
    \begin{equation}
        \norm{e^{-i\tau\widetilde A}-\sum_{j=0}^{n-1}\widetilde{\beta}_j\widetilde{\mathbf{T}}_j\left(\widetilde A\right)}
        \leq c\sum_{j=n}^{\infty}
        \frac{1}{j!}\left(\frac{\tau}{2}\right)^j
        \kappa_Sd_{\max}^j
        \leq \frac{c\kappa_S}{\sqrt{2\pi}}\sum_{j=n}^{\infty}
        \left(\frac{ed_{\max}\tau}{2j}\right)^j,
    \end{equation}
    where we have used the bound $\sqrt{2\pi j}\left(\frac{j}{e}\right)^j\leq j!\leq e\sqrt{j}\left(\frac{j}{e}\right)^j$. Assuming $n\geq ed_{\max}\tau$, we continue the calculation to get
    \begin{equation}
        \norm{e^{-i\tau\widetilde A}-\sum_{j=0}^{n-1}\widetilde{\beta}_j\widetilde{\mathbf{T}}_j\left(\widetilde A\right)}
        \leq\frac{c\kappa_S}{\sqrt{2\pi}}\left(\frac{ed_{\max}\tau}{2n}\right)^n\sum_{j=n}^{\infty}
        \left(\frac{1}{2}\right)^{j-n}
        =\frac{2c\kappa_S}{\sqrt{2\pi}}\left(\frac{ed_{\max}\tau}{2n}\right)^n.
    \end{equation}
    The claimed bound is now established.
\end{proof}

The above bound essentially quantifies the error of Chebyshev truncation for producing an unnormalized solution state. The effect of normalization is considered by the following bound, which follows from a similar reasoning as in \cor{lin_sys_perturb}.
\begin{lemma}[Quantum state transformation with perturbation]
\label{lem:state_transform_perturb}
    Let $C$ and $\widetilde C$ be invertible matrices of the same size, acting on (normalized) quantum states $\ket{\psi}$ and $\ket{\widetilde\psi}$. We have
	\begin{equation}
		\norm{\frac{\widetilde{C}\ket{\widetilde\psi}}{\norm{\widetilde{C}\ket{\widetilde\psi}}}-\frac{C\ket{\psi}}{\norm{C\ket{\psi}}}}
		\leq\frac{2\norm{C}\norm{\ket{\widetilde\psi}-\ket{\psi}}}{\norm{C\ket{\psi}}}
        +\frac{2\norm{\widetilde{C}-C}}{\norm{C\ket{\psi}}}.
	\end{equation}
\end{lemma}
\begin{proof}
    We use the triangle inequality to upper bound the left-hand side as
    \begin{equation}
    \begin{aligned}
        \norm{\frac{C\ket{\psi}}{\norm{C\ket{\psi}}}-\frac{\widetilde{C}\ket{\widetilde\psi}}{\norm{\widetilde{C}\ket{\widetilde\psi}}}}
		\leq\norm{\frac{C\ket{\psi}}{\norm{C\ket{\psi}}}-\frac{\widetilde{C}\ket{\widetilde\psi}}{\norm{C\ket{\psi}}}}
		+\norm{\frac{\widetilde{C}\ket{\widetilde\psi}}{\norm{C\ket{\psi}}}-\frac{\widetilde{C}\ket{\widetilde\psi}}{\norm{\widetilde{C}\ket{\widetilde\psi}}}}.
    \end{aligned}
    \end{equation}
    For the first term, we have
    \begin{equation}
        \norm{\frac{C\ket{\psi}}{\norm{C\ket{\psi}}}-\frac{\widetilde{C}\ket{\widetilde\psi}}{\norm{C\ket{\psi}}}}
        =\frac{\norm{C\ket{\psi}-\widetilde{C}\ket{\widetilde\psi}}}{\norm{C\ket{\psi}}},
    \end{equation}
    whereas the second term can be further bounded similarly as
    \begin{equation}
        \norm{\frac{\widetilde{C}\ket{\widetilde\psi}}{\norm{C\ket{\psi}}}-\frac{\widetilde{C}\ket{\widetilde\psi}}{\norm{\widetilde{C}\ket{\widetilde\psi}}}}
        =\norm{\widetilde{C}\ket{\widetilde\psi}}\abs{\frac{1}{\norm{C\ket{\psi}}}-\frac{1}{\norm{\widetilde{C}\ket{\widetilde\psi}}}}
        \leq\frac{\norm{C\ket{\psi}-\widetilde{C}\ket{\widetilde\psi}}}{\norm{C\ket{\psi}}}.
    \end{equation}
    The claimed bound then follows from
    \begin{equation}
        \norm{C\ket{\psi}-\widetilde{C}\ket{\widetilde\psi}}
        \leq\norm{C\ket{\psi}-C\ket{\widetilde\psi}}
        +\norm{C\ket{\widetilde\psi}-\widetilde{C}\ket{\widetilde\psi}}
        \leq
        \norm{C}\norm{\ket{\psi}-\ket{\widetilde\psi}}
        +\norm{C-\widetilde{C}}.
    \end{equation}
\end{proof}

\begin{theorem}[Quantum differential equation algorithm]
\label{thm:diff_eq}
    Let $A$ be a square matrix with only real eigenvalues, such that $A/\alpha_A$ is block encoded by $O_A$ with some normalization factor $\alpha_A\geq\norm{A}$.
    Suppose that $A/\alpha_A=SJS^{-1}$ has a Jordan form decomposition with upper bound $\kappa_S\geq\norm{S}\norm{S}^{-1}$ on the Jordan condition number.
    Let $O_\psi\ket{0}=\ket{\psi}$ be the oracle preparing the initial state.
    
    Then, applying \thm{qevt} to the function $e^{-i\alpha_Atx}$ truncated at order
    \begin{equation}
        n=\mathbf{O}\left(\alpha_At+\log\left(\frac{\kappa_S}{\epsilon}\right)\right)
    \end{equation}
    produces the state
    \begin{equation}
        \frac{e^{-itA}\ket{\psi}}{\norm{e^{-itA}\ket{\psi}}}
    \end{equation}
    with accuracy $\epsilon$ and probability $1-\pfail$. The algorithm uses
    \begin{equation}
        \mathbf{O}\left(\frac{\alphaT}{\alphaFPsi}\alphaU
        \left(\alpha_At+\log\left(\frac{\kappa_S}{\epsilon}\right)\right)
        \log\left(\frac{\alphaT}{\alphaFPsi\epsilon}\right)\log\left(\frac{1}{\pfail}\right)\right)
    \end{equation}
    queries to controlled-$O_A$, controlled-$O_\psi$, and their inverses,
    where $\alphaU$ satisfies \eq{alphaU2}, $\alphaT$ satisfies \eq{alphaT_alphaPPsi}, and
    \begin{equation}
        \alphaFPsi\leq\norm{e^{-itA}\ket{\psi}}
    \end{equation}
    is a lower bound on size of the solution vector.
\end{theorem}
\begin{proof}
    There are two sources of error, one from Chebyshev truncation of the matrix exponential function, and the other from the application of QEVT.
    Most of our effort will be spent on choosing the truncate order $n$, such that $\frac{e^{-itA}\ket{\psi}}{\norm{e^{-itA}\ket{\psi}}}$ approximates $\frac{\sum_{j=0}^{n-1}\widetilde{\beta}_j\widetilde{\mathbf{T}}_j\left(\frac{A}{\alpha_A}\right)\ket{\psi}}{\norm{\sum_{j=0}^{n-1}\widetilde{\beta}_j\widetilde{\mathbf{T}}_j\left(\frac{A}{\alpha_A}\right)\ket{\psi}}}$ to accuracy $\epsilon/2$.
    We know from the perturbation bound of \lem{state_transform_perturb} that
    \begin{equation}
        \norm{\frac{e^{-itA}\ket{\psi}}{\norm{e^{-itA}\ket{\psi}}}
        -\frac{\sum_{j=0}^{n-1}\widetilde{\beta}_j\widetilde{\mathbf{T}}_j\left(\frac{A}{\alpha_A}\right)\ket{\psi}}{\norm{\sum_{j=0}^{n-1}\widetilde{\beta}_j\widetilde{\mathbf{T}}_j\left(\frac{A}{\alpha_A}\right)\ket{\psi}}}}
        \leq
        \frac{2\norm{e^{-itA}-\sum_{j=0}^{n-1}\widetilde{\beta}_j\widetilde{\mathbf{T}}_j\left(\frac{A}{\alpha_A}\right)}}{\norm{e^{-itA}\ket{\psi}}}.
    \end{equation}
    Furthermore, \lem{trunc_matrix_exp} implies that the numerator has the scaling
    \begin{equation}
        \norm{e^{-itA}-\sum_{j=0}^{n-1}\widetilde{\beta}_j\widetilde{\mathbf{T}}_j\left(\frac{A}{\alpha_A}\right)}
        =\mathbf{O}\left(\kappa_S\left(\frac{ed_{\max}\alpha_At}{2n}\right)^n\right),
    \end{equation}
    as long as $\alpha_A\geq2\norm{A}$, which can always be satisfied using the rescaling trick of \eq{block_rescaling}. It remains to analyze $1/\norm{e^{-itA}\ket{\psi}}$.

    We have from \eq{denominator_bnd} that
    \begin{equation}
        \frac{1}{\norm{e^{-itA}\ket{\psi}}}
        \leq\norm{e^{itA}}
        =\norm{\sum_{j=0}^{n-1}\widetilde{\beta}_j\widetilde{\mathbf{T}}_j\left(\frac{A}{\alpha_A}\right)}
        +\mathbf{O}\left(\kappa_S\left(\frac{ed_{\max}\alpha_At}{2n}\right)^n\right).
    \end{equation}
    Here, the first term is a degree-($n-1$) matrix polynomial function, which can be bounded again using \cor{matrix_poly_bound} from \append{analysis_cheby} as
    \begin{equation}
    \begin{aligned}
        \norm{\sum_{j=0}^{n-1}\widetilde{\beta}_j\widetilde{\mathbf{T}}_j\left(\frac{A}{\alpha_A}\right)}
        &=\mathbf{O}\left(\kappa_Sn^{d_{\max}-1}\norm{\sum_{j=0}^{n-1}\widetilde{\beta}_j\widetilde{\mathbf{T}}_j}_{\max,[-1,1]}\right)\\
        &=\mathbf{O}\left(\kappa_Sd_{\max}^n\left(\norm{e^{-i\alpha_At}}_{\max,[-1,1]}+\left(\frac{e\alpha_At}{2n}\right)^n\right)\right)\\
        &=\mathbf{O}\left(\kappa_Sd_{\max}^n\left(1+\left(\frac{e\alpha_At}{2n}\right)^n\right)\right).
    \end{aligned}
    \end{equation}

    To approximate with accuracy $\epsilon/2$, it thus suffices to choose a truncate order $n$ scaling like $n=\mathbf{O}\left(\alpha_At+\log\left(\frac{\kappa_S}{\epsilon}\right)\right)$. We now apply QEVT with accuracy $\epsilon/2$ as well. The claimed complexity follows from \thm{qevt}.
\end{proof}
\begin{remark}
    Methods for bounding $\alphaU$ and $\alphaT$ are discussed in the remarks succeeding \thm{generate_history} and \thm{qevt}, whereas a bound for $1/\alphaFPsi$ is given in the above proof. For diagonalizable coefficient matrices with purely imaginary spectra, we show in \append{analysis_cheby_bernstein} and \append{analysis_cheby_carleson} that $\alphaU=\mathbf{O}(\kappa_S)$, and that $\frac{\alphaT}{\alphaFPsi}=\mathbf{O}\left(\kappa_S\log(n)\right)$ in the worst case. But the $\log(n)$ factor can be shaved off for an average input, resulting in the complexity
    \begin{equation}
        \mathbf{O}\left(\kappa_S^2\left(\alpha_At+\log\left(\frac{\kappa_S}{\epsilon}\right)\right)\log\left(\frac{\kappa_S}{\epsilon}\right)\log\left(\frac{1}{\pfail}\right)\right)
    \end{equation}
    strictly linear in the evolution time.

    In the circuit implementation, we use \thm{fourier_coeff} to prepare the shifted Chebyshev coefficients; see the remark succeeding that theorem. To this end, we need to implement the oracle for the complex exponential function, which can be achieved in a standard way using a classical reversible computation. Moreover, we have $\norm{f}_{\max,[-1,1]}=\mathbf{O}(1)$ and $\norm{f(\cos)\sin}_{2,[-\pi,\pi]}=\mathbf{\Omega}(1)$ for the exponential function $f(x)=e^{-i\alpha_A tx}$. Therefore, the gate complexity for preparing the shifted Chebyshev coefficients is polylogarithmic in the input parameters.

It is unclear whether our method can be used to solve differential equations with time-dependent coefficients. This is somewhat reminiscent of the limitation that Chebyshev-based approach~\cite{Low2016HamSim} is not directly applicable to the quantum simulation of time-dependent Hamiltonians.
Note however that the above algorithm can be adapted to solve an inhomogeneous linear differential equation $\frac{\mathrm{d}x(t)}{\mathrm{d}t}=Cx(t)+b$,
which has the exact solution given by $x(t)=e^{t C}x_0+\frac{e^{tC}-I}{C}b$.
After the substitution $A=iC$, we consider the Chebyshev expansion of both
\begin{equation}
    f(x)=e^{-i\alpha_Atx}
    \approx\sum_{j=0}^{n-1}\beta_j\mathbf{T}_j(x)
\end{equation}
and
\begin{equation}
    g(x)=\frac{e^{-i\alpha_Atx}-1}{-i\alpha_Ax}
    \approx\sum_{j=0}^{n-1}\xi_j\mathbf{T}_j(x).
\end{equation}
We can truncate $f$ and $g$ at order $n=\mathbf{O}\left(\alpha_At+\log\left(\frac{1}{\delta}\right)\right)$ as both functions can be extended to be analytic on the entire complex plane. Then, we use the Chebyshev generating function as before to generate a state proportional to
\begin{equation}
    \ket{0}\sum_{j=0}^{n-1}\beta_j\mathbf{T}_j\left(\frac{A}{\alpha_A}\right)\ket{x_0}
    +\ket{1}\sum_{j=0}^{n-1}\xi_j\mathbf{T}_j\left(\frac{A}{\alpha_A}\right)\ket{b}.
\end{equation}
Finally, we perform amplitude amplification toward the state $(\norm{x_0}\ket{0}+\norm{b}\ket{1})/\sqrt{\norm{x_0}^2+\norm{b}^2}$ in the first register.

Note that the query complexity of initial state preparation can be improved using the block preconditioning technique of~\cite{OptInit}.
\end{remark}

\subsection{Quantum algorithm for ground state preparation}
\label{sec:app_ground}

As a second application, we present a quantum algorithm that prepares the ground state of an input matrix with real eigenvalues.

Let $A$ be the input matrix with only real eigenvalues and block encoded as $A/\alpha_A$. Suppose that the Jordan form decomposition $A/\alpha_A=SJ S^{-1}$ holds with upper bound $\kappa_S\geq\norm{S}\norm{S^{-1}}$ on the Jordan condition number and size $d_{\max}$ of the largest Jordan block. 
To simplify the analysis, we assume that $\lambda_0$ is the smallest eigenvalue of $A$ with eigenstate $\ket{\psi_0}$ that is \emph{nondefective} and \emph{nonderogatory}. In other words, there is only one Jordan block correponding to $\lambda_0$, and size of that block is $1$. 
Then our goal is to approximately prepare $\ket{\psi_0}$, given an initial trial state expanded in the Jordan basis as
\begin{equation}
    \ket{\psi}=\gamma_0\ket{\psi_0}+\sum_{l=1}^{d-1}\gamma_l\ket{\psi_l}.
\end{equation}

Following previous conventions~\cite{Lin2020nearoptimalground}, we assume that $\lambda_0$ is separated from the next eigenvalue $\lambda_1$:
\begin{equation}
    \lambda_0\leq-\frac{\delta_A}{2}<0<\frac{\delta_A}{2}\leq\lambda_1
\end{equation}
with some spectral gap $\delta_A>0$. Here, we have placed $\lambda_0$ and $\lambda_1$ on different sides of the origin, which is without loss of generality as we can always shift and rescale the input block encoding~\cite{Lin2020nearoptimalground}. 
In practice, it is also natural to consider quantum algorithms for preparing an arbitrary eigenstate as opposed to the ground state, but such an extension is fairly straightforward and will not be discussed here.

As is to be expected, we will solve the ground state preparation problem by applying QEVT to implement a truncated Chebyshev expansion. In the special case where the input is Hermitian, this route was pursued by previous work~\cite{Lin2020nearoptimalground} with QSVT. In that case, the input matrix can be unitarily diagonalized and one only needs to find a Chebyshev truncation for the scalar error function as in \prop{trunc_erf}. However, such a truncation result is not applicable here per se, because our matrices are not necessarily Hermitian (or even diagonalizable).
Instead, we prove the following bound for truncating matrix functions.

\begin{lemma}[Chebyshev truncation of matrix sign functions]
\label{lem:trunc_matrix_sign}
    Let $\widetilde{A}$ be a matrix with eigenvalues belonging to $[-\frac{1}{2},-2\widetilde\delta]\cup[2\widetilde\delta,\frac{1}{2}]$.
    Suppose that $\widetilde{A}=SJS^{-1}$ has a Jordan form decomposition with upper bound $\kappa_S\geq\norm{S}\norm{S^{-1}}$ on the Jordan condition number and size $d_{\max}$ of the largest Jordan block.
    Given $c>0$, let $\frac{1-\mathbf{Erf}(cx)}{2}=\sum_{j=0}^{\infty}\widetilde{\beta}_j\widetilde{\mathbf{T}}_j(x)$ be the Chebyshev expansion of the (shifted and rescaled) error function $\mathbf{Erf}(x)
    =\frac{2}{\sqrt{\pi}}\int_{0}^{x}\mathrm{d}u\ e^{-u^2}$. Then,
    \begin{small}
    \newmaketag
    \begin{equation}
        \norm{\frac{I-\mathbf{Sgn}\left(\widetilde A\right)}{2}
        -\sum_{j=0}^{n-1}\widetilde{\beta}_j\widetilde{\mathbf{T}}_j\left(\widetilde A\right)}
        =\mathbf{O}\left(\kappa_S\left(\frac{n}{\sqrt{\widetilde\delta}}\right)^{d_{\max}-1}
        \left(\frac{e^{-c^2\widetilde\delta^2}}{c\widetilde\delta}
        +\frac{c}{n}e^{-\frac{n^2}{2m}}
        +\frac{c}{n}e^{-\frac{c^2}{2}}\left(\frac{ec^2}{2m}\right)^m\right)\right),
    \end{equation}
    \end{small}
    with the sign function
    \begin{equation}
        \mathbf{Sgn}(x)=
        \begin{cases}
            -1,\quad &x<0,\\
            0,&x=0,\\
            1,&x>0.
        \end{cases}
    \end{equation}
\end{lemma}
\begin{proof}
    Similar to the proof of \lem{trunc_matrix_exp}, we will also be using the matrix polynomial bound~\cor{matrix_poly_bound} derived in \append{analysis_cheby}. However, the challenge is that we need to handle two different polynomials here: we have $1-\sum_{j=0}^{n-1}\widetilde{\beta}_j\widetilde{\mathbf{T}}_j(x)$ over the interval $[-1,-\widetilde\delta]$, and $-\sum_{j=0}^{n-1}\widetilde{\beta}_j\widetilde{\mathbf{T}}_j(x)$ over the interval $[\widetilde\delta,1]$. Therefore, we will separate Jordan blocks of $\widetilde A$ accordingly, based on intervals to which the eigenvalues belong.

    Let us start with the estimate
    \begin{equation}
        \norm{\frac{I-\mathbf{Sgn}\left(\widetilde A\right)}{2}
        -\sum_{j=0}^{n-1}\widetilde{\beta}_j\widetilde{\mathbf{T}}_j\left(\widetilde A\right)}
        \leq\kappa_S
        \norm{\frac{I-\mathbf{Sgn}\left(J\right)}{2}
        -\sum_{j=0}^{n-1}\widetilde{\beta}_j\widetilde{\mathbf{T}}_j\left(J\right)}.
    \end{equation}
    Note that by our assumption, all eigenvalues of $\widetilde A$ belong to $[-\frac{1}{2},-2\widetilde\delta]$ and $[2\widetilde\delta,\frac{1}{2}]$, whereas the sign function is analytic on larger intervals $[-1,-\widetilde\delta]$ and $[\widetilde\delta,1]$. Thus, the above matrix functions are indeed well defined as per \eq{jordan_form_transformation}.
    
    We now collect all Jordan blocks with negative eigenvalues into $J_{-}$ and those with positive eigenvalues into $J_{+}$ (size of $J_-$ and $J_-$ sums up to that of $J$). This gives
    \begin{equation}
        \norm{\frac{I-\mathbf{Sgn}\left(J\right)}{2}
        -\sum_{j=0}^{n-1}\widetilde{\beta}_j\widetilde{\mathbf{T}}_j\left(J\right)}
        \leq\norm{I
        -\sum_{j=0}^{n-1}\widetilde{\beta}_j\widetilde{\mathbf{T}}_j\left(J_-\right)}
        +\norm{
        \sum_{j=0}^{n-1}\widetilde{\beta}_j\widetilde{\mathbf{T}}_j\left(J_+\right)}.
    \end{equation}
    Invoking \cor{matrix_poly_bound}, we obtain
    \begin{equation}
    \begin{aligned}
        \norm{I
        -\sum_{j=0}^{n-1}\widetilde{\beta}_j\widetilde{\mathbf{T}}_j\left(J_-\right)}
        &=\mathbf{O}\left(\left(\frac{n}{\sqrt{\widetilde\delta}}\right)^{d_{\max}-1}\norm{1-\sum_{j=0}^{n-1}\widetilde{\beta}_j\widetilde{\mathbf{T}}_j}_{\max,[-1,-\widetilde\delta]}\right),\\
        \norm{\sum_{j=0}^{n-1}\widetilde{\beta}_j\widetilde{\mathbf{T}}_j\left(J_+\right)}
        &=\mathbf{O}\left(\left(\frac{n}{\sqrt{\widetilde\delta}}\right)^{d_{\max}-1}\norm{\sum_{j=0}^{n-1}\widetilde{\beta}_j\widetilde{\mathbf{T}}_j}_{\max,[\widetilde\delta,1]}\right).
    \end{aligned}
    \end{equation}
    The claim now follows from the Chebyshev truncation bounds for the scalar error and sign functions~\cite{Low17,Wan22}.
\end{proof}

\begin{theorem}[Quantum ground state preparation algorithm]
\label{thm:ground}
    Let $A$ be a square matrix with only real eigenvalues, such that $A/\alpha_A$ is block encoded by $O_A$ with some normalization factor $\alpha_A\geq\norm{A}$.
    Suppose that $A/\alpha_A=SJS^{-1}$ has a Jordan form decomposition with upper bound $\kappa_S\geq\norm{S}\norm{S}^{-1}$ on the Jordan condition number. Let eigenvalues of $A$ be ordered nondecreasingly, with $\lambda_0$ the smallest one with eigenstate $\ket{\psi_0}$, which is nondefective and nonderogatory satisfying the condition
    \begin{equation}
        \lambda_0\leq-\frac{\delta_A}{2}<0<\frac{\delta_A}{2}\leq\lambda_1
    \end{equation}
    for some spectral gap $\delta_A>0$. 
    Let $O_\psi\ket{0}=\ket{\psi}$ be the oracle preparing the initial state with the Jordan basis expansion
    \begin{equation}
        \ket{\psi}=\gamma_0\ket{\psi_0}+\sum_{l=1}^{d-1}\gamma_l\ket{\psi_l}.
    \end{equation}
    
    Then, applying \thm{qevt} to the error function $1-\mathbf{Erf}(cx)$ with a rescaling factor $c=\mathbf{O}\left(\frac{\alpha_A}{\delta_A}
        \sqrt{\log\left(\frac{\alpha_A}{\delta_A}\frac{\kappa_S}{|\gamma_0|\epsilon}\right)}\right)$ truncated at order 
    \begin{equation}
        n=\mathbf{O}\left(\frac{\alpha_A}{\delta_A}
        \log\left(\frac{\alpha_A}{\delta_A}\frac{\kappa_S}{|\gamma_0|\epsilon}\right)\right)
    \end{equation}    
    produces the ground state $\ket{\psi_0}$ with accuracy $\epsilon$, probability $1-\pfail$, and the global phase factor $\gamma_0/\abs{\gamma_0}$. 
    The algorithm uses
    \begin{equation}
        \mathbf{O}\left(\frac{\alphaT}{\abs{\gamma_0}}\alphaU
        \frac{\alpha_A}{\delta_A}\log\left(\frac{\alpha_A}{\delta_A}\frac{\kappa_S}{|\gamma_0|\epsilon}\right)
        \log\left(\frac{\alphaT}{\abs{\gamma_0}\epsilon}\right)
        \log\left(\frac{1}{\pfail}\right)\right)
    \end{equation}
    queries to controlled-$O_A$, controlled-$O_\psi$, and their inverses, 
    where $\alphaU$ satisfies \eq{alphaU} and $\alphaT$ satisfies \eq{alphaT_alphaPPsi}.
\end{theorem}
\begin{proof}
    If $2\norm{A}>\alpha_A\geq\norm{A}$, we first rescale the input block encoding by a factor of $2$ using the trick of \eq{block_rescaling}. So in what follows, we will assume $\alpha_A\geq2\norm{A}$ without loss of generality. This rescaling may increase the Jordan condition number by $2^{d_{\max}-1}$~\cite[Corollary 3.1.21]{horn2012matrix} which we treat as constant. The eigenvalues of the block encoded matrix now satisfy the condition
    \begin{equation}
        -\frac{1}{2}
        \leq\frac{\lambda_0}{\alpha_A}
        \leq-\frac{\delta_A}{2\alpha_A}
        <0
        <\frac{\delta_A}{2\alpha_A}
        \leq\frac{\lambda_1}{\alpha_A}
        \leq\cdots
        \leq\frac{\lambda_{d-1}}{\alpha_A}
        \leq\frac{1}{2}.
    \end{equation}

    We then apply \lem{trunc_matrix_sign} to get
    \begin{small}
    \newmaketag
    \begin{equation}
    \label{eq:trunc_matrix_sign_block}
        \norm{\frac{I-\mathbf{Sgn}\left(\frac{A}{\alpha_A}\right)}{2}
        -\sum_{j=0}^{n-1}\widetilde{\beta}_j\widetilde{\mathbf{T}}_j\left(\frac{A}{\alpha_A}\right)}
        =\mathbf{O}\left(\kappa_S\left(\frac{n}{\sqrt{\widetilde\delta}}\right)^{d_{\max}-1}
        \left(\frac{e^{-c^2\widetilde\delta^2}}{c\widetilde\delta}
        +\frac{c}{n}e^{-\frac{n^2}{2m}}
        +\frac{c}{n}e^{-\frac{c^2}{2}}\left(\frac{ec^2}{2m}\right)^m\right)\right),
    \end{equation}
    \end{small}\\
    where
    \begin{equation}
        \widetilde\delta=\frac{\delta_A}{4\alpha_A}.
    \end{equation}
    From the perturbation bound of \lem{state_transform_perturb},
    \begin{equation}
        \norm{\frac{\frac{I-\mathbf{Sgn}\left(\frac{A}{\alpha_A}\right)}{2}\ket{\psi}}{\norm{\frac{I-\mathbf{Sgn}\left(\frac{A}{\alpha_A}\right)}{2}\ket{\psi}}}
        -\frac{\sum_{j=0}^{n-1}\widetilde{\beta}_j\widetilde{\mathbf{T}}_j\left(\frac{A}{\alpha_A}\right)\ket{\psi}}{\norm{\sum_{j=0}^{n-1}\widetilde{\beta}_j\widetilde{\mathbf{T}}_j\left(\frac{A}{\alpha_A}\right)\ket{\psi}}}}
        \leq\frac{2\norm{\frac{I-\mathbf{Sgn}\left(\frac{A}{\alpha_A}\right)}{2}
        -\sum_{j=0}^{n-1}\widetilde{\beta}_j\widetilde{\mathbf{T}}_j\left(\frac{A}{\alpha_A}\right)}}{\norm{\frac{I-\mathbf{Sgn}\left(\frac{A}{\alpha_A}\right)}{2}\ket{\psi}}}.
    \end{equation}
    Note that denominator of the above bound is exactly $\abs{\gamma_0}$, because
    \begin{equation}
        \norm{\frac{I-\mathbf{Sgn}\left(\frac{A}{\alpha_A}\right)}{2}\ket{\psi}}
        =\norm{S\frac{I-\mathbf{Sgn}\left(J\right)}{2}S^{-1}\ket{\psi}}
        =\norm{\ket{\psi_0}\gamma_0}=\abs{\gamma_0}.
    \end{equation}
    To see the second equality, we permute the Jordan blocks and Jordan basis so that eigenvalues are ordered increasingly.
    Then, we have
    \begin{equation}
        S=\begin{bmatrix}
            \ket{\psi_0} & \cdots & \ket{\psi_{d-1}}
        \end{bmatrix},\qquad
        \frac{I-\mathbf{Sgn}\left(J\right)}{2}
        =\ketbra{0}{0},\qquad
        \ket{\psi}
    =S\begin{bmatrix}
        \gamma_0\\
        \vdots\\
        \gamma_{d-1}
    \end{bmatrix},
    \end{equation}
    from which the above calculation is justified.

    Our goal now is to choose parameters $c,m,n$ so that Chebyshev truncation of the matrix sign function in \eq{trunc_matrix_sign_block} has an error at most $\widetilde\epsilon$. We will make sure that $n=\mathbf{\Omega}(c)$ and $c=\mathbf{\Omega}\left(\frac{1}{\widetilde\delta}\right)$, so the error scaling simplifies to
    \begin{equation}
        \mathbf{O}\left(\kappa_Sn^{\widetilde d_{\max}}
        \left(e^{-c^2\widetilde\delta^2}
        +e^{-\frac{n^2}{2m}}
        +\left(\frac{ec^2}{2m}\right)^m\right)\right)
    \end{equation}
    for $\widetilde d_{\max}=\frac{3d_{\max}-3}{2}=\mathbf{O}(1)$.
    Let us first try
    \begin{equation}
        c\sim\frac{1}{\widetilde\delta}\sqrt{\log\left(\frac{1}{\widetilde\epsilon}\right)},\qquad
        m\sim c^2+\log\left(\frac{1}{\widetilde\epsilon}\right)
        \sim\frac{1}{\widetilde\delta^2}\log\left(\frac{1}{\widetilde\epsilon}\right),\qquad
        n\sim\sqrt{m\log\left(\frac{1}{\widetilde\epsilon}\right)}
        \sim\frac{1}{\widetilde\delta}\log\left(\frac{1}{\widetilde\epsilon}\right),
    \end{equation}
    which ensures that
    \begin{equation}
        e^{-c^2\widetilde\delta^2}
        +e^{-\frac{n^2}{2m}}
        +\left(\frac{ec^2}{2m}\right)^m
        =\mathbf{O}\left(\widetilde\epsilon\right).
    \end{equation}
    However, we have an additional factor due to the non-Hermitian nature of the input matrix, which contributes to an asymptotic scaling of
    \begin{equation}
        \kappa_Sn^{\widetilde d_{\max}}
        \sim\frac{\kappa_S}{\widetilde\delta^{\widetilde d_{\max}}}\log^{\widetilde d_{\max}}\left(\frac{1}{\widetilde\epsilon}\right)
        \lesssim\frac{\kappa_S}{\widetilde\delta^{\widetilde d_{\max}}\sqrt{\widetilde\epsilon}}.
    \end{equation}
    Hence to compensate for this contribution, we choose
    \begin{equation}
        n=\mathbf{\Theta}\left(\frac{1}{\widetilde\delta}\log\left(\frac{\kappa_S}{\widetilde\delta\widetilde\epsilon}\right)\right).
    \end{equation}
    
    Finally, we take the normalization into account and set
    \begin{equation}
        \widetilde\epsilon=\mathbf{\Theta}\left(\abs{\gamma_0}\epsilon\right).
    \end{equation}
    The proof is now complete by noting that
    \begin{equation}
        \frac{\frac{I-\mathbf{Sgn}\left(\frac{A}{\alpha_A}\right)}{2}\ket{\psi}}{\norm{\frac{I-\mathbf{Sgn}\left(\frac{A}{\alpha_A}\right)}{2}\ket{\psi}}}
        =\frac{S\frac{I-\mathbf{Sgn}\left(J\right)}{2}S^{-1}\ket{\psi}}{\norm{S\frac{I-\mathbf{Sgn}\left(J\right)}{2}S^{-1}\ket{\psi}}}
        =\frac{\ket{\psi_0}\gamma_0}{\norm{\ket{\psi_0}\gamma_0}}
        =\frac{\gamma_0}{\abs{\gamma_0}}\ket{\psi_0}.
    \end{equation}
\end{proof}
\begin{remark}
    For a discussion on the scaling of $\alphaU$ and $\alphaT$, see the remarks succeeding \thm{generate_history} and \thm{qevt}. For input matrices that are diagonalizable with real spectra, we show in \append{analysis_cheby_bernstein} that $\alphaU=\mathbf{O}(\kappa_S)$, and that $\alphaT=\mathbf{O}(\kappa_S\log(n))$ in the worst case where the $\log(n)$ factor can be dropped for an average input (\append{analysis_cheby_carleson}), resulting in the complexity
    \begin{equation}
        \mathbf{O}\left(\frac{\alpha_A}{\delta_A}\frac{\kappa_S^2}{\abs{\gamma_0}}
        \log^2\left(\frac{\kappa_S}{\abs{\gamma_0}\epsilon}\right)
        \log\left(\frac{1}{\pfail}\right)\right).
    \end{equation}
    Note that we have also removed a factor of $\frac{\alpha_A}{\delta_A}$ from inside of a logarithmic factor, due to the fact that $d_{\max}=1$ in \eq{trunc_matrix_sign_block}. 
    Anyway, our result naturally reduces to the nearly optimal ground state preparation result from previous work~\cite{Lin2020nearoptimalground} in the special case where the input matrix is a Hermitian Hamiltonian.

    The complexity of our algorithm depends on the expansion coefficient $\gamma_0$ of the initial trial state under the Jordan basis. This is compatible with previous work~\cite{Lin2020nearoptimalground} that uses the notion of initial overlap, because the basis is orthonormal when the input matrix is Hermitian.
    We have also explicitly worked out the global phase factor $\gamma_0/\abs{\gamma_0}$, and adopted the Euclidean distance as the accuracy metric of the output state. This necessarily implies that our output state has a large overlap/fidelity with the true ground state in the language of~\cite{Lin2020nearoptimalground}, since the inequality $\abs{\langle\psi_0|\varphi_2\rangle}\geq\abs{\langle\psi_0|\varphi_1\rangle}-\norm{\ket{\varphi_1}-\ket{\varphi_2}}$ holds for arbitrary quantum states $\ket{\psi_0},\ket{\varphi_1},\ket{\varphi_2}$.

    For a circuit implementation of the shifted Chebyshev coefficients preparation, see \thm{fourier_coeff} and the succeeding remark. To this end, we need to implement the oracle for the error function $\mathbf{Erf}$, which can be achieved by translating the efficient classical algorithm from~\cite{CHEVILLARD201272}. Moreover, we have $\norm{f}_{\max,[-1,1]}=\mathbf{O}(1)$ and $\norm{f(\cos)\sin}_{2,[-\pi,\pi]}=\mathbf{\Omega}(1)$ for the rescaled error function, assuming the spectral gap is at most constant. Therefore, the gate complexity for preparing the shifted Chebyshev coefficients is polylogarithmic in the input parameters.

    Finally, note that the query complexity of initial state preparation can be improved using the block preconditioning technique of~\cite{OptInit}.
\end{remark}

\section{Eigenvalue processing over the complex plane}
\label{sec:faber}
We have so far focused on quantum eigenvalue algorithms for input matrices with real spectra.
In this section, we show that many of these results can be carried over to the complex plane.
This is achieved using the Faber expansion that provides a nearly optimal polynomial basis for approximation over a compact set of the complex plane, the preliminaries of which will be reviewed in \sec{faber_prelim}.
We then describe efficient quantum algorithms for generating the Faber history state in \thm{faber_history} of \sec{faber_history}, and for transforming eigenvalues of input matrices with complex spectra in \thm{qevt_faber} of \sec{faber_qevt}.
Finally, we present in \sec{faber_app} a quantum algorithm for solving differential equations with general coefficient matrices (\thm{diff_eq_faber}), as well as a quantum algorithm for estimating leading eigenvalues (\thm{qeve_extreme}).

\subsection{Preliminaries on Faber expansion}
\label{sec:faber_prelim}

Suppose that we have an input matrix whose eigenvalues are enclosed by a subset $\mathcal{E}$ of the complex plane. When the matrix has only real spectra, we may choose $\mathcal{E}$ to be a closed real interval, and approximate functions over $\mathcal{E}$ using the Chebyshev expansion. But here, we relax this assumption to handle more general operators with complex eigenvalues. 

Specifically, we require $\mathcal{E}$ to be a nonempty simply connected compact set in the complex plane with a simply closed (or \emph{Jordan curve}) boundary, which we refer to as a \emph{Faber region}.
By the Riemann mapping theorem, there exists a unique function, known as the exterior Riemann map,
\begin{equation}
	\mathbf{\Phi}:\mathcal{E}^c\rightarrow\mathcal{D}^c,\qquad
	\mathbf{\Phi}(z)=w,
\end{equation}
which sends the complement of $\mathcal{E}$ conformally onto the exterior of the unit disk $\mathcal{D}=\{|w|\leq1\}$ and satisfies the conditions
\begin{equation}
	\mathbf{\Phi}(\infty)=\infty,\qquad
	\mathbf{\Phi}'(\infty)=\lim_{z\rightarrow\infty}\frac{\mathbf{\Phi}(z)}{z}=\zeta>0.
\end{equation}
Here, the complement is taken with respect to the extended complex plane $\mathbb{C}\cup\{\infty\}$.
This implies that $\mathbf{\Phi}$ has its Laurent expansion in some neighborhood of $\infty$ as
\begin{equation}
	\mathbf{\Phi}(z)=\zeta z+\zeta_0+\frac{\zeta_1}{z}+\frac{\zeta_2}{z^2}+\cdots
\end{equation}
Then the $n$th Faber polynomial $\mathbf{F}_n(z)$ for the domain $\mathcal{E}$ is taken to be the polynomial part of the Laurent expansion of $\mathbf{\Phi}^n(z)$.
We let the inverse of $\mathbf{\Phi}$ be
\begin{equation}
	\mathbf{\Psi}:\mathcal{D}^c\rightarrow\mathcal{E}^c,\qquad
	\mathbf{\Psi}(w)=\mathbf{\Phi}^{-1}(w)=z,
\end{equation}
which maps the exterior of the unit disk $\mathcal{D}$ conformally onto the complement of $\mathcal{E}$, with the Laurent expansion
\begin{equation}
	\mathbf{\Psi}(w)=\varsigma w+\varsigma_0+\frac{\varsigma_1}{w}+\frac{\varsigma_2}{w^2}+\cdots
\end{equation}
for $\abs{w}>1$. 
By the Carath\'{e}odory's theorem, the maps $\mathbf{\Phi}$ and $\mathbf{\Psi}$ can be extended continuously to the boundaries $\partial\mathcal{E}$ and $\partial\mathcal{D}$ respectively.
See \fig{faber_def} for an illustration of these definitions, and references~\cite{suetin1998series,markushevich2005theory} for more introductory material on Faber polynomials.

By restricting $\mathcal{E}$ to be special subsets of the complex plane, one can recover familiar examples of
polynomial basis for nearly best uniform approximation of functions. For instance, consider first the case where $\mathcal{E}=[-1,1]$. Then we have
\begin{equation}
    \mathbf{\Psi}(w)=\frac{1}{2}\left(w+\frac{1}{w}\right)
\end{equation}
as the \emph{Joukowsky map}, which has inverse $\mathbf{\Phi}(z)=z+\sqrt{z^2-1}$ with the branch of square root satisfying $\lim_{z\rightarrow\infty}\frac{\sqrt{z^2-1}}{z}=1$. This means $z-\sqrt{z^2-1}$ has no polynomial part in its Laurent expansion, so the polynomial part of $\mathbf{\Phi}^n(z)$ is the same as the polynomial part of $\left(z+\sqrt{z^2-1}\right)^n+\left(z-\sqrt{z^2-1}\right)^n$. But the $n$th Chebyshev polynomial of the first kind satisfies the equality $\mathbf{T}_n(x)=\frac{1}{2}\left(\left(x+\sqrt{x^2-1}\right)^n+\left(x-\sqrt{x^2-1}\right)^n\right)$ for $x\in\mathbb R$. We thus conclude that $\left(z+\sqrt{z^2-1}\right)^n+\left(z-\sqrt{z^2-1}\right)^n$ is a degree-$n$ polynomial itself, and that $\mathbf{F}_n(z)=2\mathbf{T}_n(z)$ for $n\geq1$, or 
\begin{equation}
    \mathbf{F}_n(z)=2\widetilde{\mathbf{T}}_n(z)
\end{equation}
for all nonnegative integers $n$.
As another example, consider the case where $\mathcal{E}$ is just the unit disk $\mathcal{D}=\{|w|\leq1\}$ itself. Then we have
\begin{equation}
    \mathbf{\Psi}(w)=w
\end{equation}
and the inverse $\mathbf{\Phi}(z)=z$ both as identity maps, so
\begin{equation}
    \mathbf{F}_n(z)=z^n
\end{equation}
are the power functions. But the significance of Faber polynomials is that they provide a unifying approach for function approximations over compact subsets of the complex plane, of which Chebyshev polynomials and power series are two special cases.

The generating functions for Faber polynomials and their derivatives have the form~\cite{CurtissPaper71,schiffer1948faber}
\begin{equation}
\label{eq:faber_gen}
    \sum_{j=0}^\infty\mathbf{F}_j(z)y^j=\frac{\mathbf{\Psi}'(y^{-1})}{y\left(\mathbf{\Psi}(y^{-1})-z\right)},\qquad
    \sum_{j=1}^\infty\frac{\mathbf{F}'_j(z)}{j}y^{j-1}=\frac{1}{y\left(\mathbf{\Psi}(y^{-1})-z\right)},
\end{equation}
for $\abs{y}<1$. This is reminiscent of the generating functions
\begin{equation}
    \sum_{j=0}^\infty \widetilde{\mathbf{T}}_j(x)y^j
    =\frac{1-y^2}{2(1+y^2-2yx)},\qquad
    \sum_{j=0}^\infty\mathbf{U}_j(x)y^j=\frac{1}{1+y^2-2yx},
\end{equation}
for Chebyshev polynomials of the first and second kind, where $\widetilde{\mathbf{T}}_j'(x)=\mathbf{T}_j'(x)=j\mathbf{U}_{j-1}(x)$ for $j\geq 1$.
In our Faber algorithms, we will implement a matrix version of the generating function $\sum_{j=0}^\infty L^j\otimes\mathbf{F}_j\left(\frac{A}{\alpha_A}\right)=\frac{\mathbf{\Psi}'(L^{-1})\otimes I}{L\mathbf{\Psi}(L^{-1})\otimes I-L\otimes \frac{A}{\alpha_A}}$. The challenge here is that we need to handle operators like $\mathbf{\Psi}'(L^{-1})$ and $L\mathbf{\Psi}(L^{-1})$ which need not have finite Taylor expansions, unlike the Chebyshev case. We overcome this difficulty by reformulating it as a problem of generating Fourier coefficients, which we solve using techniques we develop in~\sec{fourier}.

Given a function analytic over a Faber region in the complex plane, we may expand the function into a series of Faber polynomials associated with that region. This is formalized by the following lemma.
\begin{lemma}
    Let $\mathcal{E}$ be a Faber region with the corresponding conformal maps $\mathbf{\Phi}:\mathcal{E}^c\rightarrow\mathcal{D}^c$, $\mathbf{\Psi}:\mathcal{D}^c\rightarrow\mathcal{E}^c$ and Faber polynomials $\mathbf{F}_n(z)$. For any function $f:\mathbb{C}\rightarrow\mathbb{C}$ analytic on $\mathcal{E}$, the following statements hold:
    \begin{enumerate}
        \item (Existence~\cite[Page 52]{suetin1998series}): There exists an expansion
        \begin{equation}
            f(z)=\sum_{j=0}^\infty \beta_j\mathbf{F}_j(z)
        \end{equation}
        converging uniformly on the entire $\mathcal{E}$.
        \item (Uniqueness~\cite[Page 109]{suetin1998series}): For any expansion $f(z)=\sum_{j=0}^\infty \beta_j\mathbf{F}_j(z)$ converging uniformly on the entire $\mathcal{E}$,
        \begin{equation}
        \label{eq:faber_coeff_unit}
            \beta_j=\frac{1}{2\pi i}\int_{\partial\mathcal{D}}\mathrm{d}w\ \frac{f(\mathbf{\Psi}(w))}{w^{j+1}}
            =\frac{1}{2\pi}\int_{0}^{2\pi}\mathrm{d}\theta\ e^{-ij\theta}f\left(\mathbf{\Psi}(e^{i\theta})\right),
        \end{equation}
        where $\partial\mathcal{D}=\{\abs{w}=1\}$ is the unit circle.
    \end{enumerate}
    Moreover, if $f$ is analytic on a region containing $\mathbf{\Psi}(r\partial\mathcal{D})$ ($r\geq 1$) and its interior, then the Faber coefficients can also be computed with the rescaled contour:
\begin{equation}
    \beta_j=\frac{1}{2\pi i}\int_{r\partial\mathcal{D}}\mathrm{d}w\ \frac{f(\mathbf{\Psi}(w))}{w^{j+1}}
    =\frac{1}{2\pi r^j}\int_{0}^{2\pi}\mathrm{d}\theta\ e^{-ij\theta}f\left(\mathbf{\Psi}(re^{i\theta})\right).
\end{equation}
\end{lemma}

Finally, we consider the size of Faber polynomials $\mathbf{F}_j$ over the Faber region $\mathcal{E}$, which is useful for bounding the complexity of our quantum algorithms. Applying Cauchy's integral formula to the Faber generating function, we have
\begin{equation}
    \mathbf{F}_j(z)=\frac{1}{2\pi i}\int_{(1+\delta)\partial\mathcal{D}}\mathrm{d}w\ \frac{w^j\mathbf{\Psi}'(w)}{\mathbf{\Psi}(w)-z}
\end{equation}
for any $z\in\mathcal{E}$ and $\delta>0$, which implies $\norm{\mathbf{F}_j}_{\max,\mathcal{E}}
    \leq c_\delta(1+\delta)^j$ for constant $c_\delta=\frac{(1+\delta)\norm{\mathbf{\Psi}'}_{\max,(1+\delta)\partial\mathcal{D}}}
    {\mathbf{Dist}\left(\mathbf{\Psi}((1+\delta)\partial\mathcal{D}),\mathcal{E}\right)}$ independent of $j$,
and hence
\begin{equation}
    \limsup_{j\rightarrow\infty}\sqrt[j]{\norm{\mathbf{F}_j}_{\max,\mathcal{E}}}\leq1.
\end{equation}
This estimate holds for an arbitrary Faber region $\mathcal{E}$. But under additional assumptions of the region, it is possible to get a tighter estimate. For instance, if the boundary $\partial\mathcal{E}$ is of bounded \emph{total rotation}, we have the following integral representation of Faber polynomials
\begin{equation}
    \mathbf{F}_j(\mathbf{\Psi}(e^{i\theta}))
    =\frac{1}{\pi}\int_{0}^{2\pi}e^{ij\varphi}\mathrm{d}_\varphi v(\varphi,\theta),\quad
    j\geq1,
\end{equation}
where $v(\varphi,\theta)=\mathbf{Arg}\left(\mathbf{\Psi}(e^{i\varphi})-\mathbf{\Psi}(e^{i\theta})\right)$ is an angular function with the jump at $\varphi=\theta$ equal to the exterior angle of $\partial\mathcal{E}$ at $\theta$. 
From this, we obtain the following bound on the maximum size of Faber polynomials~\cite[Page 182]{suetin1998series}~\cite{Ellacott1983}
\begin{equation}
    \norm{\mathbf{F}_j}_{\max,\mathcal{E}}
    =\norm{\mathbf{F}_j}_{\max,\partial\mathcal{E}}
    =\norm{\mathbf{F}_j\left(\mathbf{\Psi}(e^{i(\cdot)})\right)}_{\max,[0,2\pi]}
    \leq\max_{\theta\in[0,2\pi]}\frac{1}{\pi}\int_{0}^{2\pi}\abs{\mathrm{d}_\varphi v(\varphi,\theta)}
    \leq\frac{\mathbf{V}(\partial\mathcal{E})}{\pi},
\end{equation}
where the first equality follows from the maximum modulus principle, and the last inequality follows from a bound due to Radon. Here, $\mathbf{V}(\partial\mathcal{E})$ is the variation or \emph{total rotation} of the curve $\partial\mathcal{E}$. It is known that $\mathbf{V}(\partial\mathcal{E})\leq4\pi$ if $\mathcal{E}$ is simply connected, and that $\mathbf{V}(\partial\mathcal{E})\geq2\pi$ always holds, with the equality $\mathbf{V}(\partial\mathcal{E})=2\pi$ if and only if $\mathcal{E}$ is convex~\cite[Page 147]{elliott1978construction}. Thus, we have
\begin{equation}
    \norm{\mathbf{F}_j}_{\max,\mathcal{E}}
    \leq2
\end{equation}
for a convex Faber region, which is reminiscent of the familiar bound
\begin{equation}
    \norm{\mathbf{T}_j}_{\max,[-1,1]}
    \leq1
\end{equation}
for Chebyshev polynomials.

\subsection{Faber history state generation}
\label{sec:faber_history}

Suppose that the input matrix is block encoded as $A/\alpha_A$ with some normalization factor $\alpha_A\geq\norm{A}$, whose eigenvalues are enclosed by a Faber region $\mathcal{E}$. Denote the corresponding conformal maps as $\mathbf{\Phi}:\mathcal{E}^c\rightarrow\mathcal{D}^c$, $\mathbf{\Psi}:\mathcal{D}^c\rightarrow\mathcal{E}^c$ and the Faber polynomials as $\mathbf{F}_n(z)$.
As aforementioned, the main idea behind our approach is to use a matrix Faber generating function of the form
\begin{equation}
    \sum_{j=0}^{n-1} L^j\otimes\mathbf{F}_j\left(\frac{A}{\alpha_A}\right)
    =\sum_{j=0}^\infty L^j\otimes\mathbf{F}_j\left(\frac{A}{\alpha_A}\right)=\frac{\mathbf{\Psi}'(L^{-1})\otimes I}{L\mathbf{\Psi}(L^{-1})\otimes I-L\otimes \frac{A}{\alpha_A}}.
\end{equation}
This follows from \eq{faber_gen} by substituting $z=I\otimes\frac{A}{\alpha_A}$ and $y=L\otimes I$. This substitution is mathematically valid because $L$ has zero eigenvalues only, whereas both sides of \eq{faber_gen} have the same derivatives at $y=0$ of any order. See~\cite[Chapter 6]{roger1994topics} or~\cite[Chapter 2]{henrici1974applied} for a complete mathematical justification. Note that the Laurent series of $\mathbf{\Psi}'(w)$ and $w\mathbf{\Psi}(w^{-1})$ only contains terms with nonnegative exponents, so operator functions $\mathbf{\Psi}'(L^{-1})$ and $L\mathbf{\Psi}(L^{-1})$ are well defined even though the lower shift matrix $L$ is not invertible per se.

Now consider the problem of eigenvalue processing. Toward implementing a truncated Faber expansion of the form $\sum_{k=0}^{n-1}\beta_k\mathbf{F}_k$, we apply the matrix generating function to the initial state
\begin{equation}
	\frac{\sum_{k=0}^{n-1}{\beta}_k\ket{n-1-k}}{\norm{ \beta}}\ket{\psi}.
\end{equation}
Similar to the Chebyshev case, we obtain up to a normalization factor
\begin{equation}
	\begin{aligned}
		\left(\sum_{j=0}^{n-1}L^j\otimes {\mathbf{F}}_j\left(\frac{A}{\alpha_A}\right)\right)
		\left(\sum_{k=0}^{n-1}{\beta}_k\ket{n-1-k}\ket{\psi}\right)
		&=\sum_{l=0}^{n-1}\ket{l}
		\sum_{k=n-1-l}^{n-1}{\beta}_k{\mathbf{F}}_{k+l-n+1}\left(\frac{A}{\alpha_A}\right)\ket{\psi}.
	\end{aligned}
\end{equation}
If we now measure the first register and get the outcome $l=n-1$, the second register will have the desired state proportional to
\begin{equation}
	\sum_{k=0}^{n-1}{\beta}_k\mathbf{F}_{k}\left(\frac{A}{\alpha_A}\right)\ket{\psi}.
\end{equation}
However, we will also get unwanted components for $l=0,\ldots,n-2$, leading to a failure of the algorithm.

To boost the success probability, we use the runaway padding trick to repeat the desired state $\eta n$ times. This can again be understood via the formula in \lem{lower_block_inv} for inverting lower block matrices. Specifically, we let
\begin{equation}
	A_{11}=L_n\mathbf{\Psi}(L_n^{-1})\otimes I-L_n\otimes \frac{A}{\alpha_A},
\end{equation}
which corresponds to the denominator of the generating function. Here, we have used subscripts to explicitly represent dimensions of the matrices.
Now we take
\begin{equation}
	A_{21}=\ketbra{0}{n-1}\otimes(-I)=
	\begin{bmatrix}
		0 & 0 & \cdots & -I\\
		0 & 0 & \cdots & 0\\
		\vdots & \vdots & \vdots & \vdots\\
		0 & 0 & \cdots & 0\\
	\end{bmatrix}.
\end{equation}
and set
\begin{equation}
	A_{22}=\left(I_{\eta n}-L_{\eta n}\right)\otimes I=
	\begin{bmatrix}
		I & 0 & 0 & \cdots & 0\\
		-I & I & 0 & \cdots & 0\\
		0 & -I & I & \ddots & 0\\
		\vdots & \ddots & \ddots & \ddots & \vdots\\
		0 & \cdots & \ddots & -I & I\\
	\end{bmatrix}\
	\Rightarrow\
	A_{22}^{-1}=\begin{bmatrix}
		I & 0 & 0 & \cdots & 0\\
		I & I & 0 & \cdots & 0\\
		I & I & I & \ddots & \vdots\\
		\vdots & \ddots & \ddots & \ddots & \vdots\\
		I & \cdots & \ddots & \ddots & I
	\end{bmatrix}.
\end{equation}
We will bundle the numerator of the generating function
\begin{equation}
	B_{11}=\mathbf{\Psi}'(L_n^{-1})\otimes I
\end{equation}
with the state preparation subroutine and handle it later.

To summarize, after the padding, denominator of the matrix Faber generating function becomes
\begin{equation}
\label{eq:pad_a_faber}
\NiceMatrixOptions
{
    custom-line = 
    {
        letter = I , 
        command = hdashedline , 
        tikz = {dashed,dash phase=3pt} ,
        width = \pgflinewidth
    }
}
	\begin{aligned}
		\mathbf{Pad}(A)&=\ketbra{0}{0}\otimes A_{11}+\ketbra{1}{0}\otimes A_{21}+\ketbra{1}{1}\otimes A_{22}\\
		&=\ketbra{0}{0}\otimes\left(L_n\mathbf{\Psi}(L_n^{-1})\otimes I-L_n\otimes \frac{A}{\alpha_A}\right)\\
		&\quad+\ketbra{1}{0}\otimes\ketbra{0}{n-1}\otimes (-I)
		+\ketbra{1}{1}\otimes(I_{\eta n}-L_{\eta n})\otimes I\\
        &=
        \begin{bNiceArray}{cccccIccccc}
            \Block{5-5}<\large>{L_n\mathbf{\Psi}(L_n^{-1})\otimes I\\-L_n\otimes \frac{A}{\alpha_A}} 
            & & & & & 0 & 0 & 0 & \cdots & 0\\
            & & & & & 0 & 0 & 0 & \cdots & 0\\
            & & & & & \vdots & \vdots & \vdots & \vdots & \vdots\\
            & & & & & 0 & 0 & 0 & \cdots & 0\\
            & & & & & 0 & 0 & 0 & \cdots & 0\\
            \hdashedline
            0 & 0 & \cdots & 0 & -I & I & 0 & 0 & \cdots & 0\\
            0 & 0 & \cdots & 0 & 0 & -I & I & 0 & \cdots & 0\\
            0 & 0 & \cdots & 0 & 0 & 0 & -I & I & \ddots & \vdots\\
            \vdots & \vdots & \vdots & \vdots & \vdots & \vdots & \ddots & \ddots & \ddots & 0\\
            0 & 0 & \cdots & 0 & 0 & 0 & \cdots & 0 & -I & I
        \end{bNiceArray}.\\
	\end{aligned}
\end{equation}
The numerator
\begin{equation}
\begin{aligned}
    \mathbf{Pad}(B)&=\ketbra{0}{0}\otimes B_{11}+\ketbra{1}{1}\otimes I_{\eta n}\otimes I\\
    &=\ketbra{0}{0}\otimes\mathbf{\Psi}'(L_n^{-1})\otimes I
	+\ketbra{1}{1}\otimes I_{\eta n}\otimes I
\end{aligned}
\end{equation}
will be bundled with the state preparation which is now augmented with an additional ancilla state
\begin{equation}
	\ket{0}\frac{\sum_{k=0}^{n-1}{\beta}_k\ket{n-1-k}}{\norm{ \beta}}\ket{\psi}.
\end{equation}

Using the matrix generating function
\begin{equation}
        \sum_{j=1}^n L^{j-1}\otimes\frac{\mathbf{F}'_j\left(\frac{A}{\alpha_A}\right)}{j}
	=\sum_{j=1}^\infty L^{j-1}\otimes\frac{\mathbf{F}'_j\left(\frac{A}{\alpha_A}\right)}{j}
	=\frac{1}{L\mathbf{\Psi}(L^{-1})\otimes I-L\otimes\frac{A}{\alpha_A}}
\end{equation}
for the derivative of Faber polynomials, we have
\begin{equation}
\label{eq:pad_a_inv_faber}
\NiceMatrixOptions
{
    custom-line = 
    {
        letter = I , 
        command = hdashedline , 
        tikz = {dashed,dash phase=3pt} ,
        width = \pgflinewidth
    }
}
    \mathbf{Pad}(A)^{-1}=
    \begin{bNiceArray}{ccccIccc}
        \mathbf{F}_{1}'\left(\frac{A}{\alpha_A}\right) & 0 & \cdots & 0 & 0 & \cdots & \cdots \\
        \frac{\mathbf{F}_{2}'\left(\frac{A}{\alpha_A}\right)}{2} & \mathbf{F}_{1}'\left(\frac{A}{\alpha_A}\right) & \ddots & \vdots & \vdots & \vdots & \vdots \\
        \vdots & \ddots & \ddots & 0 & \vdots & \vdots & \vdots \\
        \frac{\mathbf{F}_{n}'\left(\frac{A}{\alpha_A}\right)}{n} & \cdots & \frac{\mathbf{F}_{2}'\left(\frac{A}{\alpha_A}\right)}{2} & \mathbf{F}_{1}'\left(\frac{A}{\alpha_A}\right) & 0 & \cdots & \cdots \\
        \hdashedline
        \frac{\mathbf{F}_{n}'\left(\frac{A}{\alpha_A}\right)}{n} & \cdots & \frac{\mathbf{F}_{2}'\left(\frac{A}{\alpha_A}\right)}{2} & \mathbf{F}_{1}'\left(\frac{A}{\alpha_A}\right) & I & 0 & \cdots \\
        \frac{\mathbf{F}_{n}'\left(\frac{A}{\alpha_A}\right)}{n} & \cdots & \frac{\mathbf{F}_{2}'\left(\frac{A}{\alpha_A}\right)}{2} & \mathbf{F}_{1}'\left(\frac{A}{\alpha_A}\right) & I & I & \ddots \\
        \vdots & \vdots & \vdots & \vdots & \vdots & \vdots & \ddots 
    \end{bNiceArray}.
\end{equation}
This implies
\begin{equation}
\NiceMatrixOptions
{
    custom-line = 
    {
        letter = I , 
        command = hdashedline , 
        tikz = {dashed,dash phase=3pt} ,
        width = \pgflinewidth
    }
}
    \mathbf{Pad}(A)^{-1}\mathbf{Pad}(B)=
    \begin{bNiceArray}{ccccIccc}
        \mathbf{F}_{0}\left(\frac{A}{\alpha_A}\right) & 0 & \cdots & 0 & 0 & \cdots & \cdots \\
        \mathbf{F}_{1}\left(\frac{A}{\alpha_A}\right) & \mathbf{F}_{0}\left(\frac{A}{\alpha_A}\right) & \ddots & \vdots & \vdots & \vdots & \vdots \\
        \vdots & \ddots & \ddots & 0 & \vdots & \vdots & \vdots \\
        \mathbf{F}_{n-1}\left(\frac{A}{\alpha_A}\right) & \cdots & \mathbf{F}_{1}\left(\frac{A}{\alpha_A}\right) & \mathbf{F}_{0}\left(\frac{A}{\alpha_A}\right) & 0 & \cdots & \cdots \\
        \hdashedline
        \mathbf{F}_{n-1}\left(\frac{A}{\alpha_A}\right) & \cdots & \mathbf{F}_{1}\left(\frac{A}{\alpha_A}\right) & \mathbf{F}_{0}\left(\frac{A}{\alpha_A}\right) & I & 0 & \cdots \\
        \mathbf{F}_{n-1}\left(\frac{A}{\alpha_A}\right) & \cdots & \mathbf{F}_{1}\left(\frac{A}{\alpha_A}\right) & \mathbf{F}_{0}\left(\frac{A}{\alpha_A}\right) & I & I & \ddots \\
        \vdots & \vdots & \vdots & \vdots & \vdots & \vdots & \ddots 
    \end{bNiceArray}
\end{equation}
as desired.

Let us now consider circuit implementation of block encoding \eq{pad_a_faber}. To this end, we bundle $-L_n\otimes \frac{A}{\alpha_A}$ with the two blocks at the bottom to get
\begin{equation}
\begin{aligned}
    &\ \ketbra{0}{0}\otimes L_n\otimes\left(-\frac{A}{\alpha_A}\right)
    +\ketbra{1}{0}\otimes\ketbra{0}{n-1}\otimes (-I)
    +\ketbra{1}{1}\otimes(I_{\eta n}-L_{\eta n})\otimes I\\
    =&\ -\sum_{k=1}^{n-1}\ketbra{k}{k-1}\otimes\frac{A}{\alpha_A}
    -\sum_{k=n}^{n\eta+n-1}\ketbra{k}{k-1}\otimes I
    +\sum_{k=n}^{n\eta+n-1}\ketbra{k}{k}\otimes I.
\end{aligned}
\end{equation}
Here, the first two terms can be rewritten as
\begin{equation}
    -\left(\sum_{k=0}^{n-1}\ketbra{k}{k}\otimes\frac{A}{\alpha_A}
    +\sum_{k=n}^{n\eta+n-1}\ketbra{k}{k}\otimes I
    \right)
    \left(\sum_{k=1}^{n\eta+n-1}\ketbra{k}{k-1}\otimes I\right)
\end{equation}
and can thus be block encoded with normalization factor $1$, using \lem{shift}. The third term is reexpressed as
\begin{equation}
    \sum_{k=0}^{n\eta+n-1}\frac{1+(-1)^{\mathbf{Ind}_{[0,n-1]}(k)}}{2}\ketbra{k}{k}\otimes I,
\end{equation}
with the indicator function $\mathbf{Ind}_{[0,n-1]}(k)=1$ if and only if $0\leq k\leq n-1$, and can also be block encoded with normalization factor $1$.

We already know from \sec{fourier} that $L_n\mathbf{\Psi}(L_n^{-1})$ can be block encoded as
\begin{equation}
\label{eq:psi_block_preamp}
    \frac{L_n\mathbf{\Psi}(L_n^{-1})}{\left(\frac{\pi}{2}+\ln(n)\right)\alpha_{{\scriptscriptstyle \mathbf{\Psi}},\max}}
\end{equation}
for any upper bound on maximum value of the conformal map over the unit circle
\begin{equation}
    \alpha_{{\scriptscriptstyle \mathbf{\Psi}},\max}\geq\norm{(\cdot)^{-1}\mathbf{\Psi}(\cdot)}_{\max,\partial\mathcal{D}},
\end{equation}
with a normalization factor scaling like $\sim\log(n)\norm{(\cdot)^{-1}\mathbf{\Psi}(\cdot)}_{\max,\partial\mathcal{D}}$.
To prevent this logarithmic overhead from ruining the query complexity, we perform a uniform amplification of the block encoding as in \lem{amp_block}. This gives
\begin{equation}
\label{eq:psi_block_amp}
    \frac{L_n\mathbf{\Psi}(L_n^{-1})}{2\alphaPsiFaber}
\end{equation}
for any upper bound on the spectral norm
\begin{equation}
\label{eq:alpha_psi_faber}
    \alphaPsiFaber\geq\norm{L_n\mathbf{\Psi}(L_n^{-1})}.
\end{equation}
By taking a final linear combination, we can thus block encode
\begin{equation}
\label{eq:psi_block_lcu}
    \frac{\mathbf{Pad}(A)}{2\alphaPsiFaber+2}.
\end{equation}

    We see that the dependence on $\alphaPsiFaber$ comes from the necessity of block encoding $L_n\mathbf{\Psi}(L_n^{-1})$ in implementing the matrix Faber generating function. It is obvious that $\alphaPsiFaber=\mathbf{O}(1)$ when $\mathbf{\Psi}$ has a finite Laurent expansion. This happens for instance in the generation of Chebyshev history states, where $\mathbf{\Psi}(w)=\frac{1}{2}\left(w+\frac{1}{w}\right)$ is the Joukowsky map containing only two terms. Even when $\mathbf{\Psi}$ has an infinite Laurent expansion, $\alphaPsiFaber$ can still be constant. For instance, if $\partial\mathcal{E}=\mathbf{\Psi}(\partial\mathcal{D})$ is an analytic curve, then by the generalized Schwarz reflection principle, $\mathbf{\Psi}$ can be analytically continued across the unit circle $\partial\mathcal{D}$~\cite[Page 299]{lang2013complex}. Therefore, the Laurent series of $\mathbf{\Psi}$ converges absolutely on $\partial\mathcal{D}$, implying a constant value of $\alphaPsiFaber$ as well. 
    In fact, the scaling $\alphaPsiFaber=\mathbf{O}(1)$ holds in general. This can be proved using the Crouzeix-Palencia theorem (\lem{crouzeix_palencia} of \append{analysis_faber}), together with the fact that the numerical range of $L_n$ is the disk $\cos\left(\frac{\pi}{n+1}\right)\mathcal{D}$~\cite{VANDANJAV201476}, whereas the function $(\cdot)\mathbf{\Psi}((\cdot)^{-1})$ is analytic on $\{\abs{w}<1\}$ and can be extended continuously to the entire $\mathcal{D}$:
    \begin{equation}
    \begin{aligned}
        \norm{L_n\mathbf{\Psi}(L_n^{-1})}
        &\lesssim\norm{(\cdot)\mathbf{\Psi}((\cdot)^{-1})}_{\max,\mathcal{W}(L_n)}
        =\norm{(\cdot)\mathbf{\Psi}((\cdot)^{-1})}_{\max,\cos\left(\frac{\pi}{n+1}\right)\mathcal{D}}\\
        &\leq\norm{(\cdot)\mathbf{\Psi}((\cdot)^{-1})}_{\max,\mathcal{D}}=\mathbf{O}(1).
    \end{aligned}
    \end{equation}
    Note that constant prefactor of the above analysis can be tightened using von Neumann's inequality~\cite[Proposition 1]{RoystonThesis} and~\cite[1.6.P26]{roger1994topics}.
    It is for this reason that the dependence on $\alphaPsiFaber$ will be dropped in the asymptotic analysis of Faber-based algorithms.

\begin{theorem}[Faber history state generation]
\label{thm:faber_history}
Let $A$ be a square matrix such that $A/\alpha_A$ is block encoded by $O_A$ with some normalization factor $\alpha_A\geq\norm{A}$. Suppose that eigenvalues of $A/\alpha_A$ are enclosed by a Faber region $\mathcal{E}$ with associated conformal maps $\mathbf{\Phi}:\mathcal{E}^c\rightarrow\mathcal{D}^c$, $\mathbf{\Psi}:\mathcal{D}^c\rightarrow\mathcal{E}^c$ and Faber polynomials $\mathbf{F}_n(z)$.
Let $O_\psi\ket{0}=\ket{\psi}$ be the oracle preparing the initial state, and $O_{\beta}\ket{0}=\mathbf{\Psi}'(L_n^{-1})\sum_{k=0}^{n-1}\beta_k\ket{n-1-k}/\norm{\mathbf{\Psi}'(L_n^{-1})\sum_{k=0}^{n-1}\beta_k\ket{n-1-k}}$ be the oracle preparing the shifting of coefficients $\beta$.
Then, the quantum state
	\begin{equation}
		\frac{\ket{0}\sum_{l=0}^{n-1}\ket{l}
			\sum_{k=n-1-l}^{n-1}{\beta}_k{\mathbf{F}}_{k+l-n+1}\left(\frac{A}{\alpha_A}\right)\ket{\psi}
			+\sum_{s=1}^{\eta}\ket{s}\sum_{l=0}^{n-1}\ket{l}
			\sum_{k=0}^{n-1}{\beta}_k{\mathbf{F}}_{k}\left(\frac{A}{\alpha_A}\right)\ket{\psi}}
        {\sqrt{\sum_{l=0}^{n-1}\norm{\sum_{k=n-1-l}^{n-1}\beta_k\mathbf{F}_{k+l-n+1}\left(\frac{A}{\alpha_A}\right)\ket{\psi}}^2
        +\eta n\norm{\sum_{k=0}^{n-1}\beta_k\mathbf{F}_{k}\left(\frac{A}{\alpha_A}\right)\ket{\psi}}^2}}
	\end{equation}
can be prepared with accuracy $\epsilon$ and probability $1-\pfail$ using
\begin{equation}
    \mathbf{O}\left(\alphaFP n(\eta+1)\log\left(\frac{1}{\epsilon}\right)\log\left(\frac{1}{\pfail}\right)\right)
\end{equation}
queries to controlled-$O_A$, controlled-$O_\psi$, controlled $O_{\widetilde\beta}$, and their inverses, where
\begin{equation}
\label{eq:alphaFP_alphaPsiFaber}
    \alphaFP\geq\max_{j=1,\ldots,n}\norm{\frac{\mathbf{F}_j'\left(\frac{A}{\alpha_A}\right)}{j}}
\end{equation}
is an upper bound on the derivative of Faber polynomials.
\end{theorem}
\begin{proof}
    The analysis proceeding the above theorem shows that one can block encode $\mathbf{Pad}(A)/(2\alphaPsiFaber+2)$ with an arbitrary precision using only $1$ query to $O_A$.
    The quantum linear system algorithm of \lem{opt_lin} then outputs a state $\epsilon$-close to
    \begin{equation}
        \frac{\frac{2\alphaPsiFaber+2}{\mathbf{Pad}(A)}
        \left(\ket{0}\frac{\mathbf{\Psi}'(L_n^{-1})\sum_{k=0}^{n-1}\beta_k\ket{n-1-k}}{\norm{\mathbf{\Psi}'(L_n^{-1})\sum_{k=0}^{n-1}\beta_k\ket{n-1-k}}}\ket{\psi}\right)}
        {\norm{\frac{2\alphaPsiFaber+2}{\mathbf{Pad}(A)}\left(\ket{0}\frac{\mathbf{\Psi}'(L_n^{-1})\sum_{k=0}^{n-1}\beta_k\ket{n-1-k}}{\norm{\mathbf{\Psi}'(L_n^{-1})\sum_{k=0}^{n-1}\beta_k\ket{n-1-k}}}\ket{\psi}\right)}}
        =\frac{\mathbf{Pad}(A)^{-1}\mathbf{Pad}(B)\left(\ket{0}\sum_{k=0}^{n-1}\beta_k\ket{n-1-k}\ket{\psi}\right)}
        {\norm{\mathbf{Pad}(A)^{-1}\mathbf{Pad}(B)\left(\ket{0}\sum_{k=0}^{n-1}\beta_k\ket{n-1-k}\ket{\psi}\right)}}.
    \end{equation}
    This is exactly the padded Faber history state.

    We have the norm bound on the inverse padded matrix
    \begin{equation}
        \norm{\mathbf{Pad}(A)^{-1}}=\mathbf{O}\left((\eta+1)n\alphaFP\right)
    \end{equation}
    for any upper bound $\alphaFP\geq\max_{j=1,\ldots,n}\norm{\frac{\mathbf{F}_j'\left(\frac{A}{\alpha_A}\right)}{j}}$ on the derivative of Faber polynomials (normalized by the degrees), which follows from \lem{block_norm} and the matrix representation of $\mathbf{Pad}(A)^{-1}$ in \eq{pad_a_inv_faber}.
    The claimed complexity now follows from \eq{opt_lin_cost}.
\end{proof}
\begin{remark}
    For the purpose of generality, we have expressed the complexity of our algorithm in terms of $\alphaFP$. This analysis can be further refined when the algorithm is applied in a concrete setting. For instance, if the Faber region $\mathcal{E}$ encloses the numerical range $\mathcal{W}(A/\alpha_A)$ or the pseudospectrum $\mathcal{S}_\delta(A/\alpha_A)$ with a sufficiently smooth boundary $\partial\mathcal{E}$, then it holds $\alphaFP=\mathbf{O}(1)$. This extends previous analysis of quantum differential equation algorithms~\cite{Krovi2023improvedquantum} based on the notion of numerical abscissa, which avoids the issue with a ill-conditioned Jordan basis that could potentially arise in the Chebyshev case. See \append{analysis_faber_crouzeix} and \append{analysis_faber_pseudospectrum} for more details.

    In the circuit implementation of this algorithm, we need to choose the precision with which \eq{psi_block_preamp} is block encoded. This can be determined as follows. We first invoke~\cor{lin_sys_perturb} to see that the padded Faber history state has accuracy $\sim\epsilon$, as long as \eq{psi_block_lcu} is block encoded with accuracy $\sim\frac{\epsilon}{(\eta+1)n\alphaFP\alphaPsiFaber}$. This is satisfied as long as we block encode \eq{psi_block_amp} using \lem{amp_block} with the same asymptotic accuracy. Up to a polylogarithmic factor, we can then block encode \eq{psi_block_preamp} with accuracy $\sim\frac{\epsilon}{(\eta+1)n\alphaFP\alpha_{{\scriptscriptstyle \mathbf{\Psi}},\max}}$. However, the specific choice of block encoding accuracy does not change the query complexity of the algorithm, and its contribution to the gate complexity is only polylogarithmically as per \thm{fourier_coeff}.

    Finally, note that the query complexity of initial state preparation can be improved using the block preconditioning technique of~\cite{OptInit}.
\end{remark}

\subsection{Quantum eigenvalue transformation, Faber version}
\label{sec:faber_qevt}

Now that we have an efficient quantum algorithm for preparing the Faber history state, we let the padding parameter $\eta=1$ and further perform a (fixed-point) amplitude amplification to solve the eigenvalue transformation problem. This is formally stated as follows.

\begin{theorem}[Quantum eigenvalue transformation, Faber version]
\label{thm:qevt_faber}
Let $A$ be a square matrix such that $A/\alpha_A$ is block encoded by $O_A$ with some normalization factor $\alpha_A\geq\norm{A}$. Suppose that eigenvalues of $A/\alpha_A$ are enclosed by a Faber region $\mathcal{E}$ with associated conformal maps $\mathbf{\Phi}:\mathcal{E}^c\rightarrow\mathcal{D}^c$, $\mathbf{\Psi}:\mathcal{D}^c\rightarrow\mathcal{E}^c$ and Faber polynomials $\mathbf{F}_n(z)$.
Let $p(z)=\sum_{k=0}^{n-1}\beta_k{\mathbf{F}}_{k}(z)$ be the Faber expansion of a degree-($n-1$) polynomial $p$.
Let $O_\psi\ket{0}=\ket{\psi}$ be the oracle preparing the initial state, and $O_{\beta}\ket{0}=\mathbf{\Psi}'(L_n^{-1})\sum_{k=0}^{n-1}\beta_k\ket{n-1-k}/\norm{\mathbf{\Psi}'(L_n^{-1})\sum_{k=0}^{n-1}\beta_k\ket{n-1-k}}$ be the oracle preparing the shifting of coefficients $\beta$.
Then, the quantum state
    \begin{equation}
        \frac{p\left(\frac{A}{\alpha_A}\right)\ket{\psi}}{\norm{p\left(\frac{A}{\alpha_A}\right)\ket{\psi}}}
    \end{equation}
can be prepared with accuracy $\epsilon$ and probability $1-\pfail$ using
\begin{equation}
    \mathbf{O}\left(\frac{\alphaF}{\alphaPPsi}\alphaFP n\log\left(\frac{\alphaF}{\alphaPPsi\epsilon}\right)\log\left(\frac{1}{\pfail}\right)\right)
\end{equation}
queries to controlled-$O_A$, controlled-$O_\psi$, controlled $O_{\widetilde\beta}$, and their inverses, where $\alphaFP$ satisfies \eq{alphaFP_alphaPsiFaber} and
\begin{equation}
\label{eq:alphaF_alphaPPsi}
    \alphaF\geq\max_{l=0,1,\ldots,n-1}\norm{\sum_{k=l}^{n-1}\beta_k\mathbf{F}_{k-l}\left(\frac{A}{\alpha_A}\right)\ket{\psi}},\qquad
    \alphaPPsi\leq\norm{p\left(\frac{A}{\alpha_A}\right)\ket{\psi}}.
\end{equation}
are upper bound on the shifted Faber partial sum and lower bound on the transformed state respectively.
\end{theorem}
\begin{proof}
    This follows from a similar reasoning as that for \thm{qevt}.
\end{proof}
\begin{remark}
    For a discussion on the scaling of $\alphaFP$, see the remark succeeding \thm{faber_history}. The parameter $\alphaF$ denotes maximum size of the shifted Faber expansion, which can be further upper bounded in terms of the max-norm of the polynomial $p$. For instance, if the numerical range $\mathcal{W}(A/\alpha_A)\subseteq\mathcal{E}$ or the pseudospectrum $\mathcal{S}_\delta(A/\alpha_A)\subseteq\mathcal{E}$ is contained in the Faber region, then we prove in \append{analysis_faber_crouzeix} and \append{analysis_faber_pseudospectrum} that $\alphaF=\mathbf{O}\left(\log(n)\norm{p}_{\max,\partial\mathcal{E}}\right)$. Furthermore, we can shave off the $\log(n)$ factor (albeit picking up the Jordan condition number $\kappa_S$) when running the algorithm on an average diagonalizable input matrix (\append{analysis_faber_carleson}), giving $\alphaF=\mathbf{O}\left(\kappa_S\norm{p}_{\max,\partial\mathcal{E}}\right)$.

    In the circuit implementation, we need to prepare a quantum state encoding the Faber expansion coefficients. This can be realized using \thm{fourier_coeff} in a similar way as in the Chebyshev case. This is because in computing the Faber coefficients \eq{faber_coeff_unit}, we can perform a contour integral along the unit circle, which can be reformulated as the generation of Fourier coefficients after a change of variable.

    Finally, note that the query complexity of initial state preparation can be improved using the block preconditioning technique of~\cite{OptInit}.
\end{remark}

\subsection{Applications}
\label{sec:faber_app}

We now apply the Faber eigenvalue processing algorithm to solve linear differential equations and estimate leading eigenvalue of the input matrix.

Consider a system of linear differential equations $\frac{\mathrm{d}}{\mathrm{d}t}x(t)=Ax(t)$, whose solution is given by $x(t)=e^{tA}x(0)$. We assume that a block encoding of $A/\alpha_A$ is given as input, and that the numerical range $\mathcal{W}(A/\alpha_A)\subseteq\mathcal{E}$ is enclosed by a Faber region. This consideration is a generalization of previous analysis of differential equation algorithms based on the notion of numerical abscissa~\cite{Krovi2023improvedquantum}, which avoids the stability issue that could arise in transforming the Jordan basis. Then, we write the Faber expansion of the matrix exponential $e^{tA}=\sum_{j=0}^{\infty}\beta_j\mathbf{F}_j\left(\frac{A}{\alpha_A}\right)$ and aim to choose a truncate order $n$ sufficiently large to achieve a target accuracy. We prove the following matrix Faber truncation bound by adapting previous results~\cite{Moret01,NOVATI2003201,Beckermann} to our setting.

\begin{lemma}[Faber truncation of matrix exponentials]
\label{lem:trunc_matrix_exp_faber}
Let $\widetilde A$ be a matrix, such that its numerical range $\mathcal{W}(\widetilde A)$ is enclosed by a Faber region $\mathcal{E}$ with associated conformal maps $\mathbf{\Phi}:\mathcal{E}^c\rightarrow\mathcal{D}^c$, $\mathbf{\Psi}:\mathcal{D}^c\rightarrow\mathcal{E}^c$ and Faber polynomials $\mathbf{F}_n(z)$.
Given $\tau>0$, let $e^{\tau z}=\sum_{j=0}^\infty\beta_j\mathbf{F}_j(z)$ be the Faber expansion of the complex exponential function $e^{\tau z}$.
Assume that $\mathcal{E}$ is convex and symmetric with respect to the real axis, lying on the left half of the complex plane ($\Re(\mathcal{E})\leq0$).
Then,
    \begin{equation}
        \norm{e^{\tau\widetilde A}-\sum_{j=0}^{n-1}\beta_j\mathbf{F}_j\left(\widetilde A\right)}
        =\mathbf{O}\left(\left(\frac{e^{\zeta}\tau}{n}\right)^n\right)
    \end{equation}
    for $n=\mathbf{\Omega}(\tau)$ sufficiently large, where $\zeta=\mathbf{\Psi}'(\infty)>0$ is the \emph{capacity} of the Faber region.
\end{lemma}
\begin{proof}
    We start with the estimate
    \begin{equation}
        \norm{e^{\tau\widetilde A}-\sum_{j=0}^{n-1}\beta_j\mathbf{F}_j\left(\widetilde A\right)}
        \lesssim\norm{e^{\tau(\cdot)}-\sum_{j=0}^{n-1}\beta_j\mathbf{F}_j}_{\max,\mathcal{E}},
    \end{equation}
    which follows from the Crouzeix-Palencia theorem~\lem{crouzeix_palencia} from \append{analysis_faber}. Using the convexity of $\mathcal{E}$ and evaluating the Faber coefficients with a circular contour of radius $r>1$,
    \begin{equation}
    \begin{aligned}
        \norm{e^{\tau(\cdot)}-\sum_{j=0}^{n-1}\beta_j\mathbf{F}_j}_{\max,\mathcal{E}}
        &=\norm{\sum_{j=n}^{\infty}\beta_j\mathbf{F}_j}_{\max,\mathcal{E}}
        \leq\sum_{j=n}^{\infty}\abs{\beta_j}\norm{\mathbf{F}_j}_{\max,\mathcal{E}}\\
        &\leq2\sum_{j=n}^{\infty}\abs{\frac{1}{2\pi r^j}\int_{0}^{2\pi}\mathrm{d}\theta\ e^{-ij\theta}\exp\left(\tau\mathbf{\Psi}(re^{i\theta})\right)}\\
        &\leq2\norm{e^{\tau(\cdot)}}_{\max,\mathbf{\Psi}\left(r\partial\mathcal{D}\right)}\sum_{j=n}^{\infty}\frac{1}{r^j}
        =2\norm{e^{\tau(\cdot)}}_{\max,\mathbf{\Psi}\left(r\partial\mathcal{D}\right)}\frac{\left(\frac{1}{r}\right)^n}{1-\frac{1}{r}}.
    \end{aligned}
    \end{equation}
    Assuming $n>\tau$, let us choose
\begin{equation}
    r=\frac{n}{\tau}>1,
\end{equation}
    so that
    \begin{equation}
        \norm{e^{\tau\widetilde A}-\sum_{j=0}^{n-1}\beta_j\mathbf{F}_j\left(\widetilde A\right)}
        \lesssim\norm{e^{\tau(\cdot)}}_{\max,\mathbf{\Psi}\left(r\partial\mathcal{D}\right)}
        \frac{n}{n-\tau}\left(\frac{\tau}{n}\right)^n.
    \end{equation}

    Because $\mathcal{E}$ is convex and symmetric relative to the real axis, we have 
    \begin{equation}
        \norm{e^{\tau(\cdot)}}_{\max,\mathbf{\Psi}\left(r\partial\mathcal{D}\right)}
        =e^{\tau\mathbf{\Psi}(r)}.
    \end{equation}
    Now, by \cite[Eqs.\ (4.4) and (4.5)]{Beckermann},
    \begin{equation}
        \abs{\mathbf{\Psi}(r)-\mathbf{\Psi}(1)-\zeta\left(r-1\right)}
        \leq\zeta\left(1-\frac{1}{r}\right)
        \quad\Rightarrow\quad
        \mathbf{\Psi}(r)\leq\mathbf{\Psi}(1)
        +\zeta\left(r-\frac{1}{r}\right)
        \leq\mathbf{\Psi}(1)+\zeta r,
    \end{equation}
    where $\zeta=\mathbf{\Psi}'(\infty)>0$ is the capacity.
    This implies
    \begin{equation}
        \norm{e^{\tau(\cdot)}}_{\max,\mathbf{\Psi}\left(r\partial\mathcal{D}\right)}
        \leq e^{\tau\mathbf{\Psi}(1)+\zeta n}.
    \end{equation}
    Combining with the estimate from the previous paragraph, we obtain
    \begin{equation}
        \norm{e^{\tau\widetilde A}-\sum_{j=0}^{n-1}\beta_j\mathbf{F}_j\left(\widetilde A\right)}
        \lesssim
        \begin{cases}
            \frac{n}{n-\tau}\left(\frac{e^{\zeta}\tau}{n}\right)^n,\quad&\mathbf{\Psi}(1)\leq0,\\
            \frac{n}{n-\tau}\left(\frac{e^{\zeta+\mathbf{\Psi}(1)}\tau}{n}\right)^n,\quad&\mathbf{\Psi}(1)>0.\\
        \end{cases}
    \end{equation}
    This establishes the claimed bound. Note that requirements on the convexity, symmetry, and nonpositiveness of $\Re(\mathcal{E})$ can be relaxed along similar lines of~\cite{NOVATI2003201,Beckermann}.
\end{proof}

\begin{theorem}[Quantum differential equation algorithm, Faber version]
\label{thm:diff_eq_faber}
Let $A$ be a square matrix, such that $A/\alpha_A$ is block encoded by $O_A$ with some normalization factor $\alpha_A\geq\norm{A}$.
Suppose that the numerical range $\mathcal{W}(A/\alpha_A)$ is enclosed by a Faber region $\mathcal{E}$, which is convex and symmetric with respect to the real axis, lying on the left half of the complex plane ($\Re(\mathcal{E})\leq0$), with associated conformal maps $\mathbf{\Phi}:\mathcal{E}^c\rightarrow\mathcal{D}^c$, $\mathbf{\Psi}:\mathcal{D}^c\rightarrow\mathcal{E}^c$ and Faber polynomials $\mathbf{F}_n(z)$.
Let $O_\psi\ket{0}=\ket{\psi}$ be the oracle preparing the initial state.

Then, applying \thm{qevt_faber} to the function $e^{\alpha_Atz}$ ($t>0$) truncated at order
    \begin{equation}
        n=\mathbf{O}\left(\alpha_At+\log\left(\frac{1}{\alphaFPsi\epsilon}\right)\right)
    \end{equation}
    produces the state
    \begin{equation}
        \frac{e^{tA}\ket{\psi}}{\norm{e^{tA}\ket{\psi}}}
    \end{equation}
    with accuracy $\epsilon$ and probability $1-\pfail$. The algorithm uses
\begin{equation}
    \mathbf{O}\left(\frac{\alphaF}{\alphaFPsi}\alphaFP 
    \left(\alpha_At+\log\left(\frac{1}{\alphaFPsi\epsilon}\right)\right)
    \log\left(\frac{\alphaF}{\alphaFPsi\epsilon}\right)\log\left(\frac{1}{\pfail}\right)\right)
\end{equation}
    queries to controlled-$O_A$, controlled-$O_\psi$, and their inverses,
    where $\alphaFP$ satisfies \eq{alphaFP_alphaPsiFaber}, $\alphaF$ satisfies \eq{alphaF_alphaPPsi} and
    \begin{equation}
        \alphaFPsi\leq\norm{e^{tA}\ket{\psi}}
    \end{equation}
    is a lower bound on size of the solution vector.
\end{theorem}
\begin{proof}
    This is proved in a similar way as \thm{diff_eq}.
\end{proof}
\begin{remark}
    See the remaks succeeding \thm{faber_history} and \thm{qevt_faber} for a discussion on the scaling of $\alphaFP$ and $\alphaF$. In particular, we have both $\alphaFP=\mathbf{O}(1)$ and $\alphaF=\mathbf{O}\left(\log(n)\norm{p}_{\max,\partial\mathcal{E}}\right)$ satisfied for a region of the form \fig{region_general2}, which is a smooth deformed version of the Elliott semidisk \fig{region_general}. See the discussion at the end of \append{analysis_faber_crouzeix} for more details.
    Finally, note that the query complexity of initial state preparation can be improved using techniques of~\cite{OptInit}.
\end{remark}

We now explain how our techniques can be applied to estimate leading eigenvalues of an input matrix. These leading eigenvalues play an important role in classical linear algebraic algorithms as they largely determine the behavior of algorithms involving matrix power iterations.

Specifically, given matrix $A$, let $\ket{\psi_\theta}$ be an eigenstate of $A$ with eigenvalue $\lambda_{\max}e^{i\theta}$, i.e., $A\ket{\psi_\theta}=\lambda_{\max}e^{i\theta}\ket{\psi_\theta}$. Here, $\lambda_{\max}>0$ is the largest absolute value of eigenvalues of $A$ which is known a priori, and our goal is to estimate the phase angle $\theta$.
The underlying idea for solving this problem is to generate the Faber history state corresponding to the disk $\lambda_{\max}\mathcal{D}$. In this case, Faber polynomials are given by power functions, and one can directly implement the power series
\begin{equation}
    \sum_{j=0}^{n-1}L_n^j\otimes\frac{A^j}{\lambda_{\max}^j}
    =\sum_{j=0}^{\infty}L_n^j\otimes\frac{A^j}{\lambda_{\max}^j}
    =\frac{1}{I_n\otimes I-L_n\otimes\frac{ A}{\lambda_{\max}}}
\end{equation}
without invoking the full power of the Faber mechanism.

Suppose that the input matrix is block encoded as $A/\alpha_A$ with some normalization factor $\alpha_A\geq\norm{A}\geq\lambda_{\max}$.
Then, using the fact that the $n$-by-$n$ lower shift matrix $L_n$ can be block encoded with normalization factor $1$ (\lem{shift}), we can take the tensor product and obtain a block encoding of $L_n\otimes\frac{ A}{\alpha_A}$. By taking a further linear combination using the ancilla state
\begin{equation}
    \frac{\sqrt{\frac{\alpha_A}{\lambda_{\max}}}\ket{0}+\ket{1}}{\sqrt{\frac{\alpha_A}{\lambda_{\max}}+1}},
\end{equation}
we can block encode $\frac{I_n\otimes I-L_n\otimes\frac{ A}{\lambda_{\max}}}{\frac{\alpha_A}{\lambda_{\max}}+1}$.

We now invert the block encoded matrix using the optimal scaling quantum linear system solver~\lem{opt_lin} with the initial state $\ket{0}\ket{\psi}$. 
We have 
\begin{equation}
\begin{aligned}
    \norm{\frac{\frac{\alpha_A}{\lambda_{\max}}+1}{I_n\otimes I-L_n\otimes\frac{ A}{\lambda_{\max}}}}
    &=\left(\frac{\alpha_A}{\lambda_{\max}}+1\right)\norm{\sum_{j=0}^{n-1}L_n^j\otimes\frac{A^j}{\lambda_{\max}^j}}
    \leq\left(\frac{\alpha_A}{\lambda_{\max}}+1\right)
    \sum_{j=0}^{n-1}\frac{\norm{A^j}}{\lambda_{\max}^j}\\
    &\leq n\kappa_S\left(\frac{\alpha_A}{\lambda_{\max}}+1\right),
\end{aligned}
\end{equation}
where $\kappa_S$ is the Jordan condition number of $A$. 
Thus, to generate an $\epsilon_{\text{lin}}$-approximation of the history state, the query complexity of the algorithm is asymptotically
\begin{equation}
    \mathbf{O}\left(\frac{\alpha_A}{\lambda_{\max}}n\kappa_S\log\left(\frac{1}{\epsilon_{\text{lin}}}\right)\right).
\end{equation}
Note that the query complexity of initial state preparation can be improved using the block preconditioning technique of~\cite{OptInit}.
Assuming the input state $\ket{\psi}=\ket{\psi_\theta}$ is the exact eigenstate, we then obtain
\begin{equation}
\label{eq:leading_fourier_state}
    \frac{\frac{\frac{\alpha_A}{\lambda_{\max}}+1}{I_n\otimes I-L_n\otimes\frac{ A}{\lambda_{\max}}}\ket{0}\ket{\psi_\theta}}
    {\norm{\frac{\frac{\alpha_A}{\lambda_{\max}}+1}{I_n\otimes I-L_n\otimes\frac{ A}{\lambda_{\max}}}\ket{0}\ket{\psi_\theta}}}
    =\frac{\sum_{j=0}^{n-1}L_n^j\otimes\frac{A^j}{\lambda_{\max}^j}\ket{0}\ket{\psi_\theta}}
    {\norm{\sum_{j=0}^{n-1}L_n^j\otimes\frac{A^j}{\lambda_{\max}^j}\ket{0}\ket{\psi_\theta}}}
    =\frac{1}{\sqrt{n}}\sum_{j=0}^{n-1}e^{ij\theta}\ket{j}\ket{\psi_\theta}.
\end{equation}
The phase angle $\theta$ can now be estimated using the standard quantum phase estimation.

\begin{theorem}[Quantum eigenvalue estimation, leading eigenvalues]
\label{thm:qeve_extreme}
Let $A$ be a square matrix, such that $A/\alpha_A$ is block encoded by $O_A$ with some normalization factor $\alpha_A\geq\norm{A}$.
Assume that $A/\alpha_A=SJS^{-1}$ has a Jordan form decomposition with upper bound $\kappa_S\geq\norm{S}\norm{S}^{-1}$ on the Jordan condition number.
Suppose that oracle $O_\psi\ket{0}=\ket{\psi}$ prepares an initial state within distance $\norm{\ket{\psi}-\ket{\psi_{\theta}}}=\mathbf{O}\left(\lambda_{\max}\sqrt{\epsilon}/(\alpha_A\kappa_S)\right)$ from an eigenstate such that $A\ket{\psi_{\theta}}=\lambda_{\max}e^{i\theta}\ket{\psi_{\theta}}$, where $\lambda_{\max}>0$ is the largest absolute value of eigenvalues of $A$.
Assume that the numerical range $\mathcal{W}(A)$ is enclosed by the disk $\lambda_{\max}\mathcal{D}$.
Then, there exists a quantum algorithm that outputs a value $\widetilde\theta$ with accuracy in centered modulus
\begin{equation}
    \abs{\mathbf{CMod}_{2\pi}\left(\widetilde\theta-\theta\right)}\leq\epsilon
\end{equation}
and probability $1-\pfail$, using
\begin{equation}
    \mathbf{O}\left(\frac{\alpha_A}{\lambda_{\max}\epsilon}\kappa_S\log\left(\frac{1}{\pfail}\right)\right)
\end{equation}
queries to controlled-$O_A$, controlled-$O_\psi$, and their inverses.
\end{theorem}
\begin{proof}
    Assume that the input state $\ket{\psi}=\ket{\psi_{\theta}}$ is the exact eigenstate, and that the quantum linear system solver makes no error. Then the standard phase estimation allows us to estimate $\theta$ with an accuracy $\epsilon$ (in centered modulus) and a constant probability strictly larger than $1/2$, for some choice of
    \begin{equation}
        n=\mathbf{O}\left(\frac{1}{\epsilon}\right).
    \end{equation}

    Now consider the general case, where the quantum linear system solver has accuracy $\epsilon_{\text{lin}}$ and the initial state has distance $\norm{\ket{\psi}-\ket{\psi_{\theta}}}=\epsilon_{\text{init}}$ to the true eigenstate. Then \cor{lin_sys_perturb} implies that the output state is close to the Fourier state \eq{leading_fourier_state} with Euclidean distance at most
    \begin{equation}
        \epsilon_{\text{lin}}+\mathbf{O}\left({\sqrt{n}\kappa_S\frac{\alpha_A}{\lambda_{\max}}\epsilon_{\text{init}}}\right).
    \end{equation}
    The remaining analysis proceeds in a similar way as in \thm{qeve}.
\end{proof}

\begin{remark}
Note that the above query complexity is expressed in terms of the Jordan condition number, although this can be improved by a tighter estimate of $\norm{A^j}$. If we were to impose the assumption on the numerical range $\mathcal{W}(A)\subseteq\lambda_{\max}\mathcal{D}$, we would find the leading eigenvalue $\lambda_{\max}e^{i\theta}$ on the boundary of $\mathcal{W}(A)$. Consequently, $A$ can be unitarily block diagonalized separating $\lambda_{\max}e^{i\theta}$ away from the remaining spectra~\cite[Theorem 1.6.6]{roger1994topics}. The problem can then be directly solved by QSVT, and there is no need to invoke the above eigenvalue estimation algorithm.
\end{remark}

\section{Discussion}
\label{sec:discuss}
In this work, we have developed quantum algorithms to estimate and transform eigenvalues of high-dimensional matrices accessed by a quantum computer. Our eigenvalue estimation algorithm is provably optimal in the inverse accuracy and failure probability for a diagonalizable input, whereas our eigenvalue transformation algorithm achieves an average performance comparable to previous results for singular value transformation. As immediate applications, we present a quantum differential equation solver for matrices with imaginary spectra achieving a strictly linear time scaling for an average diagonalizable input, as well as a ground state preparation algorithm for matrices with real spectra nearly optimal in the combined scaling with the inverse spectral gap and inverse accuracy. We have extended the results to more general matrices with complex eigenvalues, obtaining a new differential equation solver and a quantum algorithm for estimating leading eigenvalues.
Our work thus provides a unified toolbox for processing eigenvalues of non-normal matrices on quantum computers---a practical problem that is out of reach of the pre-existing quantum singular value algorithm and its descendants.

Our main technical contribution is a method to efficiently generate the Chebyshev history state, which encodes Chebyshev polynomials of the input matrix in quantum superposition. Prior to our work, it was known how to create such a state for Hermitian inputs via discrete-time quantum walk, but no such a mechanism was available for non-normal operators. Our new approach employs a matrix version of the Chebyshev generating function, which is then implemented using the optimal scaling quantum linear system solver. As Chebyshev polynomials provide a close-to-best minimax approximation for functions defined over a real interval, the query complexity of our solution is expected to be nearly optimal. However, our methodology is by no means restricted to only the Chebyshev expansion. We show how to estimate leading eigenvalues by generating a power series of the input matrix. More generally, we present an efficient quantum algorithm to generate a history state of Faber polynomials that provide a nearly optimal basis for function approximations on compact subsets of the complex plane, of which Chebyshev polynomials and power series are two special cases.

When the initial state is prepared close to an eigenstate, the Chebyshev history state we produce contains information about the corresponding eigenvalue in the phase of coefficients. We have developed a Chebyshev phase estimation algorithm that estimates the phase (and thus eigenvalue of the target matrix) with an asymptotically optimal query complexity. However, it is plausible that alternative methods exist that can extract the phase from a Chebyshev state with a better performance in practice. It would also be useful to consider a setting where the initial state is an arbitrary superposition of eigenstates, which is relevant for applications such as period finding. We leave a detailed analysis of this case as a subject for future work.

Recent work developed an alternative approach for implementing functions of non-normal matrices based on contour integrals~\cite{Fang2023timemarchingbased,Takahira2020QuantumCauchy,Takahira21}. 
That approach requires a coherent implementation of the discretized integral on a quantum computer, and its complexity depends largely on the choice of contours. With a circular contour, that method gives a differential equation solver whose complexity has a quadratic dependence on the evolution time~\cite{Fang2023timemarchingbased}. In contrast, our approach implements the Chebyshev expansion of the target function with a predetermined truncate order, and is conceptually quite different. However, contour integrals have appeared in our analysis of Faber polynomials, suggesting a deeper connection between our approach and the contour integral method from previous work. It would also be of interest to construct quantum eigenvalue algorithms with fewer queries to the initial state preparation, on which the recent technique of linear combination of Hamiltonian simulation~\cite{AnChildsLin23} may provide insight.

We have shown that our eigenvalue transformation algorithm achieves a better performance for a random choice of input matrix. This speedup in turn follows from the fact that Fourier series converges faster on average---a powerful result known as the Carleson-Hunt theorem. This complements many recent work that explored the use of randomness in speeding up quantum simulation algorithms~\cite{Childs2019fasterquantum,Campbell18,ChenBrandao21,Zhao21}. Most of those results have focused on improvement of the product-formula algorithm and its variants. Our work demonstrates that randomness can also be used to speed up more advanced quantum algorithms through the faster convergence of Fourier expansion, which has applications to quantum simulation and beyond.

For the purpose of generality, we have expressed the complexity of our algorithms in terms of various matrix functions, with the understanding that these functions can be bounded differently for different problems (see \tab{nonnormal} for a summary of results). In the Chebyshev case, we have further derived bounds for input matrices having Jordan forms, generalizing previous analysis of differential equation algorithms based on the spectra abscissa~\cite[Section 3.2]{Krovi2023improvedquantum}. 
This introduces a dependence on the Jordan condition number $\kappa_S$, which is a common measure of nonnormality in numerical linear algebra~\cite[Page 444]{Trefethen05} and can be upper bounded using techniques from~\cite{JordanCondition}. 
It seems possible to establish a query lower bound for solving differential equations in terms of $\kappa_S$, by analyzing the speed at which a non-Hermitian Hamiltonian evolves quantum states~\cite{Bender07}, though the details remain to be worked out.
In the Faber case, we have derived bounds assuming numerical range or pseudospectrum are enclosed by the Faber region, extending previous differential equation analysis based on the notion of numerical abscissa~\cite[Section 3.1]{Krovi2023improvedquantum}. 

We emphasize again that this is not the only way to analyze our algorithms. 
For instance, one can also apply the Jordan-form analysis in the Faber case, obtaining a dependence on the Jordan condition number.
The reverse direction is however not very interesting: numerical range is enclosed by a real interval, if and only if the matrix itself is Hermitian, per the characterizations from \sec{prelim_matrix}, whereas the pseudospectrum is an open set that can never be enclosed by the real interval.
Anyway, it could be fruitful for future work to establish stronger bounds on matrix functions, which would have a direct influence on the performance of our eigenvalue algorithms.

The output of our QEVT algorithm is a quantum state proportional to the transformed input matrix applied to the initial state. By measuring copies of this state and post-processing the measurement outcomes, one can learn and make predictions about properties of the input matrix. However, it is plausible that more efficient methods exist that bypass this two-step procedure and target directly at properties of high-dimensional non-normal matrices. Our algorithm also maintains a considerable amount of coherence to produce the output state in quantum superposition. But this resource requirement may be relaxed using recent methods developed for Hermitian eigenvalue estimation and transformation~\cite{LinTong22,DongLinTong22}.
On a different note, we have also studied an alternative version of QEVT where the input matrix is transformed as a block encoding, which may be useful when QEVT is invoked as a subroutine in designing other quantum algorithms.

As further developments unfold, we hope the work initiated here will reveal the potential of quantum computers in processing non-normal operators, opening up applications that were previously unexplored.

\section*{Acknowledgements}
Y.S.\ thanks Yu Tong for helpful discussions.

\newpage
\appendix
\section{Analysis of Chebyshev-based algorithms}
\label{append:analysis_cheby}
In this appendix, we analyze the Chebyshev-based eigenvalue algorithms in more detail. Specifically, we review previous bounds for matrix exponentials in \append{analysis_cheby_prelim} based on the spectral abscissa. We show in \append{analysis_cheby_bernstein} how this analysis can be generalized for a polynomial function with the help of Bernstein's theorem. Finally, we bound the average complexity of our algorithms in \append{analysis_cheby_carleson} using the Carleson-Hunt theorem.

\subsection{Matrix exponential bound based on spectral abscissa}
\label{append:analysis_cheby_prelim}

We begin by reviewing a bound on norm of the matrix exponential function used in the analysis of quantum differential equation algorithms~\cite{Krovi2023improvedquantum} \cite[Appendix D]{leveque2007finite}. Given a square input matrix $C$, our goal is to bound $\norm{e^{\tau C}}$ for $\tau>0$. Assuming that $C$ has the Jordan form decomposition $C=SJS^{-1}$, we have $\norm{e^{\tau C}}\leq \kappa_S\norm{e^{\tau J}}$ with any upper bound $\kappa_S\geq\norm{S}\norm{S^{-1}}$ on the Jordan condition number.

We know from \eq{jordan_form_transformation} that $e^{\tau J}$ contains only diagonal blocks of the form
\begin{equation}
    e^{J(\lambda_l,d_l)}=
    \begin{bmatrix}
        e^{\tau\lambda_l} &  &  &  &  & \\
        \tau e^{\tau\lambda_l} & e^{\tau\lambda_l} &  &  &  & \\
        \frac{\tau^2}{2!}e^{\tau\lambda_l} & \tau e^{\tau\lambda_l} & e^{\tau\lambda_l} &  &  & \\
        \vdots & \frac{\tau^2}{2!}e^{\tau\lambda_l} & \tau e^{\tau\lambda_l} & \ddots &\\
        \vdots & \ddots & \ddots & \ddots & \ddots & \\
        \frac{\tau^{d_l-1}}{(d_l-1)!}e^{\tau\lambda_l} & \cdots & \cdots & \frac{\tau^2}{2!}e^{\tau\lambda_l} & \tau e^{\tau\lambda_l} & e^{\tau\lambda_l}\\
    \end{bmatrix}
\end{equation}
Using the scalar version of \lem{block_norm}, we conclude that $\norm{e^{J(\lambda_l,d_l)}}=\mathbf{\Theta}\left(\abs{\tau^{d_l-1}e^{\tau\lambda_l}}\right)$, treating size $d_l$ of the Jordan blocks as constant. This implies the following bound on the matrix exponential function.

\begin{proposition}
    Let $C$ be a square matrix with the Jordan form decomposition $C=SJS^{-1}$. Then,
    \begin{equation}
        \norm{e^{\tau C}}
        =\mathbf{\Theta}\left(\kappa_S\max_l\abs{\tau^{d_l-1}e^{\tau\lambda_l}}\right),
    \end{equation}
    where $\kappa_S\geq\norm{S}\norm{S^{-1}}$ is an upper bound on the Jordan condition number, and the maximization is over all Jordan blocks $J(\lambda_l,d_l)$ with eigenvalue $\lambda_l$ and size $d_l$.

    Depending on the specific value of $\lambda_l$ and $d_l$, the factor $\mathbf{\Theta}\left(\abs{\tau^{d_l-1}e^{\tau\lambda_l}}\right)$ behaves differently as follows:
\begin{enumerate}
    \item $\Re(\lambda_l)<0$: In this case, $\mathbf{\Theta}\left(\abs{\tau^{d_l-1}e^{\tau\lambda_l}}\right)=\mathbf{\Theta}\left(e^{\tau(\Re(\lambda_l)+\mathbf{o}(1))}\right)$ decays exponentially with $\tau$.
    \item $\Re(\lambda_l)=0$ and $d_l=1$: In this case, $\lambda_l$ is a \emph{nondefective} eigenvalue, and $\mathbf{\Theta}\left(\abs{\tau^{d_l-1}e^{\tau\lambda_l}}\right)=\mathbf{\Theta}(1)$ is bounded as a function of $\tau$. 
    \item $\Re(\lambda_l)=0$ and $d_l\geq2$: In this case, $\mathbf{\Theta}\left(\abs{\tau^{d_l-1}e^{\tau\lambda_l}}\right)=\mathbf{\Theta}\left(\tau^{d_l-1}\right)$ grows polynomially.
    \item $\Re(\lambda_l)>0$: In this case, $\mathbf{\Theta}\left(\abs{\tau^{d_l-1}e^{\tau\lambda_l}}\right)=\mathbf{\Theta}\left(e^{\tau(\Re(\lambda_l)+\mathbf{o}(1))}\right)$ grows exponentially with $\tau$.
\end{enumerate}
\end{proposition}

Thus, the scaling of $e^{\tau C}$ largely depends on the \emph{spectral abscissa} defined as
\begin{equation}
    \max_l\Re\left(\lambda_l(C)\right).
\end{equation}
Specifically, $\max_l\Re(\lambda_l)<0$ implies that $\norm{e^{\tau C}}$ asymptotically decays as a function $\tau$; $\max_l\Re(\lambda_l)=0$ corresponds to the case where the growth of $\norm{e^{\tau C}}$ is at most polynomial further determined by size of the Jordan blocks with imaginary spectra; and $\max_l\Re(\lambda_l)>0$ means $\norm{e^{\tau C}}$ blows up exponentially. This estimate is particularly useful for analyzing existing quantum differential equation algorithms such as \cite{Krovi2023improvedquantum}, because the complexity of that algorithm is expressed in terms of spectral norm of the matrix exponential function.

\subsection{Matrix polynomial bound with Bernstein's theorem}
\label{append:analysis_cheby_bernstein}

Unlike \cite{Krovi2023improvedquantum}, our Chebyshev-based eigenvalue algorithms have complexity depending on the spectral norm of various matrix polynomials. In deriving a similar bound as above, the challenge here is to handle the high-order derivatives in \eq{jordan_form_transformation} for a general polynomial function. We overcome this using the following Bernstein's theorem.
\begin{lemma}[Bernstein's theorem {\cite[Eqs.\ (6) and (12)]{Kalmykov21}}]
    Given a degree-$j$ polynomial $p_j$, it holds
    \begin{equation}
        \abs{p_j'(x)}\leq\frac{j}{\sqrt{(x-a)(b-x)}}\norm{p_j}_{\max,[a,b]}
    \end{equation}
    for any $a<x<b$.
\end{lemma}

This bound should not be confused with the related Markov brothers' inequality often used in the study of query complexity~\cite{childs2017lecture}, which instead reads
\begin{equation}
    \norm{p_j'}_{\max,[a,b]}\leq\frac{2j^2}{b-a}\norm{p_j}_{\max,[a,b]}.
\end{equation}
Here, the Markov inequality has a prefactor scaling quadratically with the polynomial degree, whereas the Bernstein inequality introduces a prefactor linear in the polynomial degree. We thus have
\begin{equation}
    \norm{p_j'}_{\max,[a+\delta,b-\delta]}
    \leq\frac{j}{\sqrt{\delta(b-a-\delta)}}\norm{p_j}_{\max,[a,b]},
\end{equation}
as long as we use a nonzero margin $\delta>0$.
By recursively applying the above analysis, we obtain the following bound.
Alternatively, one may also use the sharper estimate from~\cite[Equation (37)]{Kalmykov21}.
\begin{corollary}[Recursive Bernstein's theorem]
    For degree-$j$ polynomials $p_j$ and $\delta>0$, it holds that
    \begin{equation}
        \norm{p_j^{(d-1)}}_{\max,[a+\delta,b-\delta]}
        =\mathbf{O}\left(\left(\frac{j}{\sqrt{\delta}}\right)^{d-1}\norm{p_j}_{\max,[a,b]}\right),
    \end{equation}
    with $b-a=\mathbf{\Omega}(1)$ and a fixed positive integer $d$.
\end{corollary}
\begin{proof}
    For $d=2$ and a margin $\delta>0$, we have from the Bernstein's theorem that
    \begin{equation}
        \norm{p_j'}_{\max,[a+\delta,b-\delta]}
        \leq\max_{x\in[a+\delta,b-\delta]}\frac{j}{\sqrt{(x-a)(b-x)}}\norm{p_j}_{\max,[a,b]}
        =\frac{j}{\sqrt{\delta(b-a-\delta)}}\norm{p_j}_{\max,[a,b]}.
    \end{equation}
    We shrink the interval by a margin $\delta$ again for $d=3$:
    \begin{equation}
    \begin{aligned}
        \norm{p_j''}_{\max,[a+2\delta,b-2\delta]}
        &\leq\frac{j-1}{\sqrt{\delta(b-a-3\delta)}}\norm{p_j'}_{\max,[a+\delta,b-\delta]}\\
        &\leq\frac{(j-1)j}{\sqrt{\delta(b-a-3\delta)}\sqrt{\delta(b-a-\delta)}}\norm{p_j}_{\max,[a,b]}.
    \end{aligned}
    \end{equation}
    Performing this recursively, 
    \begin{equation}
        \begin{aligned}
            \norm{p_j^{(d-1)}}_{\max,[a+(d-1)\delta,b-(d-1)\delta]}
        &\leq\left(\frac{j}{\sqrt{\delta\left(b-a-(2d-3)\delta\right)}}\right)^{d-1}\norm{p_j}_{\max,[a,b]}\\
        &=\mathbf{O}\left(\left(\frac{j}{\sqrt{\delta}}\right)^{d-1}\norm{p_j}_{\max,[a,b]}\right),
        \end{aligned}
    \end{equation}
    assuming $b-a=\mathbf{\Omega}(1)$ and $d$ is a fixed constant.

    The claimed bound then follows from the rescaling $\delta\mapsto\frac{\delta}{d-1}$, which only introduces a constant prefactor and does not change the asymptotic estimate.
\end{proof}
\begin{corollary}
\label{cor:matrix_poly_bound}
    Let $C$ be a matrix with eigenvalues belonging to the real interval $[a+\delta,b-\delta]$ with $\delta>0$. Suppose that $C=SJS^{-1}$ has a Jordan form decomposition with upper bound $\kappa_S\geq\norm{S}\norm{S^{-1}}$ on the Jordan condition number and size $d_{\max}$ of the largest Jordan block.
    For degree-$j$ polynomials $p_j$, it holds that
    \begin{equation}
        \norm{p_j(C)}=\mathbf{O}\left(\kappa_S\left(\frac{j}{\sqrt{\delta}}\right)^{d_{\max}-1}\norm{p_j}_{\max,[a,b]}\right),
    \end{equation}
    assuming $b-a=\mathbf{\Omega}(1)$ and $d_{\max}=\mathbf{O}(1)$.
\end{corollary}

We now apply this result to analyze the performance of Chebyshev-based eigenvalue algorithms. Let us first bound $\max_{j=0,1,\ldots,n-1}\norm{\mathbf{U}_{j}\left(\frac{A}{\alpha_A}\right)}$,
which appears in the asymptotic complexity expression of the algorithm for generating Chebyshev history states. By using the rescaling trick \eq{block_rescaling} for block encoding, we may assume without loss of generality that the normalization factor satisfies $\alpha_A\geq2\norm{A}$. For a fixed value of $j$, we thus have
\begin{equation}
    \norm{\mathbf{U}_{j}\left(\frac{A}{\alpha_A}\right)}
    =\mathbf{O}\left(\kappa_Sj^{d_{\max}-1}\norm{\mathbf{U}_j}_{\max,[-\frac{3}{4},\frac{3}{4}]}\right)
    =\mathbf{O}\left(\kappa_Sj^{d_{\max}-1}\norm{\mathbf{T}_j}_{\max,[-1,1]}\right)
    =\mathbf{O}\left(\kappa_Sj^{d_{\max}-1}\right),
\end{equation}
using the fact that $\mathbf{T}_j^2(x)-(x^2-1)\mathbf{U}_{j-1}^2(x)=1$. This implies the scaling
\begin{equation}
    \alphaU=\mathbf{O}\left(\kappa_Sn^{d_{\max}-1}\right)
\end{equation}
claimed in the remark succeeding \thm{generate_history}.

Next, we consider $\max_{l=0,1,\ldots,n-1}\norm{\sum_{k=l}^{n-1}\widetilde\beta_k\widetilde{\mathbf{T}}_{k-l}\left(\frac{A}{\alpha_A}\right)\ket{\psi}}$, which is used to describe complexity of the eigenvalue estimation algorithm. We again assume $\alpha_A\geq2\norm{A}$ without loss of generality, obtaining for a fixed $l$
\begin{equation}
\begin{aligned}
    \norm{\sum_{k=l}^{n-1}\widetilde\beta_k\widetilde{\mathbf{T}}_{k-l}\left(\frac{A}{\alpha_A}\right)\ket{\psi}}
    &\leq\norm{S}\norm{\sum_{k=l}^{n-1}\widetilde\beta_k\widetilde{\mathbf{T}}_{k-l}\left(J\right)}\norm{S^{-1}\ket{\psi}}\\
    &=\mathbf{O}\left(\norm{S}\norm{S^{-1}\ket{\psi}}
    n^{d_{\max}-1}
    \norm{\sum_{k=l}^{n-1}\widetilde\beta_k\widetilde{\mathbf{T}}_{k-l}}_{\max,[-1,1]}\right).
\end{aligned}
\end{equation}
Here, the shifted Chebyshev partial sum can be further bounded as
    \begin{equation}
    \begin{aligned}
    \abs{\sum_{k=l}^{n-1}\widetilde{\beta}_k \widetilde{\mathbf{T}}_{k-l}\left(x\right)}
    &=\abs{\sum_{k=l}^{n-1}\widetilde{\beta}_k \cos\left((k-l)\arccos\left(x\right)\right)}\\
    &=\abs{\sum_{k=l}^{n-1}\widetilde{\beta}_k\frac{e^{i(k-l)\arccos\left(x\right)}+e^{-i(k-l)\arccos\left(x\right)}}{2}}\\
    &\leq\frac{1}{2}\abs{\sum_{k=l}^{n-1}\widetilde{\beta}_ke^{ik\arccos\left(x\right)}}
    +\frac{1}{2}\abs{\sum_{k=l}^{n-1}\widetilde{\beta}_ke^{-ik\arccos\left(x\right)}}\\
    &\leq\frac{1}{2}\abs{\sum_{k=0}^{n-1}\widetilde{\beta}_ke^{ik\arccos\left(x\right)}}
    +\frac{1}{2}\abs{\sum_{k=0}^{l-1}\widetilde{\beta}_ke^{ik\arccos\left(x\right)}}\\
    &\quad+\frac{1}{2}\abs{\sum_{k=0}^{n-1}\widetilde{\beta}_ke^{-ik\arccos\left(x\right)}}
    +\frac{1}{2}\abs{\sum_{k=0}^{l-1}\widetilde{\beta}_ke^{-ik\arccos\left(x\right)}}.\\
    \end{aligned}
    \end{equation}
Note that all the four terms in the last step are one-sided Fourier partial sums, but we would recover the original polynomial $p$ if we had the two-sided Fourier series. Thus, by \eq{worst_partial}, they all have the scaling $\sim\norm{p}_{\max,[-1,1]}\log(n)$ in the worst case. Hence,
\begin{equation}
    \alphaT=\mathbf{O}\left(\norm{S}\norm{S^{-1}\ket{\psi}}
    n^{d_{\max}-1}\log(n)
    \norm{p}_{\max,[-1,1]}\right).
\end{equation}
This establishes the asymptotic scaling asserted in the remark succeeding \thm{qevt}.

\subsection{Average-case analysis with Carleson-Hunt theorem}
\label{append:analysis_cheby_carleson}

In the previous subsection, we have analyzed asymptotic scaling of the Chebyshev-based eigenvalue algorithms using a recursive version of the Bernstein's theorem. For the quantum differential equation algorithm, we have a polynomial that approximates the exponential function $e^{-i\alpha_Atx}$, so the above analysis yields
\begin{equation}
    \norm{\sum_{k=l}^{n-1}\widetilde\beta_k\widetilde{\mathbf{T}}_{k-l}\left(\frac{A}{\alpha_A}\right)\ket{\psi}}
    \lesssim\norm{S}\norm{S^{-1}\ket{\psi}}
    \log(n)
    \norm{e^{-i\alpha_At(\cdot)}}_{\max,[-1,1]}
    =\norm{S}\norm{S^{-1}\ket{\psi}}
    \log(n)
\end{equation}
for a diagonalizable input $A=S\Lambda S^{-1}$.
On the other hand,
\begin{equation}
\begin{aligned}
    \norm{e^{-itA}\ket{\psi}}
    &=\sqrt{\bra{\psi}S^{-\dagger} e^{it\Lambda}
    S^\dagger S
    e^{-it\Lambda}S^{-1}\ket{\psi}}\\
    &\geq\sqrt{\lambda_{\min}\left(S^\dagger S\right)\bra{\psi}S^{-\dagger} S^{-1}\ket{\psi}}\\
    &=\sqrt{\frac{1}{\norm{(S^\dagger S)^{-1}}}}
        \norm{S^{-1}\ket{\psi}}
    =\frac{\norm{S^{-1}\ket{\psi}}}{\norm{S^{-1}}}.
\end{aligned}
\end{equation}
This implies that 
\begin{equation}
    \frac{\alphaT}{\alphaFPsi}=\mathbf{O}\left(\kappa_S\log(n)\right)
\end{equation}
in the worst case as is claimed in the remark succeeding \thm{diff_eq}.

However, we show that one can shave off the $\log(n)$ factor when running the algorithm on an average input. To this end, recall that the one-sided Fourier expansion $\widetilde{\mathbf{H}}(g)(\omega)=\sum_{j=0}^\infty \xi_je^{-ij\omega}$ relates to the two-sided expansion $g(\omega)=\sum_{j=-\infty}^\infty \xi_je^{-ij\omega}$ via the Hilbert transform $\mathbf{H}(g)(\omega)$ as
\begin{equation}
    \widetilde{\mathbf{H}}(g)(\omega)
    =-\frac{i}{2}\mathbf{H}(g)(\omega)
    +\frac{1}{4\pi}\int_{-\pi}^{\pi}\mathrm{d}u\
    g(u)
    +\frac{1}{2}g(\omega).
\end{equation}
Therefore, by the Riesz inequality (\lem{riesz}), the one-sided expansion as a function of $\omega$ has the $\mathcal{L}_2$-norm
\begin{equation}
    \norm{\widetilde{\mathbf{H}}(g)}_{2,[-\pi,\pi]}
    \leq\frac{1}{2}\norm{\mathbf{H}(g)}_{2,[-\pi,\pi]}
    +\frac{1}{2\sqrt{2\pi}}\norm{g}_{1,[-\pi,\pi]}
    +\frac{1}{2}\norm{g}_{2,[-\pi,\pi]]}
    =\mathbf{O}\left(\norm{g}_{\max,[-\pi,\pi]}\right).
\end{equation}
Using the Carleson-Hunt theorem (\lem{carleson_hunt}), we have
\begin{equation}
    \norm{\mathbf{S}_*\widetilde{\mathbf{H}}(g)}_{2,[-\pi,\pi]}
    =\mathbf{O}\left(\norm{\widetilde{\mathbf{H}}(g)}_{2,[-\pi,\pi]}\right)
    =\mathbf{O}\left(\norm{g}_{\max,[-\pi,\pi]}\right)
\end{equation}
for the Fourier maximal function
$\mathbf{S}_*\widetilde{\mathbf{H}}(g)(\omega)=\sup_{n=0,1,\ldots}\abs{\sum_{j=0}^{n-1} \xi_je^{-ij\omega}}$.

Given a polynomial $p$ and its Chebyshev expansion $p(x)=\sum_{j=0}^{n-1}\widetilde{\beta}_j\widetilde{\mathbf{T}}_j(x)$, we can bound the shifted Chebyshev partial sum as
\begin{equation}
\begin{aligned}
    \norm{\sum_{k=l}^{n-1}\widetilde\beta_k\widetilde{\mathbf{T}}_{k-l}\left(\frac{A}{\alpha_A}\right)\ket{\psi}}
    &=\sqrt{\bra{\psi}S^{-\dagger}\left(\sum_{j=l}^{n-1}\widetilde\beta_j\widetilde{\mathbf{T}}_{j-l}\left(\frac{\Lambda}{\alpha_A}\right)\right)^\dagger
    S^\dagger S\sum_{k=l}^{n-1}\widetilde\beta_k\widetilde{\mathbf{T}}_{k-l}\left(\frac{\Lambda}{\alpha_A}\right)S^{-1}\ket{\psi}}\\
    &\leq\norm{S}\sqrt{\bra{\psi}S^{-\dagger}\left(\sum_{j=l}^{n-1}\widetilde\beta_j\widetilde{\mathbf{T}}_{j-l}\left(\frac{\Lambda}{\alpha_A}\right)\right)^\dagger
    \sum_{k=l}^{n-1}\widetilde\beta_k\widetilde{\mathbf{T}}_{k-l}\left(\frac{\Lambda}{\alpha_A}\right)S^{-1}\ket{\psi}},
\end{aligned}
\end{equation}
where the diagonal entries satisfy
    \begin{equation}
    \begin{aligned}
    \abs{\sum_{k=l}^{n-1}\widetilde{\beta}_k \widetilde{\mathbf{T}}_{k-l}\left(x\right)}^2
    &\leq\abs{\sum_{k=0}^{n-1}\widetilde{\beta}_ke^{ik\arccos\left(x\right)}}^2
    +\abs{\sum_{k=0}^{l-1}\widetilde{\beta}_ke^{ik\arccos\left(x\right)}}^2\\
    &\quad+\abs{\sum_{k=0}^{n-1}\widetilde{\beta}_ke^{-ik\arccos\left(x\right)}}^2
    +\abs{\sum_{k=0}^{l-1}\widetilde{\beta}_ke^{-ik\arccos\left(x\right)}}^2.\\
    \end{aligned}
    \end{equation}
The four terms above are all Fourier partial sums, and we only bound one of them without loss of generality.
Now denote the eigenvalues of $\arccos(A/\alpha_A)$ as $\omega_0,\ldots,\omega_{d-1}$, and suppose that they satisfy a probability distribution with density $q(\omega_0,\ldots,\omega_{d-1})$.
Then on average, 
\begin{equation}
\begin{aligned}
    \int_{-\pi}^{\pi}\mathrm{d}\omega_0\cdots\int_{-\pi}^{\pi}\mathrm{d}\omega_{d-1}\
    q(\omega_0,\ldots,\omega_{d-1})
    \abs{\sum_{j=0}^{n-1}\widetilde{\beta}_j e^{-ij\omega_l}}^2
    &\leq 
    \int_{-\pi}^{\pi}\mathrm{d}\omega_l\
    q_l(\omega_l)
    \mathbf{S}_*\widetilde{\mathbf{H}}(p(\cos(\omega_l)))^2\\
    &\leq\norm{q_l}_{\max,[-\pi,\pi]}
    \norm{\mathbf{S}_*\widetilde{\mathbf{H}}(p(\cos))}_{2,[-\pi,\pi]}^2\\
    &=\mathbf{O}(\norm{q_l}_{\max,[-\pi,\pi]}\norm{p}_{\max,[-1,1]}^2),
\end{aligned}
\end{equation}
where
\begin{equation}
    q_l(\omega_l)=\int_{-\pi}^{\pi}\mathrm{d}\omega_0\cdots\int_{-\pi}^{\pi}\mathrm{d}\omega_{l-1}\int_{-\pi}^{\pi}\mathrm{d}\omega_{l+1}\cdots\int_{-\pi}^{\pi}\mathrm{d}\omega_{d-1}\
        q(\omega_0,\ldots,\omega_{l-1},\omega_l,\omega_{l+1},\ldots,\omega_{d-1})
\end{equation}
is the $l$th marginal density function.
This means that on average, 
\begin{equation}
\begin{aligned}
    &\int_{-\pi}^{\pi}\mathrm{d}\omega_0\cdots\int_{-\pi}^{\pi}\mathrm{d}\omega_{d-1}\
    q(\omega_0,\ldots,\omega_{d-1})
    \sqrt{\bra{\psi}S^{-\dagger}\sum_{j=0}^{n-1}\widetilde{\beta}_j e^{ij\arccos\left(\frac{\Lambda}{\alpha_A}\right)}
    \sum_{k=0}^{n-1}\widetilde{\beta}_k e^{-ik\arccos\left(\frac{\Lambda}{\alpha_A}\right)}S^{-1}\ket{\psi}}\\
    &\leq\sqrt{\bra{\psi}S^{-\dagger}\int_{-\pi}^{\pi}\mathrm{d}\omega_0\cdots\int_{-\pi}^{\pi}\mathrm{d}\omega_{d-1}\
    q(\omega_0,\ldots,\omega_{d-1})\sum_{j=0}^{n-1}\widetilde{\beta}_j e^{ij\arccos\left(\frac{\Lambda}{\alpha_A}\right)}
    \sum_{k=0}^{n-1}\widetilde{\beta}_k e^{-ik\arccos\left(\frac{\Lambda}{\alpha_A}\right)}S^{-1}\ket{\psi}}\\
    &=\mathbf{O}\left(\sqrt{\max_l\norm{q_l}_{\max,[-\pi,\pi]}}\norm{p}_{\max,[-1,1]}\norm{S^{-1}\ket{\psi}}\right),
\end{aligned}
\end{equation}
where we have used Jensen's inequality and positive semidefinite property of matrices.
Assuming $\max_l\norm{q_l}_{\max,[-\pi,\pi]}=\mathbf{O}(1)$ is independent of the target polynomial $p(x)$, this gives the scaling
\begin{equation}
    \alphaT=\mathbf{O}\left(\norm{S}\norm{S^{-1}\ket{\psi}}
    \norm{p}_{\max,[-1,1]}\right)
\end{equation}
in analyzing the average-case runtime of eigenvalue transformation and
\begin{equation}
    \frac{\alphaT}{\alphaFPsi}=\mathbf{O}\left(\kappa_S\right)
\end{equation}
for the differential equation solver,
justifying the claims in the remarks following \thm{qevt} and \thm{diff_eq}.

We now briefly explain how our analysis can be modified to handle additional logarithmic factors, which is the case for the complexity of our algorithms. To this end, suppose we have some nonnegative function $g(\omega_0,\ldots,\omega_{d-1})$ of the eigenvalues of $\arccos(A/\alpha_A)$.
Then, by the Cauchy-Schwarz inequality,
\begin{equation}
\begin{aligned}
    &\ \int_{-\pi}^{\pi}\mathrm{d}\omega_0\cdots\int_{-\pi}^{\pi}\mathrm{d}\omega_{d-1}\
    q(\omega_0,\ldots,\omega_{d-1})
    g(\omega_0,\ldots,\omega_{d-1})\log^r\left(g(\omega_0,\ldots,\omega_{d-1})\right)\\
    \leq&\
    \sqrt{\int_{-\pi}^{\pi}\mathrm{d}\omega_0\cdots\int_{-\pi}^{\pi}\mathrm{d}\omega_{d-1}\
    q(\omega_{d-1},\ldots,\omega_{d-1})
    g^2(\omega_0,\ldots,\omega_{d-1})}\\
    &\ \cdot\sqrt{\int_{-\pi}^{\pi}\mathrm{d}\omega_0\cdots\int_{-\pi}^{\pi}\mathrm{d}\omega_{d-1}\
    q(\omega_0,\ldots,\omega_{d-1})
    \log^{2r}\left(g(\omega_0,\ldots,\omega_{d-1})\right)}.
\end{aligned}
\end{equation}
We already know how to bound the first term. To handle the second term, we need to prove the concavity of $\ln^{2r}(y)$ for large values of $y$. Since
\begin{equation}
\begin{aligned}
    \frac{\mathrm{d}\ln^{2r}(y)}{\mathrm{d}y}
    &=\frac{2r\ln^{2r-1}(y)}{y},\\
    \frac{\mathrm{d}^2\ln^{2r}(y)}{\mathrm{d}y^2}
    &=\frac{2r(2r-1)\ln^{2r-2}(y)-2r\ln^{2r-1}(y)}{y^2}
    =\frac{2r\ln^{2r-2}(y)}{y^2}\left(2r-1-\ln(y)\right),
\end{aligned}
\end{equation}
the function $\ln^{2r}(y)$ is indeed concave for $y$ sufficiently large. Therefore, Jensen's inequality yields
\begin{equation}
\begin{aligned}
    &\ \int_{-\pi}^{\pi}\mathrm{d}\omega_0\cdots\int_{-\pi}^{\pi}\mathrm{d}\omega_{d-1}\
    q(\omega_0,\ldots,\omega_{d-1})
    \log^{2r}\left(g(\omega_0,\ldots,\omega_{d-1})\right)\\
    \leq&\ \log^{2r}\left(\int_{-\pi}^{\pi}\mathrm{d}\omega_0\cdots\int_{-\pi}^{\pi}\mathrm{d}\omega_{d-1}\
    q(\omega_0,\ldots,\omega_{d-1})
    g(\omega_0,\ldots,\omega_{d-1})\right).
\end{aligned}
\end{equation}
The remaining analysis now proceeds as before.

\section{Analysis of Faber-based algorithms}
\label{append:analysis_faber}
In this appendix, we analyze the Faber-based eigenvalue algorithms in more detail. Specifically, we review previous bounds for matrix exponentials in \append{analysis_faber_prelim} based on the numerical abscissa. We show in \append{analysis_faber_crouzeix} how this idea can be generalized to bound a matrix function with the help of numerical range and Crouzeix-Palencia theorem. We prove an analogous bound in \append{analysis_faber_pseudospectrum} based on the notion of pseudospectrum. Finally, we bound the average complexity of our algorithms in \append{analysis_faber_carleson} using the Carleson-Hunt theorem.

\subsection{Matrix exponential bound based on numerical abscissa}
\label{append:analysis_faber_prelim}

We begin by reviewing a previous bound on norm of the matrix exponential function used in the analysis of quantum differential equation algorithms~\cite{Krovi2023improvedquantum} \cite[Appendix D]{leveque2007finite}. Specifically, our goal is to bound $\norm{e^{\tau C}}$ for $\tau>0$ and square matrix $C$. The bound we will present does not depend on the Jordan condition number, which avoids the potential issue of a ill-conditioned Jordan basis arising in the setting of \append{analysis_cheby}.

To this end, recall that any square matrix $C$ can be uniquely written as~\cite[Theorem 4.1.2]{horn2012matrix}
\begin{equation}
    C=\Re(C)+i\Im(C)
\end{equation}
for Hermitian matrices $\Re(C)$ and $\Im(C)$. Indeed, the existence of such a \emph{Toeplitz decomposition} follows by setting
\begin{equation}
    \Re(C)=\frac{C+C^\dagger}{2},\qquad
    \Im(C)=\frac{C-C^\dagger}{2i},
\end{equation}
whereas an equality $\Re_1(C)+i\Im_1(C)=\Re_2(C)+i\Im_2(C)$ would imply that $\Re_1(C)-\Re_2(C)=i\Im_2(C)-i\Im_1(C)$ is both Hermitian and anti-Hermitian, forcing $\Re_1(C)=\Re_2(C)$ and $\Im_1(C)=\Im_2(C)$. We then define the \emph{numerical abscissa} as
\begin{equation}
    \max_l\lambda_l(\Re(C))=\lambda_{\max}(\Re(C))=\lambda_{\max}\left(\frac{C+C^\dagger}{2}\right).
\end{equation}
As the following proposition shows, the numerical abscissa can be used to bound the growth of matrix exponentials~\cite[Theorem 17.1]{Trefethen05}. We present a short proof based on the Lie-Trotter product formula which may be more familiar to readers of this work, due to its role in the study of quantum simulation and beyond~\cite{CSTWZ19}.

\begin{proposition}
    For $\tau>0$ and square matrix $C$,
    \begin{equation}
        \norm{e^{\tau C}}\leq e^{\tau\lambda_{\max}(\Re(C))}.
    \end{equation}
\end{proposition}
\begin{proof}
    We start with the Lie-Trotter splitting:
    \begin{equation}
        \norm{e^{\tau C}-e^{\tau\Re(C)}e^{i\tau\Im(C)}}
        \leq\frac{\tau^2}{2}\norm{\left[\Re(C),\Im(C)\right]}e^{\tau(\norm{\Re(C)}+\norm{\Im(C)})}
        \leq\frac{\tau^2}{2}\norm{C}^2e^{2\tau\norm{C}}.
    \end{equation}
    This implies through the triangle inequality that
    \begin{equation}
        \norm{e^{\tau C}}
        \leq\norm{e^{\tau\Re(C)}e^{i\tau\Im(C)}}+\norm{e^{\tau C}-e^{\tau\Re(C)}e^{i\tau\Im(C)}}
        \leq e^{\tau\lambda_{\max}(\Re(C))}+\frac{\tau^2}{2}\norm{C}^2e^{2\tau\norm{C}},
    \end{equation}
    where we have used the fact that $e^{i\tau\Im(C)}$ is unitary and that $\norm{e^{\tau\Re(C)}}=e^{\tau\lambda_{\max}(\Re(C))}$.

    This bound works well when $\tau\rightarrow0$. To analyze a long-time evolution, we divide it into $r$ steps and apply the above estimate within each step, obtaining
    \begin{equation}
        \norm{e^{\tau C}}
        \leq\norm{e^{\frac{\tau}{r}C}}^r
        \leq\left(e^{\frac{\tau}{r}\lambda_{\max}(\Re(C))}+\frac{\tau^2}{2r^2}\norm{C}^2e^{2\frac{\tau}{r}\norm{C}}\right)^r.
    \end{equation}
    Let us choose
    \begin{equation}
        r\geq\max\left\{2\tau\norm{C},\frac{e\tau^2\norm{C}^2}{2\epsilon}\right\}
    \end{equation}
    so that
    \begin{equation}
        e^{2\frac{\tau}{r}\norm{C}}\leq e,\qquad
        \frac{\tau^2}{2r^2}\norm{C}^2e\leq\frac{\epsilon}{r}.
    \end{equation}
    This implies
    \begin{equation}
    \begin{aligned}
        \norm{e^{\tau C}}
        &\leq\left(e^{\frac{\tau}{r}\lambda_{\max}(\Re(C))}+\frac{\epsilon}{r}\right)^r
        =e^{\tau\lambda_{\max}(\Re(C))}\left(1+\frac{\epsilon}{re^{\frac{\tau}{r}\lambda_{\max}(\Re(C))}}\right)^r\\
        &\leq e^{\tau\lambda_{\max}(\Re(C))} \exp\left(\frac{\epsilon}{e^{\frac{\tau}{r}\lambda_{\max}(\Re(C))}}\right)
        \leq e^{\tau\lambda_{\max}(\Re(C))}\exp\left(\sqrt{e}\epsilon\right).
    \end{aligned}
    \end{equation}
    Here, the second inequality follows from the estimate $(1+y)^{1/y}\leq e$ for $y>0$ and the last inequality follows from the observation that
    \begin{equation}
        \abs{\frac{\tau}{r}\lambda_{\max}(\Re(C))}
        \leq\frac{\abs{\lambda_{\max}(\Re(C))}}{2\norm{C}}
        \leq\frac{1}{2}.
    \end{equation}
    The claimed bound is now established by letting $\epsilon\rightarrow0$.
\end{proof}

We thus have that $\norm{e^{\tau C}}$ is bounded for all $\tau>0$ when the numerical abscissa satisfies $\lambda_{\max}(\Re(C))\leq0$, and it decays exponentially with $\tau$ when the strict inequality $\lambda_{\max}(\Re(C))<0$ holds. The bound does not depend on the Jordan condition number, and is thus applicable even when the Jordan basis is ill conditioned. Note however that any matrix $C$ satisfying $\lambda_{\max}(\Re(C))\leq0$ must automatically satisfy $\max_l\Re\left(\lambda_l(C)\right)\leq0$, so the numerical abscissa leads to a stronger
assumption on the input matrix in this respect.

\subsection{Matrix function bound with Crouzeix-Palencia theorem}
\label{append:analysis_faber_crouzeix}

Unlike~\cite{Krovi2023improvedquantum}, our Faber-based eigenvalue algorithms have complexity scaling with the spectral norm of various matrix functions. To derive a similar bound independent of the Jordan condition number, we will need the notion of numerical range which extends the numerical abscissa introduced above. See~\cite{roger1994topics,Wu2023,LiNote,SHAPIRONote} for more detailed discussions about this concept.

Given a square matrix $C$, we define its \emph{numerical range} (also known as the \emph{field of values}) to be the set
\begin{equation}
    \mathcal{W}(C)=\left\{\bra{\psi}C\ket{\psi}\ |\ \norm{\ket{\psi}}=1\right\}.
\end{equation}
It is clear that when $C$ is a normal matrix, $\mathcal{W}(C)$ is the convex hull generated by all its eigenvalues. In general, $\mathcal{W}(C)$ is still a convex and compact set containing all eigenvalues of $C$ following the Toeplitz-Hausdorff theorem, although it is no longer related to the spectra of $C$ in a straightforward way.

As an illustration of this concept, we have that $H$ is a Hermitian matrix if and only if $\mathcal{W}(H)\subseteq\mathbb{R}$ is a closed and bounded real interval. This follows directly from the characterizations of Hermitian matrices in \sec{prelim_matrix}. Also, $N$ is a normal matrix if and only if all its eigenvalues lie on the boundary $\partial\mathcal{W}(N)$ of its numerical range~\cite[Condition 66]{GRONE1987213}. 
In general, when an eigenvalue of a matrix lies on the boundary of its numerical range $\lambda_l(C)\in\mathcal{W}(C)$, then $\lambda_l(C)$ is nondefective and the eigenspace $\mathbf{Ker}\left(C-\lambda_l(C)I\right)$ is perpendicular to the remaining Jordan subspaces~\cite[Theorem 1.6.6]{roger1994topics}.
As yet another example, we point out that the $n$-by-$n$ lower shift matrix has numerical range $\mathcal{W}(L_n)=\cos(\frac{\pi}{n+1})\mathcal{D}$ given by the disk of radius $\cos(\frac{\pi}{n+1})$ centered at the origin~\cite{VANDANJAV201476}. This fact will be used to analyze the block encoding of the matrix Faber generating function in \sec{faber_history}.

The notion of numerical range provides a useful tool for bounding the spectral norm of matrix functions. This is confirmed by the following Crouzeix-Palencia theorem.

\begin{lemma}[Crouzeix-Palencia theorem {\cite{CrouzeixPalencia17,Ransford18}}]
\label{lem:crouzeix_palencia}
    Given a square matrix $C$, for any function $g$ analytic in the interior of $\mathcal{W}(C)$ and continuous up to the boundary $\partial\mathcal{W}(C)$,
    \begin{equation}
        \norm{g(C)}\leq\left(1+\sqrt{2}\right)\norm{g}_{\max,\mathcal{W}(C)}.
    \end{equation}
\end{lemma}

To illustrate the power of this result, let us prove a numerical-abscissa-type bound similar to that of the previous subsection with only a one-line calculation:
\begin{equation}
    \norm{e^{\tau C}}
    \lesssim\norm{e^{\tau(\cdot)}}_{\max,\mathcal{W}(C)}
    =\max e^{\tau\Re(\mathcal{W}(C))}
    =\max e^{\tau\mathcal{W}(\Re(C))}
    =e^{\tau\lambda_{\max}(\Re(C))}.
\end{equation}
Here, we have used the property that the calculation of numerical range ``commutes'' with taking the real part, which is formally stated as:
\begin{proposition}
    For any square matrix $C$,
    \begin{equation}
        \Re(\mathcal{W}(C))=\mathcal{W}(\Re(C)).
    \end{equation}
    Here $\Re$ on the left-hand side takes the real part of a complex number, whereas $\Re$ on the right-hand side returns the Hermitian part of a square matrix.
\end{proposition}
\begin{proof}
    The claimed equality follows from a direct verification
    \begin{equation}
        \Re(\mathcal{W}(C))
        =\left\{\Re(\bra{\psi}C\ket{\psi})\ |\ \norm{\ket{\psi}}=1\right\}
        =\left\{\bra{\psi}\frac{C+C^\dagger}{2}\ket{\psi}\ |\ \norm{\ket{\psi}}=1\right\}
        =\mathcal{W}(\Re(C)).
    \end{equation}
\end{proof}

Therefore, by applying the Crouzeix-Palencia theorem, we can reproduce the numerical abscissa bound from the previous subsection up to a constant prefactor of $1+\sqrt{2}$. The advantage of this approach, however, is that it can be easily generalized to other analytic functions and regions of the complex plane, through which our Faber-based eigenvalue algorithms can be analyzed.

Let us first bound $\norm{\frac{\mathbf{F}_j'\left(\frac{A}{\alpha_A}\right)}{j}}$, which appears in the asymptotic complexity expression of the algorithm for generating Faber history states. To this end, we assume that the numerical range is enclosed by the Faber region as $\mathcal{W}\left(\frac{A}{\alpha_A}\right)\subseteq\mathcal{E}$. Using the Crouzeix-Palencia theorem, we have
\begin{equation}
    \norm{\frac{\mathbf{F}_j'\left(\frac{A}{\alpha_A}\right)}{j}}
    \lesssim\norm{\frac{\mathbf{F}_j'}{j}}_{\max,\mathcal{W}\left(\frac{A}{\alpha_A}\right)}
    =\norm{\frac{\mathbf{F}_j'}{j}}_{\max,\partial\mathcal{W}\left(\frac{A}{\alpha_A}\right)}
    \lesssim\norm{\mathbf{F}_j}_{\max,\partial\mathcal{W}\left(\frac{A}{\alpha_A}\right)}
    \leq\norm{\mathbf{F}_j}_{\max,\mathcal{E}}
    =\mathbf{O}(1),
\end{equation}
where the second inequality follows from the Bernstein's theorem for twice continuously differentiable curves stated below, and the last equality holds as long as boundary of the Faber region has a finite total rotation $\mathbf{V}(\partial\mathcal{E})$. Under all these assumptions on the Faber region (enclosing the numerical range $\mathcal{W}\left(\frac{A}{\alpha_A}\right)$ with a twice continuously differentiable and finite rotation boundary), we have
\begin{equation}
    \alphaFP=\mathbf{O}(1),
\end{equation}
which justifies the claim in the remark succeeding \thm{faber_history}.
\begin{lemma}[Bernstein's theorem for twice continuously differentiable Jordan curves {\cite[Eqs.\ (26) and (27)]{Kalmykov21}}]
    Given a degree-$j$ polynomial $p_j$ and a Jordan curve $\mathcal{C}$ twice continuously differentiable in a neighborhood of $z\in\mathcal{C}$, it holds
    \begin{equation}
        \abs{p_j'(z)}\leq\left(1+\mathbf{o}(1)\right)j2\pi\omega_{\mathcal{C}}(z)\norm{p_j}_{\max,\mathcal{C}},
    \end{equation}
    where $\omega_{\mathcal{C}}(z)$ is the equilibrium density. Thus, if the entire curve $\mathcal{C}$ is twice continuously differentiable,
    \begin{equation}
        \norm{p_j'}_{\max,\mathcal{C}}
        \leq\left(1+\mathbf{o}(1)\right)j2\pi\max_{z\in\mathcal{C}}\omega_{\mathcal{C}}(z)\norm{p_j}_{\max,\mathcal{C}}
        =\mathbf{O}\left(j\norm{p_j}_{\max,\mathcal{C}}\right).
    \end{equation}
\end{lemma}

Next, we consider $\max_{l=0,1,\ldots,n-1}\norm{\sum_{k=l}^{n-1}\beta_k\mathbf{F}_{k-l}\left(\frac{A}{\alpha_A}\right)\ket{\psi}}$, which is used to describe complexity of the Faber eigenvalue transformation algorithm. Assuming that $\mathcal{W}\left(\frac{A}{\alpha_A}\right)\subseteq\mathcal{E}$, we apply the Crouzeix-Palencia theorem to get
\begin{equation}
\begin{aligned}
    \norm{\sum_{k=l}^{n-1}\beta_k\mathbf{F}_{k-l}\left(\frac{A}{\alpha_A}\right)\ket{\psi}}
    &\leq\norm{\sum_{k=l}^{n-1}\beta_k\mathbf{F}_{k-l}\left(\frac{A}{\alpha_A}\right)}
    \lesssim\norm{\sum_{k=l}^{n-1}\beta_k\mathbf{F}_{k-l}}_{\max,\mathcal{W}\left(\frac{A}{\alpha_A}\right)}\\
    &\leq\norm{\sum_{k=l}^{n-1}\beta_k\mathbf{F}_{k-l}}_{\max,\mathcal{E}}
    =\norm{\sum_{k=l}^{n-1}\beta_k\mathbf{F}_{k-l}}_{\max,\partial\mathcal{E}}
    =\norm{\sum_{k=l}^{n-1}\beta_k\mathbf{F}_{k-l}}_{\max,\mathbf{\Psi}(\partial\mathcal{D})}.
\end{aligned}
\end{equation}
Recall that for regions with finite total rotation, we have the integral representation of Faber polynomials
\begin{equation}
    \mathbf{F}_{k}(\mathbf{\Psi}(e^{i\omega}))
    =
    \begin{cases}
        \frac{1}{\pi}\int_{0}^{2\pi}e^{ik\varphi}\mathrm{d}_{\varphi}v(\varphi,\omega),\qquad &k\geq1,\\
        \frac{1}{2\pi}\int_{0}^{2\pi}\mathrm{d}_{\varphi}v(\varphi,\omega)=1,\qquad &k=0,\\
    \end{cases}
\end{equation}
with the angular function $v(\varphi,\omega)=\mathbf{Arg}\left(\mathbf{\Psi}(e^{i\varphi})-\mathbf{\Psi}(e^{i\omega})\right)$.
This means that
\begin{equation}
\begin{aligned}
    \abs{\sum_{k=n-j}^n\beta_k\mathbf{F}_{k+j-n}(\mathbf{\Psi}(e^{i\omega}))}
    &=\abs{\sum_{k=n-j}^n\beta_k
    \frac{1}{\pi}\int_{0}^{2\pi}e^{i(k+j-n)\varphi}\mathrm{d}_{\varphi}v(\varphi,\omega)-\beta_{n-j}}\\
    &\leq\frac{1}{\pi}\int_{0}^{2\pi}\abs{\sum_{k=n-j}^n\beta_k e^{ik\varphi}}\abs{\mathrm{d}_{\varphi}v(\varphi,\omega)}+\abs{\beta_{n-j}}\\
    &\leq\frac{1}{\pi}\int_{0}^{2\pi}
    \abs{\mathbf{S}_{[n-j,n]}(\varphi)}
    \abs{\mathrm{d}_{\varphi}v(\varphi,\omega)}+\norm{p}_{\max,\partial\mathcal{E}}\\
    &\leq\frac{\mathbf{V}(\partial\mathcal{E})}{\pi}\left(\frac{4}{\pi^2}\log(n)+\mathbf{O}(1)\right)\norm{p}_{\max,\partial\mathcal{E}}
    +\norm{p}_{\max,\partial\mathcal{E}},
\end{aligned}
\end{equation}
where the coefficients $\beta_k$ are the Fourier coefficients
\begin{equation}
    \beta_k=\frac{1}{2\pi}\int_{0}^{2\pi}\mathrm{d}v\
    e^{iv}\frac{p(\mathbf{\Psi}(e^{iv}))}{e^{i(k+1)v}}.
\end{equation}
Hence,
\begin{equation}
    \alphaF=\mathbf{O}\left(\log(n)\norm{p}_{\max,\partial\mathcal{E}}\right),
\end{equation}
which establishes the claimed asymptotic complexity in the remark following \thm{qevt_faber}.

Finally, for solving linear differential equations, we have the following requirements on the Faber region $\mathcal{E}$ that encloses numerical range of the input matrix.
\begin{enumerate}
    \item To ensure that $\alphaFP=\mathbf{O}(1)$ in \eq{alphaFP_alphaPsiFaber}, we require that the boundary $\partial\mathcal{E}$ is twice continuously differentiable.
    \item To ensure that $\alphaF=\mathbf{O}\left(\log(n)\norm{p}_{\max,\partial\mathcal{E}}\right)$ in \eq{alphaF_alphaPPsi}, we need $\partial\mathcal{E}$ to have a finite total rotation.
    \item To ensure the validity of \lem{trunc_matrix_exp_faber}, we assume that $\mathcal{E}$ is convex and symmetric with respect to the real axis, lying on the left half of the complex plane.
\end{enumerate}
One possible Faber region satisfying all the above requirements is shown in \fig{region_general2}, which is a smooth deformation of the Elliott semidisk of \fig{region_general}. This is sufficient to prove the asymptotic query complexity of the differential equation solver. Note however that this is not the only option for choosing $\mathcal{E}$. One may instead numerically construct conformal maps that are easier to compute leading to better constant prefactors in the gate complexity. We leave a complete study of such optimizations as a subject for future work.

\subsection{Matrix function bound based on pseudospectrum}
\label{append:analysis_faber_pseudospectrum}
Using the notion of numerical range and Crouzeix-Palencia theorem, we have shown that complexity of the Faber-based method can be independent of the condition number of the basis transformation. This yields an efficient QEVT algorithm even when the input matrix has an ill-conditioned Jordan basis.
We now provide an alternative analysis of the Faber-based algorithms based on the concept of \emph{pseudospectrum}, which also applies to matrices with ill-conditioned Jordan basis. We begin by reviewing the definition of pseudospectrum and its properties most relevant to our paper, referring the reader to~\cite{Trefethen05,PseudospectraNotes} for further discussions of this topic.

For a given square matrix $C$ and $\delta>0$, we define the \emph{$\delta$-pseudospectrum} to be the set of complex numbers satisfying any one of the following equivalent conditions:
\begin{equation*}
\begin{aligned}
    \mathcal{S}_\delta(C)&=\left\{z\in\mathbb{C},\ \norm{C\ket{\psi}-z\ket{\psi}}<\delta\text{ for some }\norm{\ket{\psi}}=1\right\}\\
    &=\left\{z\in\mathbb{C},\ z\text{ is an eigenvalue of }\widetilde C\text{ for some }\norm{\widetilde C-C}<\delta\right\}\\
    &=\left\{z\in\mathbb{C},\ \norm{\frac{1}{z-C}}>\frac{1}{\delta}\right\}\\
    &=\left\{z\in\mathbb{C},\ \sigma_{\min}(C-z)<\delta\right\}.
\end{aligned}
\end{equation*}
Thus, pseudospectrum relaxes the definition of eigenvalues $C\ket{\psi}=\lambda\ket{\psi}$ or $\sigma_{\min}(C-\lambda)=0$ to approximately hold as $\norm{C\ket{\psi}-z\ket{\psi}}<\delta$ or $\sigma_{\min}(C-z)<\delta$, whereas the usual spectrum is recovered from the infinite intersection $\cap_{\delta>0}\mathcal{S}_\delta(C)$. This is an open and bounded set in the complex plane that contains all $\delta$-balls around eigenvalues of $C$:
\begin{equation*}
    \mathcal{S}_\delta(C)\supseteq\left\{\lambda_j(C)\right\}+\mathcal{D}(0,\delta),
\end{equation*}
with equality if and only if $C$ is a normal matrix. On the other hand, it is upper bounded by the numerical range up to a perturbation of $\delta$:
\begin{equation*}
    \mathcal{S}_\delta(C)\subseteq\mathcal{W}(C)+\mathcal{D}(0,\delta).
\end{equation*}

By definition, we have that $z\in\mathcal{S}_\delta(C)$ if and only if $\norm{\frac{1}{z-C}}>\frac{1}{\delta}$, and $z\notin\mathcal{S}_\delta(C)$ if and only if $\norm{\frac{1}{z-C}}\leq\frac{1}{\delta}$. We then use a continuity argument to conclude that $\norm{\frac{1}{z-C}}=\frac{1}{\delta}$ on the boundary $z\in\partial\mathcal{S}_\delta(C)$. More generally, letting $\mathcal{C}$ be a contour enclosing the pseudospectrum $\mathcal{S}_\delta(C)$ that are all contained in the region of analyticity of $f$, we have from the Cauchy's integral formula
\begin{equation*}
    f(C)=\frac{1}{2\pi i}\int_{\mathcal{C}}\mathrm{d}z\ f(z)(zI-C)^{-1}
\end{equation*}
that
\begin{equation*}
    \norm{f(C)}\leq\frac{\mathbf{L}_\delta(\mathcal{C})}{2\pi\delta}\norm{f}_{\max,\mathcal{C}}
\end{equation*}
for $\mathbf{L}_\delta(\mathcal{C})=\int_{\mathcal{C}}\abs{\mathrm{d}z}$ arc length of the contour $\mathcal{C}$.
This pseudospectrum-type bound for matrix functions can be compared to that of \lem{crouzeix_palencia} based on the notion of numerical range.

We now analyze the performance of Faber-based eigenvalue transformation algorithms using pseudospectrum. We first bound $\norm{\frac{\mathbf{F}_j'\left(\frac{A}{\alpha_A}\right)}{j}}$, which shows up in the asymptotic complexity of creating the Faber history state. Similar to the assumption on numerical range, we assume that the pseudospectrum is encosed by the Faber region as $\mathcal{S}_\delta\left(\frac{A}{\alpha_A}\right)\subseteq\mathcal{E}$. We have
\begin{equation*}
    \norm{\frac{\mathbf{F}_j'\left(\frac{A}{\alpha_A}\right)}{j}}
    \lesssim\norm{\frac{\mathbf{F}_j'}{j}}_{\max,\partial\mathcal{E}}
    \lesssim\norm{\mathbf{F}_j}_{\max,\partial\mathcal{E}}
    =\norm{\mathbf{F}_j}_{\max,\mathcal{E}}
    =\mathbf{O}(1),
\end{equation*}
where in the second asymptotic estimate we have used Bernstein's theorem assuming the boundary $\mathcal{E}$ is sufficiently smooth, and in the last asymptotic estimate we have assumed that boundary of the Faber region has a finite total rotation. This means that
\begin{equation*}
    \alphaFP=\mathbf{O}(1),
\end{equation*}
which justifies the claim in the remark succeeding \thm{faber_history}.

Next, we consider $\max_{l=0,1,\ldots,n-1}\norm{\sum_{k=l}^{n-1}\beta_k\mathbf{F}_{k-l}\left(\frac{A}{\alpha_A}\right)\ket{\psi}}$, which appears in the asymptotic complexity of the Faber eigenvalue transformation algorithm. Assuming that $\mathcal{S}_\delta\left(\frac{A}{\alpha_A}\right)\subseteq\mathcal{E}$, we obtain
\begin{equation*}
    \norm{\sum_{k=l}^{n-1}\beta_k\mathbf{F}_{k-l}\left(\frac{A}{\alpha_A}\right)\ket{\psi}}
    \leq\norm{\sum_{k=l}^{n-1}\beta_k\mathbf{F}_{k-l}\left(\frac{A}{\alpha_A}\right)}
    \lesssim\norm{\sum_{k=l}^{n-1}\beta_k\mathbf{F}_{k-l}}_{\max,\partial\mathcal{E}}
    =\norm{\sum_{k=l}^{n-1}\beta_k\mathbf{F}_{k-l}}_{\max,\mathbf{\Psi}(\partial\mathcal{D})}.
\end{equation*}
Proceeding in a similar way as in the numerical range case, we have
\begin{equation*}
    \alphaF=\mathbf{O}\left(\log(n)\norm{p}_{\max,\partial\mathcal{E}}\right),
\end{equation*}
which establishes the claimed complexity in the remark following \thm{qevt_faber}.

\subsection{Average-case analysis with Carleson-Hunt theorem}
\label{append:analysis_faber_carleson}

In the previous subsection, we have analyzed asymptotic scaling of the Faber-based eigenvalue algorithms based on either the numerical range or the pseudospectrum. However, there is a $\sim\log(n)$ factor in the bound of $\alphaF$. Similar to the Chebyshev case, this logarithmic factor can be removed for an average diagonalizable input matrix with the help of Carleson-Hunt theorem (by trading in a scaling with the Jordan condition number), although the reasoning is somewhat different which we present below.

For any $z$ in the interior of the Faber region $\mathcal{E}$, we apply the Cauchy integral theorem to the Faber generating function and get
\begin{equation}
\begin{aligned}
    \abs{\sum_{k=n-j}^n\beta_k\mathbf{F}_{k+j-n}(z)}
    &=\abs{\sum_{k=n-j}^n\beta_k
    \frac{1}{2\pi i}\int_{\partial\mathcal{D}}\mathrm{d}\xi\ \frac{\xi^{k+j-n}\mathbf{\Psi}'(\xi)}{\mathbf{\Psi}(\xi)-z}}\\
    &\leq\frac{1}{2\pi}\int_{\partial\mathcal{D}}\mathrm{d}\xi\ \abs{\sum_{k=n-j}^n\beta_k\xi^{k}
    \frac{\mathbf{\Psi}'(\xi)}{\mathbf{\Psi}(\xi)-z}}.
\end{aligned}
\end{equation}
The result can be further upper bounded by a linear combination of partial sums of the form $\frac{1}{2\pi}\int_{\partial\mathcal{D}}\mathrm{d}\xi \abs{\sum_{k=0}^n\beta_k\xi^{k}
    \frac{\mathbf{\Psi}'(\xi)}{\mathbf{\Psi}(\xi)-z}}$.
By definition of the contour integral,
\begin{equation}
\begin{aligned}
    \frac{1}{2\pi}\int_{\partial\mathcal{D}}\mathrm{d}\xi\ \abs{\sum_{k=0}^n\beta_k\xi^{k}
    \frac{\mathbf{\Psi}'(\xi)}{\mathbf{\Psi}(\xi)-z}}
    &=\frac{1}{2\pi}\int_{0}^{2\pi}\mathrm{d}u\
    \abs{\sum_{k=0}^n\beta_ke^{iku}
    \frac{\mathbf{\Psi}'(e^{iu})}{\mathbf{\Psi}(e^{iu})-z}},
\end{aligned}
\end{equation}
where the coefficients $\beta_k$ are the Fourier coefficients
\begin{equation}
    \beta_k=\frac{1}{2\pi}\int_{0}^{2\pi}\mathrm{d}v\
    e^{iv}\frac{p(\mathbf{\Psi}(e^{iv}))}{e^{i(k+1)v}}.
\end{equation}
Therefore, $\sum_{k=0}^n\beta_ke^{iku}$ is exactly the one-sided Fourier partial sum for the function $p(\mathbf{\Psi}(e^{iu}))$ or, equivalently, the two-sided Fourier partial sum for the Hilbert transform of $p(\mathbf{\Psi}(e^{iu}))$. Therefore, we have
\begin{equation}
    \frac{1}{2\pi}\int_{0}^{2\pi}\mathrm{d}u\
    \abs{\sum_{k=0}^n\beta_ke^{iku}
    \frac{\mathbf{\Psi}'(e^{iu})}{\mathbf{\Psi}(e^{iu})-z}}
    \leq\frac{1}{2\pi}\int_{0}^{2\pi}\mathrm{d}u\
    \abs{\mathbf{S}_n\widetilde{\mathbf{H}}\left(p(\mathbf{\Psi}(e^{iu}))\right)}
    \abs{\frac{\mathbf{\Psi}'(e^{iu})}{\mathbf{\Psi}(e^{iu})-z}}.
\end{equation}

Here, the second factor can be further bounded as
\begin{equation}
    \abs{\frac{\mathbf{\Psi}'(e^{iu})}{\mathbf{\Psi}(e^{iu})-z}}
    \leq\frac{\norm{\mathbf{\Psi}'}_{\max,\partial\mathcal{D}}}{\mathbf{Dist}(z,\partial\mathcal{E})},
\end{equation}
whereas the first factor is handled by the Carleson-Hunt theorem as
\begin{equation}
\begin{aligned}
    \int_{0}^{2\pi}\mathrm{d}u\
    \abs{\mathbf{S}_n\widetilde{\mathbf{H}}\left(p(\mathbf{\Psi}(e^{iu}))\right)}
    &\lesssim\sqrt{\int_{0}^{2\pi}\mathrm{d}u\
    \abs{\mathbf{S}_n\widetilde{\mathbf{H}}\left(p(\mathbf{\Psi}(e^{iu}))\right)}^2}\\
    &\lesssim \norm{p\left(\mathbf{\Psi}\left( e^{i(\cdot)}\right)\right)}_{\max,[0,2\pi]}
    =\norm{p}_{\max,\partial\mathcal{E}}.
\end{aligned}
\end{equation}
Altogether,
\begin{equation}
    \abs{\sum_{k=n-j}^n\beta_k\mathbf{F}_{k+j-n}(z)}
    \lesssim\frac{\norm{\mathbf{\Psi}'}_{\max,\partial\mathcal{D}}}{\mathbf{Dist}(z,\partial\mathcal{E})}
    \norm{p}_{\max,\partial\mathcal{E}}.
\end{equation}

Now, suppose that $z$ is randomly distributed in $\mathcal{E}$ with probability density function $q(z)$ such that $\int_{\mathcal{E}}\mathrm{d}z\ q(z)=1$. Then,
\begin{equation}
    \int_{\mathcal{E}}\mathrm{d}z\ q(z)\abs{\sum_{k=n-j}^n\beta_k\mathbf{F}_{k+j-n}(z)}^2
    \lesssim\norm{\mathbf{\Psi}'}_{\max,\partial\mathcal{D}}^2\norm{p}_{\max,\partial\mathcal{E}}^2
    \int_{\mathcal{E}}\mathrm{d}z\ \frac{q(z)}{\mathbf{Dist}^2(z,\partial\mathcal{E})}.
\end{equation}
We assume that $\int_{\mathcal{E}}\mathrm{d}z\ \frac{q(z)}{\mathbf{Dist}^2(z,\partial\mathcal{E})}$ is finite, which happens if eigenvalues of the input matrix have less chance to appear near boundary $\partial\mathcal{E}$ of the Faber region. In any case, this is a constant depending only on properties of the Faber region.
Proceeding as in the Chebyshev case, we obtain
\begin{equation}
    \alphaF=\mathbf{O}\left(\norm{S}\norm{S^{-1}\ket{\psi}}\norm{p}_{\max,\partial\mathcal{E}}\right)
\end{equation}
for diagonalizable input matrices.
This justifies the claimed scaling in the remark succeeding \thm{qevt_faber}.

We end this section with a brief comment on the analysis of quantum eigenvalue algorithms. We have bounded query complexity of the Chebyshev-based algorithms in \append{analysis_cheby} using the Jordan condition number as well as size of the largest Jordan block, whereas we have handled the Faber-based algorithms in \append{analysis_faber} mostly using the notion of numerical range and pseudospectrum. See \tab{nonnormal} for a summary of results. This is however not the only way in which our algorithms can be analyzed. For instance, our above analysis in \append{analysis_faber_carleson} works for diagonalizable input matrices and is essentially a generalization of the Jordan condition number bound. The reverse direction however does not lead to interesting new results. This is because, in order to apply the Crouzeix-Palencia theorem, we would need to enclose numerical range of the input matrix inside of a real interval (the pseudospectrum is an open set and can never be enclosed by a real interval). But per the characterizations of Hermitian matrices in \sec{prelim_matrix}, this only happens if the input matrix is Hermitian, which is already handled by the existing QSVT algorithm.
In any case, it would be fruitful to prove tighter bounds on the matrix functions used by our algorithms, as well as to construct alternative methods for quantum eigenvalue processing that scale with different matrix functions.

\clearpage
\bibliographystyle{myhamsplain2}
\bibliography{QEVP.bib}

\end{document}